\documentclass[a4paper,twoside,11pt]{report}
\usepackage[margin=1in]{geometry}
\usepackage[toc,page]{appendix}
\usepackage[utf8]{inputenc}
\usepackage[english]{babel}
\usepackage{xcolor}
\usepackage{url}
\usepackage{hyperref}
\usepackage{tikz}
\usepackage{multirow}
\usepackage{tcolorbox}
\usepackage{mathtools, blkarray, bigstrut, bigdelim}
\usepackage[braket,qm]{qcircuit}
\usepackage{soul} 
\usepackage{enumerate}
\usepackage{pdfpages}
\DeclarePairedDelimiter{\ceil}{\lceil}{\rceil}
\usepackage{graphicx,subcaption,caption, placeins} 
\usepackage{array, makecell}

\usepackage{amsthm, amsfonts, amssymb, bbold}
\usepackage{bm}
\graphicspath{{images/}}

\newtheorem{theorem}{Theorem}[section]

\newtheorem{property}[theorem]{Property}

\newtheorem{definition}{Definition}[section]
\theoremstyle{remark}

\newcommand{\ketbra}[2]{{|#1 \rangle \langle #2|}}

\title{The resource cost of large scale quantum computing}
\author{Marco Fellous Asiani}
\begin{document}
\includepdf[pages=-]{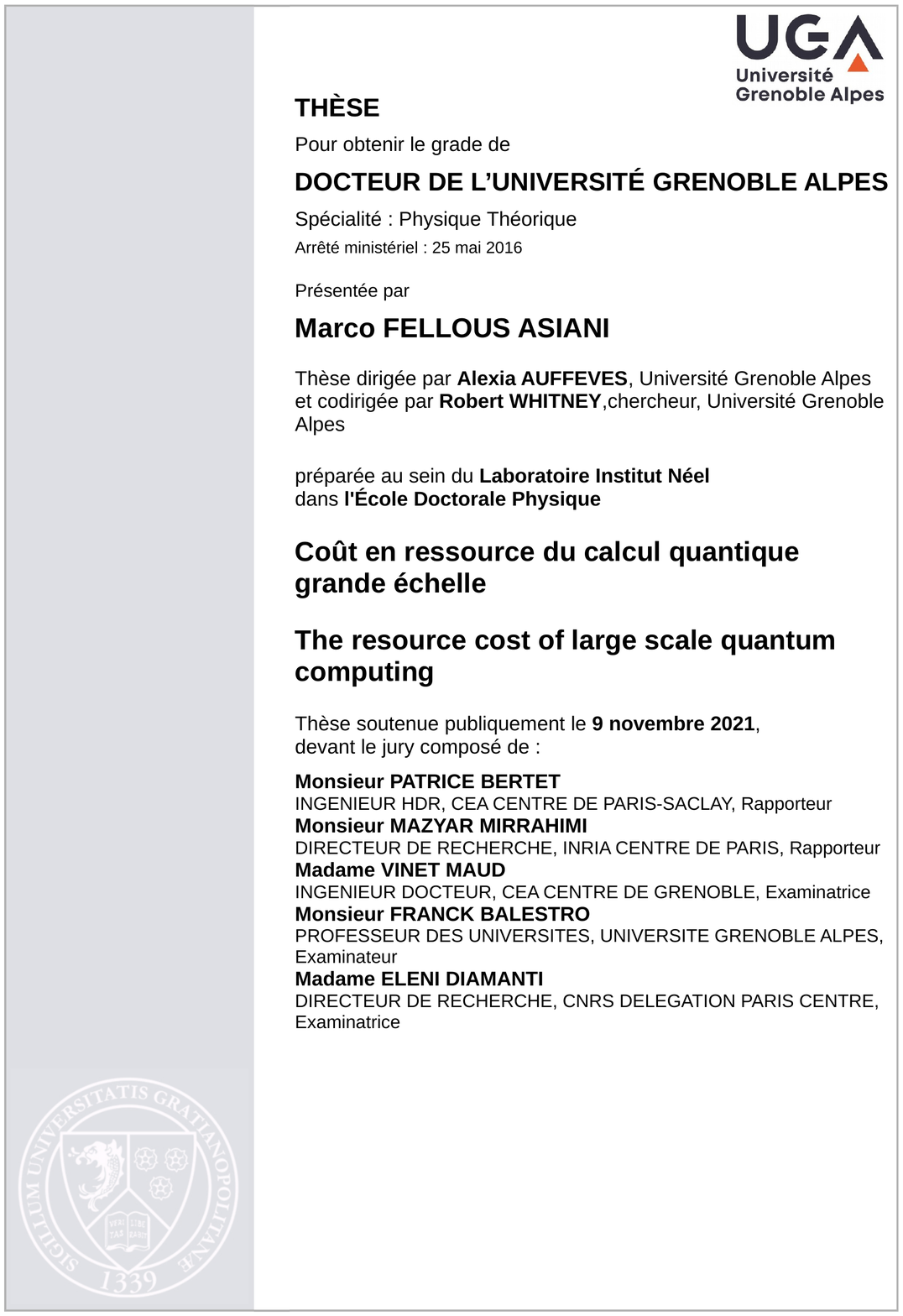}
\maketitle
\section*{Remerciements}
Je tiens tout d'abord à remercier ma directrice de thèse, Alexia Auffèves, ainsi que mon co-directeur de thèse, Robert Whitney pour m'avoir proposé ce sujet de thèse, initialement intitulé "thermodynamique de l'information quantique" mais que j'ai, grâce à eux, pu orienter d'avantage vers les aspects de scalabilité du calcul, qui m'intéressent tout particulièrement. Je les remercie donc pour la liberté qui fut possible vis à vis des orientations choisies ainsi que pour leur précieuse aide tout au long de ma thèse. I would like to thank Hui Khoon Ng as well for all the help she provided in the aspects of fault-tolerance which are central in this PhD, and without who this PhD wouldn't have been possible. Je remercie aussi Cyril Branciard qui fut impliqué sur le projet de l'étude énergétique du quantum switch, non détaillé dans ce manuscrit, mais qui est un projet "annexe" sur lequel j'ai pris beaucoup de plaisir à travailler, en partie grâce à Cyril. De manière générale, merci à l'ensemble de mes encadrants pour leurs conseils et guidance tout au long du déroulement de cette thèse.

Je tiens aussi à remercier tous les protagonistes impliqués dans l'ambiance du département PLUM ainsi que du feu bureau F218 désormais remplacé par le F323, qui aura eu sa période de gloire lors de l'ère pré-covid. Je pense évidemment à Alexandra, Julian, Hippolyte, Augustin, Pierre, Juliette, Smaïl, Noé, Patrice (qui s'est découvert une passion pour la raclette depuis son arrivée en France), ainsi qu'aux membres provenant d'autres départements tels que Guilliam (étant donné ses hautes responsabilités administratives au conseil de laboratoire, j'étais forcé de le mettre dans les remerciements, sous peine de risquer des coupures d'électricité et d'internet dans mon bureau), voir n'étant même pas à l'institut Néel (Marie et Juan vous saurez vous reconnaître). Je garde cependant un souvenir mitigé des parties de civilization V durant lesquelles moultes alliances à mon encontre furent créées (souvent sous l'impulsion d'Alexandra et Marie). 

Je remercie également les différents membres du groupe, tels que Maria, Léa, ainsi que les derniers arrivés que j'ai moins eu le temps de connaitre en raison du covid, Bruno, Nicolò, Stephen et Irénée. Merci aussi à Raphaël qui va bientôt commencer sa thèse à la suite d'Hippolyte. Je pense en particulier aux discussions passionnantes que nous avons eu, à savoir si il faut calculer les $\chi$-matrix à l'ordre 2 ou si bah en fait l'ordre un, franchement ça suffit large. I would also like to give a special thanks to Jing Hao (also a new member of the group), who suffered from my many questions on fault-tolerance when I was learning it. In addition, I take the opportunity to wish good luck to Samyak which will start his PhD this year and that showed us his amazing badminton skills the last time he was there. In return, we showed you the cute little animals we have in the mountains in Grenoble such as ticks that gave us nice "kisses" during our hiking session. But please, don't thank me back by showing me cute tigers and snakes if I come to India one day, I prefer those in pictures. Thank you also to Alessandro who is doing an exchange in Grenoble and trying to convince everyone that the Mont Blanc is an Italian Mountain. Fortunately I arrived on time to stop this Italian propaganda ! Enfin, je remercie Yvain Thonnart, Matias Urdampilleta, Tristan Meunier, Benjamin Huard, Olivier Buisson et Vladimir Milchakov qui ont su être disponibles à plusieurs reprises lors de mes questions sur l'électronique classique pour le premier,certains aspects de physique des micro-ondes et des qubits de spins pour les seconds et troisième, et diverses petites questions liées aux réalités expérimentales des qubits supraconducteurs pour les trois derniers. Sans leur aide, cette thèse n'aurait pas pu se dérouler de la même manière. Je remercie aussi mes amis à Paris ainsi que mes parents et ma famille, et bien entendu Laurie qui a été d'un grand soutien tout au long de cette thèse. 

To finish those acknowledgements, I think that the two great physicists that are Alice and Bob should be acknowledged for all their contributions in the field of physics. From their groundbreaking discovery in quantum causality, in quantum gravity and their courage to try to do very dangerous experiment, sometimes risking their life by going inside blackholes \cite{fuentes2005alice}. I believe that I speak for all the community of physicists to tell them: thank you for all the risk you took for the pure goal of improving humanity's knowledge.

Merci aux membres du jury d'avoir accepté d'en faire partie.
\section*{Short summary}
This thesis deals with the problematics of the scalability of fault-tolerant quantum computing. This question is studied under the angle of estimating the resources needed to set up such computers. Now that the first prototypes of quantum computers exist, it is time to start making such estimates. What we call a resource is, in principle, very general; it could be the power, the energy, or even the total bandwidth allocated to the different qubits. However, we focus mainly on the energetic cost of quantum computing within this thesis, although most of the approaches used can be adapted to deal with any other resource.

We first study what is the maximum accuracy a fault-tolerant quantum computer can achieve in the presence of a scale-dependent noise, i.e., a noise that increases with the number of qubits and physical gates present in the computer. Indeed, this regime may violate an assumption behind the central theorem of fault-tolerance: the quantum threshold theorem. This theorem states that the accuracy of algorithms implemented on a quantum computer can be arbitrarily high if they are protected by quantum error correction, if enough physical elements (qubits and gates) are available and if the noise strength is below a certain threshold. Since this last assumption must be satisfied regardless of the number of physical elements in the computer, scale-dependent noise can violate it. In the case where this scale-dependent noise can be expressed as a function of a resource, these estimates allow (i) to estimate the maximum precision that the computer can achieve in the presence of a fixed quantity of this resource (which makes possible to deduce the maximum size of the algorithms that the computer will be able to implement, in order to know if the scale-dependent noise is a real problem) and vice versa (ii) to estimate the minimum quantity of resource allowing to reach a given accuracy. Throughout this thesis, our calculations are based on the concatenated Steane error-correcting code (because it is a theoretically well-documented construction that protects the qubits against an arbitrary error and allows us to make analytical calculations).

In a second study, we generalize these approaches in order to be able to estimate the resource cost of a calculation in the most general case. By asking to find the minimum amount of resources required to perform a computation \textit{under the constraint} that the algorithm provides a correct answer with a targeted probability, it is possible to optimize the entire architecture of the computer to minimize the resources spent while being sure to have a correct answer with a high probability. We apply this approach to a complete model of fault-tolerant quantum computer based on superconducting qubits. Our results indicate that for algorithms implemented on thousands of logical qubits, our method makes possible to reduce the energetic cost by a factor of 100 in regimes where, without optimizing, the power consumption could exceed the gigawatt. This work illustrates the fact that the energetic cost of quantum computing should be a criterion in itself, allowing to evaluate the scaling potential of a given quantum computer technology. It also illustrates that optimizing the architecture of a quantum computer via inter-disciplinary methods, including algorithmic considerations, quantum physics, and engineering aspects, such as the ones that we propose, can prove to be a powerful tool, clearly improving the scaling potential of quantum computers. Finally, we provide general hints about how to make fault-tolerant quantum computers energy efficient.
\newpage

\newpage
\section*{Court résumé}
Cette thèse traite des questions de mise à l'échelle du calcul quantique tolérant aux fautes. Ces questions sont étudiées sous l'angle de l'estimation des ressources nécessaires à la mise en place de tels ordinateurs: maintenant que les premiers prototypes d'ordinateurs quantiques existent, il est temps de commencer à réaliser de telles estimations. Ce que nous appelons ressource est en principe très général, il pourrait s'agir de la puissance, de l'énergie, ou même de la bande passante totale allouée aux différents qubits. Cependant, nous nous focalisons particulièrement sur le coût énergétique du calcul quantique au sein de cette thèse, bien que la plupart des approches utilisées puissent être adapté pour traiter une quelconque autre ressource. 

Nous étudions dans un premier temps quelle est la précision maximale qu'un ordinateur quantique tolérant aux fautes peut atteindre en présence d'un bruit dépendant de l'échelle, c'est à dire un bruit qui augmente avec le nombre de qubits ou de portes physiques présents dans l'ordinateur. En effet, ce régime peut violer une hypothèse derrière le théorème central de la tolérance aux fautes: le théorème du seuil (« quantum threshold theorem ») qui stipule que la précision des algorithmes implémentés sur un ordinateur quantique peut être arbitrairement grande si ils sont protégés par de la correction d'erreur quantique, si suffisamment de qubits et portes physiques sont à disposition, et si le taux de bruit est en dessous d'un certain seuil. Cette dernière hypothèse devant être vérifiée peu importe le nombre d'éléments physiques dans l'ordinateur, un bruit dépendant de l'échelle peut la violer. Dans le cas où ce bruit dépendant de l'échelle peut être exprimé en fonction d'une ressource, ces estimations permettent (i) d'estimer la précision maximale que l'ordinateur peut atteindre en présence d'une quantité fixée de cette ressource (ce qui permet de déduire la taille maximale des algorithmes que l'ordinateur pourra implémenter, afin de savoir si le bruit dépendant de l'échelle est un réel problème) et réciproquement (ii) d'estimer la quantité de ressource minimale permettant d'atteindre une précision donnée. Dans toute cette thèse, nos calculs sont basés sur le code correcteur d'erreur Steane concaténé (car c'est une construction bien documentée théoriquement, permettant de protéger les qubits contre une erreur arbitraire et permettant de faire des calculs analytiques).

Dans un second temps, nous généralisons ces approches afin de pouvoir estimer le coût en ressource d'un calcul dans le cas le plus général. En demandant de trouver la quantité de ressource minimale requise pour effectuer un calcul \textit{sous la contrainte} que l'algorithme fournisse une réponse correcte avec une probabilité ciblée, il est possible d'optimiser l'intégralité de l'architecture de l'ordinateur permettant de minimiser la dépense en ressource tout en ayant une réponse correcte. Nous appliquons cette démarche à un modèle complet d'ordinateur tolérant aux fautes basé sur des qubits supraconducteurs. Nos résultats indiquent que pour des algorithmes implémentés sur plusieurs milliers de qubits logiques, notre méthode permet de réduire la facture énergétique d'un facteur 100, dans des régimes où sans optimisation la consommation en puissance pourrait dépasser le gigawatt. Ce travail illustre le fait que le coût énergétique du calcul quantique devrait être un critère en soit permettant d'évaluer le potentiel de mise à l'échelle des ordinateurs quantiques. Il illustre aussi que l'optimisation de l'architecture d'un ordinateur quantique, via des méthodes transversales, incluant les aspects algorithmiques, de physique quantique, et d'ingénierie, telles que celles que nous proposons, peut se révéler être un outil puissant permettant d'améliorer grandement le potentiel de mise à l'échelle. Enfin, nous donnons des premières pistes permettant de savoir comment réaliser des ordinateurs quantiques économes en énergie.
\tableofcontents
\chapter*{Introduction}
\section*{General context of this thesis}
In recent years, quantum technologies have grown in interest. This is part of the context of the second quantum revolution, which follows the first one that occurred in the last century. When physicists understood that classical physics was only an approximation of more fundamental quantum laws describing our world: quantum mechanics, they used this knowledge in order to design new technologies such as computers, lasers, solar cells, allowing us to enter in the information age. While quantum mechanics was necessary to design those devices, physicists did not really "engineer" the quantum effects to create the technology. Instead, they took those effects the way they already existed in nature and directly tried to create devices with those. 

The goal of the second quantum revolution is to go beyond this limitation and to design new devices whose working principle would strongly rely on the engineering of all the fundamental quantum effects, such as superposition and entanglement \cite{dowling2003quantum,jaeger2018second}. It relies on having the ability to manipulate individual quantum particles. This is becoming possible because in the last $20$ to $30$ years, experiments showing the possibility to control individual quantum particles with a high level of control have been shown experimentally \cite{rauschenbeutel1999coherent,nakamura2001rabi,haroche2006exploring}. The second quantum revolution thus opens possibilities of technological innovations in a wide range of domains. We can think about quantum sensing \cite{degen2017quantum, pirandola2018advances}, quantum cryptography \cite{pirandola2020advances,shenoy2017quantum}, quantum communication \cite{chen2021review,gisin2007quantum}, and of course, what is at the heart of this PhD, quantum computing. 

The ultimate goal of quantum computing is to build large-scale quantum computers as they would have a computational power order of magnitudes bigger than today's best supercomputers for some computational tasks. Because of their very high computational power, quantum computers could impact the national and industrial sovereignty \cite{moerel2021reflections}. One example clearly showing it is the Shor factoring algorithm which allows factorizing an integer into prime numbers exponentially faster than the best-known classical algorithm. It would allow breaking RSA encryption protocol that is currently used in the banking and military communications\footnote{The example of Shor is a clear example of the potential exponential speedup in the computing time that quantum computers would have compared to classical supercomputers. We should, however, emphasize that there exist post-quantum cryptography methods \cite{bernstein2017post} which would allow transferring data in a secured manner even against the threat of quantum computers. Also, while quantum computing could be a threat to the standard encryption technique, quantum cryptography would at the same time provide a much more secure way to transmit data such that even quantum computers would not be able to decrypt \cite{pirandola2020advances}.}. A recent estimation considered that with 20 million noisy qubits, it could do this in about eight hours \cite{gidney2021factor} where in comparison, a classical supercomputer could take at least millions of years to do it. Quantum computers already exist today. As an example of that, we can think about the Sycamore quantum processor, which has been associated with a recent claim of quantum supremacy by Google \cite{arute2019quantum}. Other laboratories or companies also have quantum computers, we can think about IBM \cite{tornow2020non,garcia2017five}, Intel \cite{intelquantumcomputing}, Rigetti \cite{motta2020determining} to give a few examples. Actually, the first experimental implementation of a quantum algorithm goes back to the year 1998, when the Deutsch-Jozsa algorithm was implemented on a two qubits quantum computer \cite{jones1998implementation}. However, all those examples are computers of very modest size, containing less than a hundred noisy physical qubits. In order to reach all the potential of the quantum speedup, there is a need to scale up the devices to make large-scale quantum computers.

Many challenges have to be solved to make large-scale quantum computers a reality. One of the main issues is related to the fact that quantum information is very fragile. The qubits are very sensitive to environmental noise, which can modify their state. If it occurs, the answer provided by the algorithm would not be trustable. The intrinsic fragility of quantum information has been seen for some time as a real threat, making quantum computers seen as an unreachable goal \cite{raimond1996quantum}. To face this issue, strategies such as quantum error correction and fault-tolerance have been developed. To explain the difference between those concepts briefly, quantum error correction consists in developing algorithms that, if implemented \textit{perfectly} (i.e., if the gates applied on the qubits during the error correction are ideal) can detect and correct errors occurring on the qubits during the calculation. Fault-tolerance takes into account the fact that error correction is unfortunately a noisy process, and provides explicit circuits allowing to implement error correction in such a way that despite the fact it is a noisy process, because correction is performed, the net effect is an improvement in the protection of the encoded information. Thanks to those protocols, the quantum information can, in principle, be preserved for a long enough time in order to get a trustable answer at the end of the algorithm. A central theorem guaranteeing this is the quantum threshold theorem \cite{aliferis2005quantum,aliferis2007accuracy}. It states that if the noise per physical gate is lower than a fixed "threshold" value, one can regroup multiple physical qubits and gates to create a logical qubit and a logical gate. On a logical level, those elements would do the same operation as their physical counterparts, excepted that because error correction is being performed, everything would be as if the qubits were less noisy. Using enough of those physical elements in a logical element, under the hypothesis that the noise is below a threshold, the noise of qubits and gates on the logical level can be reduced as much as desired. The quantum threshold theorem is a significant result. It guarantees the experimentalists that once they would be able to create good enough physical qubits and implement physical gates in a good enough manner, no further improvement would be necessary from the hardware perspective. Quantum error correction would make the necessary improvements to allow the computer to get to any targetted accuracy. For this reason, this theorem has been a big motivation behind the development of quantum technologies for quantum computing. 

However, there are two potential issues with the quantum threshold theorem. The first one is that, because fault-tolerance and error correction require many additional physical qubits and gates to be implemented, we might expect to have a significant overhead in terms of physical components required. The estimations of their number vary, but between hundreds of thousands to potentially billions of qubits might be required to implement a fault-tolerant algorithm showing a clear quantum supremacy \cite{chamberland2017overhead,suchara2013comparing,gidney2021factor}. This naturally raises the general question of the resources (energy or power, for instance) that would be required to build such large-scale quantum computers. Given the number of qubits needed, we could naturally expect that the energetic cost of such computers might be high. An even better question to answer would be to see if it is possible to design robust methods allowing to \textit{minimize} the resources expenses of such computers. The second issue with the quantum threshold theorem is that the physical components must already be of good quality in order to make error correction useful: the noise must be below the threshold, and it must \textit{remain} below this value even if the computer is composed of the potentially billions of physical qubits required for quantum error correction. This can be very challenging as having a well-controlled environment is much easier when the number of qubits is low than when many qubits are inside the computer. The work done in this Ph.D. focuses on those problems.

\section*{Positioning of the present Ph.D. work}
In this Ph.D., we are studying the question of the scalability of quantum computing by estimating how much resources it would require. In principle, what we call a resource is very general: it could be energy, power, the total frequency bandwidth allowed for the qubits, etc. Now that quantum computing is becoming a reality, and because we have access to the characteristics of the first noisy prototypes of quantum computers, such estimations can, and should be done in order to design the next generation of quantum computers which are expected to be fault-tolerant.

Our goal in this thesis is first to design a general approach that could be used to estimate the largest algorithms that could be implemented with quantum error correction, assuming a fixed amount of resources available. It is also, reciprocally, to evaluate the minimum amount of resources required to implement a given algorithm. 

While our approach is general, we then focus on the energetic cost of large-scale quantum computing. There are two reasons for that. First, the energy (or power) required to run a quantum computer is a "good variable" to know if an architecture is scalable. This is because energetic quantities encompass many criteria coming from various fields. For instance, a low energetic cost is likely to be associated with a "reasonable" design from an engineering perspective. The other reason we focus on energy and power is that the energy cost required to create large-scale quantum computers could be considerable. It is essential "in itself" to know how much energy would be needed to implement concrete algorithms.

For this reason, our ending goal is to try to make an in-depth estimation of the power that would be required to implement algorithms on a superconducting qubit fault-tolerant quantum computer. The reason why we focus on superconducting qubits is because it is one of the most mature technologies used for quantum computers today (many experimental values are accessible). Such qubits must be maintained at very low temperatures, which might require a large amount of cryogenic power. We consider using the Steane quantum error correction code (it is a code that protects qubits against arbitrary single-qubit errors). We implement it fault-tolerantly with the so-called "concatenated construction". The concatenated construction is a way to implement the code allowing to reach an arbitrarily high level of protection, assuming the noise is below the threshold initially (it is a construction in which the quantum threshold theorem can be rigorously derived). The reason why we choose such code and fault-tolerant construction is because it is very well documented theoretically, we have access to the concrete circuits allowing us to implement it, and it will enable us to perform analytic calculations, which is a requirement to keep our approach simple to explain.

Finally, the overall philosophy behind our work is to design \textit{inter-disciplinary} methods allowing us to estimate the amount of resources required by involving characteristics coming from quantum error correction, quantum algorithm, engineering, cryogenics, physics of the quantum gate, etc. A quantum computer being a multi-disciplinary device, inter-disciplinary approaches to the design are necessary, especially to minimize resource expenses. Those methods are mainly provided in the last two chapters of this thesis. They are said to be full-stack as each "stack" in the quantum computer will be modeled (where for us, a stack can be the physics of the quantum gate, the algorithm implemented, the fault-tolerance construction considered, etc).

\section*{Our work in more details}
The first question we will investigate is finding the largest algorithm a quantum computer can successfully implement, using quantum error correction in the presence of a scale-dependent noise. A scale-dependent noise is a noise whose strength grows with the number of physical components inside the quantum computer. The connection with resource estimation is that in the presence of a limited amount of resources for all the computer, when scaling up, each of the physical components in the computer will receive a fewer amount of this resource. It will typically induce noisier operations. For instance, for a fixed amount of available frequencies for the qubits, the more qubits there are inside the computer, and the more probable crosstalk issues (i.e., the fact to address extra other qubits than the ones that are targetted by the driving signals) might occur \cite{schutjens2013single,theis2016simultaneous}. Many other examples of scale-dependent noise can occur. The general idea is that it might be possible to create qubits in an environment where they experience a low amount of noise, but what is complicated is to maintain this low noise environment when more and more qubits are added \cite{rosenberg2020solid,monroe2013scaling}. Having a noise that is growing in intensity with the size of the computer is already annoying in itself, but what can be worse is that this condition can violate one crucial hypothesis behind the quantum threshold theorem: the requirement that the noise of the physical component must be below the threshold. Indeed, if the noise grows with the size of the computer, it is possible that while being lower than the threshold for a small number of qubits, it gets higher than this value when all the additional qubits used to perform error correction are included. If it occurs, the accuracy of the computer is then intrinsically limited because either no error correction or only a limited amount of it would be possible. Those issues motivated the work presented in the third chapter of the manuscript. There, we study how to maximize the accuracy of the logical gates in the presence of a scale-dependent noise, allowing us to deduce the largest algorithm the computer would be able to run successfully. From the connection between scale-dependent noise and limited resource we established, we provide a first approach allowing to find the minimum resource required to implement the algorithm, and reciprocally to find what is the maximum accuracy the logical gates of a computer can get to in the presence of a limited fixed amount of resource. However, the approaches developed in this chapter do not allow us to treat any kind of resource optimization.

This is why we investigated further to find a general way to formulate the problem. The first thing to acknowledge is that the question of resource estimation of large-scale quantum computing is a question that is at the frontier of many different disciplines; one cannot only focus on the noise felt by the physical qubits to answer this question. For instance, it requires knowledge from quantum physics but also from computer science, cryogenics, and engineering. One important issue it brings is the almost omnipresence of contradictory behaviors that might intersect all those different fields. To give a few examples, ion trap technologies are associated to a very long lifetime for the qubits. Coherence times about $600s$ \cite{bruzewicz2019trapped} and even hours \cite{wang2021single} have been reached experimentally. All this while having quantum gates that can last for $100 \mu s$  \cite{bruzewicz2019trapped}, providing the ratio of coherence time divided by the gate duration, which can be close to $10^6$, much higher than in many other technologies \cite{kjaergaard2020superconducting,hanson2007spins}. From this perspective, they could be seen as ideal candidates for large-scale, fault-tolerant quantum computers as they would be associated with a low overhead in terms of physical qubit per logical qubit. But on the other hand, there are real challenges putting many ions together in a single trap while guaranteeing good connectivity and high coherence times \cite{bruzewicz2019trapped} which is a clear drawback for scalability. This is why strategies consisting of regrouping the ions in smaller groups are considered. However, it introduces other challenges, such as how to move the ions efficiently to implement the appropriate interactions. This is one example of contradictory behaviors or challenges that can cross different fields. If we think about superconducting qubits now, many of them could be put on a small size chip such that they are not facing the issue ion trap are\footnote{Crosstalk issues could occur but in a less problematic way than for ions.}. But their lifetime is not comparable to what it is possible to do with ion traps, and those qubits must typically be at very low temperatures, close to $10mK$ \cite{krantz2019quantum,kjaergaard2020superconducting}. Assuming that a superconducting quantum chip would have many physical qubits, potentially between hundred of thousands and the billion depending on the size of the fault-tolerant algorithm that is supposed to run and the quality of the qubits, it could lead to potential issues concerning the energetic cost. More generally, many other technologies have their pros and cons in terms of scalability, resource cost, and level of maturity. There is quantum computing based on photons (linear optics quantum computation) \cite{slussarenko2019photonic,kok2007linear}, nuclear magnetic resonance quantum computing \cite{criger2012recent}, quantum computing based on spin qubit, and the list continues. Spin qubits \cite{hanson2007spins} are seen as excellent candidates in terms of scalability as they can benefit from the maturity of CMOS technology in terms of integration \cite{maurand2016cmos,vinet2018towards}. Still, this approach to quantum computing is much more recent and does not benefit yet from the same maturity as superconducting qubits.

Here, we mainly talked about competition phenomena occurring between different fields of science, but they may also arise within a given domain. For instance, it is believed that to scale up superconducting or spin qubits quantum computers, electronics controlling or generating the signals that will implement quantum gates on the qubits should be put inside of the cryostat \cite{bardin201929,patra2020scalable,mcdermott2018quantum}. But different technologies would allow for that: CMOS technology can be put at cold temperature and can generate signals of good quality, but it is associated with a more important thermal load than other approaches such as superconducting circuits \cite{semenov2003sfq,mcdermott2018quantum} or adiabatic computing \cite{debenedictis2020adiabatic} (which in return do not have the same level of maturity in term of performances). Choosing "the best" technology is then not easy. And we mainly talked about what happens for the hardware, but many different solutions can also be chosen on the software side of a quantum computer. For instance, there are many different possible strategies on the quantum-error correction side to consider. Thinking about the various areas of research and strategies to implement quantum error correction, we can give the examples of topological quantum error correction such as the surface code \cite{fowler2012surface,dennis2002topological}, bosonic codes \cite{terhal2020towards,ofek2016extending,noh2020fault}, concatenated constructions \cite{aliferis2005quantum,aharonov2008fault} etc. All the different codes and ways to implement them can be associated with different overheads in terms of the number of physical qubits and performance in reducing the errors. Benchmarking quantum error correction to find which code is the most resource-efficient is a complicated task because of the many different aspects that could be benchmarked. As we said before, even if one code appears to be better than another one from a specific performance in noise reduction, it will not necessarily be the best one when the full quantum computer is considered (if it has good performance in terms of noise reduction but requires very heavy classical processing, it might not be a very good candidate). As we see, quantum computing is a field full of contradictions!

What this discussion illustrates is mainly the fact that \textit{inter-disciplinary} approaches to the problem of scalability must be considered as all the different components involved in the design of a quantum computer are strongly interconnected. But what would be desirable would be to have a well-defined and \textit{unique} question to answer that would lead the entire design of the quantum computer. Indeed one issue is also that "too many" choices are possible in the design of the computer and it is hard to choose the best one. To phrase a "good" question, we can make the following analysis: even though at first view many elements in the quantum computer seem to be very far from the "quantum world" (cryogenics or signal generation are good examples), they are actually \textit{intrinsically} connected to it. Indeed, the ultimate goal of a quantum computer is to provide a trustable answer to some algorithm that has been implemented. The whole design of the quantum computer is made so that this condition must hold. For instance, a cryostat has to be designed because the qubits must be maintained cool, and this is because the quantity of noise felt by the qubits must be kept low to have a successful answer for some algorithm. Thus the design of a cryostat is indirectly connected to algorithmic aspects. Seeing the problem under this angle allows to see that connecting the algorithmic aspect, and more precisely, the probability that the algorithm succeeds to all the engineering involved in the design of a quantum computer can give the appropriate constraint to know how to design the computer. Now, only asking to solve this question might lead to many choices in the design, and many of them might be unreasonable. For instance, if one design will satisfy this condition but will induce a quantum computer consuming hundreds of gigawatts, it will not be a good choice "for all practical purposes". This is why one step further is to lead the design by asking to minimize a given resource \textit{under the constraint} that the algorithm succeeds with a targetted probability. Taking the power consumption as a resource and phrasing the question this way will then provide the appropriate constraints on the design of the quantum computer in such a way that it implements the algorithm successfully (which is its ultimate goal) while spending the least amount of power to do this (which will lead to a "reasonable" design). The constraints this question gives can then directly guide the engineers, physicists, and computer scientists to design the computer together. In the fourth chapter of this thesis, we will provide the conceptual elements required to apply this method properly. In the last chapter, we are going to use it in a complete model of a quantum computer where we will optimize the amount of error correction to perform, the optimal temperature of the different stages of the cryostat (which contains the qubits but also the electronics generating the signals) and the level of attenuation that is chosen on the driving line\footnote{As we are going to see, attenuation is being put on the coaxial cable where the signals driving the qubit are injected to reduce the thermal noise.} to implement fault-tolerant algorithms involving thousands of logical qubits with a high enough targetted probability of success. We will see that orders of magnitude of power consumption can be gained in regimes, where without our optimization, a power consumption bigger than the gigawatt could be involved. It indicates that our approach can significantly enhance the scalability of the architecture. We will also see how the design of the computer depends on the characteristics of the implemented algorithm: it shouldn't be surprising that the size of the algorithm and the way it is implemented can impact the optimal design of the computer (and, of course, its energetic cost). Finally, because we are going to find the \textit{minimum} power required to implement an algorithm, the way we formulate the problem can allow to actually \textit{define} properly the question of the energetic cost of quantum computing.

\section*{Outline of the manuscript}

This thesis is organized as follows. The first chapter is dedicated to provide the essential tools we need from circuit quantum electrodynamic theory to understand how noisy gates performed on superconducting qubits are and how much power is required to implement them. We also give the state-of-the-art values for superconducting qubits we will consider using in the rest of the thesis. The second chapter is dedicated to quantum error correction and fault-tolerance. We provide all the theoretical results allowing us to understand this theory. Those two chapters do not contain any original results; they just introduce the tools required to understand the last three chapters. 

The third chapter is the first one providing results from this Ph.D. It is dedicated to understand what happens for fault-tolerance when the noise felt by the qubits grows with the number of qubits: what is occurring in this regime, and is it necessarily an issue for scalability. This chapter will also give first intuitions about how it is possible to estimate the minimum resources required to implement a fault-tolerant algorithm as having a scale-dependent noise is often related to resource constraints. The fourth chapter is dedicated to explain precisely the method allowing to find how much resources a calculation requires. More specifically, we will show that the minimum power required and the optimal architecture the quantum computer should have to reach this minimum can be found. The last chapter applies those concepts in a complete model of quantum computer based on superconducting qubits where we do quantitative estimations of the energetic cost required to run different kinds of large-scale algorithms. Our work indicates that the inter-disciplinary approach to the question of energetics we propose can reduce the power consumption of the computer by orders of magnitudes in a regime where the consumption could otherwise be higher than the gigawatt. We will also give some first intuitions about what is essential to optimize to make fault-tolerant quantum computing energy-efficient. This work illustrates that the energetic cost of quantum computing should be a figure of merit by itself on the scorecard of qubits technology to assess their potential for scalability. It also shows that optimizing the architecture of a quantum computer through methods like the one we are proposing can be a powerful tool allowing to clearly improve the potential in terms of scalability.
\chapter{Physics of superconducting qubits}
In all this Ph.D. thesis, the physical examples we will consider will be based on superconducting qubits. The goal of this chapter is, first, to explain the basics of the physics they rely on. Then, we will explain what is the origin of the noise when they are being manipulated before calculating the energetic cost to perform single-qubit gate operations. We will also give there the characteristics of the qubits and gates we will use along this thesis. This energetic cost is one important building block that we will use in the rest of this Ph.D. thesis. This chapter, apart from the section \ref{sec:infidelity_operations}, does not contain any original result; we just provide the tools we need to understand the rest of this thesis.
\section{Designing superconducting qubits}
Our goal here is to explain how to make a superconducting qubit and what is the transmon regime. For this purpose, we need to be able to describe quantum phenomena in electrical circuits. This is usually done through the canonical quantization procedure that we are going to explain. Some references on the subject can be found in \cite{vool2017introduction,gu2017microwave,wendin2005superconducting,nigg2012black}. We will then apply this method for the electrical circuits we will consider: we will start by quantizing a simple LC circuit before explaining the quantization of a superconducting qubit. Our explanations are mainly taken from \cite{vool2017introduction}.
\subsection{Quantization of electromagnetic circuits}
\label{sec:canonical}
\subsubsection{Canonical quantization}
One way that is used to quantize a classical theory is called the canonical quantization. But in order to describe it, we need to make a few reminders about Lagrangian and Hamiltonian mechanics. 

\textbf{Basics of Lagrangian and Hamiltonian formulation of classical mechanics}

Let us consider a classical system. This system can be described by the mean of its coordinates in the phase space. If one considers a system with $N$ degrees of freedom, its state is entirely described by $N$ coordinates $\{q_k\}_{k=1}^N$ and associated velocities $\{\dot{q}_k\}_{k=1}^N$. A Lagrangian $\mathcal{L}$ is a function $\mathcal{L}(\{q_k\}_{k=1}^N,\{\dot{q}_k\}_{k=1}^N)$ from which the classical equations of motion can be deduced by the mean of the Euler-Lagrange equations \cite{basdevant2014principes}:
\begin{align}
\forall k \in [|1,n|]: \frac{d}{dt}\frac{\partial \mathcal{L}}{\partial \dot{q}_k}=\frac{\partial \mathcal{L}}{\partial q_k}.
\label{eq:euler_lagrange}
\end{align}
To fix ideas, we can take the example of a one dimensional mechanical harmonic oscillator. Calling $x$ the relative distance to the rest position, this physical system satisfies the equation of motion
\begin{align}
\ddot{x}+\omega_0^2 x=0
\label{eq:1Dharmonic_oscillator}
\end{align}
where $\omega_0$ is the resonant frequency of this oscillator. An appropriate Lagrangian to describe this dynamic would be:
\begin{align}
\mathcal{L}(x,\dot{x})=\frac{m}{2}\dot{x}^2-\frac{m \omega_0^2}{2}x^2.
\label{eq:lagrangian_1D_harmonic_oscillator}
\end{align}
This is because applying the Euler Lagrange equation \eqref{eq:euler_lagrange} in this situation would give:
\begin{align}
\frac{d}{dt}\frac{\partial \mathcal{L}}{\partial \dot{x}}=\frac{\partial \mathcal{L}}{\partial x} \Leftrightarrow m\ddot{x}=-m \omega_0^2 x,
\end{align}
which provides the appropriate equation of motion \eqref{eq:1Dharmonic_oscillator}\footnote{Actually, many equivalent Lagrangian can describe properly a system \cite{basdevant2014principes}. Removing the term $m$, homogeneous to a mass would for instance, also provide the appropriate dynamic and is thus not strictly necessary. It is only in order to get quantities homogeneous to energy that we considered it here.}. We also notice that this Lagrangian has the expression of the kinetic energy of the system minus the potential energy. A valid Lagrangian does not necessarily have this shape, but in many situations, it will occur to be the case \cite{goldstein2002classical}. 

In order to understand the quantization procedure later on, we must also introduce the Hamiltonian formulation of classical mechanics. First, the Lagrangian allows to define the generalized momentum $p_i$ associated to any of the generalized coordinate $q_i$:
\begin{align}
\forall i \in [|1,n|]: p_i \equiv \frac{\partial \mathcal{L}}{\partial \dot{q}_i}.
\end{align}
It also allows to define a Hamiltonian for the system through the equation:
\begin{align}
H(\{q_k\}_{k=1}^N,\{p_k\}_{k=1}^N) \equiv \sum_{k=1}^N \dot{q}_k p_k -\mathcal{L}(\{q_k\}_{k=1}^N,\{\dot{q}_k\}_{k=1}^N)
\end{align}
The equation of motions in the Hamiltonian formalism can be written as:
\begin{align}
&\forall i \in [|1,n|]: \dot{p}_i=-\frac{\partial H}{\partial q_i} \label{eq:momentum_hamiltonian} \\
&\forall i \in [|1,n|]: \dot{q}_i=\frac{\partial H}{\partial p_i}
\label{eq:velocity_hamiltonian} 
\end{align}
Those first order differential equations are mathematically equivalent to the second order Euler Lagrange equations \eqref{eq:euler_lagrange} \cite{basdevant2014principes}.

Keeping the example of the harmonic oscillator, we would for instance find that the momentum associated to the position $x$ is simply:
\begin{align}
p=\frac{\partial \mathcal{L}}{\partial \dot{x}}=m \dot{x}
\end{align}
And the Hamiltonian is:
\begin{align}
H(x,p)=\dot{x}p-\mathcal{L}(x,\dot{x})=\frac{p^2}{2m}+\frac{m \omega_0^2}{2}x^2
\end{align}

The last concept we will need in order to understand the quantization of a classical theory is the notion of Poisson bracket. Let us consider an arbitrary function $A(\{q_i\},\{p_i\},t)$. The equation of motion of this function satisfies:
\begin{align}
\frac{d A}{d t}=\sum_i \left( \frac{\partial f}{\partial q_i} \frac{\partial q_i}{\partial t}+\frac{\partial A}{\partial p_i} \frac{\partial p_i}{\partial t} \right)+\frac{\partial A}{\partial t}
\end{align}
And, using \eqref{eq:momentum_hamiltonian} and \eqref{eq:velocity_hamiltonian}, it can be rewritten as:
\begin{align}
\frac{d A}{d t}=\sum_i \left( \frac{\partial A}{\partial q_i} \frac{\partial H}{\partial p_i}-\frac{\partial A}{\partial p_i} \frac{\partial H}{\partial q_i}\right)+\frac{\partial A}{\partial t}=\{A,H\}+\frac{\partial A}{\partial t}
\label{eq:evol_A}
\end{align}
Where the Poisson Bracket $\{.,.\}$ is defined in general by:
\begin{align}
\{A,B\} \equiv \sum_i \left( \frac{\partial A}{\partial q_i} \frac{\partial B}{\partial p_i}-\frac{\partial A}{\partial p_i} \frac{\partial B}{\partial q_i} \right)
\end{align}
We notice in particular:
\begin{align}
&\{q_i,p_j\}=\delta_{i,j}
\label{eq:poisson1}\\
&\{q_i,q_j\}=\{p_i,p_j\}=0
\label{eq:poisson2}
\end{align}
We will see that the structure of the equations of motions in the Hamiltonian formalism, when expressed with the Poisson bracket, are very similar to the equation of motions of operators in the Heisenberg formalism of quantum mechanics. This will be the starting point of the canonical quantization principle.

\textbf{equation of motion of a quantum system in Heisenberg picture}

Now we consider a quantum system described by a Hamiltonian $\widehat{H}$. We assume that this Hamiltonian is composed of generalized coordinates operators $\{\widehat{q}_i\}$ and associated momentum $\{\widehat{p}_i\}$. In the Heisenberg picture of quantum mechanics, the equation of motion of any operator $\widehat{A}_H$ acting on S follows the equation:
\begin{align}
&\frac{d \widehat{A}_H(t)}{dt}=\frac{1}{i \hbar}[\widehat{A}_H(t),\widehat{H}]+\frac{\partial \widehat{A}_H(t)}{\partial t} \label{eq:heisenberg}
\end{align}
The coordinates and momenta satisfy the commutations relations:
\begin{align}
&[\widehat{q}_i(t),\widehat{p}_j(t)]=i\hbar \delta_{ij}
\label{eq:commutation1}\\
&[\widehat{q}_i(t),\widehat{q}_j(t)]=[\widehat{p}_i(t),\widehat{p}_j(t)]=0
\label{eq:commutation2}
\end{align}
\eqref{eq:heisenberg} and \eqref{eq:commutation1},\eqref{eq:commutation2} contain all the information to deduce the behavior of any operator at time $t$ which is a polynome in the variables $\{p_i(t)\}$ and $\{q_i(t)\}$.

\textbf{Quantization procedure}

At this point, we can notice the strong similarity between \eqref{eq:heisenberg} and \eqref{eq:evol_A}, and also between \eqref{eq:commutation1},\eqref{eq:commutation2} and \eqref{eq:poisson1},\eqref{eq:poisson2}. Indeed, the quantum equations of motion can be obtained by replacing the Poisson bracket $\{A,B\}$ in those equations by a commutator $\frac{1}{i \hbar}[A,B]$ equal to this Poisson bracket. This will provide the appropriate equation of motion for the operator, as well as define the commutation relations properly. In practice, to obtain the quantum theory, it will be enough to promote the Hamiltonian into an operator by imposing the commutation relations:
\begin{align}
&[\widehat{q}_i(t),\widehat{p}_j(t)]=i \hbar\delta_{ij} \label{eq:commutation_quantization0}\\
&[\widehat{q}_i(t),\widehat{q}_j(t)]=0
\label{eq:commutation_quantization1}\\
&[\widehat{p}_i(t),\widehat{p}_j(t)]=0 \label{eq:commutation_quantization2}
\end{align}
Further problems might occur in practice (for instance, if in the classical Hamiltonian terms like $x.p$ appears there is a choice of ordering to take because while $x$ and $p$ commute classically, they don't commute anymore when promoted into operators), but we will not be facing them in what follows such that we can ignore those issues. This "resemblance" in the classical and quantum theory has originally been discovered by Dirac in 1925 \cite{basdevant2014principes}, and the mapping we describe here defines the so-called "canonical quantization" of a classical theory. This should not be understood as a general proof of how a quantum theory can be deduced from a classical theory but more as a guess which provides a recipe that works for many systems. Indeed, the natural order of things would be to deduce a classical theory from the quantum one, the latter being more fundamental. However, such a procedure will at least work for the class of system we are going to study in this Ph.D.

\subsubsection{An application: quantum description of an LC circuit}
To introduce the concepts we need behind the quantization of electromagnetic circuits, we can start with the LC circuit represented in the figure \ref{LC_circuit}.

\begin{figure}[h!]
\begin{center}
\includegraphics[width=0.3\textwidth]{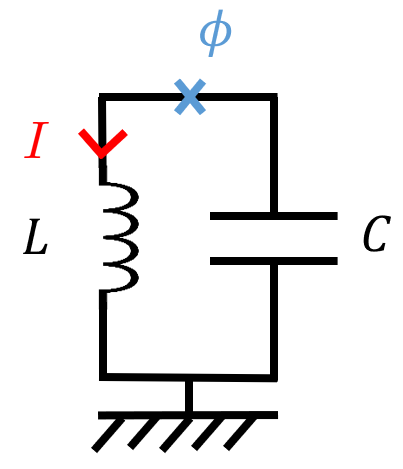}
\caption{The LC circuit we are considering}
\label{LC_circuit}
\end{center}
\end{figure}

In order to anticipate the usual way superconducting qubits are described, we will introduce a new physical quantity: the generalized flux variable. If we consider a point in the circuit where the electric potential is $V(t)$, the generalized flux variable is defined as \cite{vool2017introduction}:
\begin{align}
\phi(t) \equiv \int_{-\infty}^t V(t') dt'
\label{eq:generalized_flux}
\end{align}
In all what follows we will also assume that all electrical quantities vanish at $t=-\infty$. Calling $\phi$ the generalized flux associated to the voltage difference around the inductor, the current in the circuit satisfies $I=\phi/L=-C \ddot{\phi}$. It gives an equation of motion being, with $\omega_0=\frac{1}{\sqrt{LC}}$.
\begin{align}
\ddot{\phi}+\omega_0^2 \phi=0
\end{align}
To quantize the description we will proceed as explained in the section \ref{sec:canonical}. We consider $\phi$ being our generalized coordinate. A possible Lagrangian describing this dynamic is:
\begin{align}
\mathcal{L}(\phi,\dot{\phi})=\frac{C}{2}\dot{\phi}^2-\frac{\phi^2}{2 L}
\end{align}
The conjugated momentum associated to $\phi$ is:
\begin{align}
P \equiv \frac{\partial \mathcal{L}}{\partial \dot{\phi}}=C \dot{\phi}
\end{align}
We notice that $C \dot{\phi}$ physically corresponds to the charge on the negative plate of the capacitor. We can deduce the Hamiltonian:
\begin{align}
H(\phi,P) \equiv P \dot{\phi}-\mathcal{L}(\phi,\dot{\phi})=\frac{P^2}{2 C}+\frac{\phi^2}{2 L}
\label{eq:hamiltonian_LC}
\end{align}
The canonical quantization principle explained in \ref{sec:canonical} tells us that to obtain the quantum theory of this system, we just have to promote $\phi$ and $P$ into operators by imposing the commutation relation:
\begin{align}
&[\widehat{\phi},\widehat{P}]=i \hbar
\end{align}
We see that our choice of Lagrangian consists in considering the energy stored in the capacitor as kinetic energy (because associated to $P$), while the energy stored in the inductor is seen as potential energy. But there is nothing fundamental about that; it simply comes from our \textit{choice} of choosing $\phi$ as the coordinate. Another approach would have been to choose the charge stored by the capacitor as being the generalized coordinate. In this situation, the interpretation of kinetic and potential energy between inductor and capacitor would have been reversed. 

\textbf{A systematic approach to quantize electrical circuits}

In all that follows, the Lagrangian is always going to be written:
\begin{align}
\mathcal{L}=\mathcal{E}_{C}-\mathcal{E}_L,
\end{align} 
where $\mathcal{E}_{C}$ corresponds to the energy stored by the capacitive elements of the circuit, and $\mathcal{E}_L$ to the inductive ones. The generalized coordinate is always going to be the generalized flux \eqref{eq:generalized_flux}. A precise definition of capacitive and inductive elements can be found in \cite{vool2017introduction}. In this work, we will only have capacitors, inductors, and Josephson junction (that we are going to introduce). The Josephson junction will be considered as an inductive element, and its associated potential energy is going to be taken into account in $\mathcal{E}_L$. This approach to quantize electrical circuits is the one proposed in \cite{vool2017introduction}.
\subsection{Superconducting qubit and transmon regime}
Now that the basic principles of how one can have a quantum description of an electrical circuit have been explained, we are going to see how we can engineer a superconducting qubit in practice and what is the so-called transmon regime.
\subsubsection{The need for non-linear oscillators}
A quantum computer usually requires quantum systems having quantum states living in a two-dimensional Hilbert space: qubits. Ideally, we would like to create a physical system for which the quantum state exactly lives in a bi-dimensional Hilbert space. In the context of superconducting qubits, it is not possible. The strategy is then to create an approximation of a qubit. This is often done by considering a two-dimensional subspace of a physical system of high dimensions.

One way to do it experimentally is to realize a non-linear quantum oscillator. A non-linear quantum oscillator has energy levels that are not equally spaced, as represented in Fig \ref{potential_anharmonic}, in comparison to a harmonic oscillator. The advantage in using such a physical system is that assuming the initial state of the system lives in the Hilbert space spanned by the two lowest energy eigenstates, by sending signals resonant with the associated energy transition, the system will remain in this bi-dimensional subspace.

\begin{figure}[h!]
\begin{center}
\includegraphics[width=0.7\textwidth]{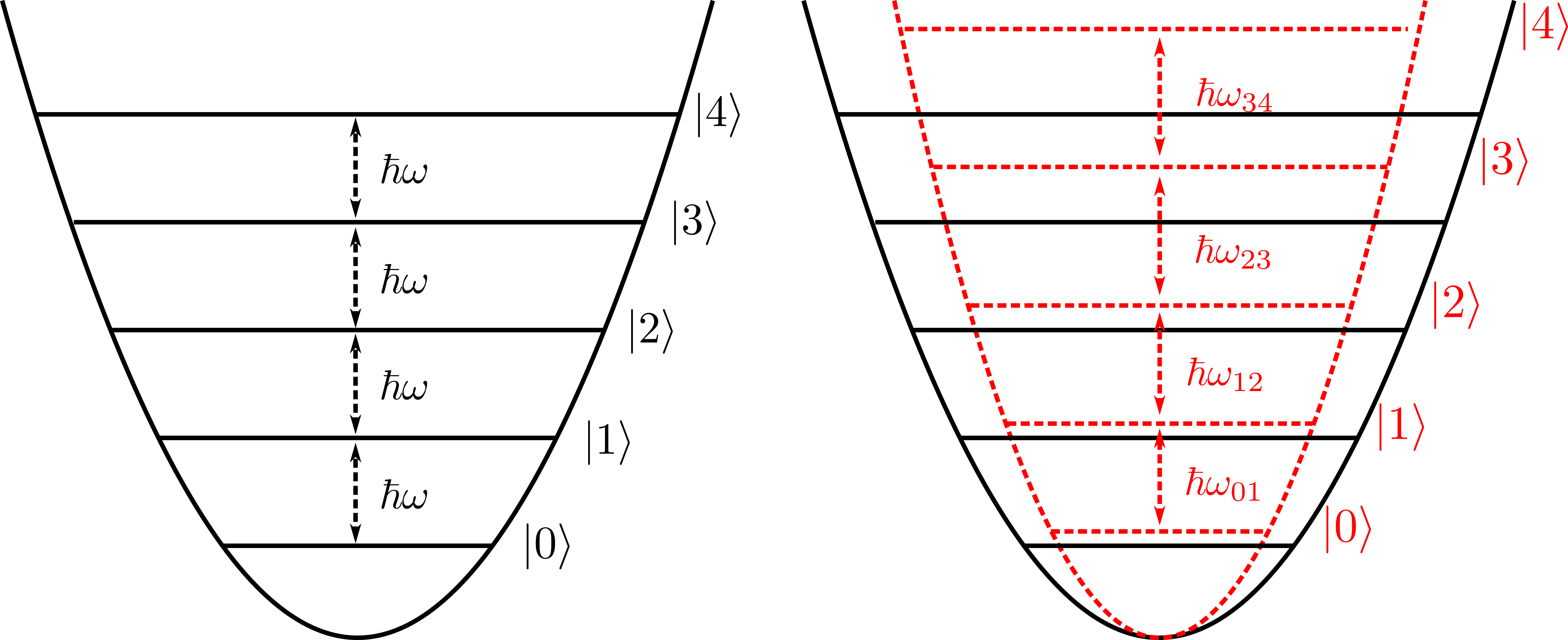}
\caption{Left: an harmonic potential. All level are equally spaced by the energy $\hbar \omega$. Right: an anharmonic potential. Because the potential is no longer quadratic, the quantum levels will not be equally spaced.}
\label{potential_anharmonic}
\end{center}
\end{figure}

As we are going to see, designing such systems can be done by replacing the inductor composing the LC oscillator with an element called the Josephson junction that we are now going to introduce. 
\subsubsection{The Josephson junction}
A Josephson junction is a two-terminal electric component composed of two superconducting electrodes separated via an insulator. In terms of description, it can be modeled by a capacitor in parallel to a Josephson tunnel element as shown in Figure \ref{josephson_junction}. To describe the electrical properties of a Josephson tunnel element, we will use again the generalized flux corresponding to the voltage applied to this element.  

\begin{figure}[h!]
\begin{center}
\includegraphics[width=0.8\textwidth]{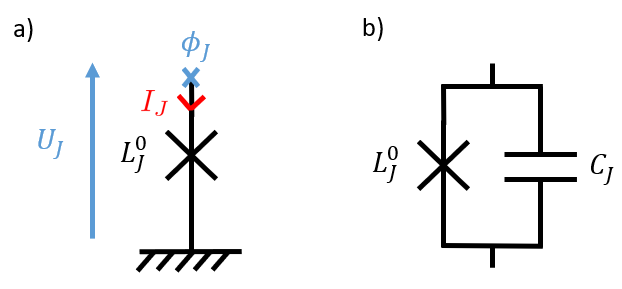}
\caption{\textbf{a)} Electrical representation of the Josephson tunnel element. \textbf{b)} Electrical model of a Josephson junction:  it is composed of a Josephson tunnel element in parallel to a capacitor.}
\label{josephson_junction}
\end{center}
\end{figure}

Calling $\phi_J$ the generalized flux associated to the voltage around the Josephson tunnel element, the electrical equations are:
\begin{align}
&\phi_J(t) \equiv \int_{-\infty}^t U_J(t') dt'\\
&I_J(t)=I_0 \sin\left(\frac{\phi_J(t)}{\varphi_0}\right)
\label{eq_current_josephson}\\
&\varphi_0 \equiv \frac{\hbar}{2e}
\end{align}
$\varphi_0$ is called the superconducting flux quantum, $2e$ physically represents the charge of a Cooper pair which are the particle leading to the superconducting current. Typical experimental values for the current $I_0$ lie between the $\mu A$ and $nA$ \cite{vool2017introduction}.

This element can be seen as a non-linear inductance. To understand what it precisely means, we recall that the relationship between current and voltage around an inductor of inductance $L$ is: $U(t)=L \frac{d I(t)}{dt}$. Expressing those quantities in term of the generalized flux associated to $U$, we can see that the inductance $L$ satisfies $\frac{d \phi}{dt}= L \frac{d I}{dt} \Leftrightarrow L=\frac{d \phi}{d I}$. The relationship between $d \phi$ and $dI$ is linear as $L$ does not depends on $\phi$. It motivates us to define the inductance of the pure Josephson element as:
\begin{align}
&L_J(\phi_J) \equiv \frac{d \phi_J}{d I_J}=\frac{L_J^0}{\cos(\phi_J/\varphi_0)}\\
&L_J^0 \equiv \frac{\varphi_0}{I_0}
\end{align}
We notice that now, the inductance is no longer constant and depends on the flux $\phi_J$. It implies that the relationship between $d \phi$ and $d I$ is not linear. This is what is usually meant by considering that the Josephson junction is a non-linear element. We also introduced the quantity $L_J^0$, which is called the zero-flux Josephson inductance.

Finally, in order to quantize the description, we need to express the energy stored by the Josephson tunnel element. The energy it receives between $t_0$ and $t_1$ is:
\begin{align}
\mathcal{E}=\int_{t_0}^{t_1} dt \ U(t) I(t) = \int_{t_0}^{t_1} dt \  \dot{\phi} I_0 \sin(\frac{\phi}{\varphi_0})= \left[-I_0 \varphi_0 \cos(\frac{\phi(t)}{\varphi_0})\right]_{t_0}^{t_1}=\left[-\frac{\varphi_0^2}{L_J^0} \cos(\frac{\phi(t)}{\varphi_0})\right]_{t_0}^{t_1},
\end{align}
where $[f(x)]_{a}^{b} \equiv f(b)-f(a)$. This energy is thus stored into potential energy and has the expression:
\begin{align}
&E_p(\phi) = -E_J \cos \left(\frac{\phi(t)}{\varphi_0} \right)\\
&E_J \equiv \frac{\varphi_0^2}{L_J^0},
\end{align}
where $E_J$ is called the Josephson energy \cite{krantz2019quantum}. 

At this point, we explained how a Josephson junction can be understood from an electrical point of view. This is enough to understand how we can engineer and then control a superconducting qubit.
\subsubsection{Engineering a superconducting qubit} 
Our goal is to use the Josephson junction in order to create a potential that will deviate from a harmonic oscillator. The reason is that it will provide us with two "well isolated" energy levels, which will allow us to define the superconducting qubit properly, as explained in a previous section. To do it, the usual way is to put the Josephson junction in parallel to a capacitor as represented in Figure \ref{superconducting_qubit}. This extra capacitor is just here to increase the total capacitance in the circuit for a reason given below.

\begin{figure}[h!]
\begin{center}
\includegraphics[width=0.8\textwidth]{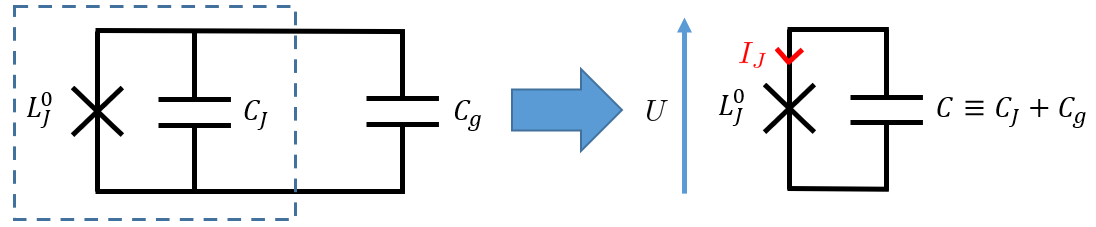}
\caption{A Josephson junction (inside the blue dotted box) is put in parallel to a capacitor of capacitance $C_g$. The resulting entity is an equivalent Josephson junction with a now greater capacitance. It will correspond to the "hardware" implementation of a superconducting qubit.}
\label{superconducting_qubit}
\end{center}
\end{figure}

Now, we can write down a Lagrangian describing the equation of motion in order to quantize the circuit. We will consider the generalized coordinate being the generalized flux associated to the potential difference $U$ as represented in Figure \ref{superconducting_qubit}. As we saw previously, a good "hint" as a Lagrangian is to consider it being equal to the energy stored by the capacitive element, minus the one stored by the inductive. Considering $\phi$ as the generalized coordinate, we get:
\begin{align}
&\mathcal{L}(\phi,\dot{\phi}) = \mathcal{E}_C-\mathcal{E}_L=\frac{1}{2} C \dot{\phi}^2+E_J \cos(\frac{\phi(t)}{\varphi_0})\\
\end{align}
Applying Euler Lagrange equations \eqref{eq:euler_lagrange}, it provides the expected equations of motion\footnote{Indeed, with the convention used in Figure \ref{superconducting_qubit}, the voltage $U$ at the boundary of the capacitor satisfies: $U=-Q/C$ where $Q$ is the charge on the positive plate of the capacitor. Replacing $U$ with the generalized flux and taking the derivative to time, we obtain: $\ddot{\phi}=-\frac{I}{C}=-\frac{I_0 \sin(\frac{\phi}{\varphi_0})}{C}$ which is equivalent to \eqref{eq:euler_lagrange_superconducting_qubit}.}:
\begin{align}
\frac{d}{dt} \frac{\partial \mathcal{L}}{\partial \dot{\phi}}=\frac{\partial \mathcal{L}}{\partial \phi} \Leftrightarrow  \ddot{\phi} = -\frac{\varphi_0}{C L_J^0} \sin(\frac{\phi(t)}{\varphi_0})=-\frac{I_0}{C}\sin(\frac{\phi(t)}{\varphi_0})
\label{eq:euler_lagrange_superconducting_qubit}
\end{align}
We can now find the momentum and Hamiltonian:
\begin{align}
& P \equiv \frac{\partial \mathcal{L}}{\partial \dot{\phi}}=C \dot{\phi}\\
& H(\phi,P) \equiv P \dot{\phi} - \mathcal{L}(\phi,\dot{\phi})=\frac{E_C P^2}{e^2}-E_J \cos(\frac{\phi}{\varphi_0})\\
&E_C \equiv \frac{e^2}{2C}
\end{align}
We notice that the momentum physically represents the charge (up to a sign) that is stored inside the capacitor. The quantity $E_C$ is called the charging energy \cite{krantz2019quantum}.
Promoting $\phi$ and $P$ into operators and imposing the commutation relation $[\widehat{\phi},\widehat{P}]=i \hbar$, the theory is now quantized.  

Now, we will adopt a perturbative approach. Performing a Taylor expansion of the $\cos$ in the variable $\phi/\varphi_0$, and neglecting any constant term which doesn't play a role in the dynamic, it gives, up to order four in $\widehat{\phi}$:

\begin{align}
H(\widehat{\phi},\widehat{P})=\frac{E_C \widehat{P}^2}{e^2}+E_J \frac{\widehat{\phi}^2}{\varphi_0^2} -E_J \frac{\widehat{\phi}^4}{4! \varphi_0^4}+ O\left(\widehat{\phi}^6\right)
\label{eq:hamiltonian_superconducting_qubit_firstorders1}
\end{align}
Expressing $E_C$ and $E_J$ as a function of $C$ and $L_J^0$, we realize that up to order two in $\widehat{\phi}$, \eqref{eq:hamiltonian_superconducting_qubit_firstorders1} has an expression analog to \eqref{eq:hamiltonian_LC}:
\begin{align}
H=\frac{\widehat{P}^2}{2C}+ \frac{\widehat{\phi}^2}{2 L_J^0} +O\left(\widehat{\phi}^4\right)
\end{align} 
It means that at the lowest order, the device we engineered behaves as a quantum harmonic oscillator. But the higher-order terms introduce the an-harmonicity that we are looking for.

\subsubsection{From anharmonic oscillator to superconducting qubit}
Now that we understand how the Josephson junction allowed us to create an anharmonic potential, we can see how it allows us to find two isolated energy levels that are going to play the role of the computational states $\ket{0}$ and $\ket{1}$ for the qubit. Up to order two, the Hamiltonian has the shape of a mechanical harmonic oscillator having a mass $m=e^2/2 E_C=C$ and frequency $\omega_0=\sqrt{8 E_J E_C}/\hbar=\frac{1}{\sqrt{L_J^0 C}}$. Thus, if we introduce the annihilation and creation operators associated to a harmonic oscillator, defined as \cite{cohen1998mecanique}:
\begin{align}
&\widehat{a} \equiv \sqrt{\frac{C \omega_0}{2 \hbar}} (\widehat{\phi}+\frac{i}{m \omega_0} \widehat{P}) \label{eq:annihilation_operator}\\
&\widehat{a}^{\dagger} \equiv \sqrt{\frac{C \omega_0}{2 \hbar}} (\widehat{\phi}-\frac{i}{m \omega_0} \widehat{P}),
\end{align}
and if we re express the physics with those operators, we obtain (removing all the terms higher than order $4$ in $\widehat{\phi}$)
\begin{align}
&\widehat{H} = \hbar \omega_0 \widehat{a}^{\dagger} \widehat{a} +\widehat{V}
\label{eq:hamiltonian_superconducting_qubit_firstorders2}\\
&\widehat{V} \equiv - \frac{E_C}{12}(\widehat{a}+\widehat{a}^{\dagger})^4
\label{eq:potential_V}
\end{align}
The energy levels of the qubit will correspond to the two lowest energy states of this Hamiltonian. First order perturbation theory provides the energy gap between the two first levels, $\hbar \omega_{01}$, and the two following $\hbar \omega_{12}$:
\begin{align}
&\hbar \omega_{01}=\hbar \omega_0-E_C \label{eq:omega01}\\
&\hbar \omega_{12}=\hbar \omega_{01}+E_C
\label{eq:omega12}
\end{align}
The derivations leading to those results are performed in further details in the appendix \ref{app:superconducting_qubit_frequency}.
We notice that the charging energy is the energy that will give the strength of the an-harmonicity of the potential. For low values of $E_C$, $\omega_{01} \approx \omega_{12}$: the behavior gets closer to the harmonic oscillator ($\omega_{01}$ is typically in the GHz range for superconducting qubits, we are going to provide typical values in the section \ref{sec:state_of_the_art}). For this reason, it seems at first view preferable to have a high value for the charging energy, more precisely: $E_C \gg \hbar \omega_{01} \Leftrightarrow E_C \gg E_J$ in order to have a good qubit. Indeed the potential will be greatly an-harmonic, and the two first levels would have a better "insulation" when being coherently driven. However, in the regime $E_C \gg E_J$ the qubit becomes more sensible to charge noise which is experimentally challenging to suppress \cite{krantz2019quantum}. For those reasons, the community has actually chosen to design superconducting qubits in the regime $E_J \gg E_C$ which induces a low anharmonicity of the potential but makes the qubit less sensitive to charge noise. Qubits realized in this regime are called transmon qubits. We can also understand better the role of the extra capacitance $C_g$ that we have added: it will allow experimentally to increase the charging energy in order to reach the transmon regime. Finally, the Hamiltonian describing our qubit can simply be approximated as:
\begin{align}
H = \hbar \omega_0 \widehat{a}^{\dagger} \widehat{a} +\widehat{V} \approx -\frac{\hbar \omega_{01}}{2} \widehat{\sigma}_z.
\end{align}
The reason for the minus sign ($-\hbar \omega_{01}/2$) is related to the fact that in all this thesis, we take the conventions used in quantum information (which are sometimes different than the ones used in the quantum optic community). For us, the Bloch sphere \cite{cohen1998mecanique} which allows representing the state of a two-level system has its north pole being $\ket{0}$, which corresponds to the ground state (i.e., the state of lowest energy) of the system. On the south pole, there is the state $\ket{1}$ which corresponds to the excited state. Considering that all the matrices are written in the ordered basis $(\ket{0},\ket{1})$, it implies that we must have $H=-\hbar \omega_{01}/2 \widehat{\sigma}_z$. The figure \ref{fig:bloch_sphere} summarizes this discussion.
\begin{figure}[h!]
\begin{center}
\includegraphics[width=0.5\textwidth]{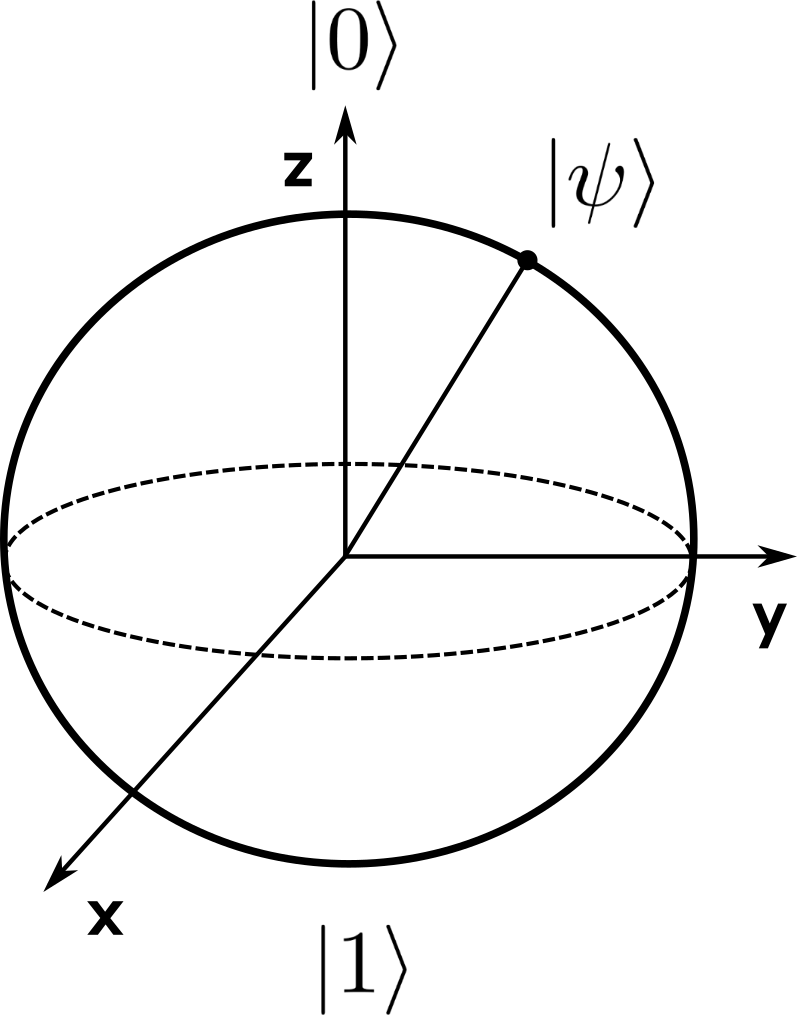}
\caption{The convention for the Bloch sphere we are taking in this PhD follows the convention from the quantum information community where the $\ket{0}$ state is at the north pole, and the $\ket{1}$ at the south pole.}
\label{fig:bloch_sphere}
\end{center}
\end{figure}
\subsubsection{Driving the qubit}
So far, we have explained how to obtain two well-isolated energy levels in order to have a qubit. But for an information processing task, we must be able to control it, i.e., to perform single-qubit gate operations (and also two-qubit gate operations, but we will only briefly comment on those in the section \ref{sec:state_of_the_art_2qb}. 

One possible way to realize single qubit gates is to couple capacitively the qubit to a voltage source as represented in Figure \ref{driving_superconducting_qubit}. 
\begin{figure}[h!]
\begin{center}
\includegraphics[width=0.8\textwidth]{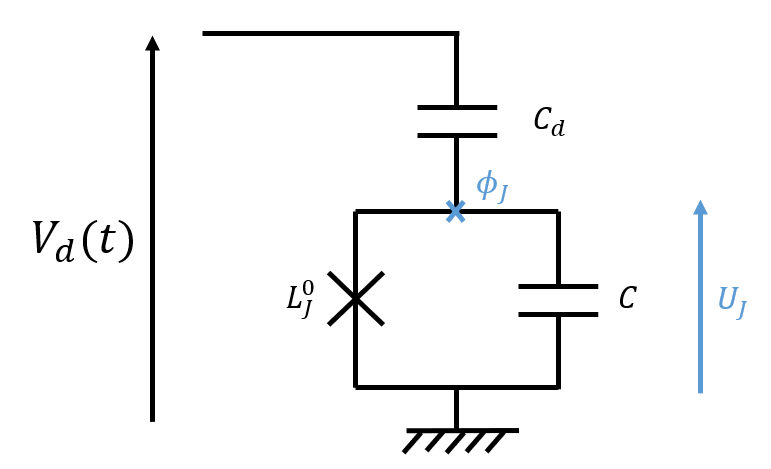}
\caption{Putting an extra capacitance $C_d$ and driving it with a voltage $V_d(t)$ allows to implement single qubit gate operations.}
\label{driving_superconducting_qubit}
\end{center}
\end{figure}

One Lagrangian describing the dynamic represented in Figure \ref{driving_superconducting_qubit} is:
\begin{align}
\mathcal{L}(\phi_J,\dot{\phi}_J)=\frac{1}{2} C_d (V_d(t)-\dot{\phi}_J)^2 +  \frac{1}{2} C_J \dot{\phi}_J^2 -(-E_J \cos(\frac{\phi_J}{\varphi_0})),
\end{align}
where the momentum conjugated to $\phi_J$ is:
\begin{align}
p_J \equiv \frac{\partial \mathcal{L}}{\partial \dot{\phi}_J} = (C_J+C_d)\dot{\phi}_J-C_d V_d(t)
\end{align}
The Hamiltonian can be deduced:
\begin{align}
&H\equiv p_J \dot{\phi}_J - \mathcal{L}=\frac{p_J^2}{2 C_{\Sigma}} - E_J \cos(\frac{\phi_J}{\varphi_0}) + \frac{C_d}{C_{\Sigma}} p_J V_d(t)\\
&C_{\Sigma} \equiv C+C_d
\end{align}
We now promote the flux and momentum into quantum operators, imposing the canonical commutation relation $[\widehat{\phi}_J,\widehat{p}_J]=i \hbar$, and, up to order four in the variable $\widehat{\phi}_J$, we get:
\begin{align}
&\widehat{H}=\widehat{H}_0+\widehat{H}_d\\
&\widehat{H}_0 \equiv \frac{\widehat{p}_J^2}{2 C_{\Sigma}} + \frac{E_J^0}{2 \varphi_0^2} (\widehat{\phi}_J)^2 - \frac{E_J^0}{6 \varphi_0^4} (\widehat{\phi}_J)^4 + O((\widehat{\phi}_J)^6)\\
&\widehat{H}_d \equiv \frac{C_d}{C_{\Sigma}} \widehat{p}_J  V_d(t).
\end{align}
We recognize $\widehat{H}_0$ as corresponding to the free Hamiltonian of a superconducting qubit having a Josephson energy $E_J$ and a charge energy $E_C=2C_{\Sigma}/e^2$. $\widehat{H}_d$ correspond to the interaction of the qubit with the time-varying voltage that will perform the rotations. In order to see it more clearly, we will express all the operators as a function of the harmonic oscillator annihilation and creation operators defined in \eqref{eq:annihilation_operator}. Neglecting the terms higher or equal to order four in $\widehat{\phi}_J$, we get:
\begin{align}
&\widehat{H}_0 = \hbar \omega_{0} \widehat{a}^{\dagger} \widehat{a}\\
&\widehat{H}_d(t)=\frac{C_d}{C_{\Sigma}} \sqrt{\frac{\hbar}{2}\sqrt{\frac{C_{\Sigma}}{L_J^0}}}i(\widehat{a}^{\dagger}-\widehat{a}) V_d(t)\\
&\omega_0 \equiv \frac{1}{\sqrt{L_J^0 C_{\Sigma}}}.
\label{eq:transmon_approx_harmonic_oscillator_frequency}
\end{align}
Finally, if we "cut" the Hilbert space to the two firsts energy levels, which means in this context replacing: $\widehat{a} \to \widehat{\sigma}_-$, $\widehat{a}^{\dagger} \to \widehat{\sigma}_+$ we get:
\begin{align}
&\widehat{H}(t)=-\frac{\hbar \omega_{0}}{2} \widehat{\sigma}_z + \hbar g V_d(t) \widehat{\sigma}_y \label{eq:driven_superconducting_qubit}\\
&g \equiv  \frac{C_d}{C_{\Sigma}} \sqrt{\frac{1}{2 \hbar}\sqrt{\frac{C_{\Sigma}}{L_J^0}}}
\end{align}
We recognize in \eqref{eq:driven_superconducting_qubit} the Hamiltonian of a two-level system driven classically. To see more clearly the operation done on the qubit, we can go in the interacting picture with respect to the free Hamiltonian $-\frac{\hbar \omega_{0}}{2}\widehat{\sigma}_z$. Assuming a resonant drive: $V_d(t)=V_0 \cos(\omega_{0} t+\psi+\pi/2)$, and applying the rotating wave approximation \cite{haroche2006exploring} the Hamiltonian of the evolution reads:
\begin{align}
\widehat{H}_I(t)=\frac{\hbar g V_0}{2} \vec{n}.\vec{\widehat{\sigma}},
\end{align}
where $\overrightarrow{\widehat{\sigma}} \equiv (\widehat{\sigma}_x,\widehat{\sigma}_y,\widehat{\sigma}_z)$. It induces a rotation of the qubit along the $\vec{n}=(\cos(\psi),\sin(\psi),0)$ axis of the Bloch sphere at the Rabi frequency $\Omega = g V_0$: the phase in the voltage drive is selecting the axis of rotation.
\section{Scaling up the devices: arranging superconducting qubits in waveguide}
Up to this point, we described how a superconducting qubit can be designed and what are the main approximations behind such construction. We also showed how it can be controlled experimentally in order to implement gate operations. We saw that single-qubit gates are typically performed by applying an appropriate oscillating voltage on the superconducting qubit. However, our approach relied on two main assumptions. The first one is that we assumed that the voltage can be applied "instantaneously" on the qubit, i.e., there is no delay between the moment the voltage is generated and the moment it interacts with the qubit. Unfortunately, if we want to scale up quantum computers, it will not be possible to put the signal generation close enough to all the qubits in order to make this assumption valid. In typical microwave experiments, the signal propagation delay cannot be neglected \cite{pozar2011microwave}. The second one is that we assumed the voltage is a classical entity. Understanding it this way doesn't allow us to understand phenomenons like spontaneous emission, which are due to the interaction between a two-level system and a quantized environment composed of a continuum of modes \cite{gardiner1985input}. Spontaneous emission, being one of the major limitations in the qubit lifetime, has to be taken into account with the modeling in order to have an accurate description of the physics. The role of this section is thus to have more accurate models.
\subsection{Classical description of a waveguide}
We will begin by providing the classical description of a transmission line such as a coaxial cable or a coplanar waveguide. The typical wavelength of signals propagating into coaxial cables or waveguides is in the cm to mm range which can be comparable or smaller than the typical size of the circuits considered. In order to properly model the electromagnetic phenomenon, we then have to take into account the propagation phenomenon. One model to describe an electrical waveguide is based on a lumped-element circuit model \cite{pozar2011microwave} as represented on Figure \ref{fig:lossless_transmission_line}. The motivation behind this model is to acknowledge the fact that to send a voltage, one necessarily needs two lines between which the voltage difference will occur. Those metallic lines are separated by some insulator, and because of that, a "parasitic" capacitance will be formed, represented by the capacitance per unit length $c_0$ on Figure \ref{fig:lossless_transmission_line}. Also, the current flowing on those lines might generate a flux inducing a counter-acting electromotive force. This is taken into account by considering a "parasitic" inductance per unit length $l_0$.
\begin{figure}[h!]
\begin{center}
\includegraphics[width=0.8\textwidth]{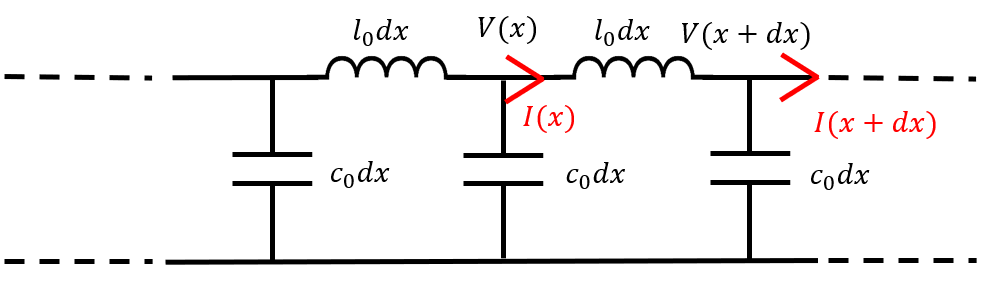}
\caption{Lumped-element model for a lossless transmission line. The inductance per unit length $l_0$ represent the flux that might be induced by the current between the positive and ground planes. The capacitance per unit length $l_0$ the charge accumulation between positive and mass line.}
\label{fig:lossless_transmission_line}
\end{center}
\end{figure}

Now, we can apply Kirchhoff laws between $x$ and $x+dx$ in order to find the equations that current and voltage are following. We find:
\begin{align}
&V(x,t)-V(x+dx,t)=l_0 dx \frac{\partial I(x,t)}{\partial t}\\
&c_0 dx \frac{\partial V(x+dx,t)}{\partial t}=I(x,t)-I(x+dx,t),
\end{align}
from which we deduce that telegraph equations \cite{pozar2011microwave}:
\begin{align}
&\frac{\partial V(x,t)}{\partial x}=-l_0 \frac{\partial I(x,t)}{\partial t} \label{eq:telegraph1}\\
&\frac{\partial I(x,t)}{\partial x}=-c_0 \frac{\partial V(x,t)}{\partial t}.
\label{eq:telegraph2}
\end{align}
Those equations imply that the voltage and current follow a one dimensional wave equation with a velocity $c=1/\sqrt{l_0 c_0}$:
\begin{align}
&\frac{\partial^2 V}{\partial x^2}-\frac{1}{c^2} \frac{\partial^2 V}{\partial t^2}=0 \label{eq:telegraphwavesV}\\
&\frac{\partial^2 I}{\partial x^2}-\frac{1}{c^2} \frac{\partial^2 I}{\partial t^2}=0 \label{eq:telegraphwavesI}
\end{align}
Both voltage and current can then be expressed as sum of forward and backward propagating modes:
\begin{align}
&V(x,t)=V_{\rightarrow}(x-ct)+V_{\leftarrow}(x+ct)\\
&I(x,t)=I_{\rightarrow}(x-ct)+I_{\leftarrow}(x+ct)
\end{align}
where we can decompose the waves on Fourier modes, for $F \in \{I,V\}$:
\begin{align}
&F_{\rightarrow}(x-ct)=\sum_{\omega > 0} \widetilde{F}_{\rightarrow}(\omega)e^{j(\omega(t-x/c))}+c.c\\
&F_{\leftarrow}(x+ct)=\sum_{\omega > 0} \widetilde{F}_{\leftarrow}(\omega)e^{j(\omega(t+x/c))}+c.c
\end{align}
In these equations, $c.c$ means complex conjugate. Using the telegraph equation \eqref{eq:telegraph2}, we realize that:
\begin{align}
&\widetilde{V}_{\rightleftarrows}(\omega)=\pm Z_0 \widetilde{I}_{\rightleftarrows}(\omega) \label{eq:relation_voltage_current_freq} \\
&V_{\rightleftarrows}(x-ct)=\pm Z_0 I_{\rightleftarrows}(x-ct)\\
&Z_0 \equiv \sqrt{\frac{l_0}{c_0}}
\end{align}
Basically, the right/left voltage moving waves have the same/opposite sign to the associated currents up to the quantity $Z_0$ called the impedance of the line, typically about $50 \Omega$. The power flow associated with the right or left moving waves satisfies \cite{clerk2010introduction}:
\begin{align}
P_{\rightleftarrows} \equiv V_{\rightleftarrows} I_{\rightleftarrows} = \pm \frac{V_{\rightleftarrows}^2}{Z_0}
\label{eq:power_wave_voltage_classical}
\end{align}
This result can be shown from an electromagnetic treatment of the lines (it comes from the Poynting vector integrated on the waveguide cross-section \cite{pozar2011microwave}).

From this, we understand that right moving waves will be associated with a positive power flow. More precisely, let's assume we are at the position $x$. A product $V_{\rightarrow} I_{\rightarrow} > 0$ represents a power "dissipated" to the $[x,+\infty]$ portion of the waveguide, while a negative product is a power generated by this same part. This interpretation comes from the convention of sign for voltage and currents as defined in Figure \ref{fig:lossless_transmission_line}. The reason for this interpretation, while having only non-dissipative elements in the lumped-element model, is that for an observer in $x$, a right moving wave at this position represents a power that will be going on the $[x,+\infty]$ portion of the waveguide and will thus be "lost" for the observer. An opposite interpretation holds for the left moving waves. 
\begin{figure}[h!]
\begin{center}
\includegraphics[width=0.5\textwidth]{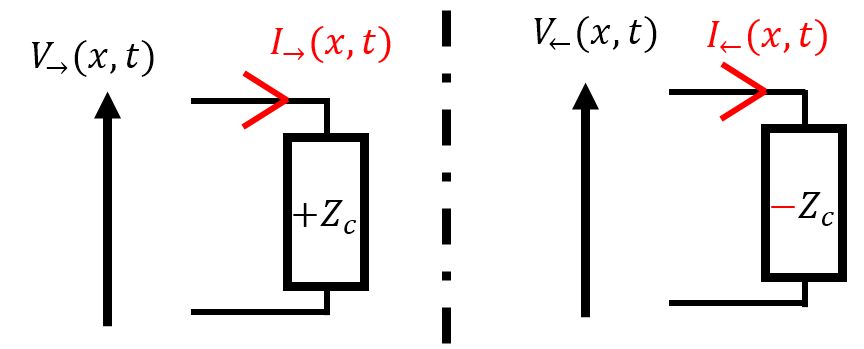}
\caption{Power flow in a waveguide. The waveguide can be seen as a resistor dissipating power for the right moving waves, while it is seen as a negative resistor for the left moving ones.}
\label{fig:diff_convention_power}
\end{center}
\end{figure}
\subsection{Quantum description and coupling to a qubit}
\subsubsection{Quantization of the line}
Now, we will go into the full quantum framework by considering the waveguide being quantized interacting with a superconducting qubit. 

We consider a semi infinite waveguide grounded in $x=0$. Signals are being injected from $x=-\infty$. In $x=x_k=-L$ a superconducting qubit, capacitively coupled to the waveguide has been put. The setup is illustrated on Figure \ref{fig:transmon_capacitively_coupled_semi_infinite_waveguide}.
\begin{figure}[h!]
\begin{center}
\includegraphics[width=0.8\textwidth]{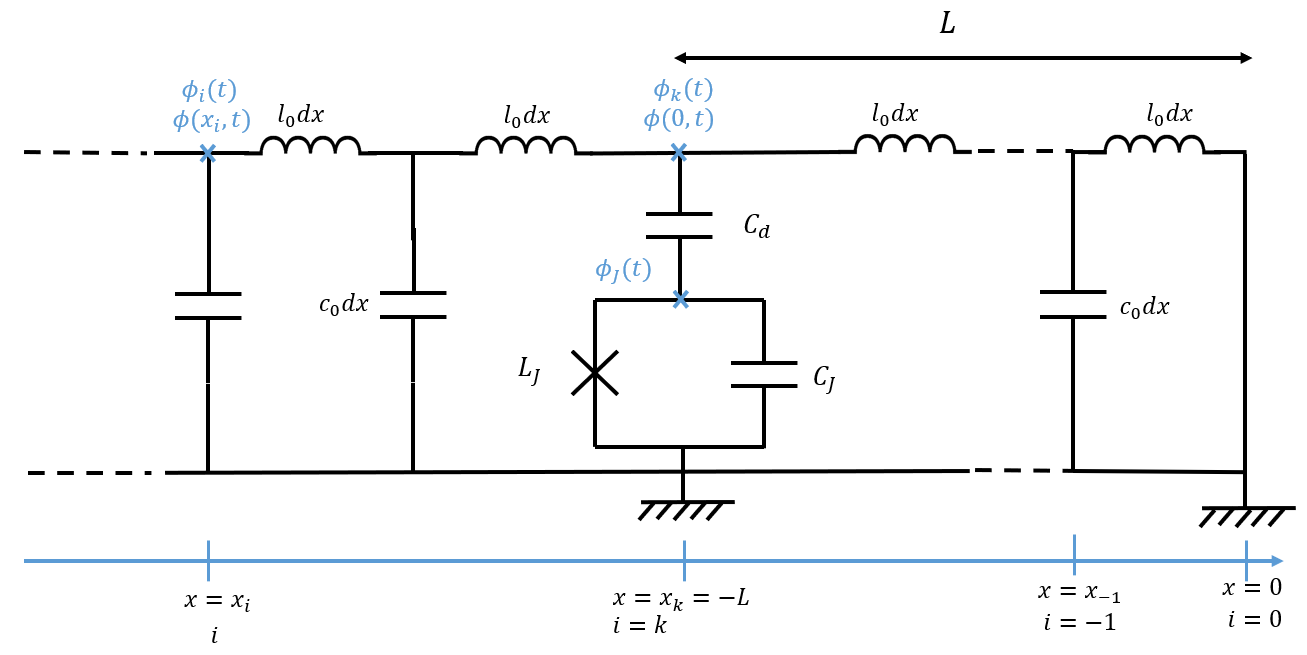}
\caption{A superconducting qubit is capacitively coupled to a semi-infinite waveguide. We first describe the physics with a discretized model, where the flux variable associated to the waveguide are indexed by integer numbers: $\{\phi_i(t)\}_{i=0,-1...-\infty}$ before considering the continuous version $\phi_i(t) \to \phi(x_i,t)$.}
\label{fig:transmon_capacitively_coupled_semi_infinite_waveguide}
\end{center}
\end{figure}

We first consider a discretized model. The Lagrangian of both the line and the superconducting qubit can be written as:
\begin{align}
\mathcal{L}=\sum_{i=-1, i \neq k}^{-\infty} \frac{1}{2} c_0 dx \dot{\phi}_i^2  +\frac{1}{2} C_d (\dot{\phi}_{k}-\dot{\phi}_{J}(t))^2+\frac{1}{2} C_J (\dot{\phi}_{J}(t))^2+E_J \cos\left(\frac{\phi_{J}}{\varphi_0}\right)- \sum_{i=-1}^{-\infty} \frac{1}{2 l_0 dx} (\phi_i-\phi_{i+1})^2
\end{align}
We deduce the canonical momentum associated to the coordinates $\{\phi_J,\phi_0,\phi_{-1},...\}$:
\begin{align}
&p_{J} = \frac{\partial \mathcal{L}}{\partial \dot{\phi}_{J}}=\dot{\phi}_{J}(C_J+C_d)-C_d \dot{\phi}_k\\
&p_k = \frac{\partial \mathcal{L}}{\partial \dot{\phi}_k}=C_d (\dot{\phi}_k-\dot{\phi}_J)\\
&p_{i<0,i \neq k}=c_0 dx \dot{\phi}_i,\\
\end{align}
and we can find the Hamiltonian:
\begin{align}
H \equiv &\sum_{i \leq 0} p_i \dot{\phi}_i + p_J \dot{\phi}_J-\mathcal{L} \notag \\
&= \sum_{i=-1, i \neq k}^{-\infty} \frac{1}{2 c_0 dx} p_i^2  +\frac{1}{2 C_d} p_k^2+\frac{1}{2 C_J}(p_k+p_J)^2-E_J \cos(\frac{\phi_{J}}{\varphi_0})+ \sum_{i=-1}^{-\infty} \frac{1}{2 l_0 dx} (\phi_i-\phi_{i+1})^2
\end{align}
In order to take the continuum limit, we define the generalized flux and momentum density fields as:
\begin{align}
&\phi(x_i,t) \equiv \phi_i(t)\\
&p(x_i,t) \equiv \frac{p_i(t)}{dx}.
\end{align}
However, we will keep the momentum $p_k$ as an independent variable not included in the field $p(x,t)$. Assuming a low waveguide impedance $Z_0$, and performing the continuous limit, the Hamiltonian describing the waveguide and qubit can be shown \cite{wiegand2020semiclassical} to be equivalent to: 
\begin{align}
H=\int_{-\infty}^0 dx \left(\frac{1}{2 c_0} p^2(x,t)+\frac{1}{2 l_0} (\partial_x \phi(x,t))^2 \right) + \frac{p_J^2}{2(C_d+C_J)}-E_J \cos(\frac{\phi_J}{\varphi_0})-\frac{C_d}{C_d+C_J}\frac{p(-L)}{c_0} p_J
\end{align}
And here we can identify the Hamiltonian of the waveguide, of the superconducting qubit and of the interaction between both:
\begin{align}
&H_{\text{Waveguide}}=\int_{-\infty}^0 dx \left(\frac{1}{2 c_0} p^2(x,0)+\frac{1}{2 l_0} (\partial_x \phi(x,0))^2 \right) \label{eq:Hwaveguide}\\
&H_{\text{J}}=\frac{p_J(0)^2}{2 C_{\Sigma}}-E_J \cos(\frac{\phi_J(0)}{\varphi_0}) \label{eq:HJ}\\
&H_{\text{int}}=-\frac{C_d}{C_{\Sigma}}\frac{p(-L,0)}{c_0}p_J(0) \label{eq:Hint}
\end{align}
Where we defined $C_{\Sigma}\equiv C_d+C_J$. We notice that we evaluated all the operators at the instant $t=0$ in this expression (the Hamiltonian of a closed system being time-invariant, it will give the appropriate dynamic).
Before quantizing, we can check that the Hamiltonian of the waveguide produces the wave equations \eqref{eq:telegraphwavesV}, and \eqref{eq:telegraphwavesI}. Indeed, $\mathcal{H}=\left(\frac{1}{2 c_0} p^2(x,t)+\frac{1}{2 l_0} (\partial_x \phi(x,t))^2 \right)$ corresponds to the Hamiltonian density, of a free propagating field \cite{peropadre2013scattering, goldstein2002classical}, leading to the equation of motion:
\begin{align}
\ddot{\phi}(x,t)-\frac{1}{c^2}\partial_x^2 \phi(x,t)=0
\label{eq:phi_free}
\end{align}
This equation implies eqs \eqref{eq:telegraphwavesV}, and \eqref{eq:telegraphwavesI} using the definitions $V(x,t)=\dot{\phi}(x,t)$, $I(x,t)=\partial_x \phi/l_0$ such that we find the results described in the previous part\footnote{$\partial_x \phi/l_0$ can be seen as a \textit{definition} of currents here as all the physics is described by the variable $\phi(x,t)$. Expressing it in terms of current needs to define within this context what the current is.}.

Now, as explained in the appendix \ref{app:quantized_fields} and \cite{wiegand2020semiclassical}, from the boundary condition $\dot{\phi}(0,t)=V(0,t)=0$, the field once quantized takes the expression, where we recognize stationnary modes:
\begin{align}
&\widehat{\phi}(x,t)=i \sqrt{\frac{\hbar Z_0}{\pi}}\int_{0}^{+\infty} \frac{d \omega}{\sqrt{\omega}}\sin(\frac{\omega x}{c}) \left(\widehat{b}(\omega) e^{-i \omega t}-\widehat{b}^{\dagger}(\omega) e^{i \omega t}  \right) \label{eq:phi_stat} \\
&\widehat{p}(x,t)=\sqrt{\frac{\hbar c_0}{c \pi}} \int_{0}^{+\infty} d \omega \sqrt{\omega} \sin(\frac{\omega x}{c}) \left(\widehat{b}(\omega) e^{ - i \omega_k t}-\widehat{b}^{\dagger}(\omega) e^{i \omega_k t} \right) \label{eq:p_stat}
\end{align}
The operators $\widehat{b}(\omega)$ here are in principle time dependent in the Heisenberg picture (because of the interaction with the transmon). They satisfy the bosonic commutation rules:
\begin{align}
&[\widehat{b}(\omega_1),\widehat{b}^{\dagger}(\omega_2)]=\delta(\omega_1-\omega_2)\\
&[\widehat{b}(\omega_1),\widehat{b}(\omega_2)]=[\widehat{b}^{\dagger}(\omega_1),\widehat{b}^{\dagger}(\omega_2)]=0
\end{align}

Finally, approximating the transmon as a harmonic oscillator and performing the rotating wave approximation \cite{haroche2006exploring}, as shown in the appendix \ref{app:waveguide_transmon_interaction}, the interaction can be described as:
\begin{align}
&H_{\text{Waveguide}}=\int_0^{+\infty} d\omega \hbar \omega \widehat{b}^{\dagger}(\omega) \widehat{b}(\omega) \label{eq:waveguide_final}\\
&H_{\text{J}}=\hbar \omega_0 \widehat{a}^{\dagger}_J \widehat{a}_J \label{eq:transmon_final}\\
&H_{\text{int}}=\int_0^{+\infty} d \omega \hbar g(\omega) (\widehat{a}_J e^{j \phi} \widehat{b}^{\dagger}(\omega)+\widehat{a}^{\dagger}_J \widehat{b}(\omega) e^{-j \phi}) \label{eq:int_final}
\end{align}
Where, here:
\begin{align}
&\phi=-\pi/2\\
&g(\omega)=\frac{C_d}{C_{\Sigma}} \sqrt{\frac{Z_0}{2 \pi L_J^0}} \sqrt{\frac{\omega}{\omega_0}} \sin(\frac{\omega L}{c})\\
\end{align}
In this section, we obtained a complete quantized description of a transmon (at this point approximated as a harmonic oscillator) capacitively coupled to a waveguide. 
\subsubsection{Input-output relations: dynamic of the qubit at $0$ temperature}
\label{sec:input_output}

Now, we are going to find the equation of motion of a driven superconducting qubit. We will then understand the origin of the noise in single-qubit operations. From now on, we will also remove the hats denoting the operators in order to simplify the notations. 

\textbf{Step 1: Approximating the Hamiltonian}

Our goal is now to find the equations of motion of the qubit being driven in the waveguide. We are also interested in finding the energetic cost in order to perform gate operations. One common way to access such quantities is via the input-output formalism. But to get to that point, we need to make a few approximations on the model. 

The first one is that we will approximate $g \equiv g(\omega) \approx g(\omega_0)$. This is called the Markov approximation \cite{gardiner1985input}. This approximation will be valid if one considers that the populated frequencies in the driving field are not spread too far away from the qubit frequency at $\omega_0$. More precisely, if the bandwidth $\Delta \omega$ of the injected signals satisfies $\Delta \omega/\omega_0 \ll 1$.

The second approximation will consist in extending the integral over frequencies down to $-\infty$, thus introducing "new" bosonic operators valid on negative frequencies. Such approximation will be correct as long as the populated frequency in the driving signal is high enough \cite{walls2007quantum}.

Finally, we will approximate the superconducting qubit by only considering its two first levels. It gives us the resulting Hamiltonian:

\begin{align}
&H=H_{\text{Waveguide}}+H_{\text{J}}+H_{\text{int}}\\
&H_{\text{Waveguide}}=\int_{-\infty}^{+\infty} d\omega \hbar \omega b^{\dagger}(\omega) b(\omega)\\
&H_{\text{J}}=-\frac{\hbar \omega_0}{2} \sigma_z\\
&H_{\text{int}}=\hbar g(\omega_0) \int_{-\infty}^{+\infty} d \omega (\sigma_J  b^{\dagger}(\omega)e^{i \phi}+\sigma_J^{\dagger} b(\omega)e^{-i \phi} )
\end{align}

\textbf{Step 2: solving the dynamic}

At this point, we won't make any further approximations and we will just solve the dynamic. First, we find the equation of motion for the field bosonic operators. In the Heisenberg picture, we get:
\begin{align}
\dot{b}(\omega,t)=\frac{1}{i \hbar}[b(\omega,t),H]=-i \omega b(\omega,t)-i g \sigma_J(t) e^{i \phi}
\label{eq:bath_operator_t_continuum}
\end{align}
This equation is a first order differential one. Postulating a solution: $b(\omega,t)=C(t) a^{-i \omega t}$, we can find $C(t)$ that solves the equation (this is known as the variation of constants method). We get in the end:
\begin{align}
b(\omega,t)=b(\omega,0)e^{-i \omega t}-ig \int_{0}^t \sigma_J(t') e^{-i \omega(t-t')} e^{i \phi} dt'
\end{align}
We now make the same resolution but for a given system operator: $O_S(t)$. We have:
\begin{align}
\dot{O_S}(t)=\frac{1}{i \hbar}[O_S,H]=\frac{1}{i \hbar}[O_S(t),H_J]-i g(\omega_0)\int_{-\infty}^{+\infty} d \omega e^{-i \phi} [O_S(t),\sigma_J^{\dagger}(t)] b(\omega,t)+e^{i \phi} b^{\dagger}(\omega) [O_S(t),\sigma_J(t)]
\end{align}
It finally gives us:
\begin{align}
&\dot{O_S}(t)=\frac{1}{i \hbar}[O_S(t),H_J]-i g(\omega_0)\int_{-\infty}^{+\infty} d \omega  \left(e^{-i \phi} [O_S(t),\sigma_J^{\dagger}(t)] b(\omega,0)e^{-i \omega t} +e^{i \phi} b^{\dagger}(\omega,0)e^{+i \omega t}[O_S(t),\sigma_J(t)]\right)\notag \\
&+g^2\int_{-\infty}^{+\infty} d \omega  \left(  \left(\int_{0}^t \sigma_J^{\dagger}(t') e^{+i \omega(t-t')}  dt'[O_S(t),\sigma_J(t)]\right)-\left( \int_{0}^t [O_S(t),\sigma_J^{\dagger}(t)] \sigma_J(t') e^{-i \omega(t-t')}  dt'  \right) \right)
\end{align}
At this point, we define the input field $b_{in}(t)$:
\begin{align}
b_{in}(t)=\frac{1}{\sqrt{2 \pi}}\int_{-\infty}^{+\infty} d \omega b(\omega,0) e^{-i \omega t}
\end{align}
This operator physically represent the free evolution of all the bosonic operators between $t=0$ and the instant $t$. This is why we call it the input field: as it corresponds to the evolution of the operator as if no interaction occured, it will correspond to what has been injected in the waveguide, thus the "input". Using this operator, we have:
\begin{align}
&\dot{O_S}(t)=\frac{1}{i \hbar}[O_S(t),H_J]-i g(\omega_0)\sqrt{2 \pi}\left(e^{-i \phi} [O_S(t),\sigma_J^{\dagger}(t)] b_{in}(t) +e^{i \phi} b^{\dagger}_{in}(t)[O_S(t),\sigma_J(t)]\right)\notag \\
&+\pi g^2 \left(  \sigma_J^{\dagger}(t) [O_S(t),\sigma_J(t)]-[O_S(t),\sigma_J^{\dagger}(t)] \sigma_J(t)  \right)
\end{align}

\textbf{Step 4: recognizing Bloch equations}

Finally, solving this equation for $O_S=\sigma$ or $O_S=A_{11} \equiv \ketbra{1}{1}$, we get:
\begin{align}
&\dot{\sigma}(t)=-i \omega_0 \sigma-ig(\omega_0)\sqrt{2 \pi} \sigma_z e^{-i \phi} b_{in} - \pi g^2 \sigma\\
&\dot{A}_{11}(t)=ig \sqrt{2 \pi} (b_{in}(t)^{\dagger} e^{i \phi} \sigma - e^{-i \phi} b_{in}(t) \sigma^{\dagger})-2 \pi g^2 A_{11}(t)
\end{align}
And, taking the average of those quantities, assuming a coherent state at frequency $\omega_0$ for the drive: $\ket{\alpha_{\omega_0}}$, we get\footnote{We assume to simplify that $\alpha_{\omega_0}>0$: the coherent state has "no phase".}:
\begin{align}
&\dot{\rho}_{10}=-i \omega_0 \rho_{10}-ig(\omega_0) \alpha_{\omega_0}e^{-j (\omega_0 t+\phi)}(\rho_{00}-\rho_{11})  - \pi g^2 \rho_{10}\\
&\dot{\rho}_{11}=ig(\omega_0) (\alpha_{\omega_0}e^{j (\omega_0 t+\phi)} \rho_{10} - \alpha_{\omega_0}e^{-j (\omega_0 t+\phi)} \rho_{01})-2 \pi g^2 \rho_{11}
\end{align}
We recognize that those equations have the same structure as the Bloch optical equations of a two-level system in contact with a bath at 0 temperature, classically driven by a field inducing rotation around the axis $\vec{n}=(\cos(\phi),\sin(\phi),0)$ of the Bloch sphere (it is represented on the figure \ref{fig:bloch_sphere}), at the Rabi frequency $\Omega$ \cite{cohen1998atom}:
\begin{align}
&\frac{\partial \rho}{\partial t}=\frac{1}{i \hbar}[H_S,\rho]+\gamma D[\sigma](\rho)\\
& D[\sigma](\rho) \equiv \left(\sigma \rho \sigma-\frac{1}{2}(\sigma^{\dagger} \sigma \rho + \rho \sigma^{\dagger} \sigma) \right)\\
& H_S \equiv -\frac{\hbar \omega_0}{2} \sigma_z + \frac{\hbar \Omega}{2}\left(\sigma e^{i \phi} e^{i \omega_0 t} + \sigma^{\dagger} e^{-i \phi} e^{-i \omega_0 t} \right)
\end{align}
Indeed, such equation admit the solution:
\begin{align}
\dot{\rho}_{10}&=-i \omega_0 \rho_{10} -\frac{i\Omega}{2}e^{-i(\phi+\omega_0 t)}( \rho_{00}-\rho_{11})-\frac{\gamma}{2} \rho_{10} \label{eq:drho10dt} \\
\dot{\rho}_{11}&=\frac{i \Omega}{2}\left( \rho_{10} e^{i \phi} e^{i \omega_0 t}-e^{-i \phi} e^{-i \omega_0 t} \rho_{01} \right)- \gamma \rho_{11} \label{eq:drho11dt}
\end{align}
The identification gives us:
\begin{align}
&g(\omega_0)=\sqrt{\frac{\gamma}{2 \pi}}\\
&\Omega=2 g(\omega_0) \alpha_{\omega_0}=\sqrt{\frac{2 \gamma}{\pi}} \alpha_{\omega_0}
\label{eq:relation_Omega_alpha}
\end{align}

In the end, here, we described the equation of motion of a superconducting qubit interacting with a continuum of bosonic modes at zero temperature and being driven by a coherent state resonant at its frequency.
\subsubsection{Dynamic in the presence of thermal noise}
\label{sec:dynamic_in_presence_thermal_noise}
What we presented so far assumes that the only reason why qubits are noisy comes from spontaneous emission. In practice, the field state that is realizing their evolution might contain noise that will perturb their final states. Calling $\overline{n}_{\text{tot}}$ the number of noisy thermal photons\footnote{In some cases, $\overline{n}_{\text{tot}}=n_{BE}(T_{\text{Q}})$, $T_{\text{Q}}$ being the qubit temperature, and $n_{BE}(T_{\text{Q}})$ the Bose einstein population at the qubit temperature \cite{cohen1998mecanique} but we will see in the section \ref{sec:single_qubit_example} that if the signals driving the gate is generated in the laboratory, additional thermal photons might be present}, their evolution can be modeled as \cite{haroche2006exploring}:
\begin{align}
&\frac{\partial \rho}{\partial t}=\frac{1}{i \hbar}[H,\rho]+\gamma_{\text{sp}} \overline{n}_{\text{tot}} D[\sigma^{\dagger}](\rho) + \gamma_{\text{sp}}(\overline{n}_{\text{tot}}+1) D[\sigma](\rho) \label{eq:evolution_noisy_gate_with_thermal_noise}
\end{align}
We see that the presence of noise in the line coupled to the qubit will induce additional relaxation of the excited state to the ground ($D[\sigma]$ is now multiplied by $\gamma_{\text{sp}}(\overline{n}_{\text{tot}}+1)$ instead of $\gamma_{\text{sp}}$), and there is also an additional term that will excite the qubit from the ground to the excited state: $\gamma_{\text{sp}} \overline{n}_{\text{tot}}$. The reason why we say that it might excite the qubit can be understood by finding the equation of motion for the coefficients $\rho_{10}$ and $\rho_{11}$ associated to \eqref{eq:evolution_noisy_gate_with_thermal_noise}. We would find that $\dot{\rho}_{11}$ in \eqref{eq:drho11dt} would now contain an additional term $+\gamma_{\text{sp}} \overline{n}_{\text{tot}} \rho_{00}$ describing the fact that population in the ground state might be transfered to the excited state because of the presence of noise.
\section{The models we use in the rest of the Ph.D.}
Here, we explain how we will model the noise affecting the qubits, the power consumption required by the gates, and we will provide the typical characteristics for the qubits and gates we are going to consider. Because a Ph.D. is a long-term project, the models used in the section \ref{sec:physical_model_chap3} of the third chapter will be \textit{very slightly} different from the ones we present here, which corresponds to the models used in the fourth and fifth chapter of this thesis. Appropriate minor comments will be made in \ref{sec:physical_model_chap3} to explain in what the physics is modeled in a slightly different manner.

The ending goal of this Ph.D. being to have a complete description of a superconducting quantum computer in order to have access to general trends in the power consumption of fault-tolerant quantum computing; we will keep simple models describing the gate physics, allowing us to access those general trends. It implies, among other things, that the only reason why the gates will be noisy in our models will be because of the limited qubit lifetime and thermal excitations.
\subsection{Performance of the operations}
\subsubsection{Noise models and infidelities for the gates}
\label{sec:infidelity_operations}
From \eqref{eq:evolution_noisy_gate_with_thermal_noise}, we can, in principle, estimate the "quantity of noise" that will be introduced by the evolution. We will, however, slightly simplify the description\footnote{In order to have analytical results for the worst-case infidelity of the operation for an arbitrary gate duration $\tau$, we must slightly simplify \eqref{eq:evolution_noisy_gate_with_thermal_noise}.}. We will consider that the noise introduced by the different gates used can be modeled by considering that first, each qubit evolves for a duration $\tau$ through the equation \eqref{eq:evolution_noisy_gate_with_thermal_noise_NO_HAMILTONIAN}, and that the "perfect" unitary implementation of the gate is applied afterward\footnote{In the language of the quantum channels that we introduce more properly in \ref{sec:tools_describing_noise}, we model the total evolution of a noisy gate that tries to implement a unitary quantum channel $\mathcal{U}$ as being $\mathcal{U} \circ \mathcal{N}$, where $\mathcal{N}$ is the evolution associated to \eqref{eq:evolution_noisy_gate_with_thermal_noise_NO_HAMILTONIAN} integrated for the gate duration.}.
\begin{align}
&\frac{\partial \rho}{\partial t}=\gamma_{\text{sp}} \overline{n}_{\text{tot}} D[\sigma^{\dagger}](\rho) + \gamma_{\text{sp}}(\overline{n}_{\text{tot}}+1) D[\sigma](\rho) \label{eq:evolution_noisy_gate_with_thermal_noise_NO_HAMILTONIAN}
\end{align}

In sections \ref{sec:single_qubit_example} and \ref{sec:energetic_cost_nisq_algo}, the noise will be quantified with the infidelity introduced by such evolution\footnote{The worst-case or average infidelity of a noisy operation followed by a perfect one (i.e., unitary in this context) being the same as the worst-case or average infidelity of the noisy operation alone, we can focus on the noisy part of the evolution in the reasoning.}. In order to calculate it, we need to explain what the fidelity between quantum states is. Let us consider that we have two quantum states. The first one is described by a pure density matrix $\ketbra{\psi}{\psi}$ and the second one by a mixed one $\rho$. We assume that $\ketbra{\psi}{\psi}$ was the ideal state we would like to have for our system, and $\rho$ is the state that has actually been prepared. The fidelity between $\rho$ and $\ketbra{\psi}{\psi}$ is defined as: 
\begin{align}
F \equiv \bra{\psi} \rho \ket{\psi}.
\end{align}
It is possible to show \cite{nielsen2002quantum} that $0 \leq F \leq 1$ where $F=1$ is reached when $\rho=\ketbra{\psi}{\psi}$. Because we will be interested in evaluating the quantity of noise in a quantum state, we prefer to define the \textit{infidelity} of a quantum state as being the quantity $IF \equiv 1-F$ (an infidelity being equal to $0$ means that the two states are identical and the closer the infidelity will be to $1$, the noisier $\rho$ will be).

Now that the definition for states has been provided, we can extend them to evolutions. Here, we define two infidelities for the evolution \eqref{eq:evolution_noisy_gate_with_thermal_noise_NO_HAMILTONIAN}. We consider first the maximum infidelity this evolution can induce. It means that we calculate the infidelity between an arbitrary initial state and the final state after having evolved through \eqref{eq:evolution_noisy_gate_with_thermal_noise_NO_HAMILTONIAN} during a duration $\tau$. By varying the initial preparation, we can find which one induces the biggest infidelity. It is the value we consider for maximum infidelity. We will do the same calculation but for the average infidelity (this time, we consider a uniform distribution over the Bloch sphere for the initial state, we compute the infidelity for each of those preparations, and we consider the average). Doing the calculation properly, we can show that the maximum and average infidelities satisfy:
\begin{align}
&\max(IF) \equiv \max(1-F)=(1+\overline{n}_{\text{tot}})\gamma_{\text{sp}}\tau \label{eq:worst_case_IF}\\
&\overline{IF}=1-\overline{F} =  \frac{1+2 \overline{n}_{\text{tot}}}{3}\gamma_{\text{sp}}\tau \label{eq:avg_case_IF}
\end{align}
\subsubsection{Energetic cost to perform a gate}
\label{sec:energetic_cost_to_perform_a_gate}
The equation \eqref{eq:power_wave_voltage_classical} gave us the expression of the power of right moving waves in the classical regime. The corresponding quantity in the quantum regime is \cite{cottet2017observing,monsel2020energetic} (see also the appendix \ref{app:injected_power}):
\begin{align}
P =\hbar \omega_0 \langle b_{in}^{\dagger}(t) b_{in}(t) \rangle=\hbar \omega_0\langle b_{in}^{\dagger}(0) b_{in}(0) \rangle
\end{align}
For a coherent state $\ket{\alpha_{\omega_0}}$, $\alpha_{\omega_0} > 0$ injected, we find:
\begin{align}
P=\frac{\hbar \omega_0}{2 \pi} \alpha_{\omega_0}^2
\end{align}
It allows us to find the relationship between the power injected and the Rabi frequency. Using \eqref{eq:relation_Omega_alpha}, we find:
\begin{align}
P=\frac{\hbar \omega_0}{4 \gamma_{\text{sp}}}\Omega^2
\label{eq:power_fct_rabi}
\end{align}
This relation will be central for the work that we are going to present in the next chapters. It corresponds to the power required to inject a signal that will drive a single qubit gate as a function of the qubit-waveguide coupling $\gamma_{\text{sp}}$ and the Rabi frequency of the gate $\Omega$. In principle, this power could be recovered after the signal has interacted with the qubit, but we will see in the two last chapters (for instance in section \ref{sec:single_qubit_example}) that because the signals are attenuated, it will not be possible. For this reason, this power will be one of the important elements that will play a role in the energetic cost of quantum computing.
\subsection{Parameters we will consider for the superconducting qubits}
\label{sec:state_of_the_art}
\subsubsection{Qubit characteristics}
\label{sec:state_of_the_art_qubit}
Now, we can give some order of magnitude of the parameters involved to describe superconducting qubits. First, a typical frequency for transmon qubits can be taken as $\omega_0/2 \pi \approx 6 GHz$ \cite{krantz2019quantum,kjaergaard2020superconducting}. Then, the coherence time of the qubits (at a temperature close to $0K$, in practice around $10mK$ \cite{krantz2019quantum}) have greatly evolved in the last years \cite{kjaergaard2020superconducting}. The decoherence times for transmon qubits are getting very close to the millisecond. In recent experiments, decoherences times being about $0.3 ms$ \cite{place2021new} have been shown. Other types of qubits such as the fluxonium are even passing the barrier of the millisecond \cite{somoroff2021millisecond}.

In our models, we assumed so far that the only reason why qubits are noisy is because of spontaneous emission and thermal photons, as one can see from \eqref{eq:evolution_noisy_gate_with_thermal_noise} and \eqref{eq:evolution_noisy_gate_with_thermal_noise_NO_HAMILTONIAN}. This is, of course, an idealistic approximation: we consider that the effect of pure dephasing is negligible compared to spontaneous emission. Under this assumption, we can estimate that our values of $\gamma_{\text{sp}}$ can typically be given by the inverse of the qubit lifetime. Considering the fact that state-of-the-art qubits are close to the millisecond coherence time, it will allow us to consider $\gamma_{\text{sp}} \sim 1kHz$ in our calculations.

In summary, the typical characteristics we will consider for the qubits in our models are:
\begin{itemize}
\item Qubit frequency: $\omega_0/2 \pi = 6 GHz$.
\item Qubit-waveguide coupling: $\gamma_{\text{sp}}=1kHz$.
\end{itemize} 
\subsubsection{Single qubit gates}
\label{sec:state_of_the_art_1qb}
There are other parameters to fix, such as the duration of the single-qubit gate $\tau$. Single qubit gate can have a typical duration in the $10 ns$ range \cite{almudever2017engineering}, even though faster gates based on optimized pulse techniques can in principle be implemented \cite{werninghaus2021leakage}. In this thesis, we will take the reference value of $\tau_{\text{1qb}}=25ns$. The reason why we take $25ns$ and not $10ns$ is that later on, we are going to include engineering aspects in the quantum computer, which are designed for this typical gate duration of $25ns$. The power required to implement a single qubit $\pi$-pulse will then correspond to \eqref{eq:power_fct_rabi} where $\Omega=\pi/\tau_{\text{1qb}}$. To simplify the discussions, we will consider that any single-qubit gate will last for $\tau_{\text{1qb}}=25ns$ and that they will consume the same amount of power as a single qubit $\pi$-pulse.

In summary, the typical characteristics for the single-qubit gates we will consider in our models are:
\begin{itemize}
\item Gate duration (for any single qubit gate): $\tau_{\text{1qb}}=25ns$
\item Power consumption of the single qubit gate: given in \eqref{eq:power_fct_rabi} with $\Omega=\pi/\tau_{\text{1qb}}$ (all gates consume as much as single qubit $\pi$-pulse)
\item Noise model: \eqref{eq:evolution_noisy_gate_with_thermal_noise_NO_HAMILTONIAN} integrated for the duration $\tau_{\text{1qb}}=25ns$.
\end{itemize} 
\subsubsection{two qubit gates}
\label{sec:state_of_the_art_2qb}
Different ways to perform two-qubit gates exist \cite{huang2020superconducting}. For instance, a two-qubit gate between qubits $A$ and $B$ can be implemented by tuning the frequency of a coupler qubit $C$. When $C$ is put at the appropriate frequency, $A$ and $B$ are having a mediated interaction through $C$ \cite{yan2018tunable,xu2020high}. There are also schemes in which the frequencies of the qubits are unchanged. One such scheme is the cross-resonance gate \cite{chow2011simple,sheldon2016procedure} which allows making two qubits $A$ and $B$ interact by sending a microwave pulse on qubit $A$ at the frequency of qubit $B$, the interaction being mediated by a bus. We can also cite \cite{bertet2006parametric,mckay2016universal}. 

In this thesis, we will consider using in our models the cross-resonance scheme. This scheme can, in principle, allow making two-qubit interaction between any pair of qubits connected to a bus \cite{chow2011simple} while having fixed-frequency qubits. Those gates are typically longer than the single-qubit gates, but as the qubits are fixed in frequency, the qubits lifetime can also be longer in principle. In recent proposal, the gate duration was getting close to $100ns$ \cite{sheldon2016procedure}. This is the typical duration we are going to consider for cNOT gates in the rest of the thesis. We will consider that they will induce noise on each of the two qubits involved in the dynamic in such a way that each qubit will have a noise model described by \eqref{eq:evolution_noisy_gate_with_thermal_noise_NO_HAMILTONIAN} (integrated for $\tau_{\text{cNOT}}=100ns$). At the moment this thesis is being written, the dominant source of noise for such gates does not come from the sole presence of spontaneous emission (and thermal noise) affecting the qubits. However, it is toward what the community is trying to tend, and we will assume it is the case in our model. Anyway, the two qubit-gates are too noisy to be able to do fault-tolerance\footnote{At least for the kind of code we will use in the rest of this thesis.} today, so we have to make some optimistic assumptions about their performances in the future. Here we make the "bet" that in the near future, the fidelity of those gates will mainly be given by the "intrinsic" qubit lifetime. Finally, we will also model their energetic cost by assuming that it is comparable to the cost of a single qubit $\pi$-pulse.

In summary, the typical characteristics for the single qubit gates we will consider in our models are:
\begin{itemize}
\item Two-qubit gate duration (for a cNOT): $\tau_{\text{cNOT}}=100ns$
\item Power consumption of the gate: same power than a single qubit $\pi$-pulse described in \ref{sec:state_of_the_art_1qb}.
\item Noise model: the noise for \textit{each} of the two qubit involved in the gate can be modelled by  \eqref{eq:evolution_noisy_gate_with_thermal_noise_NO_HAMILTONIAN} integrated for the duration $\tau_{\text{cNOT}}=100ns$. The infidelity of the gate will thus be equal to $2 \times IF$ for IF corresponding to either worst-case or average infidelity as defined in \eqref{eq:worst_case_IF} or \eqref{eq:avg_case_IF}\footnote{For $\tau=\tau_{\text{cNOT}}$ in those equations.}.
\end{itemize} 
\section{Summary}
In this chapter, we explained the basic principle behind the quantization of electrical circuits. We also explained how superconducting qubits are designed, and driven, first by a classical modelization of the driving signal, neglecting any phenomenon of propagations. Then in a fully quantized description where both the fields and the transmon are treated as quantum objects. In this quantum description, we explained that because the transmon is coupled to a continuum of mode, by construction, it will experience noise known as spontaneous emission. Finally, we provided the quantity that one needs to access in order to estimate the energetic cost of single-qubit gate operations.
\begin{appendices}
\chapter{Injected power}
\label{app:injected_power}
The goal is to show that the average injected power in the quantum regime satisfies:
\begin{align}
&\langle P_{in} \rangle = \hbar \omega_0 \langle b_{in}^{\dagger} b_{in} \rangle\\
&b_{in}(t) \equiv \frac{1}{\sqrt{2 \pi}} \int_{- \infty}^{+ \infty} d \omega \ b(\omega) e^{-i \omega t}
\end{align}
From \eqref{eq:power_wave_voltage_classical} we know that the power of the right moving waves in the classical regime reads:
\begin{align}
P_{\rightarrow}(x-ct)=\frac{(V_{\rightarrow}(t-x/c))^2}{Z_0}
\end{align}
We also recall that the voltage is the derivative of the generalized flux so that the quantum operator describing the power of the right moving waves reads:
\begin{align}
P_{\rightarrow}(t-x/c)=\frac{(\dot{\phi}_{\rightarrow}(t-x/c))^2}{Z_0}
\label{eq:power_right_moving_0}
\end{align}
The expression of $\dot{\phi}_{\rightarrow}(t-x/c)$ can be found from \eqref{eq:app_phi_omega} by only keeping the right moving waves (i.e we decompose $\sin(\omega x/c)$ in a sum of exponentials, and we keep the part of $\phi(x,t)$ composed of exponentials of the argument $\pm(wt-x/c)$). Applying the approximations described in section \ref{sec:input_output} (Markov approximation, and extending frequency range toward $-\infty$), we find the expression of $\dot{\phi}_{\rightarrow}(t-x/c)$, and by injecting it in \eqref{eq:power_right_moving_0}, we have:
\begin{align}
P_{\rightarrow}(t-x/c) &= A+B\\
 & A \equiv \frac{\hbar \omega_0}{4 \pi} \int_{-\infty}^{+\infty} d \omega_1 \int_{-\infty}^{+\infty} d \omega_2 \notag \\
&
 -(b(\omega_1)b(\omega_2)e^{-j(\omega_1+\omega_2)(t-x/c)}+b^{\dagger}(\omega_1)b^{\dagger}(\omega_2)e^{j(\omega_1+\omega_2)(t-x/c)}) \\
 & B \equiv \frac{\hbar \omega_0}{4 \pi} \int_{-\infty}^{+\infty} d \omega_1 \int_{-\infty}^{+\infty} d \omega_2 \notag \\
 &+b^{\dagger}(\omega_1)b(\omega_2)e^{j(\omega_1-\omega_2)(t-x/c)}+b(\omega_1)b^{\dagger}(\omega_2)e^{j(\omega_1-\omega_2)(t-x/c)}, \notag \\
 \end{align}
We have, using the commutation relation for the bosonic operators:
 \begin{align}
 B&=\hbar \omega_0 \int_{-\infty}^{+\infty} \frac{d \omega_1}{\sqrt{2 \pi}} b^{\dagger}(\omega_1) e^{i \omega_1 (t-x/c)} \int_{-\infty}^{+\infty} \frac{d \omega_2}{\sqrt{2 \pi}} b_{\rightarrow}(\omega_2) e^{-i \omega_2 (t-x/c)}+\Delta\\
 &= \hbar \omega_0 b^{\dagger}_{in}(t-x/c) b_{in}(t-x/c)+\Delta
 \end{align}
where $\Delta$ is a divergence which typically occurs when working with a continuum of modes\footnote{I did not find a source properly explaining why $\langle P_{in} \rangle = \hbar \omega_0 \langle b_{in}^{\dagger} b_{in} \rangle$ even though it is the correct expression \cite{cottet2017observing,monsel2020energetic}. I am then not entirely sure of the validity of this passage, this appendix is my own attempt to understand how to derive properly $\langle P_{in} \rangle = \hbar \omega_0 \langle b_{in}^{\dagger} b_{in} \rangle$}, which can in principle be removed via renormalization techniques \cite{schwartz2014quantum}. Thus, we won't consider it in the expression of power.
 
Assuming a coherent state injection, we would have:
 \begin{align}
 \langle P_{\rightarrow} \rangle = \bra{\alpha_{\omega_0}} A \ket{\alpha_{\omega_0}} + \bra{\alpha_{\omega_0}} B \ket{\alpha_{\omega_0}}
 \end{align}
 Where:
 \begin{align}
 \bra{\alpha_{\omega_0}} A \ket{\alpha_{\omega_0}}=\frac{\hbar \omega_0}{4 \pi} (-\alpha_{\omega_0}^2 e^{-2i \omega_0 t}-(\alpha_{\omega_0}^*)^2 e^{2i \omega_0 t})
 \end{align}
If we are only interested in the average power, this term will time-average to $0$. We also notice that for a coherent state, $\langle P_{\rightarrow}(t-x/c) \rangle=\langle P_{\rightarrow}(0) \rangle$ 
 
In the end, if we are interested in the time (and quantum) averaged power of the waves injected, for monochromatic (or close to monochromatic) signals, the appropriate quantity to consider is $\hbar \omega_0 \langle b_{in}^{\dagger}(0) b_{in}(0) \rangle$.
\chapter{Superconducting qubit}
\section{Superconducting qubit frequency}
\label{app:superconducting_qubit_frequency}
Here, we show how to find \eqref{eq:omega01} and \eqref{eq:omega12}. This calculation is done by applying first order perturbation theory \cite{cohen1998mecanique}. Given an Hamiltonian $H=H_0+W$, the eigenvalues of $H$ will correspond to the eigenvalues of $H_0$ up to a correction provided by $W$. More precisely, if $\ket{E_n}$ is an eigenvalue of $H_0$ for the non degenerated eigenvalue $E_n$, in the presence of $W \ll H_0$, this eigenvalue will be modified in such a way that $E_n \to E_n + \bra{n} W \ket{n}$.

We apply this method here for $H_0=\hbar \omega_0 a^{\dagger} a$ and $V$ defined in \eqref{eq:potential_V}. Calling $\ket{n}$ the eigenstates of $H_0$ of eigenvalue $n \hbar \omega_0$, we then need to calculate $\bra{n}V\ket{n}$. In order to do it, we start by calculating $(\widehat{a}+\widehat{a}^{\dagger})^2\ket{n}$.
\begin{align}
(\widehat{a}+\widehat{a}^{\dagger})^2\ket{n}&=(\widehat{a}+\widehat{a}^{\dagger}) \left(\sqrt{n} \ket{n-1}+\sqrt{n+1} \ket{n+1} \right) \notag \\
&=\left(\sqrt{n(n-1)} \ket{n-2}+(2n+1) \ket{n}+\sqrt{(n+1)(n+2)} \ket{n+2}\right)
\end{align}
We deduce:
\begin{align}
&\bra{n}V\ket{n}=-\frac{E_C}{12} \bra{n}(\widehat{a}+\widehat{a}^{\dagger})^4\ket{n}=-\frac{E_C}{12}\left(n(n-1)+(2n+1)^2+(n+1)(n+2)\right)\\
&\bra{n+1}V\ket{n+1}-\bra{n}V\ket{n}=-E_C (n+1)
\end{align}
Because of this perturbation, the energy gap between two consecutive level then reads:
\begin{align}
&\hbar \omega_{n,n+1} =  \hbar \omega_0 + (\bra{n+1}V\ket{n+1}-\bra{n}V\ket{n}) = \hbar \omega_0 - E_C(n+1)
\end{align}
\section{Superconducting qubit in waveguide}
\subsection{Expression of the quantized fields}
\label{app:quantized_fields}
Here, we show how the expression of the field has been found. We will start by expressing the free propagating fields without the short circuit boundary condition $\dot{\phi}(0,t)=0$. And we will then impose this boundary condition to find the resulting field. The classical solutions to the wave equation \eqref{eq:phi_free}, are:
\begin{align}
&\phi(x,t)=\phi_{\rightarrow}(x,t)+\phi_{\leftarrow}(x,t)\\
&p(x,t)=p_{\rightarrow}(x,t)+p_{\leftarrow}(x,t)
\end{align}
Where $\phi_{\rightarrow}(x,t)$ (resp $\phi_{\leftarrow}(x,t)$) represent the forward (resp backward) propagating generalized flux waves. Same principle for $p_{\rightleftarrows}(x,t)$. Their expression are, with $\omega_k=c|k|$:
\begin{align}
&\phi_{\rightleftarrows}(x,t)=\int_0^{+\infty} d k \left( \phi_{\rightleftarrows}(k) e^{i(\pm kx-\omega_k t)} + \phi^*_{\rightleftarrows}(k) e^{-i(\pm kx-\omega_k t)} \right)\\
& p_{\rightleftarrows}(x,t)= c_0 \dot{\phi}_{\rightleftarrows}(x,t) = -c_0 i \int_0^{+\infty} d k \omega_k \left(\phi_{\rightleftarrows}(k) e^{i(\pm kx-\omega_k t)} - \phi^*_{\rightleftarrows}(k) e^{-i(\pm kx-\omega_k t)} \right)
\end{align}
But to make calculation simpler in what follows, we will express the field in the following way:
\begin{align}
&\phi(x,t)=\int_{-\infty}^{+\infty} d k \  \phi(k) e^{i( kx-\omega_k t)}+\phi^*(k) e^{-i( kx-\omega_k t)}\\
&p(x,t)=c_0 i \int_{-\infty}^{+\infty} d k (-\omega_k) \left(\phi(k) e^{i( kx-\omega_k t)}-\phi^*(k) e^{-i( kx-\omega_k t)} \right)
\end{align}
Where we defined: $\phi(k>0) \equiv \phi_{\rightarrow}(k)$ and $\phi(k<0) \equiv \phi_{\leftarrow}(-k)$. Thus, the positive wavector represent forward propagating waves and the negative backward ones.
Following the approach from \cite{clerk2010introduction}, one can check that we have the following relations:
\begin{align}
&\phi(k)=\frac{1}{4 \pi} \int_{-\infty}^{+\infty} dx e^{-ikx} (\phi(x,0)+\frac{i}{c_0 \omega_k} P(x,0))\\
&\phi^*(k)=\frac{1}{4 \pi} \int_{-\infty}^{+\infty} dx e^{+ikx} (\phi(x,0)-\frac{i}{c_0 \omega_k} P(x,0))
\end{align}
We start to quantize the free field without boundary conditions. In this case, we will promote $\phi(x,t)$ and $p(x,t)$ into operators imposing the commutation relation:
\begin{align}
[\widehat{\phi}(x,t),\widehat{p}(y,t)]=i \hbar \delta(x-y)
\end{align}
From this, we can find the commutation relation satisfied between $\widehat{\phi}(k)$ and $\widehat{\phi}^{\dagger}(k)$. We have:
\begin{align}
[\widehat{\phi}(k_1),\widehat{\phi}^{\dagger}(k_2)]&=\frac{1}{(4 \pi)^2} \int_{-\infty}^{+\infty} dx \int_{-\infty}^{+\infty} dy e^{-i k_1 x + i k_2 y} \left(\frac{-i}{c_0 \omega_{k_2}}[\phi(x,0),P(y,0)]+\frac{i}{c_0 \omega_{k_1}}[P(x,0),\phi(y,0)] \right)\\
&=\frac{1}{(4 \pi)^2} \int_{-\infty}^{+\infty} dx \int_{-\infty}^{+\infty} dy \frac{1}{c_0} e^{-i k_1 x + i k_2 y} \delta(x-y)( \frac{\hbar}{\omega_{k_1}}+\frac{\hbar}{\omega_{k_2}}) \\
&=\frac{1}{(4 \pi)^2} \int_{-\infty}^{+\infty} dx \frac{2 \hbar}{c_0 \omega_k} e^{-i(k_1-k_2)x}\\
&=\frac{1}{(4 \pi)^2} \int_{-\infty}^{+\infty} dx \frac{2 \hbar}{c_0 \omega_k} e^{-i(k_1-k_2)x}\\
&=\frac{\hbar}{4 \pi c_0 \omega_k} \delta(k_1-k_2)
\end{align}
Thus, defining the operators:
\begin{align}
&\widehat{b}_k \equiv \sqrt{\frac{4 \pi c_0 \omega_k}{\hbar}} \widehat{\phi}(k)\\
&\widehat{b}_k^{\dagger} \equiv \sqrt{\frac{4 \pi c_0 \omega_k}{\hbar}} \widehat{\phi}^{\dagger}(k)
\end{align}
They follow the bosonic commutation relations:
\begin{align}
&[\widehat{b}_{k_1},\widehat{b}^{\dagger}_{k_2}]=\delta(k_1-k_2)\\
&[\widehat{b}_{k_1},\widehat{b}_{k_2}]=[\widehat{b}^{\dagger}_{k_1},\widehat{b}^{\dagger}_{k_2}]=0
\end{align}
The fields can then be written as:
\begin{align}
&\widehat{\phi}(x,t)=\sqrt{\frac{\hbar}{4 \pi c_0}}\int_{-\infty}^{+\infty} \frac{dk}{\sqrt{\omega_k}} \left(\widehat{b}_k e^{i (kx-\omega_k t)}+\widehat{b}^{\dagger}_k e^{-i (kx-\omega_kt)} \right) \\
&\widehat{p}(x,t)=-i \sqrt{\frac{\hbar c_0}{4 \pi}} \int_{-\infty}^{+\infty} dk \sqrt{\omega_k} \left(\widehat{b}_k e^{i(kx-\omega_kt)}-\widehat{b}^{\dagger}_k e^{-i(kx-\omega_kt)} \right)
\end{align}
Now, we want to describe the field in the semi infinite waveguide with the boundary condition $\forall t \ \dot{\widehat{\phi}}(0,t)=0$ because the voltage $\dot{\phi}$ vanishes in $x=0$. From this condition, we can find the expression of the field:
\begin{align}
&\widehat{\phi}(x,t)=i \sqrt{\frac{\hbar}{\pi c_0}}\int_{0}^{+\infty} \frac{dk}{\sqrt{\omega_k}}sin(kx) \left(\widehat{b}_k e^{-i \omega_k t}-\widehat{b}^{\dagger}_k e^{i \omega_k t}  \right) \\
&\widehat{p}(x,t)=\sqrt{\frac{\hbar c_0}{\pi}} \int_{0}^{+\infty} dk \sqrt{\omega_k} \sin(kx) \left(\widehat{b}_k e^{ - i \omega_k t}+\widehat{b}^{\dagger}_k e^{i \omega_k t} \right)
\end{align}
We can rewrite with the bosonic operators $\widehat{b}(\omega) \equiv \frac{\widehat{b}_k}{\sqrt{c}}$:
\begin{align}
&\widehat{\phi}(x,t)=i \sqrt{\frac{\hbar Z_0}{\pi}}\int_{0}^{+\infty} \frac{d \omega}{\sqrt{\omega}}sin(\frac{\omega x}{c}) \left(\widehat{b}(\omega) e^{-i \omega t}-\widehat{b}^{\dagger}(\omega) e^{i \omega t}  \right) \label{eq:app_phi_omega}  \\
&\widehat{p}(x,t)=\sqrt{\frac{\hbar c_0}{c \pi}} \int_{0}^{+\infty} d \omega \sqrt{\omega} \sin(\frac{\omega x}{c}) \left(\widehat{b}(\omega) e^{ - i \omega_k t}+\widehat{b}^{\dagger}(\omega) e^{i \omega_k t} \right)
\end{align}
Up to a phase in the operator $\widehat{b}(\omega)$ (which corresponds to a choice in the origin of time), those expressions are similar as the ones in \cite{wiegand2020semiclassical}. Also, using the relations:
\begin{align}
\int_{-\infty}^{+\infty} dx \sin(kx) \sin(k'x)=\pi ( \delta(k-k')-\delta(k+k'))
\end{align}
\begin{align}
\int_{-\infty}^{+\infty} dx \sin(kx) \sin(k'x)=\pi ( \delta(k-k')+\delta(k+k'))
\end{align}
Injecting \eqref{eq:phi_stat} and \eqref{eq:p_stat} in \eqref{eq:Hwaveguide}, one can show that it ends up in the following Hamiltonian:
\begin{align}
H_{\text{Waveguide}}=\int_{0}^{+\infty} d \omega \ \hbar \omega \widehat{b}^{\dagger}(\omega) \widehat{b}(\omega) 
\end{align}
\subsection{Waveguide-transmon interaction}
\label{app:waveguide_transmon_interaction}
Now, we can compute the expression of the waveguide-transmon coupling as a function of the bosonic operators. We recall that we had:
\begin{align}
H_{\text{int}}=-\frac{C_d}{C_{\Sigma}}\frac{\widehat{p}(-L,0)}{c_0}\widehat{p}_J(0)
\end{align}
First, we linearize the transmon and approximate it as if it was an harmonic oscillator. Its frequency would then be $\omega_0=\frac{1}{\sqrt{L_J^0 C_{\Sigma}}}$ as shown in \eqref{eq:transmon_approx_harmonic_oscillator_frequency}. What will play the role of the mass will be the capacitance $C_{\Sigma}$. Thus, we can express the momentum of the transmon as a function of creation and annihilation operators:
\begin{align}
p_J(0)=i\sqrt{\frac{\hbar \sqrt{\frac{C_{\Sigma}}{L_J^0}}}{2}}(\widehat{a}^{\dagger}_J-\widehat{a}_J)
\end{align}
Replacing with the expression of the field \eqref{eq:p_stat}, and performing the rotating wave approximation \cite{haroche2006exploring} which consists here in removing the terms $\widehat{a}^{\dagger}_J \widehat{b}^{\dagger}(\omega)$ and $\widehat{a}_J \widehat{b}(\omega)$, we end up with the Hamiltonian:
\begin{align}
&H_{\text{int}}=i \int_0^{+\infty} d \omega \ \hbar g(\omega) \left( \widehat{b}(\omega) \widehat{a}_J^{\dagger}-\widehat{b}^{\dagger}(\omega) a_J\right)\\
&g(\omega) \equiv \frac{C_d}{C_{\Sigma}} \sqrt{\frac{\omega}{\omega_0}} \sqrt{\frac{Z_0}{2 \pi L_J^0}} \sin(\omega L)
\end{align}
\end{appendices}

\chapter{Quantum error correction and fault-tolerance}

As we saw in the previous chapter, operations performed on qubits are noisy. And this is true even if the qubits are put at zero temperature because of spontaneous emission and pure dephasing\footnote{We did not model it in the previous chapter, as we assumed that the dominant source of noise is spontaneous emission in our models.}. It implies that the data stored on them has a limited lifetime: those qubits cannot be used in algorithms that require too many quantum gates as, by the time the algorithm will be over, the information contained on those qubits will be completely corrupted. Of course, in principle, physicists could focus on improving the quality of the qubits, trying to make their lifetime longer and longer (thus reducing the spontaneous emission, and pure dephasing rate in the context of superconducting qubits). But this direction is full of experimental challenges which might not be easy to solve. What we would like to have is a way to be able to implement longer and longer algorithms \textit{without} having to improve the quality of the devices each time. A way to do it is to use quantum error correction. Typically, the strategy behind this is to design quantum computers that, on a "logical" level, are implementing the desired algorithm, but "in the background", errors occurring on qubits are detected and corrected. With this strategy, even if the qubits are noisy, with more and more error correction, on a logical level, everything could be "as if" the lifetime of the qubits was longer and longer. The cost of this method is that it requires more qubits inside of the computer because some of them are used to perform the error correction (and do not directly participate in the algorithm), and because we need to encode the information in a greater number of degree of freedom than what would be strictly necessary: some redundancy of the information is required to detect errors. But there is overall a clear interest: no more experimental improvement in the quality of the qubits would be necessary, the reduction of errors would have a well-defined strategy that would not require any challenging innovation from the qubit technology\footnote{But of course it might be challenging from an engineering perspective as it would need to put more and more qubits inside of the computer.}\footnote{As explained in a few lines, the noise would also have to be below some threshold value to make this happen. Otherwise, error correction would actually make the situation worse.}. The goal of this chapter is thus to present the basic elements of quantum error correction we need to understand our work on the questions of scalability of quantum computing. Indeed the major part of our work relies on results from quantum error correction theory. After having introduced some basic elements we need to understand this theory in section \ref{sec:fundamentals_QEC}, we give some first intuitions in the section \ref{sec:stabilizer} behind the working principle of stabilizer codes. Those codes are in the same class of code as Steane code which is the one we used in the work done during this Ph.D., and having the general (but simplified) picture we will give there will allow us to have an easier grasp on Steane. The Steane code is explained in the section \ref{sec:Steane}. In the section \ref{sec:FT} we will explain why the knowledge provided by quantum error correction is not enough to scale up a quantum computer, and why the concepts behind fault-tolerance are required to make the quantum computer able to resist detect and correct errors in practice. To summarize briefly: error correction shows that it is possible to detect and correct errors "in principle", assuming a \textit{perfect} correction procedure. Fault-tolerance explains how to do it \textit{in practice}, taking into account the fact that the correction might also introduce errors. Fault-tolerance is needed to make error correction useful. Typically, it is in this section that we will introduce the quantum threshold theorem that states that error correction is only useful if the noise is below a threshold value and that if it is the case, then, using enough physical resources, the effect of the noise can be reduced as much as desired, without using an absurdly high number of physical resources. The fault-tolerant construction we are going to consider is called the concatenated construction. The reason why we use this specific construction and the Steane code is because they are both very well understood theoretically. For instance, we know exactly the circuits allowing to implement Steane code fault-tolerantly with this concatenated construction, and it allows us in particular to perform analytic calculations. It is also a construction in which the quantum threshold theorem can be formally derived. Finally, all the important numbers we need to estimate the energetic cost of quantum computing in the next chapters: how many qubits do we need, what is the accuracy the computer can get to as a function of this number of qubits, etc, will be deduced from the results provided in this chapter. 

This chapter does not contain any original result from this Ph.D. (apart from the exact estimation behind the table \ref{table:FT_gate_breakdown} in \ref{sec:number_gates_1Rec}, but I did not calculate myself\footnote{Those tables have been calculated by Jing Hao Chai.} those numbers), it is here to introduce the necessary concepts.
\section{Fundamentals of quantum error correction}
\label{sec:fundamentals_QEC}
The topic of error correction consists in providing algorithms that allow detecting and correct errors occurring on qubits, those errors being caused by the noise affecting the qubits. At first view, it is not guaranteed that it is conceptually possible to detect errors. Indeed the no-cloning theorem \cite{nielsen2002quantum} forbid copying a quantum state which would be of great help to encode the data in a more robust manner. Also, measuring a quantum system can perturb it because of the measurement postulate of quantum mechanics \cite{preskill1998lecture,cohen1998mecanique}. For those reasons, it looks unlikely to conceptually be possible to detect and correct errors occurring on quantum systems. Fortunately, error correction theory shows that despite all those limitations, it is still possible to detect and correct errors on quantum systems. Fault tolerance that we introduce in the next section will tell us how error correction has to be implemented \textit{in practice}. For instance, what it adds to the concepts of error correction is that it acknowledges the fact that error correction might be noisy as well. It also takes into account all the ancilla qubits that are required to be able to implement error correction in practice, and it explains how to implement operations on qubits in a robust (i.e., noise resilient) manner.
\subsection{Tools allowing to describe the noise in quantum systems}
\label{sec:tools_describing_noise}
Before starting, we need to understand what quantum noise is and how it is described mathematically. For this purpose, we need to recall some basic tools and definitions behind quantum information theory. We already used some of them in the previous chapter, but as we now need solid definitions, we take a step back and properly define mathematically the different objects we need. We consider a quantum system A. The state of this quantum system can often be described by a vector $\ket{\psi_A}$ (of norm $1$) living in a Hilbert space $H_A$ of dimension $d_A$. When it is the case, we will say that the system is \textit{pure}. It will be the case in the absence of classical uncertainty or if this system is not entangled with another system. But in order to be more general and to allow to describe systems submitted to classical uncertainty (or being subsystems of entangled systems), we prefer to use a more "general" tool called the density matrix that we note with the letter $\rho$ in what follows. It defines the state of this system, allowing to take into account this "lack of knowledge". This tool is a matrix acting on the space $H_A$, that admits eigenvalues which can be interpreted as probabilities. For this reason, $\rho$ admits positive eigenvalues which sum up to $1$. Mathematically the appropriate wording is to say that $\rho_A$ is positive semi-definite (positivity of the eigenvalues), and satisfies $Tr(\rho)=1$ (they sum up to $1$). 

At this point, we talked about the description of quantum states. But those states might be evolving in time because of interactions. And this evolution might introduce noise. Under some circumstances, solely knowing its value at an instant $t$: $\rho_A$, it is possible to deduce its value $\rho_A'$ at an instant $t'>t$. It will typically be the case, for instance, in the absence of initial correlations between S and any ancillary system. In such case, for any initial preparation $\rho_A$, the final density matrix satisfies $\rho_A'=\mathcal{E}(\rho_A)$, where $\mathcal{E}$ is some "map", i.e. an object that given a density matrix returns another density matrix. In order to be physical, this map must satisfy certain properties. Typically we must actually ensure that for any $\rho_A$, $\mathcal{E}(\rho_A)$ remains a density matrix. Those "physically valid" maps are called quantum channels, which is a term coming from communication theory, and that was initially used to characterize the deformation of a signal when it is being sent between two parties. Here we see that it is in some sense the equivalent on a quantum level as it describes how a state of a quantum system is being modified after some evolution. In order to be physical, $\mathcal{E}$ must thus satisfy the following set of axioms.
\begin{definition}{Quantum channel}
\label{def:quantum_channel}

A linear map $\mathcal{E} \in \mathcal{L}(\mathcal{L}(H_A))$, i.e., a linear map that takes as input a matrix in $\mathcal{L}(H_A)$ and gives as output another matrix living in this same space is called a quantum channel if it satisfies:
\begin{itemize}
\item $\mathcal{E}$ is a convex-linear map on the set of density matrix, i.e: for any distribution of probabilities $\{p_i\}$ and family of density matrices $\{\rho_i\}$, we have $\mathcal{E}(\sum_i p_i \rho_i)=\sum_i p_i \mathcal{E}(\rho_i)$.
\item $\mathcal{E}$ is trace-preserving, i.e for any $\rho_A$, $Tr(\mathcal{E}(\rho_A))=1$
\item $\mathcal{E}$ is completely positive, i.e: for any integer $n$, for any  $\rho \in \mathcal{L}(H_A \otimes H_{B_n})$, where $H_{B_n}$ is another Hilbert space of dimension $n$, $\left( \mathcal{E} \otimes \mathcal{I}_n \right)(\rho)$ is semi definite positive (i.e, it is still a matrix having positive eigenvalues) where $\mathcal{I}_n$ represents the identity operation applied on $H_{B_n}$.
\end{itemize}
\end{definition}

The first axiom of this definition is here to ensure that $\mathcal{E}$ follows the postulate of quantum mechanics: the evolution of a quantum system is linear\footnote{It is not the case when one measures and reads the outcome, but we are not interested in those cases here.}, thus, any mixture present in the initial quantum state will remain in such mixture after the evolution. Indeed, a classical mixture between a family of quantum states that we described as the family: $\{p_i, \rho_i\}$, where $p_i$ is the probability to have the matrix $\rho_i$ in this mixture, is described by a density matrix $\rho=\sum_i p_i \rho_i$. The condition $\mathcal{E}(\sum_i p_i \rho_i)=\sum_i p_i \mathcal{E}(\rho_i)$ means that the mixture is "preserved" through the evolution\footnote{We emphasize on the fact that it does not implies that the entropy will remain constant through the evolution.}. The second and third axioms are here to ensure that the resulting object will still be a density matrix. Indeed we want it to be of trace $1$ (second condition) and to be semi-definite positive (third condition) to keep the interpretation of eigenvalues as being probabilities. But the third condition is actually more restrictive. We do not only ask the resulting density matrix to be semi-definite positive but also that if we imagine the system of interest as being a subsystem of a larger system, i.e., there exists a system $B$ such that $\rho_A=Tr_B (\rho)$ ($Tr_B$ denotes the partial trace over the system $B$, we can find the definition of such operation in \cite{cohen2021mecanique}), where $\rho \in \mathcal{L}(H_A \otimes H_B)$ (i.e., $\rho$ defines the state of the composite system $AB$), the density matrix of this bigger system remains positive after the evolution of $\rho_A$ through $\mathcal{E}$. It is a requirement for the evolution to be physical. For instance, if we did not have this requirement, we could imagine that if the quantum state was initially entangled with some other party, after the evolution, some of the eigenvalues of the density matrix would become negative and we couldn't interpret them as probabilities anymore \cite{breuer2002theory}.

To summarize briefly, quantum channels are the objects that describe a class of evolution for quantum systems. Those evolutions are the ones in which only knowing the initial state of the system that is going to evolve; its final state can be perfectly deduced. In some sense, we can say that they can describe evolution in which the effect of the potential noise is "local". Those kinds of evolutions are the ones we are going to mainly focus on in the next chapters and also the ones on which a major part of quantum error correction theory is based upon. 
\subsection{Conditions on the noise allowing to correct errors}
Now, we can give some first insight about the very general principle behind quantum error correction before explaining Steane code. It will give us the intuition of why (and how) it is possible to correct errors occurring in quantum systems.

Let's assume that we have some data encoded on a density matrix $\rho_A=\ketbra{\psi_A}{\psi_A} \in \mathcal{L}(H_A)$. Our goal is to protect this data against any noise that might corrupt it. We assume that this noise is modeled by a quantum channel $\mathcal{E}$ that we call here the \textit{error channel}.
In an ideal world, we would like to find a quantum channel $\mathcal{R}$ that we call the \textit{error correction channel} which would have the following property:
\begin{align}
\forall \rho \in \mathcal{L}(H), \ \mathcal{R} \circ \mathcal{E} (\rho) = \rho
\end{align}
Physically, it would mean that whatever the initial state before the noise was, it would be possible to restore it. However, such map $\mathcal{R}$ cannot exist for an arbitrary $\mathcal{E}$. This can be understood intuitively: typically, if $\mathcal{E}$ introduces some irreversibility, $\mathcal{E}(\rho)$ could correspond to different initial states. There is not enough information in $\mathcal{E}(\rho)$ to be able to inverse the effect of $\mathcal{E}$ in general. This remark would thus also be true for classical systems. A way to solve this issue is to make the effect of errors "locally" invertible. Basically, what we can hope to get is that there exists a subspace of the total Hilbert space $H_C \subset H$ in which:
\begin{align}
\forall \rho_C \in \mathcal{L}(H_C), \mathcal{R} \circ \mathcal{E} (\rho_C) = \rho_C
\label{eq:cond_recovery_intro}
\end{align}
As we are going to see, such a requirement is possible. In practice, it means that one could encode the information to protect on the (smaller) subspace $H_C$, and error correction would then allow to protect it. The space $H_C$ is called \textit{code space}, but in the literature it is also frequently called the \textit{the quantum code}. Proceeding this way, it will be possible to detect and then correct the errors induced by the noise that is going to perturb the quantum state.
\begin{definition}{Code space}
\label{def:code_space}

Let $H_C$ be a subspace of $H$ the total Hilbert space describing the quantum system on which we wish to encode quantum information. We call this space the code space \cite{nielsen2002quantum}. 
\end{definition}
Now, we would like to have some conditions that give us necessary and sufficient conditions the error channel $\mathcal{E}$ has to satisfy in order that an error correction channel $\mathcal{R}$ satisfying \eqref{eq:cond_recovery_intro} for some $H_C$ exists. Those conditions are called the 
Knill-Laflamme conditions \cite{nielsen2002quantum}, but to understand them, we now need to define precisely what we call an \textit{error}.

In the classical world, an error is an event easy to interpret: it is some operation that has been applied in an undesired manner on the system encoding the information in such a way that this information has been corrupted. Maybe the person running the experiment will not know if the error occurred or not in practice, but it will be clear that either the error occurred either it did not. In the quantum world, this concept becomes directly more subtle. Typically, because of superpositions, we might be in a scenario in which an error "occurred and did not occur" at the same time. In addition to that, we can expect the errors to live in a continuous set, as quantum states are living in continuous space (the state of a qubit, for instance, is described by two continuous parameters). Being able to correct continuous errors might be very challenging. Now that we have given basic intuition of the difference we might expect with classical physics, we need a solid definition for the notion of error. Let us consider any error channel $\mathcal{E}$. What we would like to have is a given set of operators associated with $\mathcal{E}$ that are acting on the system encoding the data, and we would like to refer to this set as the possible set of errors that might be occurring. 

In principle, if the system evolves in an autonomous manner (i.e., it is not being actively measured by an experimentalist), it will follow a unitary evolution, possibly with another system $B$. Assuming $A$ was in an initial state $\ket{\psi_A}$, $B$ was in an initial state $\ket{b_0}$ and considering $\{\ket{b_k}\}$ an orthonormal family for $H_B$, the evolution reads:
\begin{align}
&U_{AB} \ket{\psi_A} \ket{b_0} = \left(\sum_{k} \ketbra{b_k}{b_k} \right) U_{AB} \left( \sum_l \ketbra{b_l}{b_l} \right) \ket{\psi_A} \ket{b_0}= \sum_k M_k \ket{\psi_A} \ket{b_k} \label{eq:Psi_purified} \\
&M_k \equiv  \bra{b_k} U_{AB} \ket{b_0},
\end{align}
where we used the fact $\sum_l \ketbra{b_l}{b_l}=Id$ ($Id$ is the identity operation).
Tracing out the environment, we realize that $\ketbra{\psi_A}{\psi_A}$ followed the evolution:
\begin{align}
\ketbra{\psi_A}{\psi_A} \to \mathcal{E}(\ketbra{\psi_A}{\psi_A})=\sum_k M_k \ketbra{\psi_A}{\psi_A} M_k^{\dagger}
\end{align}
The family of operators $\{M_k\}$ is being called Kraus operators. This way of describing the effect of noise occurring on a system $A$ by considering an extra system $B$ (an "environment") is pretty standard, and we will make use of it again in the section \ref{sec:error_detection_and_correction} dedicated to explain how the error correction works with Steane code. What we just said could be generalized for any initial density matrix $\rho_A$ for $A$ (in the case $A$ was not pure initially, its final evolution could be found from the same family of operators $\{M_k\}$), and we have the theorem \ref{theorem:kraus} \cite{nielsen2002quantum}.
\begin{theorem}{Kraus decomposition of a quantum channel}

\label{theorem:kraus}
A map $\mathcal{E} \in \mathcal{L}(\mathcal{L}(H_A))$ is a quantum channel if and only if for any $\rho_A \in \mathcal{L}(H_A)$ it can be written as:
\begin{align}
\mathcal{E}(\rho_A)=\sum_{k=1}^{K} M_k \rho_A M_k^{\dagger}
\end{align}
Where $1 \leq K \leq d_A^2$, $d_A$ being the dimension of $H_A$. The operators $M_k \in \mathcal{L}(H_A)$ are called \textit{Kraus operators} and they satisfy:
\begin{align}
\sum_{k=1}^{K} M_k^{\dagger} M_k=I
\end{align}
\end{theorem}
This family of operators $\{M_k\}$ describes a set of errors that might affect our quantum system. It is possible to show that the Kraus decomposition of a given quantum channel is not unique \cite{nielsen2002quantum}; different sets of errors could describe the exact same error channel. This is also a difference with the classical world: there are infinitely many different ways to interpret the errors.

The reason why $\{M_k\}$ can really be interpreted as errors can be seen from \eqref{eq:Psi_purified}: measuring the environment in the basis $\{\ket{b_k}\}$ would modify $\ket{\psi_A}$ to $M_k \ket{\psi_A}/|| M_k \ket{\psi_A}||$ for some $k$: $A$ would have been directly "impacted" by the operators $\{M_k\}$.

We can now express the Knill-Laflamme conditions, which allow to give necessary and sufficient conditions for an error correction channel $\mathcal{R}$ to exist.
\begin{theorem}{Knill-Laflamme conditions}

\label{theorem:Knill_laflamme}
Let $H_C \subset H$ be the code space. We call $P_C$ an orthogonal projector on this space. We have the following property.

A quantum channel $\mathcal{R}$ allowing to correct for the errors introduced by an error channel $\mathcal{E}$ (which is also a quantum channel) will exist, i.e will satisfy
\begin{align}
\forall \rho_C \in \mathcal{L}(H_C), \mathcal{R} \circ \mathcal{E} (\rho_C)=\rho_C
\end{align}
if and only if it is possible\footnote{Infinitely many Kraus operator allowing to describe a quantum channel exist. Here, we ask that there exist one choice that satisfies the subsequent property.} to describe the $\mathcal{E}$ with a set of Kraus operators $\{M_i\}$ that satisfy for all $(i,j)$: 
\begin{align}
P_C M_i^{\dagger} M_j P_C = c_{ii} \delta_{ij} P_C,
\label{eq:Knill_Laflamme_diago}
\end{align}
where $c_{ii}$ are positive numbers. Those condititions are called the Knill-Laflamme conditions. The set of operators $\{M_i\}$ will be called \textit{errors} of the error channel.
\end{theorem}
It may not look at first view, but what the Knill-Laflamme conditions are asking is actually rather intuitive, and it consists in asking that the effect of two different errors would bring a given quantum state initially in the code space to two orthogonal subspaces. Then, an experimentalist could measure in which subspace the system went and would be able to deduce which error occurred, making the correction possible to apply. It also asks for this measurement to not be able to "betray" the information that was encoded, i.e., to get any information about what was the state $\ket{\psi_A}$ that had to be protected. An intuitive picture of the principle is provided in figure \ref{fig:principle_knill_laflamme}.
\begin{figure}[h!]
\begin{center}
\includegraphics[scale=0.5]{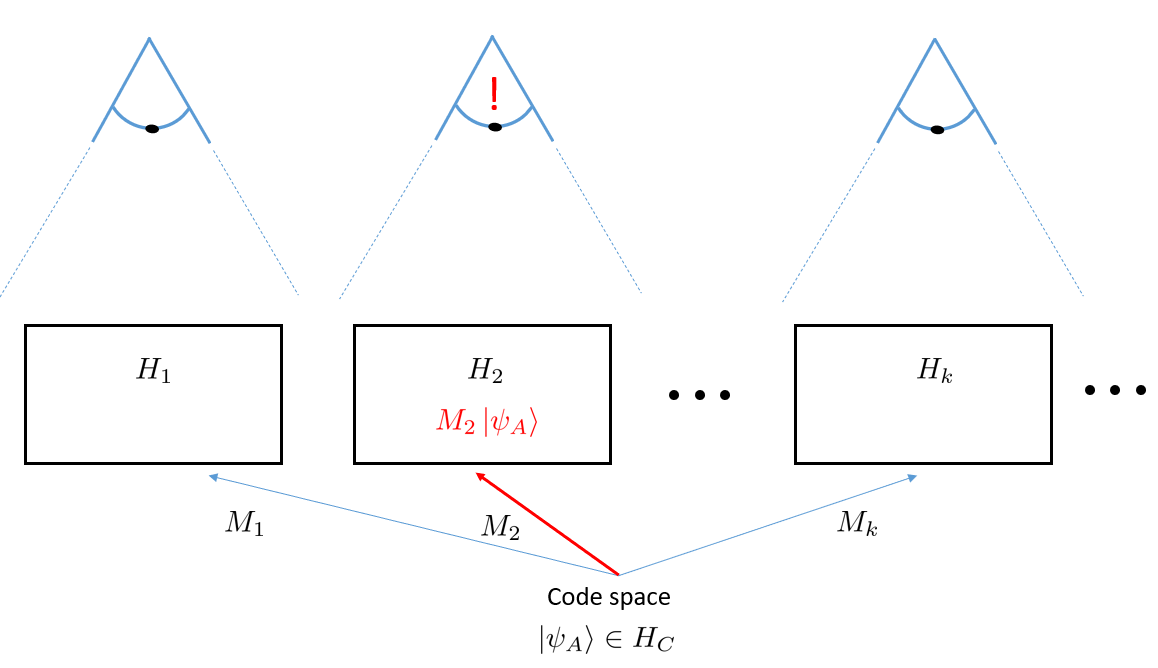}
\caption{The Knill-Laflamme conditions presented in theorem \ref{theorem:Knill_laflamme} express the fact that if there exist a Kraus-decomposition $\{M_k\}$ for an error channel $\mathcal{E}$ such that $M_i \ket{\psi_A} \perp M_j \ket{\psi_A}$ for $j \neq i$ and $\ket{\psi_A}$ initially in the code space $H_C$, then it is possible to detect which error occured by performing a projective measurement (and once the error is found, an appropriate correction can be applied). In this figure, we assumed that the error $M_2$ occured. Another requirement of the theorem is to avoid that the measurement can access any information about the specific state $\ket{\psi_A}$ that was initially prepared. The Knill-Laflamme condtions, as indicated by the theorem, are necessary \textit{and} sufficients condition for error correction.}
\label{fig:principle_knill_laflamme}
\end{center}
\end{figure}
\FloatBarrier
\textbf{Some additional comments}

The proof of this theorem is provided in the appendix \ref{app:Knill_Laflamme}. The interpretation provided in the figure \ref{fig:principle_knill_laflamme} can be found by considering a pure quantum state $\ket{\psi_C}$. The Knill Laflamme conditions imply that: $\bra{\psi_C} M_i^{\dagger} M_j \ket{\psi_C} = c_{ii} \delta_{ij}$: this term vanishes for two different errors ($i \neq j$) which illustrates the condition that two different error bring a state in two orthogonal subspaces. In such a case, the error correction basically consists in finding in which subspace the system went and apply an appropriate unitary to bring it back into the code space. The fact that this unitary always exists is shown in the proof of the theorem. However, all this interpretation relies on a specific choice of Kraus operators. Indeed, we recall that different sets of Kraus operators can describe the same channel, and as shown in the appendix, for another set of Kraus operator, the Knill-Laflamme condition would actually become:
\begin{align}
P_C M_i^{\dagger} M_j P_C = \alpha_{ij} P_C,
\label{eq:Knill_Laflamme_general}
\end{align}
with $\alpha_{ij}$ elements of an Hermitian matrix. In this case, the interpretation that two errors bring the code space into two different orthogonal spaces is no longer valid as $\alpha_{ij}$ might not vanish for $i \neq j$. But if the Kraus operators satisfy this new condition, then it is always possible to find another equivalent set of Kraus operators (i.e. describing the same error channel) such that the condition we gave in theorem \ref{theorem:Knill_laflamme} is valid, on which the interpretation is easier. 

Now, even though this theorem shows that it is possible to correct for an arbitrary error channel $\mathcal{E}$, it doesn't allow us to design a \textit{generic} error correction channel $\mathcal{R}$ which would be valid for different error channels at the same time. Indeed in the proof of the theorem, we built an error-correction channel that was a function of the error operators $\{M_i\}$. This is the reason why we need to introduce the concept of error discretization. Basically, we will see that if someone is able to correct for a finite set of error operators, it will be able to correct for an error that is any linear combination of elements of this set.
\subsection{Discretization of errors}
We finish this section with a last important general result called the discretization of errors. The Knill-Laflamme conditions gave us necessary and sufficient conditions in order to be able to correct errors coming from an error channel. But the exact correction procedure then depends on the particular error channel that is affecting the qubits. What would be desirable is to design a generic procedure which would work for a whole class of error channels. The principle allowing it is called the discretization of errors: as soon as it is possible to correct against a family of errors $\{M_i\}$, it is possible to correct for any error that is a linear combination of elements from this family. Thus, designing an error correction channel allowing to correct for the family $\{M_i\}$ will allow to resist against any linear combination of those errors.
\begin{theorem}{Discretization of errors}
\label{theorem:discretization_errors}

Let's assume $H_C \subset H$ is a code space. If the Knill-Laflamme conditions are satisfied for a set of error operators $\{M_i\}$, then they are satisfied for an arbitrary linear combination of those operators.
\end{theorem}
The proof of this theorem is straightforward but shown in the appendix as well. In practice, it will imply that if someone can protect against bit-flip (i.e $X$ error) and phase-flip (i.e $Z$ error) occurring on a single physical qubit, it is possible to correct for an arbitrary single-qubit error. Indeed, the single-qubit Pauli matrices form a basis for any single qubit operator, and we also have the fact that $Y=iXZ$. Thus it is enough to correct for $X$ and $Z$ errors to be able to correct for an arbitrary single-qubit error. Under this perspective, the error correction channel becomes somewhat independent on the exact expression of the error channel \footnote{Of course, it still depends on the error channel to some extent because it is designed to correct for a set of errors which linear combinations are able to reproduce any error of the error channel. But it won't require knowing the exact expression of the error channel for instance.}.

\section{Intuition behind stabilizer codes}
\label{sec:stabilizer}

Now that the general principle of quantum error correction has been provided, we are going to focus on a subclass of codes which are called stabilizer codes. Our work is based on Steane code which belongs in this category. However, we must keep in mind that there are actually many different classes of quantum error correction codes. We can think about CSS codes \cite{calderbank1996good,steane1996multiple}, bosonic codes \cite{noh2020fault,terhal2020towards}, topological codes \cite{dennis2002topological,noh2020fault,fowler2012surface},  large block codes \cite{zheng2018efficient,brun2015teleportation}, etc (those families are not necessarily exclusive from one another). But the family of stabilizer codes is a very important one as a wide variety of codes belongs in it. It relies on what is called the stabilizer formalism, for which we will try to give some general intuitions (we are not going to study all the details as it is a subject in itself, we refer to \cite{nielsen2002quantum, gottesman1997stabilizer,gottesman2010introduction} for further details). Then, we will apply this formalism in the simplest quantum error correction code that exists: the three-qubit bit-flip code in order to have an easy grasp of the concepts. It will also allow us to be able to easily understand Steane code which will be presented in the section \ref{sec:Steane} that follows. 
\subsection{What are stabilizers}
\label{sec:what_are_stabilizers}

In order to give a first basic intuition of what stabilizer codes are, we can start with explaining how quantum states can be described with the notion of stabilizers. Let us consider the following Bell state:
\begin{equation}
\ket{\text{Bell}} = \frac{\ket{00}+\ket{11}}{\sqrt{2}}
\end{equation}
We can notice that this state has the particularity to be an eigenstate of $+1$ eigenvalue for the operators: $Z_1 Z_2$ and $X_1 X_2$ where here (and for all the rest of this thesis), the notation $G_i$ where $G$ is any single-qubit Pauli matrix means that a Pauli operator $G$ is applied on the qubit $i$. It is also possible to show that it is actually the \textit{unique} (up to an arbitrary global phase) quantum state that is $+1$ eigenvalue of those two operators simultaneously. We say that $Z_1 Z_2$ and $X_1 X_2$ are \textit{stabilizers} of the state $\ket{\text{Bell}}$, or equivalently that the state is \textit{stabilized} by $Z_1 Z_2$ and $X_1 X_2$. This simple example illustrates the principle behind stabilizer formalism. Instead of providing an explicit expression of a quantum state, it is possible to define it by giving an appropriate set of stabilizers that define this state in a unique way. 

We can also stabilize \textit{spaces} and not only states. We can take an example with the total Hilbert space spanned by $n=3$ physical qubits:\footnote{The notation $Span(\ket{u},\ket{v})$ denotes the Hilbert space containing all the vectors than can be written as $a \ket{u}+b\ket{v}$ for any $a$ and $b$ complex coefficients.}
\begin{align}
Span \left( \ket{000},\ket{001},\ket{010},\ket{011},\ket{100},\ket{101},\ket{110},\ket{111} \right)
\label{eq:span_3physical}
\end{align}
Considering an arbitrary state in this vector space, it is possible to show that $Z_1 Z_2$ is stabilizing the space:
\begin{align}
Span \left( \ket{000},\ket{001},\ket{110},\ket{111} \right).
\end{align}
To understand it easily, we can notice that $Z_1 Z_2$ stabilizes the states where the two first qubits have the same parity, i.e., they are both either $00$ or $11$ but never $10$ or $01$. Any other computational state must then be removed from the linear span in \eqref{eq:span_3physical}. We notice that the dimension of the stabilized space has been divided by two. Considering an additional stabilizer $Z_2 Z_3$ (the two last qubits must have the same parity),
the final stabilized space is:
\begin{align}
Span \left( \ket{000},\ket{111} \right)
\end{align}
The dimension has again been divided by two. Thus, instead of talking explicitly about the space $Span \left( \ket{000},\ket{111} \right)$, we can describe it by saying that it is the space stabilized by the list of operators $\{Z_1 Z_2,Z_2 Z_3\}$. We see that each stabilizer here divided the dimension of the space by two. The intuition behind this is that one $n$-Pauli operator\footnote{We are not talking about the $n$-Pauli group, which is built from a tensor product of $n$ single-qubit Pauli matrices which can then be multiplied by $-1$, $i$ or $-i$. We try to avoid using any notion of group theory to keep the explanations simple.} (i.e., an operator built from a tensor product of $n$ single-qubit Pauli matrices) divides the space in half, each of the two halves corresponding to its eigenspace having $+1$ or $-1$ eigenvalue. Thus, as $2^n/2^{n-k}=2^k$, we need $n-k$ of those operators to reduce the dimension of the total space from $2^n$ to $2^k$. Finally, we can also add a sign $-1$ if necessary to the $n$-Pauli operator, it would then stabilize a different subspace (for instance, $\{-Z_1 Z_2, Z_2 Z_3\}$ would stabilize $Span \left( \ket{100},\ket{011} \right)$ as now $-Z_1 Z_2$ only keeps the state for which the two first qubit have different parities).

Now there are additional conditions to the sole fact to have a list of $n-k$ $n$-Pauli operators. Those extra conditions were satisfied in the previous example, but we briefly enumerate them here. Having $n-k$ operators is not the only condition necessary as we need the description to not be "redundant". Thus we ask to have a list of $n-k$ independent n-Pauli matrices in the sense that it wouldn't be possible to express one as a product of the others. For instance, $\{Z_1 Z_2,Z_2 Z_3, Z_1 Z_3\}$ is not a family of independent matrices because $Z_1 Z_3 = Z_1 Z_2 Z_2 Z_3$. We also ask for the operators that they all commute with each other; this is necessary to not stabilize a space only composed of the null vector (as we are only dealing with $n$-Pauli matrix, two operators will necessarily either commute either anti-commute, there is no other option \cite{nielsen2002quantum}). For instance if it appeared that two elements in this list, $g$ and $g'$ anti-commuted, then for any $\ket{\psi}$ stabilized by those operators we would have $g g' \ket{\psi}=-g' g \ket{\psi}$ and thus $\ket{\psi}=-\ket{\psi} \Rightarrow \ket{\psi}=0$. Finally, any product involving elements in the list cannot be equal to $-I$. Otherwise, we would also stabilizer a trivial space. This is because a product of operators stabilizing a space will stabilize the same space. If it happens that this product that we call $g$ satisfies $g=-I$ we then have for any $\ket{\psi}$ in this space $g\ket{\psi}=\ket{\psi}=(-I)\ket{\psi}$ which implies $\ket{\psi}=0$. See \cite{nielsen2002quantum} for derivations of all those properties.

\subsection{Simple example of stabilizer code: the three qubit code}
\label{sec:three_qubit_code}
Now, we can give a first example of a stabilizer code: the three-qubit code that is able to correct for a bit-flip occurring on any of the three physical qubits that are composing one logical qubit. A code using $n$ physical qubits means that the total Hilbert space describing the physics is $2^n$. We say that it protects $k<n$ logical qubits if the code space has dimension $2^k$. Indeed, in this case, one could protect $k$ "qubit unit of information", which we call logical qubits.

Here, the physics is thus occurring in the Hilbert space of dimension $2^3$ described in \eqref{eq:span_3physical}. When no error has occurred, the state to protect is in the code space. We define the code space as being the space stabilized by the operators $\{Z_1 Z_2, Z_2 Z_3\}$: it corresponds to $Span(\ket{000},\ket{111})$. If the logical qubit containing the logical information is written as $a \ket{0_L}+b\ket{1_L}$ (we put the index $L$ which means "logical" in order to describe the information that we wish to protect), it can in practice be encoded in the physical qubits as the state $a \ket{000}+b\ket{111}$. Assuming only bit-flip noise, i.e., the error channel is composed of errors being either $X_1,X_2$ or $X_3$, assuming that only one physical qubit can be affected by the noise at a time, the state can, because of that noise become either of the three possibilities below:
\begin{itemize}
\item $X_1$ error: $a \ket{100}+b\ket{011}$
\item $X_2$ error: $a \ket{010}+b\ket{101}$
\item $X_3$ error: $a \ket{001}+b\ket{110}$
\end{itemize}
The effect of the error is then to change the values of the stabilizers. For instance, if an $X_1$ error occurs, $Z_1 Z_2$ will no longer stabilize the state, but $-Z_1 Z_2$ will. By measuring the stabilizers, an experimentalist could deduce which qubit has been impacted by an error and apply from that an appropriate recovery. This is shown in the figure \ref{fig:3bitcode}. We can notice an important property. The fact an error changes a stabilizer in its opposite is related to the fact this error either commutes or anti-commutes with the stabilizer (those are the only two possibilities as both the errors and the stabilizers are n-Pauli operators here). To see it we can notice that if a state $\ket{\psi}$ is stabilized by $g$ ($g \ket{\psi}=\ket{\psi}$), then if an error $X_i$ occurred, we have: $X_i \ket{\psi}=X_i g \ket{\psi} = + g X_i \ket{\psi}$ if $[X_i,g]=0$ and $-g X_i \ket{\psi}$ if $\{X_i,g\}=0$.: $X_i \ket{\psi}$ is stabilized by $g$ if $X_i$ commutes with $g$ and by $-g$ if $X_i$ anti-commutes with $g$. We will make use of this property in the section \ref{sec:error_detection_and_correction}.
\begin{figure}[h!]
\begin{center}
\includegraphics[width=0.9\textwidth]{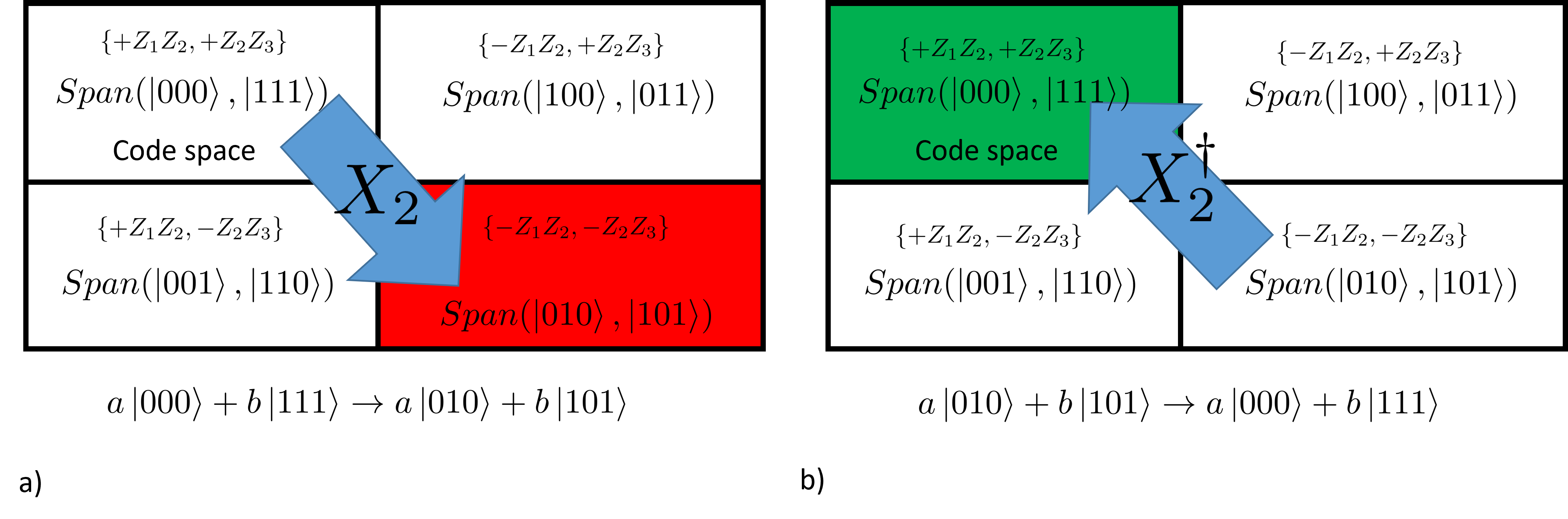}
\caption{Effect of an error and correction on a logical qubit state $a \ket{0_L} + b \ket{1_L}$ encoded on the physical qubits as $a \ket{000}+b\ket{111}$. \textbf{a)} An error $X_2$ affected the physical qubits composing the logical one. It makes the state going outside of the code space. A measurement of the stabilizers allows to find in which subspace the state has been going and to deduce from that that the error was indeed $X_2$. The state is then no longer stabilized by $\{Z_1 Z_2, Z_2 Z_3 \}$ but by $\{-Z_1 Z_2, -Z_2 Z_3 \}$ as $X_2$ anti-commuted with both $Z_1 Z_2$ and $Z_2 Z_3$. \textbf{b)} The measurement of the stabilizers allowed to identify that the error was indeed $X_2$. A correction can be applied by applying $X_2^{\dagger}$ and we recover the initial state $a \ket{000}+b\ket{111}$, putting the system back in the code space.}
\label{fig:3bitcode}
\end{center}
\end{figure}
\FloatBarrier
We see that the physics we are explaining here is analog to what we presented in \ref{fig:principle_knill_laflamme}. But it is phrased in the language of stabilizer codes, which, as we are going to see, allows to describe the physics in much simpler terms. 

We can also wonder what happens if \textit{two} errors are affecting the qubits, which is more than what the code is designed to protect. This is represented on the figure \ref{fig:3bitcode_error}.
\begin{figure}[h!]
\begin{center}
\includegraphics[width=0.9\textwidth]{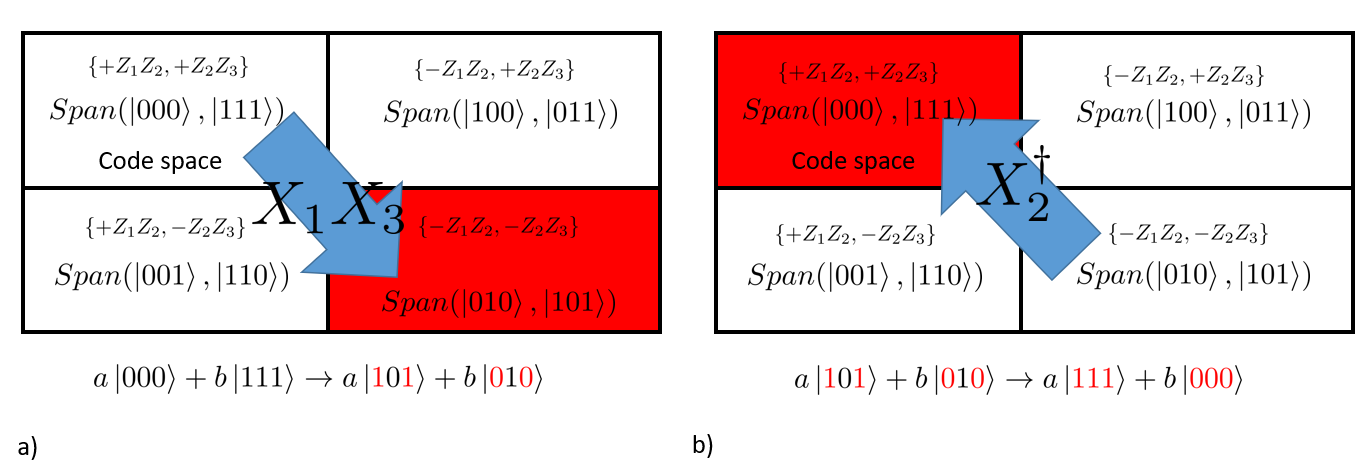}
\caption{What happens if there are more errors than what the code can protect. \textbf{a)} An error $X_1 X_3$ affected the physical qubits composing the logical one. As this operator anti-commutes with $Z_1 Z_2$ and $Z_2 Z_3$, the state is no longer stabilized by $\{Z_1 Z_2, Z_2 Z_3 \}$ but by $\{-Z_1 Z_2, -Z_2 Z_3 \}$. A measurement of the stabilizers will then wrongly interpret $X_1 X_3$ are being an $X_2$ error because the code assumes by construction that a unique bit-flip error can occur at a time.  \textbf{b)} Because the error $X_1 X_3$ was mis-interpreted as $X_2$, the correction will apply $X_2^{\dagger}=X_2$ and the state after correction becomes $X_2 X_1 X_3(a \ket{000}+b\ket{111})=(a \ket{111}+b\ket{000})$ which corresponds to a logical state $a \ket{1_L}+b\ket{0_L}$: the state of the system is back in the code-space but the value of the logical qubit is no longer the good one.}
\label{fig:3bitcode_error}
\end{center}
\end{figure}
\FloatBarrier
We see that the correction is then made in a wrong manner, and the state of the system might be back in the code space, but the logical qubit has now what is called a logical error: initially, we had a logical qubit being $a \ket{0_L}+b\ket{1_L}$, and after the error and correction, it became $a \ket{1_L}+b\ket{0_L}$. If the logical qubit was used in a computation, the outcome of the computation might then be wrong. Once a logical error occurs, it is in principle not possible to recover the initial data (excepted if, by some luck, a subsequent logical error will compensate this one).

The example of the three-qubit code is very instructive as it contains the main logic behind the stabilizer codes. For this reason, we are now ready to introduce Steane code which is able to correct for an arbitrary single-qubit error occurring on any physical qubit composing the logical one.

\section{Correcting arbitrary single qubit errors with Steane stabilizer code}
\label{sec:Steane}
The Steane code is a stabilizer code that is able to protect a quantum state from an arbitrary single-qubit error. To be able to do it, it uses $7$ physical qubits that are encoding one logical qubit. 

\subsection{The stabilizers of Steane code}
The code space of the Steane code\footnote{Actually, the definition of a code is the same as the definition of code space, the two terms are synonym: thus Steane code is actually by \textit{definition} the space that is stabilized by the stabilizers provided in \eqref{eq:list_stab_steane}.} is the space of dimension $k=1$ that is contained in the space of dimension $n=7$ stabilized by the $6$ following $7$-Pauli operators. We call those operators the \textit{stabilizers} of Steane code.
\begin{align}
& g_1=X_4 X_5 X_6 X_7 \notag\\
& g_2=X_2 X_3 X_6 X_7 \notag\\
& g_3=X_1 X_3 X_5 X_7 \notag\\
& g_4=Z_4 Z_5 Z_6 Z_7 \notag\\
& g_5=Z_2 Z_3 Z_6 Z_7 \notag\\
& g_6=Z_1 Z_3 Z_5 Z_7 \label{eq:list_stab_steane}
\end{align}
We recall that to be "stabilized" by operators means to be in the common eigenspace of eigenvalue $+1$ of those operators. This is not entirely obvious to see it here, but by performing appropriate calculations, we could show that this list of stabilizers satisfy the conditions we gave in the last paragraph of the section \ref{sec:what_are_stabilizers}. First, they all commute together. Then, any product involving those stabilizers cannot be equal to $-I$. Those two conditions ensure that the stabilized space is not reduced to $\{0\}$. They also form an independent family: any stabilizer cannot be written as a product of the others. Those conditions will imply that each of the stabilizers participates in stabilizing a smaller subspace: each one of them "divides" the dimension of the Hilbert space by two. The total Hilbert space being of dimension $2^7$, we end up with a code space being of dimension $2^7/2^6=2$: a unique logical qubit is being protected. Rigorous derivations behind our claims can be found in \cite{nielsen2002quantum}.
\subsection{Defining the logical states}
At this point, the code space is well defined. We saw that $7$ physical qubits are encoding one logical qubit, and this logical qubit lives by definition in the two-dimensional code space. But we would like to know what is now playing the role of the $\ket{0}$ and $\ket{1}$ states in the computation which we write $\ket{0_L}$ and $\ket{1_L}$ to recall that they are logical qubit states. In principle, any two orthogonal states of the code space could play this role. One possible way to define them is also based on the stabilizer formalism. For this, we can notice that for a qubit, $\ket{0}$ can be defined as the state stabilized by $Z$ (and $\ket{1}$ by the state stabilized by $-Z$). Adopting the same philosophy, we can first define the logical $Z$ operators, $Z_L$ in order to, then, define the states $\ket{0_L}$ and $\ket{1_L}$. Defining $Z_L$ actually follows the same logic as what we would need to define the code space: we want to build $Z_L$ in order that $Z_L$ added to the list of stabilizers in \eqref{eq:list_stab_steane} stabilizes a \textit{unique}\footnote{Unique up to some arbitrary global phase.} state. If the properties given in the last paragraph of \ref{sec:what_are_stabilizers} are satisfied for the resulting list, then the state will necessarily be in the code space. The stabilized state will \textit{by definition} be $\ket{0_L}$. We can also define $\ket{1_L}$ as the state stabilized by the same list, up to the modification $Z_L \to -Z_L$. One possible choice is to consider\footnote{It would satisfy the conditions given in the last paragraph of \ref{sec:what_are_stabilizers}.}:
\begin{align}
Z_L \equiv Z_1 Z_2 Z_3 Z_4 Z_5 Z_6 Z_7
\end{align}
From it (and the expression of the stabilizers), we can write down the expression of the logical states, a proper calculation \cite{nielsen2002quantum} would show that:
\begin{align}
&\ket{0_L} =\frac{1}{\sqrt{8}} \left( \ket{0000000}+\ket{0001111}+\ket{0110011}+\ket{0111100}+\ket{1010101}+\ket{1011010}+\ket{1100110}+\ket{1101001} \right)\\
&\ket{1_L}=\frac{1}{\sqrt{8}} \left(\ket{1111111}+ \ket{0101010} + \ket{1001100}+ \ket{0011001}+ \ket{1110000}+ \ket{0100101}+\ket{1000011}+\ket{0010110} \right)
\end{align}
As we see, their expression is rather cumbersome. This is also why the stabilizer formalism is used; it allows to be very economical in order to describe the quantum states. The interest in having written them is that we see that those states are really "not obvious": they are highly entangled states. But this is not a coincidence. This code is able to protect against local errors (because it can correct errors occurring on a single physical qubit). To do it, it will encode the information in a delocalized way, thus using highly entangled states.

Now that a computational basis for the logical state has been defined, an important question remains: how to manipulate this logical qubit? For any given unitary that the algorithm requires, how can we "translate" it to the logical qubit level? The answer to this question can easily be accessed by noticing that (i) the Pauli matrices $\{I,X,Y,Z\}$ form a basis for the operators acting on a qubit, (ii) $Y=iZX$. Indeed, if we are able to define the logical $X$ operator: $X_L$, then any single logical qubit operation could be described using $X_L$ and $Z_L$. For this, we can build $X_L$ in such a way that it satisfies all the properties given in the last paragraph of the section \ref{sec:what_are_stabilizers} (to make sure that $X_L$ added on the stabilizer list will stabilize a unique quantum state), and that it anti-commutes with $Z_L$ in order to properly respect the Pauli algebra. An operator satisfying those conditions is:
\begin{align}
X_L \equiv X_1 X_2 X_3 X_4 X_5 X_6 X_7
\end{align} 
And now that single logical qubit Pauli operators are defined, by using the tensor product, we can define multi logical qubit Pauli operators, and we can finally build an arbitrary operator. Thus, \textit{conceptually}, we know how to manipulate the logical state. But in practice, the situation is more complicated. Indeed what we need is to implement operations in a \textit{robust} manner such that we do not introduce "too many" errors when manipulating the qubit. This is the topic of fault tolerance that we introduce in the next pages.
\subsection{Error detection and correction}
\label{sec:error_detection_and_correction}
Now that Steane code has been described and that we understand how the logical computational states are defined, we need to understand how errors are detected and corrected. First, we recall that Steane code is able to correct for an arbitrary\footnote{As it is able to correct for $X$ and $Z$ errors, it is able to correct for an arbitrary single-qubit error from the theorem \eqref{theorem:discretization_errors}.} single-qubit error affecting \textit{one} physical qubit composing the logical qubit. Typically, if two physical qubits composing the logical one are affected by an error at the same time, it will not be possible to recover from it\footnote{This is not entirely true: Steane code cannot correct for errors such as $X_1 X_2$ but it can correct for error like $X_1 Z_2$ for instance. Typically if two \textit{different} errors are affecting two different qubits, the recovery will succeed.}. The code is first doing the detection of the error followed by the correction. We describe those two steps here.
\subsubsection{Error detection: the syndrome measurement}
Here, we are going to explain with a little bit more details the principle of error detection, also called \textit{syndrome measurement}. The principle is very similar to the measurement of the stabilizers we presented for the three-qubit code in section \ref{sec:three_qubit_code}, but in order to properly understand the correction procedure and the principle of error discretization, we consider an arbitrary error channel and explain how this correction works in practice. We will also see that even if no correction is being applied, the simple fact to detect the error already introduces a "partial" correction. 

We consider an error channel $\mathcal{E}$ described by a family of Kraus operators $\{M_k\}$. This error channel is affecting the state of the physical qubits encoding the logical qubit: $\ket{\psi_A}$. Considering in the description the environment with which $\ket{\psi_A}$ will get entangled (it is a way to model the effect of the noise as we explained in the text around \eqref{eq:Psi_purified}), we can describe the state of the system+environment once the noise has acted as $\ket{\Psi}=\sum_k M_k \ket{\psi_A} \ket{b_k}$ (the family of orthogonal states $\{\ket{b_k}\}$ are associated to the environment). Now, we can decompose the Kraus operators on some $n$-Pauli operators $\{E_i\}$ such that: $M_k = \sum_i c_{ik} E_i$ for some complex coefficients $c_{ik}$. This is always possible as the $n$-Pauli operator form a basis on which any operator acting on $n$ qubits can be decomposed \cite{nielsen2002quantum}. If we define $\ket{\widetilde{b}_i}=\sum_k c_{ik} \ket{b_k}$, we obtain:
\begin{align}
\ket{\Psi}= \sum_i E_i \ket{\psi_A} \ket{\widetilde{b}_i}
\label{eq:Psi}
\end{align}
At this point, we can clarify what it means to have only single physical qubit errors. It means that any of the $n$-Pauli operators $E_i$ actually contains \textit{one} non-trivial Pauli operator. For instance, it would mean that it is allowed to have $E_i=X_3$ but not $E_i=Z_1 Z_2$ as the latter contains Pauli operators affecting two different qubits. We can also see here that the notion of having or not having an error is less clear than with the Kraus decomposition: as the different states $\{\ket{\widetilde{b}_i}\}$ do not represent an orthonormal family, having an error $E_i$ is not entirely distinguishable from having an error $E_j$ for instance. 

Now, in order to detect the error, we have to measure all the stabilizers on the state $\ket{\Psi}$ in \eqref{eq:Psi}. In a similar fashion as what was done in the section \ref{sec:three_qubit_code}, we will make use of the that fact $E_i$ commutes or anti-commutes with any of the stabilizer $g_k$ of Steane code. If it commutes with it, $g_k$ will still stabilize the state. Otherwise, $-g_k$ will. For this reason, any term $E_i \ket{\psi_A}$ in the sum \eqref{eq:Psi} will, after measurement, bring the state in a subspace that will be stabilized by $\pm g_k$ for any stabilizer $g_k$ given in \eqref{eq:list_stab_steane}. By measuring the eigenvalues of the Stabilizers, the experimentalist will then be able to deduce which error occurred and will be able to correct it. We also notice that the syndrome measurement already performs a "partial" correction. Indeed, because a projective measurement is performed, once the syndrome measurement is being performed, any overlap of the system between different eigenspaces of the stabilizers will be destroyed. If we think about the example of the three-qubit code, the state of the system after syndrome measurement would either be in one of the four rectangles of the figure \ref{fig:3bitcode}, but it cannot extend on different rectangles at the same time. We are going to see in the section \ref{sec:FT} that this syndrome measurement will then be "enough" in practice if classical processing is keeping in memory the outcome result: the syndrome does the "most important" part of the correction.

Now we simplified a bit the story by saying that from the measurement outcome of the stabilizers, we can directly identify which error occured\footnote{For the conditions on the errors given in this paragraph, different errors can in principle lead to the same syndrome while ensuring a correction to be possible.}. We saw with the three-qubit code example that if we allowed for two physical qubit errors such as $X_1 X_3$ it is not possible to identify that it is the error that occurred. The conditions we need to allow the different errors to be corrected is that for all $(i,j)$: $E_i E_j^{\dagger}$ is either a product of the stabilizers, either anti-commutes with at least one of the stabilizers \cite{nielsen2002quantum}. For Steane code, if we assume that all the possible n-Pauli errors $\{E_i\}$ actually contain only one non-trivial Pauli operator, we can check that those conditions will be satisfied (we provided those conditions for the sake of completeness).
\subsubsection{Error correction:}

The error correction procedure is very simple; once the syndrome has been found, the experimentalist will find which error occurred, and it simply has to apply its inverse. For instance if an error $X_1$ occurred, the stabilizer of the state containing the error will be $\{g_1,g_2,g_3,g_4,g_5,-g_6\}$. The fact that the last stabilizer changed allows the experimentalist to deduce that the error $X_1$ occurred, and it will then apply $X_1$ on the system to implement the correction (of course, it works if no more than one error occurred).

\section{Fault-tolerant quantum computing}
\label{sec:FT}
\subsection{Error correction is not enough: the need for fault-tolerance}
In the previous section, we explained how the Steane code is constructed, and we gave the basic intuition behind the more general formalism of Stabilizer codes. We showed that it is conceptually possible to detect and correct errors occurring on quantum systems. But unfortunately, this is not enough in practice. Indeed, so far, we assumed that we had a given quantum state on which we want the information to be preserved, and that after the noise is occurring, we can \textit{perfectly} detect and correct for the errors it induced. This is, of course, an unrealistic assumption: those operations will be noisy as well. Fault-tolerant quantum computing is the step further. It explains how it is possible to implement \textit{successfully} quantum error correction, acknowledging that this is also a noisy operation by providing explicit circuit construction. Those circuits will be at the roots of the energetic estimations we will do in the next chapters. It will also provide us with all the theoretical tools required to understand how to scale up the level of protection: we will see that given some conditions, the effective noise felt by the logical qubits can be put as close to zero as desired.

To understand the issue with the sole use of error correction, let us design a simple circuit that allows measuring all the stabilizers of the Steane code. In order to do so, we will need the following property that allows us to measure observables easily. Its proof is given in the appendix. 
\begin{property}{Measuring observable having $\pm 1$ as eigenvalues}
\label{prop:meas_obs}

To measure an observable $M$ (possibly acting on multiple qubits) that admits eigenvalues $\pm 1$, one can design the following circuit.
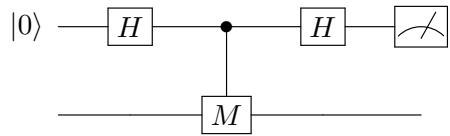
\begin{figure}[h!]
\centerline{
\Qcircuit @C=1.7em @R=1.7em {
\lstick{\ket{0}} & \gate{H} & \ctrl{1} & \gate{H} & \meter \\
 & \qw & \gate{M} & \qw & \qw \\
}}
\caption{The measurement in the $\sigma_z$ basis of the ancilla qubit measures the observable $M$ of the system measured. The gates $H$ are Hadamard \cite{nielsen2002quantum} operations. The black dot represents a controlled operation. Here we thus apply $M$ on the system qubit (second line) if the ancilla qubit (first line) is in the state $\ket{1}$ and we don't apply it otherwise. The effect for superposition of the control are deduced by the linearity of unitary transformations.}
\label{measurement_M}
\end{figure}
\FloatBarrier
If the ancilla (top line) is being found in $\ket{1}$ after measurement, it means that the system (bottom line) is in $\ket{-_M}$. If the ancilla is being found in $\ket{0}$ it means that the system is in $\ket{+_M}$. Where $\ket{\pm_M}$ are eigenstates associated to the eigenvalue $\pm 1$ of $M$.
\label{measuring_observable_pm1}
\end{property}
Based on this property, we can consider measuring the stabilizers of Steane code, i.e., perform the syndrome measurement, by using the circuit represented on the figure \ref{fig:meas_stab_Steane}. The correction is not represented, but from the measurement outcomes, one would simply have to apply on the physical qubits the appropriate unitary to correct as we explained in the section \ref{sec:error_detection_and_correction}. We also take the opportunity to introduce some further definitions. The top $7$ physical qubits in these circuits are the qubits composing the logical qubits. We call them physical \textit{data} qubits as they are qubits encoding the information used by the algorithm. The bottom $6$ qubits represented are the physical ancilla qubits. They are not data qubits as they don't directly participate in the implemented algorithm: they are only here to implement the error correction procedure. 

At first view, we could believe that we just have to implement this circuit to detect the errors. But the issue is that each of the gates in this circuit might induce errors. It is then not guaranteed at all that we are able to correct errors with it. It might be the opposite; it could be possible that \textit{more} errors would be introduced because all the gates required might be noisy. As a concrete example, let us imagine that an error occurred before one of the controlled gates. This error might propagate on the other physical data qubits involved in this gate in such a way that a unique gate that "failed" (we give a precise definition of gate failure in the following section) during this implementation will induce errors occurring on multiple physical data qubits. The sole fact to try detecting the errors will introduce more errors and actually make the situation worse. We need \textit{robust} implementation of quantum error correction for which the explanations of the rest of this chapter are dedicated to. 
\begin{figure}[h!]
\centerline{
\Qcircuit @C=1.7em @R=1.7em {
 & \qw & \qw & \qw & \qw & \sgate{X}{2} & \qw & \qw & \sgate{Z}{2} & \qw & \qw  \\
 & \qw & \qw & \qw & \sgate{X}{1} & \qw & \qw & \sgate{Z}{1} & \qw & \qw & \qw \\
 & \qw & \qw & \qw & \sgate{X}{3} & \sgate{X}{2} & \qw & \sgate{Z}{3} & \sgate{Z}{2} & \qw & \qw \\
 & \qw & \qw & \sgate{X}{1}& \qw  & \qw & \sgate{Z}{1} & \qw & \qw & \qw & \qw \\
 & \qw & \qw & \sgate{X}{1} & \qw & \sgate{X}{2} & \sgate{Z}{1} & \qw & \sgate{Z}{2} & \qw & \qw \\
 & \qw & \qw & \sgate{X}{1} & \sgate{X}{1} & \qw & \sgate{Z}{1} & \sgate{Z}{1} & \qw & \qw & \qw \\
 & \qw & \qw & \gate{X} & \gate{X} & \gate{X} & \gate{Z} & \gate{Z} & \gate{Z} & \qw & \qw \\
 & \lstick{\ket{0}} & \gate{H} & \ctrl{-1} & \qw & \qw & \qw & \qw & \qw & \gate{H} & \meter \\
 & \lstick{\ket{0}} & \gate{H} & \qw & \ctrl{-2} & \qw & \qw & \qw & \qw & \gate{H} & \meter \\
 & \lstick{\ket{0}} & \gate{H} & \qw & \qw & \ctrl{-3} & \qw & \qw & \qw & \gate{H} & \meter \\
 & \lstick{\ket{0}} & \gate{H} & \qw & \qw & \qw & \ctrl{-4} & \qw & \qw & \gate{H} & \meter  \\
 & \lstick{\ket{0}} & \gate{H} & \qw & \qw & \qw & \qw & \ctrl{-5} & \qw & \gate{H} & \meter  \\
 & \lstick{\ket{0}} & \gate{H} & \qw & \qw & \qw & \qw & \qw & \ctrl{-6} & \gate{H} & \meter \\
}}
\caption{Syndrome measurement of Steane code. The top $7$ qubits are the physical (data) qubits, i.e., the qubit that compose the logical qubit. The top $6$ below are the ancilla qubits performing the syndrome measurement. The $H$ gates are Hadamard gates. The gates involving a black dot represent controlled operations. For instance, the first of those gates, i.e., the further on the left, will apply $X_4 X_5 X_6 X_7$ on the data qubits if the first ancilla is in the state $\ket{1}$. If it is in $\ket{0}$ the identity is being applied on the data qubits. Not represented here is the correction procedure in itself that must be applied afterward in order to correct. Based on the property \ref{prop:meas_obs}, this circuit measures the six stabilizers of Steane code provided in \eqref{eq:list_stab_steane} (each measurement outcome will directly give the eigenvalue of the associated stabilizer).}
\label{fig:meas_stab_Steane}
\end{figure}
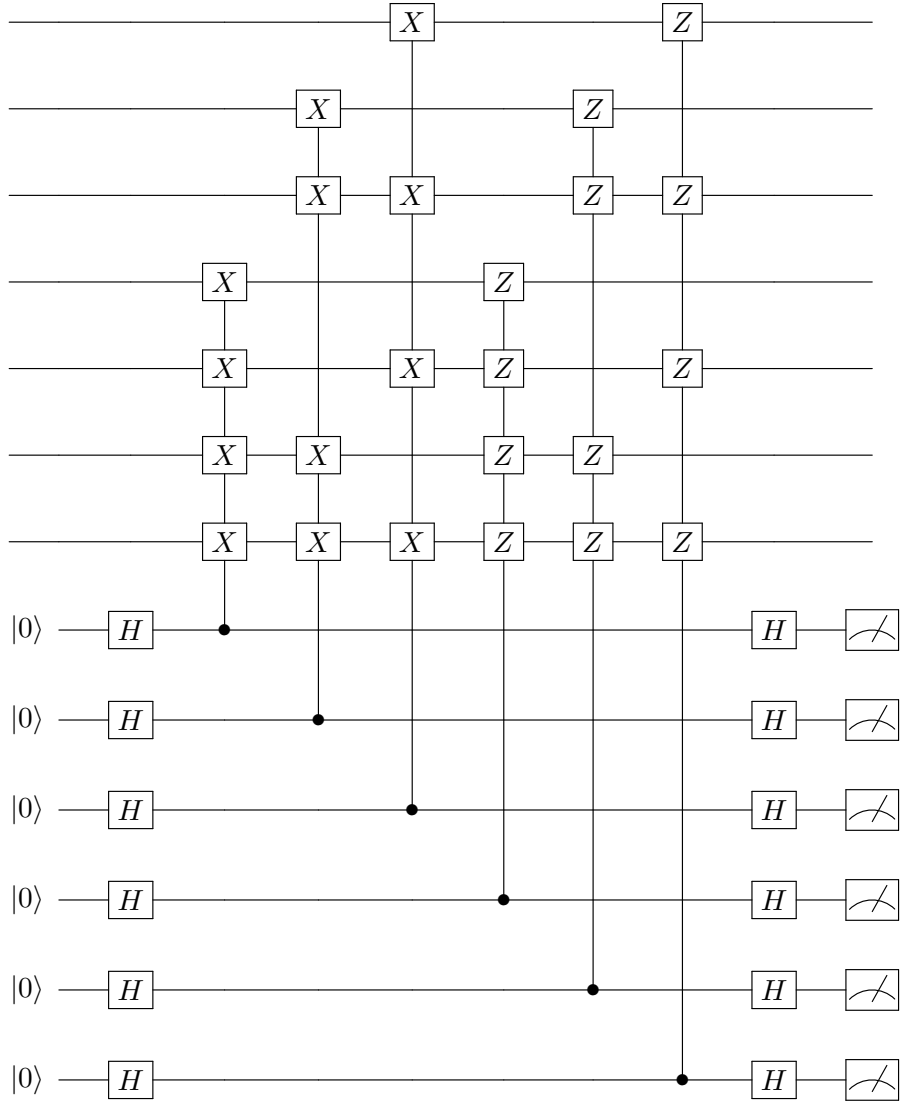
\FloatBarrier
\subsection{The principle of fault-tolerant quantum computing}
One of the main ideas behind fault-tolerant quantum computing is to perform error correction before that errors propagate "too much". To illustrate what we mean, we can take a look at Figure \ref{fig:fault_errors}.
\begin{figure}[h!]
\begin{center}
\includegraphics[scale=0.5]{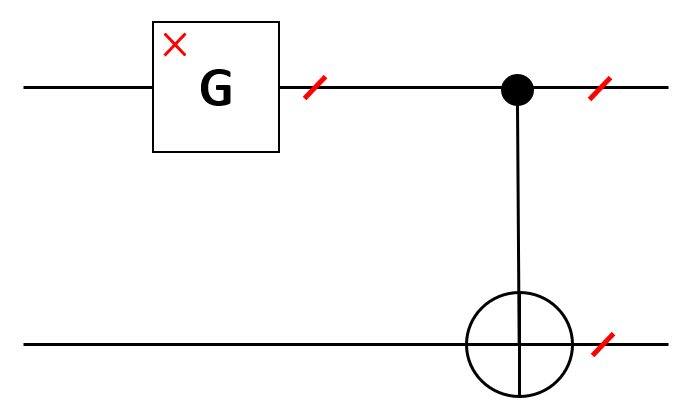}
\caption{Each black line represents a physical qubit. The red cross represents a fault: a position in the circuit where a gate "failed". The oblique red lines represent errors: physical qubits composing the logical one that has been corrupted. If an error occurred on the top qubit, it might propagate to the bottom qubit because of the two-qubit cNOT gate. This kind of behavior is to avoid for a robust design.}
\label{fig:fault_errors}
\end{center}
\end{figure}
Let us imagine that the top physical qubit has been "badly" manipulated in the gate $G$. Then, this qubit will have what we call an error: it is not in the state it is supposed to be. And this error might propagate as soon as a two-qubit gate involves this qubit, such as a cNOT with another qubit. One error will induce two errors simply because of interactions. This propagation of errors is something to avoid at all costs for the design of a robust quantum computer because the number of errors might increase in an uncontrollable manner. Fault-tolerant quantum computing provides an approach to design a quantum computer in a way that errors do not propagate "too much" before correction is applied. It is a way to design circuits that allows keeping the situation "under control".

\subsubsection{Fault and errors}
\label{sec:fault_and_errors}
The starting point is to make a clear distinction between errors and faults. For this we can still refer to Figure \ref{fig:fault_errors}. A fault (represented by a red cross on this figure) is a position in the circuit where a gate did not operate as expected. It thus induces errors (represented by oblique red lines) on qubits. Thus, the notion of errors is a concept related to the qubits, while the notion of faults is a concept related to (noisy) gates. To say things simply, a fault is the reason why errors are occurring. But to understand better we need to be more precise. Indeed, in general, all the gates will introduce a little bit of noise on the manipulated qubits. The answer to that question can rely upon the error discretization principle we already talked about, but in this context, we need to add one extra layer in the description. Indeed now we do not only have noise that is acting, but we also desire to manipulate the qubits. The issue is that both noise and gate operation are acting "at the same time", there is not, at first view a well-defined "noise" that is acting separately from the gate. This is why we are going to introduce the concept of \textit{noise channel} of a quantum gate, which will allow us to actually "separate" the noise introduced by a gate from its ideal implementation. We take an example to explain it. Let's assume that we are implementing a quantum circuit. This circuit will, in practice, be implemented by a succession of elementary operations that experimentalists are able to implement in the laboratory. This is what we are going to call quantum gates. Calling $N$ the number of (noisy) quantum gates that there are inside the circuit, assuming that each quantum gate can be modeled by a quantum channel, the total evolution of the circuit can be written as:
\begin{align}
\mathcal{G}=\mathcal{G}_N \circ \mathcal{G}_{N-1} \circ ... \circ \mathcal{G}_1,
\label{eq:seq_N_noisy_gates}
\end{align}
where each of the $\mathcal{G}_i$ is a quantum channel\footnote{In a general circuit, gates may be acting in parallel (which is mathematically described by a tensor product), or in sequence (which is described as a composition). By taking the convention that $\mathcal{G}_i$ will apply identity operation on the qubits not involved in the dynamic of the gate, we can always write the evolution of the general circuit as \eqref{eq:seq_N_noisy_gates}}. We can \textit{define} the noise channels associated with each of those maps as the maps $\mathcal{N}_i$ which satisfy:
\begin{align}
\mathcal{G}_i = \mathcal{U}_i \circ \mathcal{N}_i,
\label{eq:def_noise_channels}
\end{align}
where $\mathcal{U}_i$ is the ideal (unitary) operation that we tried to implement. This definition\footnote{Because $\mathcal{U}_i$ is unitary, it admits an inverse; thus, there is only one map $\mathcal{N}_i$ satisfying \eqref{eq:def_noise_channels}, it is $\mathcal{N}_i \equiv \mathcal{U}_i^{-1} \circ \mathcal{G}_i$.} means that a noisy gate can mathematically be seen as an ideal gate preceded by some noise map (even though in the laboratory the gate and the noise are acting "at the same time"). The fact we have chosen to make the noise acting "before" the gate is a matter of convention; we could have considered that it acts "after" by reversing the order of the composition in \eqref{eq:def_noise_channels} (it would change the expression of $\mathcal{N}_i$). Now, to understand what it means to have a fault at some position in the circuit, we can use again the vision we used around \eqref{eq:Psi_purified} in which noise acting on a system can be modeled by considering that the system of interest is getting entangled with some environment. We call $\{M^i_k\}$ the Kraus operators associated to $\mathcal{N}_i$. The density matrix at the end of the evolution: $\rho'=\mathcal{G}(\ketbra{\psi}{\psi})$ ($\ket{\psi}$ being the initial quantum state at the beginning of the algorithm) can be written as $\rho'=Tr_{E}(\ketbra{\Psi}{\Psi})$ where $E$ is the environment that has been introduced in the modeling:
\begin{align}
\ket{\Psi} = \sum_{k_1,...,k_N} U_N M^N_{k_N}...U_1 M^1_{k_1} \ket{\psi} \ket{\epsilon_{k_1,...,k_N}},
\end{align}
where $\{\ket{\epsilon_{k_1,...,k_N}}\}_{k_1,...,k_N}$ is a family of orthonormal vectors belonging in the environment $E$ used to purify. We can now follow the same approach we did around \eqref{eq:Psi}, and decompose each of the Kraus operator in the $n$-Pauli basis family, which allows us to rewrite $\ket{\Psi}$ as:
\begin{align}
\ket{\Psi} = \sum_{j_1,...,j_N} U_N E_{j_N}...U_1 E_{j_1} \ket{\psi} \ket{\widetilde{\epsilon}_{j_1,...,j_N}},
\label{eq:purification_faults}
\end{align}
where each $E_{j_i}$ is an $n$-Pauli operator, and where we have introduced the new family of states of the environment: $\ket{\widetilde{\epsilon}_{j_1,...,j_N}}$ (which is not orthonormal in general).

At this point, we can properly define what we mean by "a fault". Let us consider a non vanishing term in the sum \eqref{eq:purification_faults} (thus when $||\ket{\widetilde{\epsilon}_{j_1,...,j_N}}|| \neq 0$). We will say that the $l$'th gate in $U_1 E_{j_1}...U_N E_{j_N} \ket{\psi} \ket{\widetilde{\epsilon}_{j_1,...,j_N}}$ had $p$ faults if $E_{j_l}$ contains $p$ non trivial Pauli operators in its tensor product. For instance, if $E_{j_l}$ is equal to $X_1 X_2$, $Z_1 X_6$ or $Y_1 Z_2$ it would contain two faults with this definition. Of course, a $Q$-qubit gate can contain a maximum of $Q$ faults. We notice that we keep having some "quantum fuzziness" about the number of faults occurring: it is not very clear to know how many faults occurred, the state $\ket{\Psi}$ being a superposition between different states composed of a varying number of faults. It is only for each of the terms in this sum that the number of faults is well defined. 

In the same spirit, we can also revisit the notion of errors to make the connection with faults. An equivalent way to describe the noisy evolution is to define the quantum channel $\mathcal{N}$ such that $\mathcal{G}=\mathcal{N} \circ ( \mathcal{U}_N \circ ... \circ \mathcal{U}_1)$ (here we basically introduced the unique noise channel $\mathcal{N}$ occurring after the ideal evolution of the $N$ gates). Calling $\ket{\psi}_{\text{ideal}}=U_1 ... U_N \ket{\psi}$, we can purify the final quantum state (which is physically the same as the one in \eqref{eq:purification_faults}) and decompose the purification on the basis of $n$-Pauli operators as we did before. We find:
\begin{align}
\ket{\Psi}=\sum_i E_i \ket{\psi_{\text{ideal}}} \ket{\widetilde{\epsilon}_i} \label{eq:error_state_definition},
\end{align}
where again $E_i$ is an $n$-Pauli matrix, and the $\ket{\widetilde{\epsilon}_i}$ belong to a family of states of some environment that are not necessarily orthogonal between each other. Then, we can take any term inside this sum, and we will say that the associated quantum state had $p$ \textit{errors} if the associated $n$-Pauli operator contains $p$ non-trivial Pauli operators in its tensor product. We see that the errors are really associated with the quantum state as there is no notion of "where in the circuit" this $n$-Pauli operator occurred as opposite as what we discussed for faults. Now, surely, if no fault occurred, no error would be in the final state. But there is, in general, no direct "easy" relationship between the number of faults and the number of errors at the end of the process. One could, for instance, imagine that two consecutive faults compensate each other such that the associated event will be associated with no error. Or, more critically, a single fault could induce an error that would propagate into many errors. This is what would occur in Figure \ref{fig:fault_errors} if the fault associated with $G$ is an $X$ Pauli operator; after the cNOT, both qubits will have $X$ errors. We discuss the propagation of errors for different quantum gates in the appendix \ref{app:error_propagation}.

In summary, we saw that faults are associated with quantum gates that did not operate as expected. Errors are associated with quantum states and represent the fact that a quantum state is not in the state it is supposed to be. To make a precise definition of those concepts, we need to decompose the evolution on a basis (here Pauli operators). Once it is done, it is possible to clearly identify a number of faults or errors only for each term within those sums. The reason why such approaches are considered is that they will allow us to properly choose conditions a circuit must satisfy in order to avoid a too big propagation of errors. 

\subsubsection{Avoiding errors to propagate} 
\label{sec:axiomatic_FT}
Now that we have precise definitions of faults and errors, we can give the principle behind fault-tolerant construction. All the results presented here assume that the noise model behind the different gates is local, which actually means that each quantum gate can be modeled by a quantum channel as we already used in \eqref{eq:seq_N_noisy_gates}.
Those were some of the initial assumptions behind the quantum threshold theorem \cite{aliferis2005quantum,knill1998resilient,preskill1998reliable},  which has then been extended to non-Markovian noise\footnote{To be precise, the fact that individual quantum gates can be modeled by quantum channels includes some cases of non-Markovianity. But it does not include the cases in which the memory effect between two \textit{different} gates are not negligible. In \cite{terhal2005fault}, it is shown that the result of fault tolerance can be extended to this case.} \cite{terhal2005fault}, and to non-local, long-range correlated noise \cite{aharonov2006fault}.

We wish to design circuits that are robust to avoid error propagation, and which implement quantum error correction in order to protect the logical qubits. The basic idea is to replace each gate in an algorithm by its implementation on logical qubits, followed by quantum error correction. We call the level-0 concatenation a quantum circuit that implements an algorithm without quantum error correction. Any gate in the level-0 concatenation will be called 0-Ga, (for level-0 gate), whatever the exact single or two-qubit gate it is. A 0-Ga is thus implemented on physical qubits in the exact same manner as the one described by the algorithm. We call the level-1 concatenation (or first level of concatenation) an algorithm in which each of the physical qubits are replaced by logical qubits, and in which each 0-Ga is replaced by what we call 1-Rec, for level 1 rectangle\footnote{The element of language "Ga", "Ec", "Rec", "exRec" that we will use in this chapter are the conventional one used for instance in \cite{aliferis2005quantum}.}. The reason why it is the "first" level is because such construction would only improve the accuracy once. We will see in the following section how to increase even more the level of protection by performing more concatenations. A 1-Rec is an entity that is acting on the logical level and is composed of two elements. First, there is the implementation on the logical level of the 0-Ga gate itself; we call it the 1-Ga (for level-1 gate). This entity performs on a logical level the same operation as what the 0-Ga did (if without error correction we wanted to implement a Hadamard gate, the 1-Ga associated will implement the Hadamard gate but now on the logical level). In the 1-Rec, the 1-Ga is then followed by error correction that we call 1-Ec. The 1-Ec will first consist in performing the syndrome measurement, and then applying the appropriate correction \footnote{We will see that implementing the correction is not always necessary if the experimentalist keeps track of the errors that occurred.}, see figure \ref{fig:1_Rec}. In summary, a 1-Rec is then simply the translation of a 0-Ga now on the logical level followed by error correction. On a physical level, both 1-Ec and 1-Ga are composed of many physical gates (thus 0-Ga) that allow to perform the appropriate operations on the logical level, as represented in the figure \ref{fig:inside_level1}. To be a little bit more precise, if the 0-Ga gate was a two-qubit gate, we actually need to implement error correction on each of those two qubits when they are being replaced by logical ones: we would need \textit{two} 1-Ec boxes as represented on the figure \ref{fig:1_Rec}.
\begin{figure}[h!]
\begin{center}
\includegraphics[scale=0.5]{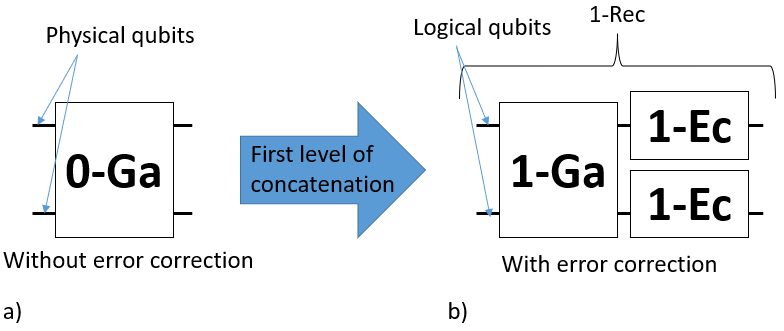}
\caption{\textbf{a)}: A gate (here a two-qubit gate) implemented on physical qubits is called 0-Ga. \textbf{b)}: When performing error correction (i.e., applying the first level of concatenation), the physical qubits are replaced by logical qubits, and the gate is replaced by a 1-Rec that is defined as a 1-Ga followed by one 1-Ec for each logical qubit on which the 1-Ga is acting. The 1-Ga performs the same operation as 0-Ga in terms of information processing, but it now acts on logical qubits. The 1-Ec are performing error correction on each logical qubit after that the gate has been applied. As 1-Ga, 1-Ec, 1-Rec are acting on a logical level, they are composed of physical elements: physical gates (i.e. 0-Ga) that are acting on the physical qubits.}
\label{fig:1_Rec}
\end{center}
\end{figure}
\FloatBarrier
\begin{figure}[h!]
\begin{center}
\includegraphics[scale=0.5]{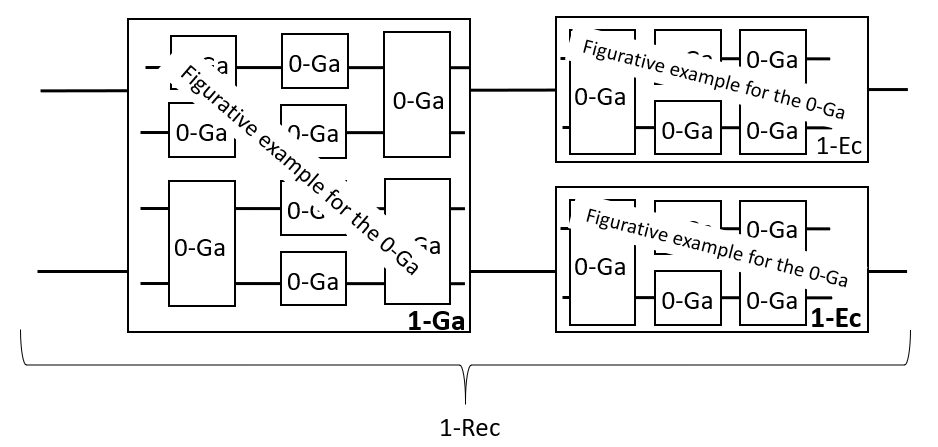}
\caption{Inside a 1-Ga or a 1-Ec, many physical gates are acting on the physical qubits to perform the appropriate logical gate (1-Ga) or error correction (1-Ec). The exact 0-Ga gates represented are just here for the illustration: the information to understand here is that inside 1-Ga and 1-Ec there are 0-Ga elements.}
\label{fig:inside_level1}
\end{center}
\end{figure}
\FloatBarrier
Now, it is possible to provide a list of axioms \cite{aliferis2005quantum} that, if satisfied, allow us to be certain that the outcome of the algorithm matches the ideal (i.e., noiseless) algorithm would provide. We will assume for simplicity that the preparation and the measurements of the logical qubits can be done perfectly (otherwise, we would have to enter in a little bit more details). The goal consists in demanding that (i) the faults must be sufficiently sparse inside the algorithm and (ii) the 1-Ga, 1-Ec do not propagate errors on different physical qubits composing a given logical one, nor create multiple errors inside a logical qubit in the case they are faulty, i.e., when they contain a 0-Ga that had a fault. If such requirements are satisfied, \textit{then}, it is possible to show that the algorithm will be successfully \textit{simulated}: i.e., its outcome will match the outcome of the ideal algorithm. 

In practice, we will ask for (i) that there is a maximum of one fault that can occur in what we call a 1-exRec for level-1 extended rectangle (this is the last technical term we introduce). A 1-exRec is composed of a 1-Rec in which we add the previous 1-Ec as represented by the colored rectangles on the figure \ref{fig:level_1_simulation_without_measprep}. The 1-exRec are central elements in the construction as they represent the appropriate group of components allowing to understand why a logical gate can be more resilient to faults and errors. The assumption of one fault maximum per 1-exRec is our assumption of "the faults are sufficiently sparce". For (ii), we will ask 1-Ga and 1-Ec to be fault-tolerant according to the two definitions that follow. In those definitions, we will say that a logical qubit contains (or not) one error. To say that a logical qubit contains (or not) one error means that its quantum state is "affected" (or not) by one error (thus by a 1-Pauli operator in the sense of the definition for errors given around \eqref{eq:error_state_definition}).
\begin{definition}{Fault-tolerant design for 1-Ga}

\label{def:FT_1Ga}
A 1-Ga will be said to be implemented fault-tolerantly if:
\begin{itemize}
\item If it contains no fault, and if there are no errors for the logical qubits at its input, then its output will contain no errors.
\item If it contains no faults, and there is one error for one of the logical qubits at its input, then the logical qubits at its ouput\footnote{There will have one logical qubit at its output for a one qubit logical gate, but two for a two-qubit logical qubit gate.} will contain one error per logical qubit maximum.
\item If there is one fault inside, and the logical qubits at its input contain no errors, then the logical qubits at its output will contain one error per logical qubit maximum.
\end{itemize}
\end{definition}
The first item asks that the 1-Ga does not introduce errors if no fault occurred and no errors were initially there. The second item asks that in the case there was initially an error, the 1-Ga does not propagate this error "too much", more precisely not on multiple physical qubits within a logical qubit. What the last item asks is somehow in the same spirit with the difference that it is now the gate that is being faulty while the input was good. It demands that the gate does not create too many errors if only one fault occurred. We can also ask for properties that a fault-tolerant 1-Ec should satisfy.
\begin{definition}{Fault-tolerant design for 1-Ec}

\label{def:FT_1Ec}
A 1-Ec will be said to be implemented fault-tolerantly if:
\begin{itemize}
\item If it contains no fault, any input logical qubit having a maximum of one error will be mapped to an output having no errors.
\item If it contains one fault, for any input logical qubit having no error, the logical qubit at its output will have one error maximum.
\end{itemize}
\end{definition}
The first item basically asks for an error correction that "properly does the job": it will detect and correct a potential error if it is not faulty\footnote{And of course not create one if there were no error initially.}. With this axiom, we assumed implicitly that the code considered can correct for one error maximum, which is the case for Steane code. The second item asks to avoid that it creates errors on multiple physical qubits composing a logical one in the case it was faulty (while having a good input). Thus it avoids that the 1-Ec will create "too many errors" if it is faulty.

We are now ready to understand the example represented in the figure \ref{fig:level_1_simulation_without_measprep}. We assume working with a unique logical qubit from the beginning to the end of the algorithm (this is also a simplification in the explanations, but the principle behind it would still work for any algorithm). The faults are still represented by the red crosses, while errors are symbolized by red oblique lines. We assume that the first 1-exRec receives at its input a logical qubit containing no errors. We also assume that this 1-exRec contains a fault in the first 1-Ec. As the input logical qubit contained no errors, the output will contain one error maximum from the second item of definition \ref{def:FT_1Ec}. The logical qubit then goes through the 1-Ga, which cannot be faulty because of the sparse assumption. This error may thus "remain" \textit{but will not be amplified}, i.e., there will still have a unique error on the logical qubit after the 1-Ga as a consequence of the second item in the definition \ref{def:FT_1Ga}. The logical qubit then enters in the last 1-Ec of the first 1-exRec which removes the errors from the first item of \ref{def:FT_1Ec} as this 1-Ec is not faulty (because again, there can only have one fault per 1-exRec). This 1-Ec is also part of the second 1-exRec (as they are overlapping), and this second 1-exRec can contain one fault that we assumed to be in the 1-Ga. This fault may induce an error on the logical qubit from the third item of definition \ref{def:FT_1Ga} that will be again corrected by the following 1-Ec. In the last 1-exRec we assumed that the last 1-Ec was faulty. It creates one error on the logical qubit. But this error will then be corrected by the following 1-Ec in the (non represented) next 1-exRec. We see that fault tolerance is basically asking for circuit designs that avoid errors propagating in an uncontrollable manner such that it is possible to keep the errors under control. In those examples, we saw that errors occurring are necessarily corrected a few steps later on. 

Here we did not talk about what happens on the boundaries: for the logical qubit preparation (called 1-prep) and its measurement (called 1-measurement). Basically, the same kind of axioms will be associated with those components, with some slight changes. For instance, we ask for the 1-measurement to be more "resilient" to errors and fault. As an example, we ask it to provide a correct output even if there were an error at its input (while it was not faulty). Thus having a unique error before the final measurement (it is what we would have in figure \ref{fig:level_1_simulation_without_measprep} if we assumed that the logical qubit is then measured) would not be an issue. Further details and the exact assumptions to make for the 1-preparation and 1-measurement can also be found in \cite{aliferis2005quantum}.
\begin{figure}[h!]
\begin{center}
\includegraphics[scale=0.5]{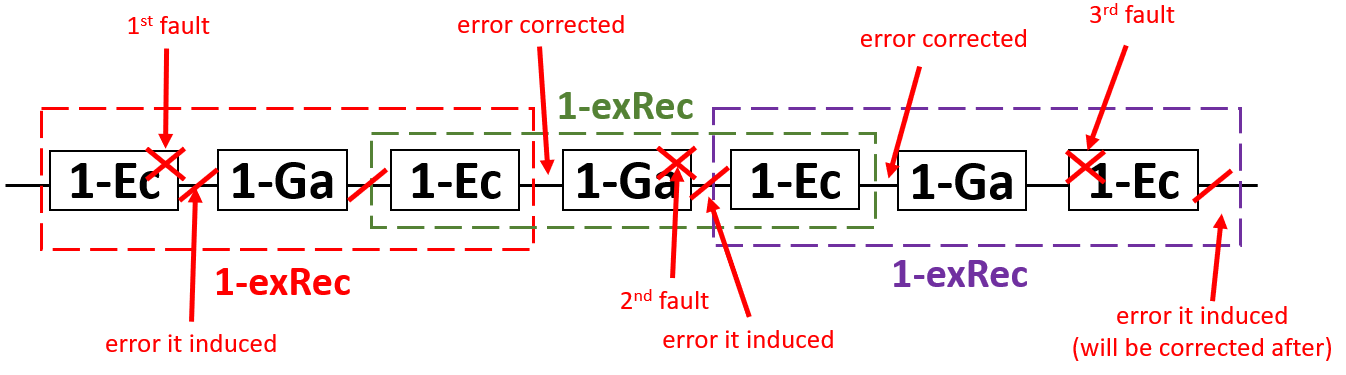}
\caption{example of a successful first concatenation level. The colored rectangles represent what is called a 1-exRec, which corresponds to a 1-Ga preceded and followed by 1-Ec. The 1-exRec are overlapping: the third 1-Ec belongs in the first and second 1-exRec at the same time. We randomly chose some faults (represented by the red crosses) occurring in a way that there is a unique fault per 1-exRec (it is the assumption of "the error are sparse". As explained in the main text, from the assumption we made on the behavior of the 1-Ec and 1-Ga, the errors (represented by the red oblique lines) are kept under control from error correction.}
\label{fig:level_1_simulation_without_measprep}
\end{center}
\end{figure}
\subsubsection{Quantitative estimation of the first level of protection}
\label{sec:quantitative_estimation_first_level_protection}
Now, we can distinguish the two requirements we made. We asked that the 1-Ga, 1-Ec satisfy some properties about error propagation. Those properties are what will constraint the circuits that are implementing the operations. As we will see, circuits satisfying those constraints exist, such that this assumption is fair to make. The other assumption was based on the fact the faults are occurring in a sufficiently sparsed manner. This is not guaranteed in general, but what can happen is that if the noise is local and weak enough, then the errors will likely be sparse. 

We are now going to estimate in a quantitative manner how much the effect of the noise can be reduced by applying quantum error correction. Those quantitative estimations will assume that the noise channels of the different gates can be described as probabilistic Pauli noise (applied on each physical qubit composing the logical one). Those assumptions will remain along with all this thesis (up to one exception in the section \ref{sec:long_range_correlated_crosstalk} of the next chapter). The principle of fault-tolerance does not depend on the specific noise model for the gates (as soon as the noise is local), but the quantitative estimations about how much error correction improves the situation depend on it. A probabilistic Pauli noise acting on a single qubit is a noise channel satisfying:
\begin{align}
\mathcal{N}(\rho)=p_0 \rho + p_1 X \rho X + p_2 Y \rho Y + p_3 Z \rho Z
\end{align}
where the family $\{p_i\}$ correspond to probabilities, i.e $\sum_i p_i=1$ and $\forall i, p_i \geq 0$.

Defining $\eta$ as being the biggest probability to have a physical gate that failed (even if we model the noise as probabilistic Pauli, different gates can have different strength for this noise), i.e., that it will apply at least one non-trivial (i.e., not identity) Pauli operator on the qubit(s) that it manipulates,  we can find that without error correction, an algorithm composed of $N_L$ gates will provide a wrong answer with a probability that is upper bounded by $N_L \eta$. To obtain this, we just sum the probability that each gate fails, under the assumption $\eta \ll 1$, which is a general assumption made in fault-tolerance (the noise must be low in order for error correction to be useful). 

To improve the level of protection, we can do one level of concatenation. The first concatenation level will have a probability to fail that is upper bounded by the probability that any of the 1-exRec failed. We say that a 1-exRec fails if there is at least a \textit{pair} of faults that is occurring inside. The reason why looking at pair of faults is relevant can be understood from the figure \ref{fig:level_1_simulation_without_measprep}: if a unique fault occurred in a 1-exRec we found that the errors did not propagate (and the algorithm would provide a correct answer). If a pair of faults are occurring within a 1-exRec then the situation might lead to uncorrectable errors. Calling $A$ the number of fault locations (i.e. the number of "places" in which faults can occur), under the assumption $\eta \ll 1$, the probability that one 1-exRec fails is then upper bounded by:
\begin{align}
p_L^{(1)}(\eta)=\binom{A}{2} \eta^2.
\label{eq:pL1}
\end{align}
The number $A$ can be considered here to be equal to the number of physical gates inside a 1-exRec (it is very close to it). But strictly speaking, it is not exactly equal to it. Indeed, for instance, the initialization of a physical qubit in $\ket{0}$ is not considered as a gate. But it is a place where something could go wrong (a bad initialization could occur). Thus it counts as a fault location. The difference between the number of physical gates and fault location being small, we will assume they are the same in those explanations.
  
The probability that the first concatenation level fails is then simply upper bounded by $N_L p_L^{(1)}$ as the number of gates the algorithm has to implement is equal to the number of 1-exRec (because each 1-Ga is inside a unique 1-exRec). As soon as $p_L^{(1)} < \eta$, the error correction would have increased the accuracy of the computation. This condition would mean that the faults are "sparse" enough. We understand from \eqref{eq:pL1} and the algorithm failure probability that the 1-exRec is the good group of components to consider when we need to think about a logical gate. Indeed, the probability that a 1-exRec fails plays a similar role as the probability that a physical gate failed without error correction played.

Now, this construction would work to improve protection once. What we would like is to find a way to "scale up" the protection and to make the algorithm able to resist against two faults, three faults, ... $k$ faults occurring in those exRecs. One way to do it is to "increase" the concatenation level, as we are now going to explain.
\subsubsection{Improving the protection to an arbitrary level: the principle of concatenations.}
\label{sec:arbitrary_accurate_quantum_computing}
The principle of concatenations is one that allows to provide an arbitrarily high level of protection given the fact that the probability that a physical gate is faulty is below some threshold. The overall principle is based on a Russian-dolls like construction. We keep the example of an algorithm composed of $N_L$ gates, and we still assume that the faults on physical gates are occurring with a probability $\eta \ll 1$. Without error correction, we "directly" implement the algorithm on physical qubits, and the probability that it fails can be estimated as $N_L \eta$. To improve the protection, we can do one concatenation. We replace each of the 0-Ga gates by a 1-Rec. As explained before, the probability that the first concatenation level would fail is now $N_L p_L^{(1)}$, with $p_L^{(1)}$ defined in \eqref{eq:pL1}. 

To increase further the level of protection, we can apply "more" quantum error correction. We apply the \textit{exact same recipe} we already applied: we replace each of the 0-Ga contained in the circuit, implementing the first concatenation level by a 1-Rec. This is called the second level of concatenation. Thus if we had one physical gate at the very beginning, this gate would have been replaced by a 1-Rec, which is composed of physical gates (this was the first level of concatenation). And we now replace all the physical gates again inside this 1-Rec by "new" 1-Rec (it defines a 2-Rec). This is really analog to a Russian dolls construction; the figure \ref{fig:recursive_concatenation} shows the particular example we are presenting and the general "Russian dolls" philosophy. 

In some sense, we can say that we do "twice" more error correction. Indeed, error correction is implemented on level-2, but also on level-1 "inside" this level-2. A 1-exRec (in the level-1) will fail if a pair of physical gates are faulty. We already estimated the probability for such event: $p_L^{(1)}$ in \eqref{eq:pL1}. But we also implement error correction on level-2. And what played the role of physical gates inside the level-1 is now being played by the 1-exRec. The construction at level-2 will then fail if a pair of 1-exRec are failing. From this, we can find the probability that a 2-exRec fails as being: \begin{align}
p_L^{(2)}(\eta)=\binom{A}{2} \left(p_L^{(1)}(\eta)\right)^2.
\end{align}
where a 2-exRec is formally defined as a 1-exRec in which all the 0-Ga have been replaced by 1-Rec. But formally, we should really think of $p_L^{(2)}$ as being the probability that a gate implemented by the algorithm fails when two concatenations are being performed. It is the exact analog of $\eta$ when no error correction was performed.
\begin{figure}[h!]
\begin{center}
\includegraphics[width=0.9\textwidth]{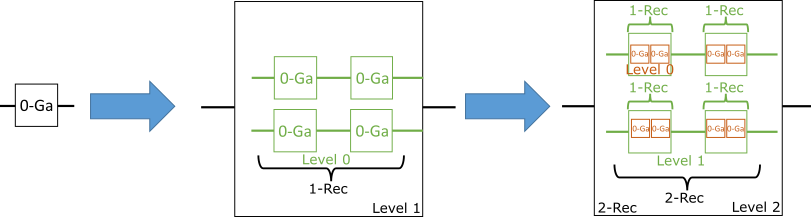}
\caption{The recursive "Russian doll" construction from 0 concatenation to 2 concatenations. Initially, there are only physical gates (i.e 0-Ga) in the algorithm. They are then replaced by logical components (more precisely, they are replaced by 1-Rec) at level-1, which introduces new physical gates to implement the logical gates and error correction. To do the second level of concatenation, we then replace "again" the 0-Ga inside the first concatenation by 1-Rec. It defines what is called a 2-Rec. In practice, for each increment of the concatenation level, the 0-Ga are replaced by 1-Rec. The 0-Ga represented inside the level-1 and the level-2 are here for illustration purposes (they do not represent a "real" circuit)}
\label{fig:recursive_concatenation}
\end{center}
\end{figure}

And we can continue to do those replacements recursively. Basically, each 0-Ga in the level-2 concatenation can be replaced again by a 1-Rec, which would give the level-3 concatenation level. And so on up to an arbitrary level $k$. The level-k exRec is defined from the level-(k-1) exRec in which each 0-Ga has been replaced by a 1-Rec. After $k$ levels of concatenations, the probability of having a bad level-k exRec, which is the good entity to consider if we want to think about a protected logical gate, is upper bounded by $p_L^{(k)} \equiv \binom{A}{2}(p_L^{(k-1)})^2$ which from a recursive reasoning is equal to:
\begin{align}
&p_L^{(k)}=\eta_{\text{thr}}\left(\frac{\eta}{\eta_{\text{thr}}}\right)^{2^k}
\label{eq:probability_error_k_concatenations}\\
& \eta_{\text{thr}} \equiv \frac{1}{\binom{A}{2}}
\end{align}
$p_L^{(k)}$ is the analog of $\eta$ when $k$ concatenations have been performed. For this reason, the probability to have an unsuccessful algorithm after $k$ level of concatenation can then be upper bounded by:
\begin{align}
p^{(k)}_{\text{unsuccessful}}=N_L p_L^{(k)}
\label{eq:probability_error_algo_k_concatenations}
\end{align}
From those equations, we understand that an arbitrary accurate computation is possible given the fact that $\eta<\eta_{\text{thr}}$. $\eta_{\text{thr}}$ is called \textit{the threshold} and is about $10^{-4}$ for a probabilistic noise model \cite{aliferis2005quantum}\footnote{This value is obtained from numerical simulations and is not strictly equal to the number of pair of fault locations, it is a "more accurate" numerical estimate.}. 

What this equation shows is that if the noise is lower than the threshold $\eta_{\text{thr}}$, then quantum error correction improves the situation. But if $\eta > \eta_{\text{thr}}$ it does not because "too many" physical components are required to implement the error correction, and the fact their probability of failing is too big, they would degrade the situation "more" than what quantum error correction "improves". 

We can also comment that the interpretation of $p_L$ as being a \textit{probability} of error for a logical gate, and the exact value of the threshold would be different outside of the probabilistic fault models we considered here \cite{aharonov2008fault}, but the general formula \eqref{eq:probability_error_k_concatenations}, and the physical intuition about how the effect of noise is being reduced would remain. 

It is actually standard to approximate quantum noise by probabilistic noise models in the literature \cite{jayashankar2021achieving,beale2018quantum}, and this is what we are going to consider in many places along this thesis. A completely rigorous treatment of the noise here would ask us to compute a norm for the process $\mathcal{N}$, and to use this norm as what plays the role of $\eta$, but this approach, even though more rigorous, would add complications and it would lead to a very poor upper bound of the estimation for the occurrence of a fault for a logical gate that is unnecessarily pessimistic.

The equation \eqref{eq:probability_error_k_concatenations} and the formula \eqref{eq:probability_error_algo_k_concatenations} are the central ones we are going to use in all this thesis. It is with them that we will estimate how many concatenations have to be performed to implement the various algorithm as a function of the algorithm size. Those formulas are also the main ingredient and spirit that are behind the quantum threshold theorem. However, what this theorem states is a bit more than only providing those expressions, and this is the reason why, for the sake of completeness of this chapter, we will express it.
\subsubsection{The quantum threshold theorem}
\label{sec:quantum_threshold_theorem}
The concatenated construction allows us to reduce in an effective manner the strength of the noise. But it comes with a cost: the number of physical elements required grows with the level of protection. What the quantum threshold theorem will tell us is that (i) the noise can be effectively reduced as much as desired when enough error correction is being done (we already saw it, but it formalizes slightly differently the notion of error of the algorithm), and (ii) the number of physical components, as well as the total algorithm duration required to make this improvement, does not grow "too fast" as a function of the size of the algorithm we want to implement on the quantum computer: error correction can be implemented using a "reasonable" amount of resources, and the algorithm will run in a "reasonable" amount of time.

More formally, we call $\{p_i^{\text{actual}}\}$ the probability distribution of the measurement performed on the logical qubits at the very end of the noisy (but protected by error correction) circuit, and $\{p_i^{\text{ideal}}\}$ this same distribution but for a perfect implementation of the circuit. The index $i$ is thus associated to a possible measurement outcome (if we have $Q_L$ logical qubits, $i \in [1,2^{Q_L}]$). Now, we can define the error of the computation as being\footnote{It corresponds to the $L^1$ distance between probability distributions \cite{aliferis2005quantum}}:
\begin{align}
\delta \equiv \sum_{i} |p_i^{\text{actual}}-p_i^{\text{ideal}}|
\end{align}
In the case all the $k$-exRec are not faulty, we would have, for all $i$, $ p_i^{\text{actual}}=p_i^{\text{ideal}}$ (because the ideal algorithm would then be perfectly simulated). We notice that in general, $p_i^{\text{ideal}}$ might not be peaked, i.e., equal to a delta Kronecker if, by essence, the algorithm does not provide an answer with certainty (it would be the case for Grover algorithm, for instance, \cite{nielsen2002quantum}).

Then, we have\footnote{In \eqref{eq:proba_i_actual} we use for each term the fact that $P(A \cap B)=P(A)P(B|A)$ where $P$ denotes a probability, $A$ and $B$ are two events and $A|B$ means $A$ knowing $B$. For instance for the first term, $A$ is "The algorithm is successful", and $B$ is "The outcome $i$ has been found".}:
\begin{align}
p_i^{\text{actual}}=(1-p^{(k)}_{\text{unsuccessful}})p_i^{\text{ideal}}+p^{(k)}_{\text{unsuccessful}}p_i^{\text{fail}}
\label{eq:proba_i_actual}
\end{align}
The probability of finding the outcome $i$ in the noisy "actual" implementation of the algorithm is equal to the probability to find this outcome whether the implementation is successful (this is the first term in the sum) or not (this corresponds to the second term of this sum). The first term corresponds to the probability to have a successful algorithm \textit{and} to find the outcome $i$. In this case, as the algorithm was successful, this last probability is equal to $p_i^{\text{ideal}}$. The second term corresponds to the probability to have an unsuccessful algorithm \textit{and} to find the outcome $i$. We don't know exactly what will be the probability to find the outcome $i$, knowing the algorithm is unsuccessful, but there exists some probability distribution for that that we call  $p_i^{\text{fail}}$.

From those expressions, we deduce that:
\begin{align}
\delta = p^{(k)}_{\text{unsuccessful}} \sum_i |p_i^{\text{fail}}-p_i^{\text{ideal}}| \leq 2 p^{(k)}_{\text{unsuccessful}} \leq N_L p_L^{(k)}, 
\end{align}
where we used the fact that $\sum_i |p_i^{\text{fail}}-p_i^{\text{ideal}}| \leq 2$ (this is a general property for $L^1$ distance between probability distributions \cite{aliferis2005quantum}). As, under the assumption $\eta < \eta_{\text{thr}}$, the right handside converges to $0$ for $k$ big enough (see \eqref{eq:probability_error_k_concatenations}), we deduce that the error $\delta$ can be put as close to $0$ as desired. This is the first result of the quantum threshold theorem.

The second result is that it can be done with a "reasonable" amount of physical resources, and the algorithm will run for a "reasonable" amount of time. More formally, calling $d$ the maximal depth of a 1-Rec, the depth of a circuit being defined as the number of timesteps the circuit is composed of\footnote{This definition only makes sense when all the physical gates have the same duration. We can make sense of this by calling $\tau_{\text{longest}}$ the duration of the longest physical gate that is used in the algorithm. Then, we can consider that any gate faster than that will have to be completed by an identity operation such that the sequence of gate+identity operation lasts for $\tau_{\text{longest}}$. Under this angle, the notion of timestep is well defined.}, and $l$ the maximum number of locations in a 1-Rec (roughly speaking, the number of gates inside), an algorithm composed of $N_L$ locations and having a depth being $D_L$ can be implemented with the concatenated construction with an error lower or equal to $\delta$ with $N_L^*$ locations and a depth $D_L^*$ such that \cite{aliferis2005quantum}:
\begin{align}
&N_L^*=O(N_L(\log(N_L))^{\log_2(l)}) \label{eq:Lstar}\\
&D_L^*=O(D_L(\log(N_L))^{\log_2(d)}), \label{eq:Dstar}
\end{align}
where $O$ is the "big-O" Landau notation. For two sequences $U_n$ and $V_n$, $U_n=O(V_n)$ means that $|U_n/V_n| \leq C$ for any $n$ bigger than some $n_0$, for some positive constant $C$. Here, \eqref{eq:Lstar}, \eqref{eq:Dstar} mean that the depth and number of locations of the error-protected circuit will not grow faster than what there is inside of the $O(.)$. 

In practice, it implies that the number of physical elements required to perform the error correction for a fixed physical error value and a fixed accuracy to reach $\delta$ will in the worst case grow "a bit faster" as to how the algorithm is expected to grow ("a bit faster" because it doesn't grow only proportionally to $N_L$ but to $N_L \log(N_L)^{\log_2(l)}$), but for instance, we know that at least the growth is not exponential. And it implies that the algorithm will run in a "reasonable" amount of time as the depth of the error protected algorithm will, in the worst case, grow "a bit faster" than the depth of the algorithm to implement ("a bit faster" because it doesn't grow only proportionally to $D_L$, but to $D_L (\log(N_L))^{\log_2(d)}$ in the worst case). This theorem is a result showing that quantum error correction can be useful in practice. But it does not give a quantitative estimation of all the resources (physical qubits, gates of each type, gates active in parallel, etc.) that will be required for a concrete problem for a given technology. Determining those elements is among the main goals of the two last chapters of this thesis.
\subsubsection{Fault-tolerant implementation of Steane code with the Steane method}
\label{sec:FT_implementation_Steane_method}
Now that we explained the principle of concatenations, we need to find the concrete circuits allowing us to implement error correction fault-tolerantly: they will allow us to perform our detailed energetic calculations later on. Here, we will provide the exact circuits allowing to do the first level of protection, i.e., the first level of concatenation. In practice, it means that we will explain here how a 1-Rec can be realized with Steane code for any 0-Ga (what is the exact quantum circuit allowing to implement it). Then, from the "Russian dolls" recursive construction, we will be able to access the number of physical gates of each type and the number of physical qubits required for any concatenation level, but this will be done in the section \ref{sec:adapting_framework_FT} of the chapter $4$. 

The first question we must address is how we can design circuits that implement gates fault-tolerantly, i.e., that implement 0-Ga in such a way that the errors do not propagate on multiple physical qubits composing a logical qubit as we saw that it is a requirement to perform concatenations (see definition \ref{def:FT_1Ga} for the exact requirements). One way to do it is based on transversal implementations. Let us assume that we want to apply a single qubit gate $G_L$ on a logical qubit. We will say that the implementation of this gate is performed transversally if to implement this gate, we need to apply the corresponding physical gate $G$ on all the physical qubits composing this logical qubit. For a two-qubit gate such as a cNOT, we will say that the gate is implemented transversally between two logical qubits if it is realized by applying the corresponding physical gate between the physical qubits, such that the $i$'th physical qubit of the first logical qubit is interacting with the $i$'th physical qubit of the second logical qubit through this physical gate\footnote{Actually the exact definition of transversal operation asks to not make two physical qubits composing a logical qubit interact between each other. We took a little bit of freedom on this definition in the main text.}. This concept is illustrated on the figure \ref{fig:logical_transversal}. Transversal gates are by construction fault-tolerant as they never make two physical qubits within the same logical qubit interact with each other. Thus one error inside a logical qubit cannot cause two errors within \textit{the same} logical qubit after that the gate has been implemented (if the gate was not faulty). In short, a transversal gate will be fault-tolerant with respect to the definition \ref{def:FT_1Ga}. Unfortunately, the Eastin-Knill theorem \cite{eastin2009restrictions} states that it is not possible to have a complete gateset of logical operation using only transversal operations: an arbitrary algorithm can then not be implemented only based on transversal gate implementations; we come back on that point later. For Steane code, it is possible to implement transversally any logical Pauli operator\footnote{We saw this explicitly in the previous section: $X_L \equiv X_1 X_2 X_3 X_4 X_5 X_6 X_7$ is the logical Pauli $X$ operator for Steane code for instance.}, cNOT, Hadamard, the $S \equiv e^{-i (\pi/4) Z }$ gate \cite{aliferis2005quantum}. This is why for the rest of this thesis, we will consider using the logical gateset $\mathcal{G}_L \equiv \{Id,H,S,X,Y,Z,cNOT\}$ ($Id$ is the logical identity gate that we anyway need). Completing this gateset with the $T \equiv e^{-i (\pi/8) Z}$ gate (which has to be implemented through another procedure than the one we describe here), an arbitrary logical gate can be implemented \cite{aliferis2005quantum}\footnote{More precisely, an arbitrary gate can be approximated to an arbitrary level of precision by implementing a finite sequence of elements in $\mathcal{G}_L$ and $T$. The fact that any algorithm can be implemented with this gateset also assumes that the qubits can only be prepared in the computational basis, and that only measurements in the computational basis can be performed.}.
\begin{figure}[h!]
\begin{center}
\includegraphics[scale=0.4]{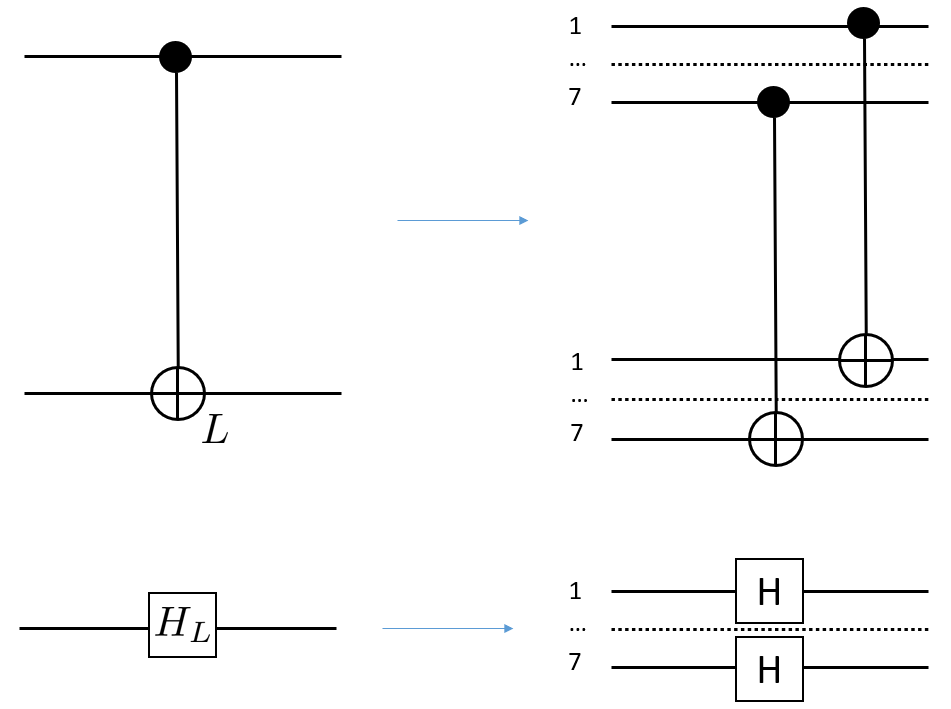}
\caption{Top: implementation of a logical cNOT transversally (in Steane code as we assume the logical qubit being composed of $7$ physical qubits). Bottom: implementation of a logical Hadamard transversally.}
\label{fig:logical_transversal}
\end{center}
\end{figure}

Now that the fault-tolerant implementation of 1-Ga has been explained, we need to explain how to implement error correction (i.e., a 1-Ec) fault-tolerantly in the sense of the definition \ref{def:FT_1Ec}. There are different ways to do this, but in this Ph.D., we considered Steane's method (not to be confused with the Steane code). This method is based on the fact that for codes in which stabilizers are a tensor product of either only $X$ or only $Z$ Pauli (which is the case for Steane code), a logical cNOT can be implemented transversally. This property can be used to extract errors affecting the data qubits into the ancilla without destroying the encoded state on the data qubits while avoiding an uncontrollable propagation of errors. To understand the principle, we can look at Figure \ref{fig:steane_method_error_propagation_Zstab} which explains how the $Z$ stabilizers can be measured in a fault-tolerant manner. 
\begin{figure}[h!]
\begin{center}
\includegraphics[scale=0.4]{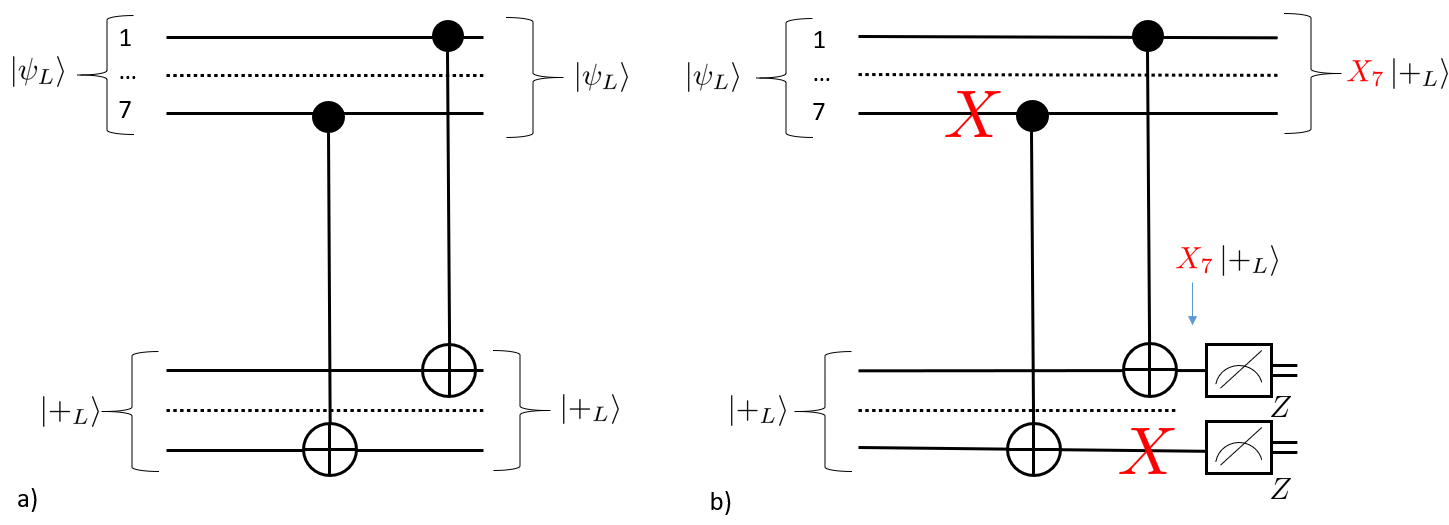}
\caption{The principle behind Steane method: measurement of the $Z$ stabilizers. We assumed that a logical qubit is composed of $7$ physical qubits as in Steane code. \textbf{a)}: A logical cNOT is performed transversally, i.e. it is implemented by realizing physical cNOT between each corresponding physical qubits composing the two logical qubits. The ancilla is prepared in $\ket{+_L}$. Because it is an eigenstate of NOT ($X$) operator, it won't evolve and the transformation is $\ket{\psi_L}\ket{+_L} \to \ket{\psi_L}\ket{+_L}$. \textbf{b)} In the presence of an error $X_i$ on the logical qubit affecting the control, the $i$'th qubit of the logical qubit affecting the target is also affected. The experimentalist then measures the observables $Z_i$, $i \in [|1,7|]$ which will allow to find out the eigenvalues of the $Z$ stabilizers as explained in the main text.}
\label{fig:steane_method_error_propagation_Zstab}
\end{center}
\end{figure}
On this figure, we represented the $7$ physical qubits composing the logical data qubit on the top. A data qubit is a qubit that participates in the implemented algorithm (as opposed to ancilla qubits which are here to perform error correction). On the bottom, we have $7$ physical qubits associated with the (logical) ancilla. The ancilla is initially prepared in the state $\ket{+_L}$ which is the logical $\ket{+}$ state. It allows to make it "insensitive" to the cNOT on the logical level, as explained in the figure caption. In practice, it means that they will not get entangled with the logical data qubits. On the figure \ref{fig:steane_method_error_propagation_Zstab} \textbf{b)}, we can see that an $X$ error affecting the $i$'th physical qubit of the logical data qubit is "transferred" on the corresponding physical qubit of the ancilla. In this sense, the ancilla will have the exact same $X$ errors as the data qubits. Then, one measures the observables $Z_i$, $i \in [|1,7|]$ of the ancilla and store the measurement outcomes in a vector $z=(z_1,...,z_7)$. From those measurements outcomes, it could deduce the measurement outcomes of any of the $Z$ stabilizers. For instance, the measurement of the stabilizer $Z_4 Z_5 Z_6 Z_7$ is being performed by calculating the dot product (modulo 2) $b.z$ with $b=(0001111)$. Now, all this does not strictly show that this circuit is fault-tolerant. But one could check that it is the case according to the definition \ref{def:FT_1Ec} that we gave (or to the axioms provided in \cite{aliferis2005quantum}). 

Here we talked about the $Z$-stabilizer measurements. But we can measure the $X$ stabilizers in a very similar way as represented on the figure \ref{fig:steane_method_error_propagation_Xstab}. 
\begin{figure}[h!]
\begin{center}
\includegraphics[scale=0.4]{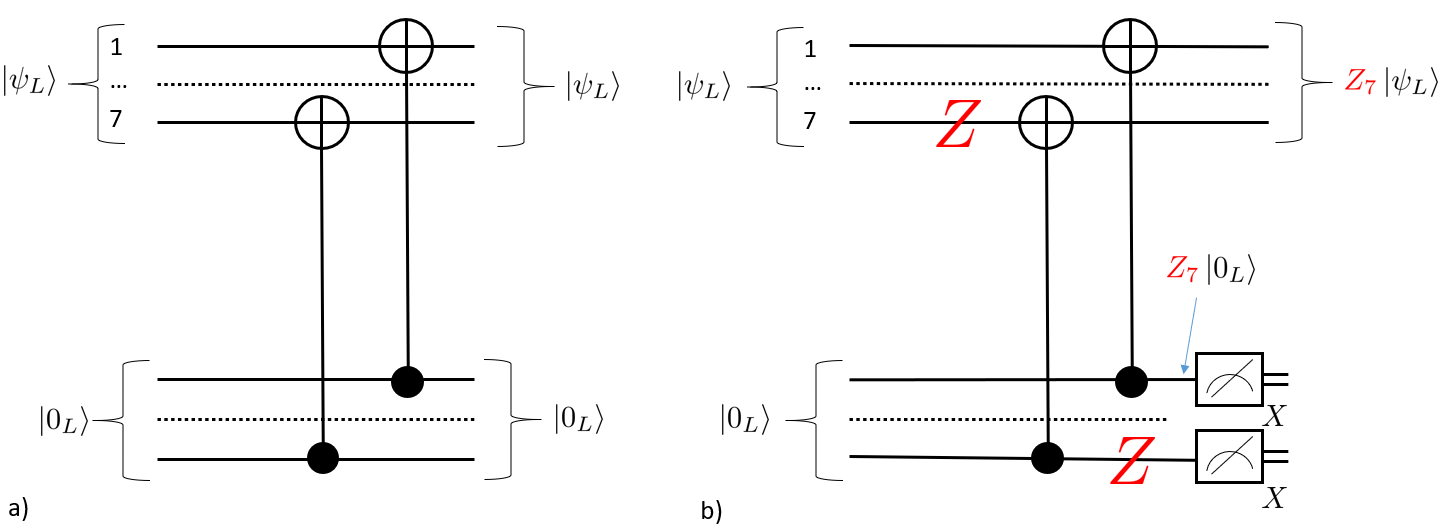}
\caption{The measurement of the $X$ stabilizers follows the same logic as for the $Z$ ones. We need to put the ancilla in a logical state that will not be affected by the logical cNOT, and in such a way that the $Z$ errors affecting the data qubits are "replicated" in the ancilla qubits. Because the $Z$ errors are propagating from the target to the control as explained in the appendix \ref{app:error_propagation}, putting the cNOT as on this graph allows to detect those errors on the ancilla. We finally measure all the observables $X_i$, $i \in [|1,7|]$.}
\label{fig:steane_method_error_propagation_Xstab}
\end{center}
\end{figure}

Assuming the states $\ket{+_L}$ and $\ket{0_L}$ of the ancilla are "given for free", those circuits can be shown to be fault-tolerant in the sense that they satisfy the axioms a 1-Ec must satisfy \cite{aliferis2005quantum} \footnote{Strictly speaking, we must add the correction. But the correction simply consists in applying an appropriate $X$ or $Z$ operator on the physical qubits composing the logical one to put them back in the code space. Such operation is fault-tolerant as it does not involve operation between different physical qubits composing a logical one. The conclusion that the circuit is fault-tolerant would then remain.}. However, to be implemented, it requires to initialize the ancilla in the states $\ket{0_L}$ or $\ket{+_L}$ depending on the stabilizers to measure. No known fault-tolerant procedure to do this is known: including directly a preparation would then lead to an uncontrollable error spreading. For instance, the states $\ket{0_L}$ and $\ket{+_L}$ can be prepared by the circuit represented on the figure \ref{fig:preparation_0logical}, where we see that if one physical gate is faulty, it might induce errors on many qubits. One way to escape this issue is to prepare multiple ancillae and to only select those that have been correctly prepared. It requires a verification procedure. This verification is represented by the orange box on the figure \ref{fig:steane_method_full}. For instance, the $\ket{0_L}$ ancilla that is required to measure the $X$ stabilizers (it is on the second line, we call it the $X$-syndrome ancilla) interacts with a logical cNOT with another ancilla qubit called verifier, initialized as well in $\ket{0_L}$. Its role is to check if the $X$-syndrome ancilla contains $X$ errors. Depending on the measurement outcome of the $Z$ measurement performed on the verifier, the ancilla is accepted or rejected. We won't detail the criteria and why exactly it ensures that the circuit is fault-tolerant, but we refer to \cite{aliferis2005quantum} for the details. In the case the ancilla is rejected, another ancilla is tried for until the verification succeeds\footnote{We say "until" the verification succeeds, but because qubits have limited lifetime, all those verifications are done "in parallel" to not make the physical qubits decohere.}. Once the ancilla has been verified, we enter in the green box, which performs the syndrome measurement exactly in the way described from the previous figures \ref{fig:steane_method_error_propagation_Xstab} and  \ref{fig:steane_method_error_propagation_Zstab}. Then an appropriate correction can be applied to the logical qubit. And the algorithm continues with the next logical gate. 
\begin{figure}[h!]
\begin{center}
\includegraphics[scale=0.4]{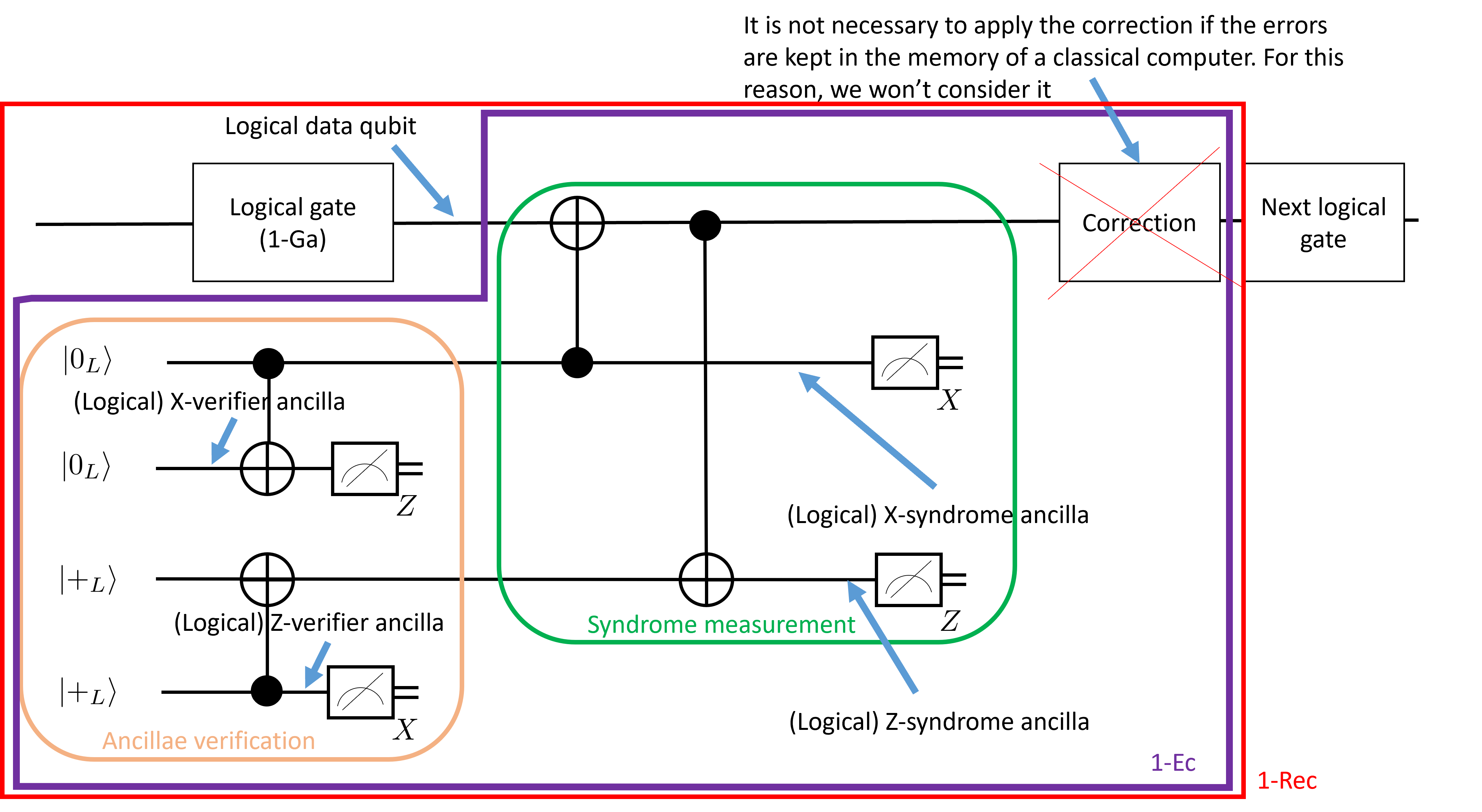}
\caption{Each line represents a logical qubit. Any gate represented here is a logical gate that is implemented transversally. The first line represents the logical data qubit, i.e., the logical qubit that implements a gate of the algorithm. The four bottom lines are four ancillae that will be here to perform the error correction. The second and fourth lines are the ancillae that participate to the $X$ and $Z$ syndrome measurement respectively, if they are accepted by the verifiers ancillae, which are on the third and last lines. The green rectangle corresponds to the syndrome measurement that we already described in the figures \ref{fig:steane_method_error_propagation_Xstab} and \ref{fig:steane_method_error_propagation_Zstab}. As long as the errors are kept in some classical memory, the correction doesn't need to be applied if all the gates implemented are in the set $\mathcal{G}_L=\{Id,H,S,X,Y,Z,cNOT\}$. The reason for that is related to the concept of Pauli frame; we refer to \cite{knill2005quantum,riesebos2017pauli} and reference therein for further information. The red box represents the 1-Rec and includes the 1-Ga followed by the 1-Ec (purple box).}
\label{fig:steane_method_full}
\end{center}
\end{figure}

\begin{figure}[h!]
\begin{center}
\includegraphics[scale=0.6]{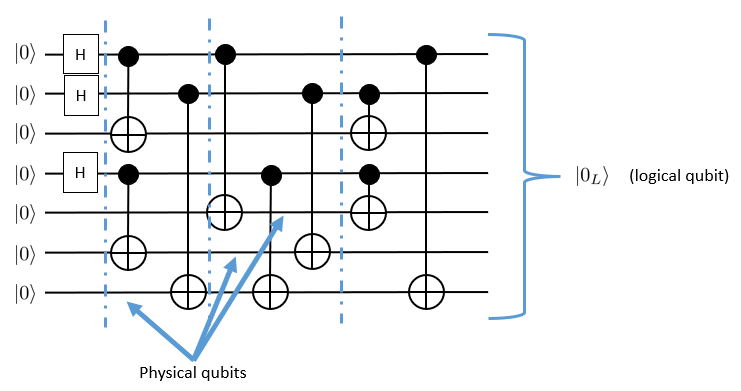}
\caption{preparation of the $\ket{0_L}$ logical state. All the lines represent physical qubits (which regrouped, in the end, create the logical qubit $\ket{0_L}$). The preparation of the $\ket{+_L}$ logical state can be done via a similar procedure where (i) the direction of all the cNOTs is being reversed (ii) four Hadamards gates are added on the four qubits not affected by Hadamard here, right after their initialization in $\ket{0}$, (iii) the three Hadamard of this figure are being removed. This circuit is not fault-tolerant in the sense that a fault occurring on a given gate during the propagation will induce an error that might propagate on all the physical qubits. This is why a verification procedure is required, as explained in the main text. This circuit can be run in $4$ timesteps which are delimited by the blue vertical dashed lines.}
\label{fig:preparation_0logical}
\end{center}
\end{figure}
The figures \ref{fig:steane_method_full} and \ref{fig:preparation_0logical} contain all the information necessary in order to know how many qubits and gates are required to implement transversal gates when we will scale up the computer. Indeed, this circuit allows us to go from a level-0 to level-1 simulation exactly. Then, following the "Russian dolls" principle of concatenation, each 0-Ga in the level-1 simulation is replaced by a 1-Rec following the schemes described on the Figures \ref{fig:steane_method_full} and \ref{fig:preparation_0logical}. And we can perform this recursive construction up to level-k. Strictly speaking, as some ancillae will be rejected by the verifiers, the total number of qubits and gates to implement will not be known exactly as there will have a probability to have a good preparation. Actually, as discussed in the appendix \ref{app:ancilla_rejected}, as the probability to be accepted is higher than to be rejected, and as this acceptance probability is actually high, it can be reasonably agreed that the number of physical qubits and gates required can be estimated by assuming the verification always succeeds (the number of "extra" ancillae and gates needed will be dominated by the number of ancillae required for a "one-shot" success). We also did not discuss how to implement the necessary non-transversal gates, which are required to have a complete gateset. Their implementation requires a different procedure. There are different proposals to do this, requiring a very different amount of physical qubits and gates. Some try to implement the gates that cannot in principle be implemented transversally by approximating them by a transversal implementation  \cite{yang2020covariant,paetznick2013universal}, such that error correction will then correct the errors of this approximation. There is also a procedure called state injection which often relies upon magic state distillation that allows to implement those gates \cite{chamberland2020very,litinski2019magic,bravyi2005universal}. In this thesis, we will not model the non-transversal gates; it belongs to the list of outlooks we are considering in conclusion. However \textit{first investigations} on this question seem to indicate that including those gates in our estimation will not modify very much our estimations. The order of magnitudes of physical components required for such gates should not dominate the one for Clifford operation if appropriates optimization are performed. 
\subsubsection{Number of gates required in a 1-Rec}
\label{sec:number_gates_1Rec}
We now provide the explicit number of physical gates that are required in the 1-Rec of the different gates that we can implement fault-tolerantly\footnote{Assuming the verification always succeeds.}. This calculation is based on an explicit counting of the number of gates there are in the figures \ref{fig:preparation_0logical} (and its analog for the $\ket{+_L}$ preparation, see the caption) and \ref{fig:steane_method_full}.

The results are presented in the table \ref{table:FT_gate_breakdown}. What we refer to as a single-qubit physical gate will either correspond to a Pauli, Hadamard, or $S$ gate (we will consider in our energetic calculations that all the single-qubit gates have the same cost). In order to make this counting, we assume that we can only do measurements in the computational basis ($Z$-measurement): measuring in the $X$ basis requires applying Hadamard gates before a $Z$-measurement. We also assume that qubits can only be available in their ground state $\ket{0}$ (preparing $\ket{+}$ requires applying a Hadamard on $\ket{0}$). With those assumptions we find that we need $7 \times 4+9 \times 4=64$ cNOTs: $7 \times 4$ for the transversal implementation of all the cNOTs represented on the figure \ref{fig:steane_method_full}, and $9 \times 4$ for the preparation of the $\ket{+_L}$ and $\ket{0_L}$ states for the ancilla (see figure \ref{fig:preparation_0logical}). For the measurement gates (we call measurement gate the action of measuring a qubit), we need $7 \times 4$ of them to measure all the ancilla. We also need single-qubit gates: $7$ for the transversal implementation of the gate at the logical data qubit, $7 \times 2$ Hadamard to apply before the $X$-measurement (because we can only do $Z$-measurements), and $3 \times 2 + 4 \times 2 = 14$ Hadamard for the preparation of the ancilla states as explained in the caption of figure \ref{fig:preparation_0logical}. It gives us $35$ single-qubit gates to apply. We also need to make some qubits necessarily wait at some point. For instance, the logical syndrome qubits have to wait for the measurement outcome of the verifier to know if they can be injected or not: they are then doing a noisy physical identity. We won't go into the detail of this counting, but in the end, we would find $36$ physical identity gates by looking at all the places the physical qubits composing the logical ones necessarily have to wait. In the end, doing an appropriate counting for each type of logical gate that has to be implemented, we have access to the table \ref{table:FT_gate_breakdown} which represents for each level-1 gate (on the line), the number of physical elements (on the column) its 1-Rec is composed of \footnote{To be very precise, the 1-Rec of a measurement is defined as 1-Ec followed by 1-measurement (the 1-Ec is "before" the 1-Ga measurement gate). This is defined this way for theoretical reason to prove the threshold \cite{aliferis2005quantum}. As we are interested in finding the number of physical elements when concatenating, the elements in this 1-Ec would already have been counted by the gate before the measurement: taking them into account "again" would lead to an overcounting of the number of physical elements. This is why we remove the elements that would be in the 1-Ec before the measurement gates, which explains why the last line of table \ref{table:FT_gate_breakdown} is composed of three zeros.}. We see from this table that, as we could expect, a logical identity is the same as a logical single-qubit gate, excepted that the transversal operation is composed of identity gates. This is why it contains $7$ less single-qubit gate but $7$ more identity gates than a logical single-qubit gate.
\begin{table*}[t]
\begin{center}
\begin{tabular}{|c|c|c|c|c|}
	\hline
	& lvl-0 cNOT & lvl-0 Single & lvl-0 Identity & lvl-0 Measurement\\
	\hline
   lvl-1 cNOT & $135$ & $56$ & $72$ & $56$ \\
   \hline
   lvl-1 Single & $64$ & $35$ & $36$ & $28$ \\
   \hline
   lvl-1 Identity & $64$ & $28$ & $43$ & $28$\\
   \hline
   lvl-1 Measurement & $0$ & $0$ & $0$ & $7$\\
   \hline
\end{tabular}
\caption{Each row lists a FT logical gate and tabulates the lower level components required for the listed gate as columns. Those level-0 components are thus 0-Ga physical gates. The single qubit gates here are either Hadamard, Pauli or $S$ gates. This table includes the gates required to prepare the ancilla and verifier and assumes that the verification always succeeds (there is no need to prepare more than one $Z$ or $X$ syndrome ancilla for instance). We discuss the motivation behind this assumption in the appendix \ref{app:ancilla_rejected}.}
\label{table:FT_gate_breakdown}
\end{center}
\end{table*}
\FloatBarrier
\subsubsection{Number of timesteps for the logical gate and the ancilla}
\label{sec:number_timesteps_for_logical_and_ancilla}
The last thing we need to access from this chapter is the number of timesteps that are required to implement the 1-Rec. It will be a piece of information we will use in the section \ref{sec:estimation_physical_qubits}, where we will estimate the number of physical qubits and the average number of gates acting in parallel. There are actually two different possible numbers associated to the number of timesteps a 1-Rec lasts for. 

First, from the point of view of the logical data qubits, a 1-Rec is only lasting for $3$ physical timesteps: the transversal implementation of the 1-Ga, the $X$ syndrome measurement, and the $Z$ syndrome measurement (as said in the caption of \ref{fig:steane_method_full}, we don't need to implement the correction if we just use the gates in $\mathcal{G}_L$.). After those three timesteps, the logical data qubits are implementing the next logical gate and its associated 1-Rec. 

From the point of view of the ancilla, however, things take much longer. Here we are only interested in finding the number of timesteps of the ancilla that take the longest time to be implemented. It corresponds to the $X$ and the $Z$ syndromes ancilla: both are taking a total amount of $9$ timesteps before being measured. For instance, for the $X$-syndrome, it takes $4$ timesteps to prepare it (see \ref{fig:preparation_0logical}), followed by $2$ timestep to be verified (one cNOT, and it then has to wait for the $Z$ measurement outcome of the $X$ verifier). Then in the syndrome measurement, it is doing one cNOT with the data qubit, and as it is being measured in the $X$ basis, it takes one Hadamard followed by the measurement: a total of $3$ timesteps here. In conclusion, it takes $9$ timesteps in total.

\section{Conclusion}
In this chapter, we explained how the Steane code works (and gave intuition behind the more general formalism of stabilizer codes). We also introduced the basics behind fault-tolerant quantum computing through the concatenated code construction. We then provided concrete circuits that allow to implement gates in a fault-tolerant manner. Those circuits and the estimation of the number of physical elements inside level-1 gates (more precisely, 1-Rec) will allow us to determine the exact number of physical gates and qubits that are required in order to reach a targetted accuracy (thus for a given concatenation level) for a fault-tolerant implementation of an algorithm. A general intuition that we can get from what we presented is that the concatenated construction seems demanding in term of physical resources. Indeed, each time a concatenation is added, the number of physical elements drastically increases. The quantum threshold theorem \ref{sec:quantum_threshold_theorem} might tell us that the number of elements required does not grow fastly with the algorithm size, but it does not give us information about the prefactors to expect. One important criterion is also to access the number of physical qubits required, and this estimation highly depends on the level of recycling it is possible to do, i.e., when one qubit ancilla has finished working, when can it be reused? A poor level of recycling might, for instance, dramatically increase the number of physical resources required. We are going to study it in section \ref{sec:estimation_physical_qubits}. All those estimations are important to know in order to build a concrete quantum computer as they might significantly impact the power consumption and the overall feasibility of the device. We will do those estimations in the sections \ref{sec:estimation_physical_gates} and \ref{sec:estimation_physical_qubits} of the fourth chapter of this thesis, and we will use them in the quantitative examples we will treat in the section \ref{sec:application_to_QFT} of the last chapter of the thesis.
\begin{appendices}
\chapter{Fundamentals of quantum error correction}
\section{Knill-Laflamme}
\label{app:Knill_Laflamme}
Here, we give the proof of the Knill-Laflamme conditions.
\begin{proof}{~}

\textbf{Knill-Laflamme conditions are sufficient:}

This proof is inspired from \cite{nielsen2002quantum}, with further details. We first prove that the Knill-Laflamme condition is sufficient for the existence of an error correction channel. To show it, we will explicitly construct this channel. We start by performing a polar decomposition \cite{nielsen2002quantum} on the operator $M_i P_C$. It allows us to know that there exists a unitary $U_i$ satisfying:
\begin{align}
M_i P_C = U_i \sqrt{(M_i P_C)^{\dagger} M_i P_C}=\sqrt{c_{ii}} U_i P_C,
\end{align}
where we made use of the Knill-Laflamme condition for the last equality. From this equation, \textit{we see that the effect of an operator $M_i$ on the code-space is to apply a unitary $U_i$} (up to the proportionality coefficient $\sqrt{c_{ii}} \geq 0$). Using the fact that if $P$ is a projector on some space $H$, then for any unitary $U$, $U P U^{\dagger}$ will be a projector on $U H$, we deduce that, when $c_{ii} \neq 0$, $M_i P_C$ has the effect of transforming any state in the code space to a state in the space having for projector:
\begin{align}
P_i =  U_i P_C U_i^{\dagger}
\end{align}
If: $c_{ii} \neq 0$, we have $P_i=\frac{M_i P_C U_i^{\dagger}}{\sqrt{c_{ii}}}$, and thus:
\begin{align}
\forall (i, j) \text{ such that } c_{ii} \neq 0, c_{jj} \neq 0 : P_i P_j = P_i^{\dagger} P_j= \frac{1}{\sqrt{c_{ii}c_{jj}}} U_i P_C M_i^{\dagger}  M_j P_C U_j^{\dagger} =\delta_{ij} P_i
\end{align}
For the cases which $c_{ii}=0$, we will have $M_i P_C=0$ and thus $P_i=0$. Considering $P_C$ was an orthogonal projector, so are $P_i$ and $P_j$. This results shows that for $i \neq j$, $M_i$ and $M_j$ either bring the codespace into orthogonal subspaces, either bring the code space to the null vector (cases in which $c_{ii}$ or $c_{jj}$ vanish, implying $P_{i}=0$ or $P_{j}=0$).

From those remarks, \textit{a legitimate "guess" of recovery operation is to: (i), perform a projective measurement associated with the projectors $\{P_i\}$ followed by (ii) the unitary operation $U_i^{\dagger}$ in order to correct the error.} We thus define $\mathcal{R}(\rho) \equiv \sum_j U_j^{\dagger} P_j \rho P_j U_j$. We have:
\begin{align}
(\mathcal{R} \circ \mathcal{E}) (P_C \rho P_C) &= \sum_{i,j} U_j^{\dagger} P_j M_i P_C \rho P_C M_i^{\dagger} P_j U_j \notag \\
&= \sum_{i,j} c_{ii} U_j^{\dagger} P_j P_i U_i \rho U_i^{\dagger} P_i^{\dagger} P_j U_j \notag\\
&=\sum_{i} c_{ii} U_i^{\dagger} P_i U_i \rho U_i^{\dagger} P_i  U_i \notag\\
&= \left(\sum_{i} c_{ii}\right) P_C \rho P_C = P_C \rho P_C
\end{align}
Where we used $\sum_{i} c_{ii}=1$ that comes from the fact that as $\mathcal{E}$ is trace preserving, then $\sum_i M_i^{\dagger} M_i = I$, and thus $\sum_i P_C M_i^{\dagger} M_i P_C = (\sum_i c_{ii}) P_C \Rightarrow \sum_i c_{ii}=1$. Thus, at this point we showed that if the noise channel satisfies Knill-Laflamme condition, then an error correction channel exists.

\textbf{Knill-Laflamme conditions are necessary:}

Let's assume that there exists a CPTP operation $\mathcal{R}$ verying: $\forall \rho: (\mathcal{R} \circ \mathcal{E}) (P_C \rho P_C) = P_C \rho P_C$. Then, defining $\{R_i\}$ one of its Kraus decomposition we get:
\begin{align}
(\mathcal{R} \circ \mathcal{E}) (P_C \rho P_C) = \sum_{ij} R_i M_j P_C \rho P_C M_j^{\dagger} R_i = P_C \rho P_C
\end{align}
The map $\mathcal{R} \circ \mathcal{E}$ admits a family of Kraus operator being $\{R_i M_j\}$. But it also admits a family of Kraus operators being $\{P_C,0,...0\}$ from the right handside. As two CPTP operations are identical iff their Kraus operator are related by some unitary transformation, we deduce that there is some complex number $c_{ij}$ such that $R_i M_j = c_{ij} P_C$. Finally, we have:
\begin{align*}
P_C M_i^{\dagger} R_k^{\dagger} R_k M_j P_C = c_{kj} c^*_{k i} P_C
\end{align*}
Which implies by summing on $k$:
\begin{align*}
P_C M_i^{\dagger} M_j P_C = \sum_k ( c_{kj} c^*_{k i} ) P_C
\end{align*}
Calling $\alpha_{ij}=\sum_k ( c_{kj} c^*_{k i} )$, the matrix $\alpha$ admitting for matrix elements the $\alpha_{ij}$ is Hermitian, and thus diagonalizable in an orthonormal basis. Calling $u$ a unitary matrix that diagonalizes $\alpha$, we have:$\alpha=u.d.u^{\dagger}$, where $d$ is diagonal. Thus, $\alpha_{ij}=\sum_{kl} u_{ik} d_{kl} u^{\dagger}_{lj}$. And, we get:
\begin{align}
P_C M_i^{\dagger} M_j P_C = \sum_{kl} u_{ik} d_{kl} u^{\dagger}_{lj} P_C \Leftrightarrow P_C \widetilde{F}_m^{\dagger} \widetilde{F}_n P_C = c_{mm} \delta_{mn} P_C.
\end{align} 
Where:
\begin{align}
&\widetilde{F}_n=\sum_{j} u_{jn} M_j \notag\\
&c_{mm} = d_{mm}
\end{align}
$\{\widetilde{F}_n\}$, as being related to the $\{M_j\}$ family through a unitary transformation is thus an equivalent set of Kraus operator describing $\mathcal{E}$ \cite{nielsen2002quantum} that satisfies the Knill-Laflamme conditions which proves the necessary condition.
\end{proof}
\begin{theorem}{Discretization of errors}

Let's assume $H_C \subset H$ is a code space. If the Knill-Laflamme conditions are satisfied for a set of error operators $\{M_i\}$, then they are satisfied for an arbitrary linear combination of those operators.
\end{theorem}
\begin{proof}{~}

We just have to check that Knill-Laflamme condition are satisfied for any family of operators $\{E_i\}$ such that: $E_i = \sum_{k} c_{ik} M_k$. We have:
\begin{align}
P_C E_i^{\dagger} E_j P_C = \sum_{kl} c_{ik}^* c_{jl} P_C M_k^{\dagger} M_l P_C = \sum_{kl} c_{ik}^* c_{jl} \alpha_{kl} P_C
\end{align}
Where we used for the last equality the fact that $\{E_i\}$ satisfies the Knill-Laflamme condition given in \eqref{eq:Knill_Laflamme_general}. The remaining thing to verify is that $\beta_{ij} \equiv \sum_{kl} c_{ik}^* c_{jl} \alpha_{kl}$ represent element of an Hermitian matrix. And indeed, we have:
\begin{align}
\beta_{ij}^* = \sum_{kl} c_{ik} c_{jl}^* \alpha^*_{kl}=\sum_{kl} c_{ik} c_{jl}^* \alpha_{lk}=\sum_{kl}  c_{il} c_{jk}^* \alpha_{kl} = \beta_{ji}
\end{align}
The Knill-Laflamme condition are thus satisfied.
\end{proof}
\chapter{Fault-tolerant quantum computing}
\section{Propagation of errors through gates}
\label{app:error_propagation}
We are interested in finding how the Pauli errors are "propagating" when a gate is acting. Given a gate $U$, asking this question means to find the operator $A_E$ such that $U E = A_E U$, where $E$ is an $n$-Pauli matrix. Indeed $A_E$ will be the "translation" on the output of the Pauli matrix that was applied on the input. Also, using the fact: $Y=iXZ$, we can reason on X or Z type Pauli error to deduce their effect on a Y error. It is possible to show the following behaviors \cite{nielsen2002quantum}:
\begin{figure}[h!]
\begin{center}
\begin{tabular}{|c|c|c|}
\hline
Gate & Input Pauli error $E$ & Output error $A_E$\\
\hline
\multirow{4}{*}{cNOT} 
& \multicolumn{1}{c|}{$X_1$} & \multicolumn{1}{c|}{$X_1 X_2$} \\ \cline{2-3}
& \multicolumn{1}{c|}{$X_2$} & \multicolumn{1}{c|}{$X_2$} \\\cline{2-3}
& \multicolumn{1}{c|}{$Z_1$} & \multicolumn{1}{c|}{$Z_1$} \\\cline{2-3}
& \multicolumn{1}{c|}{$Z_2$} & \multicolumn{1}{c|}{$Z_1 Z_2$} \\\hline
\multirow{2}{*}{H} 
& \multicolumn{1}{c|}{$X$} & \multicolumn{1}{c|}{$Z$} \\ \cline{2-3}
& \multicolumn{1}{c|}{$Z$} & \multicolumn{1}{c|}{$X$} \\ \hline
\multirow{2}{*}{X} 
& \multicolumn{1}{c|}{$X$} & \multicolumn{1}{c|}{$X$} \\ \cline{2-3}
& \multicolumn{1}{c|}{$Z$} & \multicolumn{1}{c|}{$-Z$} \\ \hline
\multirow{2}{*}{Y} 
& \multicolumn{1}{c|}{$X$} & \multicolumn{1}{c|}{$-X$} \\ \cline{2-3}
& \multicolumn{1}{c|}{$Z$} & \multicolumn{1}{c|}{$-Z$} \\ \hline
\multirow{2}{*}{Z} 
& \multicolumn{1}{c|}{$X$} & \multicolumn{1}{c|}{$-X$} \\ \cline{2-3}
& \multicolumn{1}{c|}{$Z$} & \multicolumn{1}{c|}{$Z$} \\ \hline
\end{tabular}
\end{center}
\caption{Error propagation for typical gates. For the cNOT, the qubit $1$ is the control. The qubit $2$ is the target. $E=X_1$ means for instance that the control qubit has been affected by an $X$ error before the perfect cNOT acts.}
\label{table:error_prop}
\end{figure}
\chapter{Various properties}
\section{Measuring observables with ancilla}
\begin{property}{Measuring observable having $\pm 1$ as eigenvalues}

To measure an observable $M$ (possibly acting on multiple qubits) that admits eigenvalues $\pm 1$, one can design the circuit of figure \ref{measurement_M}.
\begin{figure}[h!]
\centerline{
\Qcircuit @C=1.7em @R=1.7em {
\lstick{\ket{0}} & \gate{H} & \ctrl{1} & \gate{H} & \meter \\
 & \qw & \gate{M} & \qw & \qw \\
}}
\caption{The measurement in the $\sigma_z$ basis of the ancilla qubit measures the observable $M$ of the system measured.}
\end{figure}
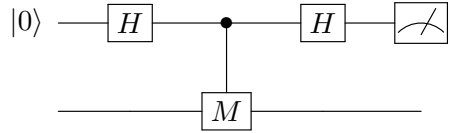
\FloatBarrier
If the ancilla (top line) is being found in $\ket{1}$ it means that the system (bottom line) is in $\ket{-_M}$. If the ancilla is being found in $\ket{0}$ it means that the system is in $\ket{+_M}$. Where $\ket{\pm_M}$ are eigenstates associated to the eigenvalue $\pm 1$ of $M$.
\end{property}
\begin{proof}{ ~ }

We start from $\ket{\psi}=\ket{0}\ket{\phi}$ where $\ket{\phi}=a \ket{+_M} + b \ket{-_M}$. After the first Hadamard, we have:
\begin{equation}
H \ket{\psi} = \frac{1}{\sqrt{2}} \left(\ket{0}\left( a \ket{+_M} + b \ket{-_M} \right) + \ket{1}\left( a \ket{+_M} + b \ket{-_M} \right) \right)
\end{equation}
We apply the controlled-M $U_M$:
\begin{equation}
U_M H \ket{\psi} = \frac{1}{\sqrt{2}} \left(\ket{0}\left( a \ket{+_M} + b \ket{-_M} \right) + \ket{1} \left( a m_+ \ket{+_M} + b m_- \ket{-_M} \right) \right)
\end{equation}
Where actually $1=m_+=-m_-$. We apply the last Hadamard:
\begin{equation}
H U_M H \ket{\psi} = \frac{1}{2} \left(\ket{0}\left( a(1+m_+) \ket{+_M} + b(1+m_-) \ket{-_M} \right) + \ket{1}\left( a(1-m_+) \ket{+_M} + b(1-m_-) \ket{-_M} \right) \right)
\end{equation}
Now, as $1=m_+=-m_-$:
\begin{equation}
H U_M H \ket{\psi} = a \ket{0} \ket{+_M} + b\ket{1} \ket{-_M}
\end{equation}
Thus, if the measurement of the ancilla returns $1$, the state of the data qubit is in $\ket{-_M}$ and otherwise, it is in $\ket{+_M}$. The outcome probabilities are the same as if we had directly measured the observable $M$ of the measured system.
\end{proof}
\end{appendices}

\chapter{Fault-tolerance with a scale-dependent noise}

In the section \ref{sec:arbitrary_accurate_quantum_computing} of the previous chapter, we saw that arbitrarily accurate quantum computing is possible provided that the probability of fault per physical gate is below a threshold. All this discussion implicitly relied on the fact that whatever the concatenation level is, the probability that a physical gate fails is the same: the noise per physical gate is assumed to be independent of the computer size. In practice, it implies that an operation performed on a quantum computer composed of one physical qubit is as noisy as this same operation performed on a physical qubit inside a quantum computer composed of millions of qubits. This is not reflecting current experiments where the noise is often scale-dependent \cite{schutjens2013single,theis2016simultaneous,rosenberg2020solid,monroe2013scaling}. 

Crosstalk issues \cite{piltz2014trapped} are a first reason for that. It is basically the fact that when more than one qubit are inside the quantum computer, they might either interact in an undesired manner (increasing the physical noise), either they do not interact, but an experimentalist trying to address a specific qubit might address other ones around at the same time in an undesired manner. What we can expect is that the greater the number of qubits there are inside the quantum computer and the bigger the strength of the associated noise is. 

A frequent scenario, which will be central in the energetic estimation of quantum computing, is what happens for a scale-dependent noise induced by the presence of limited resources (which could be energy, power, a limited amount of available frequencies for the qubits...). If the quality of operations performed on physical qubits is related to a resource $R$ that appears to be limited, such that the more this resource is available, the better the fidelity of the operation performed is, then, the greater the number of physical qubits there are inside of the quantum computer and the lower the quality of those operations on physical qubits will be. This kind of behavior will then give rise to a scale-dependent noise. One simple example of that can be the power used for cryogenics. It might be easy to maintain a few superconducting qubits at $10mK$. But maintaining millions of qubits at this temperature might be more complicated: a limited power for the cryostat would mean that the temperature of the qubits would have to increase the more qubits there are. We will see that understanding the physics behind a scale-dependent noise in this context will actually allow us to estimate the minimum amount of resources that are required to perform a calculation. 

In summary, in this chapter, we will study what happens in the presence of a scale-dependent noise for fault-tolerance theory. We will see that the accuracy of the quantum computer will then often be limited and that the physics behind such noise models can be related to resource estimation.

This chapter is the first containing original results.\footnote{I realized most of the analysis and calculations presented in this chapter, excepted the section \ref{sec:noise_growth_prop} which has mainly been done by Robert Whitney.}
\section{Scale-dependent noise: generalities}
\subsection{Description of the problem}
We recall the expression of the probability of error of a logical gate after $k$ concatenations that we have shown in \eqref{eq:probability_error_k_concatenations}, in the section \ref{sec:arbitrary_accurate_quantum_computing}:
\begin{align}
&p_L^{(k)}(\eta)=\eta_{\text{thr}}\left(\frac{\eta}{\eta_{\text{thr}}}\right)^{2^k}
\label{eq:logical_error_probability}
\end{align}
Here, $\eta$ is the probability of fault per physical gate, and $\eta_{\text{thr}} \approx 10^{-4}$ for probabilistic Pauli noise \cite{aliferis2005quantum} is called the quantum threshold\footnote{As explained in the lasts paragraphs of \ref{sec:arbitrary_accurate_quantum_computing}, for other noise models than Pauli the same expression would still hold for $p_L^{(k)}$, excepted that the value of the threshold, $\eta_{\text{thr}}$ might be different and the interpretation of $p_L^{(k)}$ as being a probability of fault might also differ, it is usually considered as being an operator norm in more general noise models. But this fact will not change the principle of the analysis we do in this chapter.}. Assuming $\eta$ constant, as soon as $\eta < \eta_{\text{thr}}$, arbitrarily accurate computation is possible. Indeed, $p_L^{(k)}$ can be put as close to $0$ as desired for a big enough value of $k$. In this chapter, we will assume that $\eta$ grows with the computer size, and thus it will, in particular, grow with $k$ as the more concatenations there are, the more physical elements there are within the computer. This scenario can be treated by directly injecting the law $\eta(k)$ in \eqref{eq:logical_error_probability}. What we can expect now is the accuracy to be limited: indeed, if there exists a concatenation level $k_0$ such that for any $k>k_0, \eta(k_0) > \eta_{\text{thr}}$, then further concatenations become necessarily counter-productive. This is what we are going to study now.
\subsection{The general situation we consider: the noise grows with the computer size}
\label{sec:general_situation}
\subsubsection{The hypotheses we make}
\label{sec:hypotheses_we_make}
To understand the problem, we can start by treating the general situation we are interested in. Let us assume that $\eta$ is an increasing function of $k$: the physical error grows with the computer size (and thus with the concatenation level). If we also assume that it was below the threshold initially, we can always rewrite it as $\eta(k)=\eta_0 f(k)$, where:
\begin{enumerate}[(i)]
\item $f(0)=1$
\item $f$ is a strictly increasing function of $k$
\item $\eta_0 < \eta_{\text{thr}}$
\end{enumerate}
As $f(0)=1$, $\eta_0$ represents the probability of fault for a physical gate without concatenating. The conditions we have put on $f$ represent a general case of noise that increases with the size of the computer. We also add the fact that the physical fault probability is below the threshold without concatenations. If it wasn't the case, then $p_L^{(k>0)} \geq \eta_0$: error correction would necessary deteriorate\footnote{Or eventually it would keep the noise equal to $\eta_0$.} the situation. 

Here, we can notice that with those hypotheses, arbitrarily accurate quantum computing could still be possible under the condition\footnote{The $\sup$ (supremum) is the smallest upper bound which corresponds to the maximum in the case the set $\{\eta(k)\}$ admits a maximum.} $\eta^{\sup} \equiv \sup_k[\eta(k)] < \eta_{\text{thr}}$. Indeed, we would have: $p_L^{(k)}(\eta(k)) < \eta_{\text{thr}}\left(\eta^{\sup}/\eta_{\text{thr}}\right)^{2^k}$ where the right handside can be put arbitrarily close to $0$ for a big enough value of $k$. We will briefly comment about what interesting features we can still have in this regime later on, but for the problem of resources estimations (and some kind of crosstalks models\footnote{Crosstalk issues can also in some cases be related to a problem of limited resources if we consider the frequency bandwidth for the qubits being the resource of interest.}), we cannot expect the noise to be bounded with the computer size (or at least to be below the threshold whatever the concatenation level is). This is why in our study, we consider adding the following condition:
\begin{enumerate}[(i)]
\setcounter{enumi}{3}
\item There exists a number $k_0$ such that $\eta(k_0) \geq \eta_0 f(k_0) = \eta_{\text{thr}}$.
\end{enumerate}
The conditions (i)-(iv) thus consist in asking that the noise grows with the computer size, that it was below the threshold initially (without concatenating), and that at some point, it gets higher or equal to the threshold. All our study is mainly based on those conditions that will physically be motivated by resource estimation later on. 

\subsubsection{The maximum accuracy of the computer is limited}
\label{sec:max_accuracy_is_limited}
The first natural conclusion we get to with our hypotheses is that the accuracy is limited because at some point, there will exist an integer $k'$ such that for all  $k > k'$, we have $\eta(k)>\eta_{\text{thr}}$. Thus there is no point in concatenating more than $k'$ times, and the accuracy is limited. We call $k_{\max}$ the maximum level of concatenation it is interesting to do in order to increase the accuracy. It is defined as the integer satisfying that for all integer $k$, $p_L^{(k_{\max})}(\eta(k_{\max})) \leq p_L^{(k)}(\eta(k))$, and if there are different solutions for that, we only keep the lowest one. Concatenating more than $k_{\max}$ times is not productive as the probability of fault of the logical gate would be higher or equal to $p_L^{(k_{\max})}(\eta(k_{\max}))$. In other words, $k_{\max}$ is the number of concatenations minimizing the logical error probability\footnote{And if there exists two concatenations level satisfying this condition, $k_{\max}$ is defined as the lowest one}. Thus, $p_L^{(k_{\max})}(\eta(k_{\max}))$ is the minimum logical error. If $k_{\max}$ is finite, then the minimal logical error is finite, which would be the case with the hypotheses (i)-(iv) we made.

In full generality, there is no easy formula to find $k_{\max}$, but we can at least upper bound it. Indeed, $k_{\max}$ is necessary strictly lower than the real number $k_0$ satisfying $\eta_0 f(k_0) = \eta_{\text{thr}}$. Because of that, we can upper-bound $k_{\max}$ as: 
\begin{align}
k_{\max} < f^{-1}\left(\frac{\eta_{\text{thr}}}{\eta_0}\right)
\end{align}
But the important (and expected) message here is simply that with the hypotheses we made, \textit{the accuracy of the quantum computer is necessarily limited}.

\subsubsection{Behavior of the maximum accuracy in different scenarios}
\label{sec:behavior_max_accuracy}
As we just said, there is no easy formula allowing us to access $k_{\max}$, and it is instructive to understand why. What could happen in general is that $p_L^{(k)}(\eta_0 f(k))$ admits non trivial variations: any of the curves represented on the schematic diagram \ref{fig:qualitative_behavior_scale_dep} are possible. 
In this figure, the probability of error for a logical gate as a function of the concatenation level is represented. The black dotted lines represent the standard fault-tolerance theory where $\eta(k)=\eta_0$. In this case, as soon as $\eta_0<\eta_{\text{thr}}$, arbitrarily accurate quantum computation is possible. The solid blue lines represent what happens in the presence of a scale-dependent noise where the noise always increases with the size of the computer (and gets higher than $\eta_{\text{thr}}$ at some point). The red points on those curves are what defines $k_{\max}$: the concatenation level leading to maximum accuracy. We see that it may be useful to perform some concatenations, but at some point, it becomes detrimental. Because of that, the accuracy the computer can get to is naturally limited (illustrated by the fact $k_{\max}$ is always finite, or equivalently, by the fact the blue curves always end up diverging).

Let us comment on the easiest examples first. The curves C2 and C3 represent the situation where the noise was higher than the threshold initially. As $\eta(k)$ is strictly increasing, concatenations cannot be helpful, represented by the fact that $\eta(k) \geq \eta_0$ on those curves. The curve C4 represents a situation in which the noise was initially below the threshold. If it increases slowly enough, we could expect error correction to be useful at the beginning before deteriorating the situation. It is represented here by the fact $k_{\max}=1$ on this curve. In the unfortunate case it increases too brutally, the behavior of the curve C1 would occur. In this example, even if $\eta_0 < \eta_{\text{thr}}$, if $f$ increases fast enough in $[0,1]$ such that $\eta_0 f(1)>\eta_{\text{thr}}$, performing one concatenation level would already "come too late" and would induce $p_L^{(1)} \geq \eta_{\text{thr}} > \eta_0$. We are going to see in the next section that this behavior will imply much more stringent conditions on the quality of physical gates than with standard fault-tolerance theory.

Then, we can expect "stranger" behaviors in principle. Sometimes (curve C6), many minima can be present. We also see on curve C5 that it is possible that error correction degrades the situation before improving it. Those examples might be surprising at first view: how non-trivial variations could occur while $\eta(k)$ is strictly increasing? We provide one such example in the appendix \ref{app:non_monotonous_behavior_pL}.
\begin{figure}[h!]
\begin{center}
\includegraphics[width=0.7\textwidth]{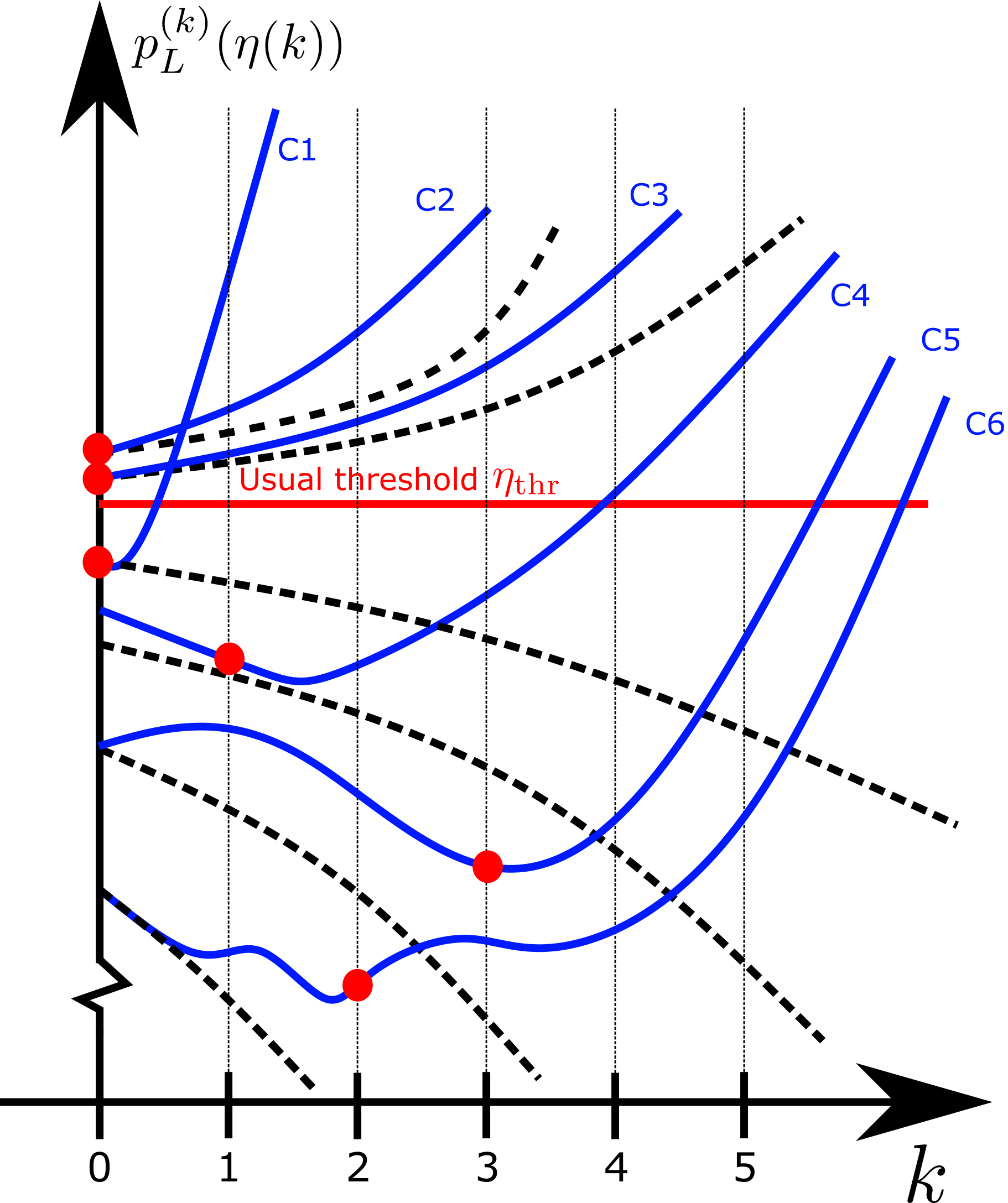}
\caption{Probability of fault for a logical gate as a function of the concatenation level. This is a schematic diagram for pedagogic purposes. The black dotted lines represents the standard fault-tolerance theory where $\eta(k)=\eta_0$. The blue solid lines represent what happens in presence of a scale-dependent noise where the noise always increases with the size of the computer (and gets higher than $\eta_{\text{thr}}$ at some point). The red points correspond to the maximum accuracy one curve can get to, and the associated concatenation level is what we call $k_{\max}$.}
\label{fig:qualitative_behavior_scale_dep}
\end{center}
\end{figure}
\FloatBarrier
In summary, in general, the hypothesis (i)-(iv) can give rise to the behavior of any curve represented in figure \ref{fig:qualitative_behavior_scale_dep}: many minima can occur, but at some point, concatenations do not help to make the computation more accurate. However, all the physical examples we are going to treat in the main text will correspond to either the curves C1 or C4.
\section{Physical examples}
\label{sec:physical_examples}
Now that we explained what we could expect in the presence of scale-dependent noise, we will study physically motivated examples. The main question we are interested in is how much resources it would cost to implement an algorithm. To answer this question, we must understand what it precisely means to "implement" an algorithm. 

If we are interested in energy or power as a resource, we can, of course, spend (almost) no energy or power to perform a computation. If we use superconducting qubits, it would mean, for instance, that the qubits are not inside a cryogenic unit, and there is no energy required to maintain them cool. But if someone does this, the answer of the algorithm will be extremely noisy and not trustable (even with error correction as the probability of fault for the physical gates would be higher than the threshold such that error correction wouldn't help). The resource is probably "minimized" but in a manner that is not satisfying. Another approach could consist in asking what is the amount of energy one has to spend in order to have the \textit{maximum} accuracy in the computation. But this question is not satisfying either as it would mean (among other things) that the qubits would have to be at $0$ Kelvin to minimize the amount of thermal noise. It would require an infinite amount of energy (or power). 

A good way to tackle this problem is thus to find the \textit{minimum} amount of resources in order to have an algorithm that reaches a \textit{targetted} success rate, decided by the experimentalist. Typically if the answer of the algorithm is good with a probability greater than $1/2$, running the algorithm a few times and doing some majority vote would allow the experimentalist to know the output of the algorithm with very high fidelity. For this reason, a commonly used target in the literature is to ask an algorithm to be successful $2/3$ of the time. The concept of finding the minimum resource under the constraint of a targetted algorithm accuracy will be at the root of all the resources estimations we will do in this chapter and in the following. 

In a second time, we will also see that there are mathematical connections between some crosstalk models and the problem of resource estimation in the sense that both will give the same type of scale-dependent noise. Thus, the calculations we are going to do for scale-dependent noise from limited resources will directly be applicable to some crosstalks models.
\subsection{Relationship between scale-dependent noise and resources estimation}
\label{sec:relation_scale_dep_resource_estimation}
Let us assume that each physical gate in the quantum computer needs some resource in order to perform an operation such that the less this resource is available, the noisier the operation will be. Calling $R$ the resource those physical gates require, what can happen is that the physical noise can be expressed as a function of this resource\footnote{This will be the case in the examples of this chapter, but in general $R$ and $\eta$ might only be "correlated": $\eta$ is not necessarily a function of $R$.}: $\eta=g(R)$, where $g$ is a \textit{decreasing} function of $R$ (the more the resource is available, the less noisy the physical gates are). From now on, we consider that we want to perform an algorithm composed of a unique logical gate (the generalization to any algorithm will be straightforward). Performing $k$ levels of concatenations, each physical gate and qubit will have been replaced by many physical gates and qubits in a recursive manner in order to perform error correction. Then, the total amount of resources available for the logical gate, we call it $R_{L}$, will be shared among all those physical elements. In the figure \ref{fig:resource_conservation}, we can see an example where the resource appears to be shared among the physical gates.
\begin{figure}[h!]
\begin{center}
\includegraphics[scale=0.4]{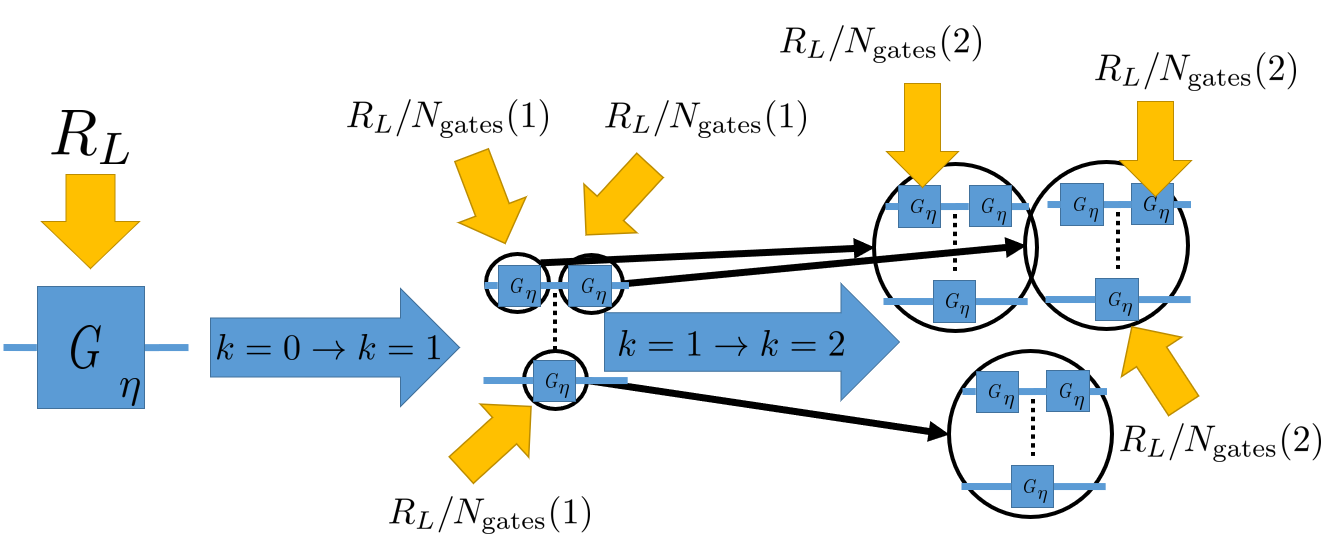}
\caption{Example of resource being shared among physical gates. If there is a total amount of resources $R_{L}$ available in the laboratory to implement a logical gate, the more concatenations are performed and the less of this resource will be available for each physical element. If the resource is shared among physical gates, calling $N_{\text{gates}}(k)$ the number of physical gates within a logical gate for a concatenation level $k$, each physical gate receive $R_{L}/N_{\text{gates}}(k)$ of this resource after $k$ concatenations. This figure follows the exact same principle of the figure \ref{fig:recursive_concatenation}.}
\label{fig:resource_conservation}
\end{center}
\end{figure}
But in general, the resource might be shared on other physical elements than the total number of gates, it can be shared on the number of gates that are active in parallel for instance (it would be the case for the power required to drive the quantum gates), or we could imagine a resource that is being shared on the total number of qubits. Calling $N(k)$ the number of physical elements on which the resource is being shared (this number increases with $k$), and assuming $N(0)=1$ (without concatenation a unique physical element is present, generalizations for $N(0)>1$ would be straightforward), each of those elements would receive an amount $R=R_L/N(k)$ of the resource (assuming that the resource is being shared equally on all the physical elements). Thus, in general, after $k$ level of concatenations, the probability of failure of the physical gates follows a law $\eta(k)=g\left(R_L/N(k)\right)$, where $g$ is a decreasing function of the variable $R=R_L/N(k)$. As $N(k)$ increases with $k$, $g$ is an increasing function of $k$: we are in the presence of a scale-dependent noise. 

To use the same notations introduced in the section \ref{sec:hypotheses_we_make}, we can rewrite $\eta(k)=\eta_0 f(k)$ with $\eta_0=g \left(R_L \right)$ and $f(k)=g \left(R_L/N(k)\right)/ \eta_0$. From the behavior of $g$, we already know that $f$ increases as a function of $k$, and $f(0)=1$ by construction. In order to have the remaining hypotheses (iii)-(iv) of \ref{sec:hypotheses_we_make} satisfied, we must have $\eta_0<\eta_{\text{thr}}$, and the fact there exist some real number $k_0$ such that $\eta(k_0) \geq \eta_{\text{thr}}$. It is what we can usually expect for a wide variety of situations, and it will be the case in the following examples. 

Now, we can establish the connection with resource estimation. We consider an algorithm composed of $N_L$ logical gates. We recall from \eqref{eq:probability_error_algo_k_concatenations} in the section \ref{sec:arbitrary_accurate_quantum_computing} that a way to estimate the probability of failure of this algorithm is through the formula $p^{(k)}_{\text{unsucessful}}=N_L p_L^{(k)}(\eta(k))$ at first order in $p_L$. Expressing this quantity in function of the resource available per logical gate $R_L$, we thus get:
\begin{align}
p^{(k)}_{\text{unsucessful}}=N_L p_L^{(k)}\left(\eta_0 g\left( \frac{R_L}{N(k)} \right)\right)
\label{eq:psuccess_algo}
\end{align}
It gives the connection between the success of the algorithm and the available resource $R_L$ provided to each of the logical gates. What is interesting is, of course, to actually relate it to the total amount of resource available for the \textit{entire} algorithm $R_\text{Algo}$. For a resource that is shared among gates, we would simply have $R_{L}=R_\text{Algo}/N_L$. Otherwise, the appropriate conservation law has to be injected. But in the end, conceptually, using the appropriate resource conservation law, \textit{we can express the probability that the algorithm fails as a function of the resource of interest and the concatenation level}, which give a law $p^{(k)}_{\text{unsucessful}}(R_{\text{Algo}})$. Then, if someone wants to implement an algorithm in such a way that the algorithm succeeds with a probability being at least $p^{\text{target}}_{\text{success}}$, while minimizing its expense, it will have to solve the following equation:
 \begin{align}
\min_k \left[R_{\text{Algo}}\right]_{\big |p^{(k)}_{\text{unsucessful}}(R_{\text{Algo}}) \leq 1- p^{\text{target}}_{\text{success}}}.
\label{eq:min_kRalgo}
\end{align}
This minimization under constraint gives the optimum level of concatenation to perform in order to minimize the resource spend. It establishes the connection between scale-dependent noise and resource estimation: we fix the total amount of resources available to some value $R_{\text{Algo}}$ which gives rise to a scale-dependent noise. Then, we can solve \eqref{eq:min_kRalgo} in order to find the minimum resource required to perform the computation. We also notice that if the accuracy of the gates is limited (i.e., $p_L^{(k)}$ is bounded), the size of the algorithms that the computer can run successfully (i.e., with a success probability being $2/3$ for instance) will be bounded as $p^{(k)}_{\text{unsucessful}}$ depends on $N_L$ as shown in \eqref{eq:psuccess_algo}. The maximum accuracy of the logical gates is directly related to the maximum size of the algorithms the computer can run.

\subsection{Noise growing proportionally with the number of physical elements}
\label{sec:noise_growth_prop}
In many cases of interest and in particular, in the one we are interested in what follows, the physical error can be expressed as $\eta(k)=\eta_0 D^{\beta k}$ for some $\beta \geq 0$ and $D \geq 0$. Such expression would come from a noise model in which the physical error is proportional to the number of components within the computer (up to some power exponent $\beta$). Indeed, because of the recursive ("Russian dolls") structure of the concatenations shown in the section \ref{sec:arbitrary_accurate_quantum_computing}, the number of elements after $k$ concatenation typically follows a law $N(k)=D^k$ where $D$ is the number of physical elements (qubits or gates) within one concatenation level\footnote{As we will justify more properly in the sections \ref{sec:estimation_physical_gates} and \ref{sec:estimation_physical_qubits} of the next chapter, this law is actually an approximation to the exact numbers.}. Thus, for a noise proportional to the number of physical elements after $k$ concatenations, up to some power $\beta$, we would have:
\begin{align}
\eta(k)=\eta_0 N(k)^{\beta}=\eta_0 D^{\beta k}
\end{align}
for some $D$. A physical realization of such noise occurs for instance if the noise of a physical gate follows a law $\alpha R^{-\beta}$ where $R$ is the resource spent to implement this gate, $\alpha \geq 0$ some proportionality coefficient and $\beta \geq 0$ (because the more resource, the better the physical gate is). We would then have $\eta(k)=g(R_L/N(k))=\alpha (R_L/N(k))^{-\beta}$ which gives the following scale-dependent noise:
\begin{align}
&\eta(k)=\eta_0 f(k)\\
&\eta_0=\alpha R_L^{-\beta}\\
&f(k)=N(k)^{\beta}=D^{\beta k}
\end{align}
Such models can easily be understood analytically. For instance the maximum accuracy one can get to with such model is straightforward to obtain. Treating $k$ as a continuous parameter, we find that (i) there is a \textit{unique} minimum for $p_L^{(k)}(\eta(k))$ and (ii) it is reached for $k_{st}$ (st for stationnary point) at:
\begin{align}
k_{st}=-\frac{1}{\ln(2)}-\frac{\ln(\eta_0 / \eta_{\text{thr}})}{\beta \ln(D)}.
\end{align}
Because there is a unique minimum, $k_{\max}$, the concatenation level leading to the maximum accuracy is reached for one of the two closest integers to $k_{st}$. More precisely, because $p_L^{(k)}(\eta(k))$ admits a (unique) minimum, there exists $\widetilde{k}$ such that $p_L^{(\widetilde{k}-1)}(\eta(\widetilde{k}-1))=p_L^{(\widetilde{k})}(\eta(\widetilde{k}))$. And $k_{\max}$ will be the lowest integer in the range $[\widetilde{k}-1,\widetilde{k}]$, thus it satisfies:
\begin{align}
&k_{\max}=\ceil[\bigg]{\widetilde{k}}-1\\
&\widetilde{k}=-\frac{\ln(\eta_0 D^{\beta}/\eta_{\text{thr}})}{\ln(D^{\beta})}
\label{eq:ktilde}
\end{align}
where $\ceil{x}$ denotes the ceiling function (i.e the function that rounds to the closest higher integer). All those different numbers are represented on figure \ref{fig:p_kmax_singleminimum} for more clarity.
\begin{figure}[h!]
\begin{center}
\includegraphics[scale=0.35]{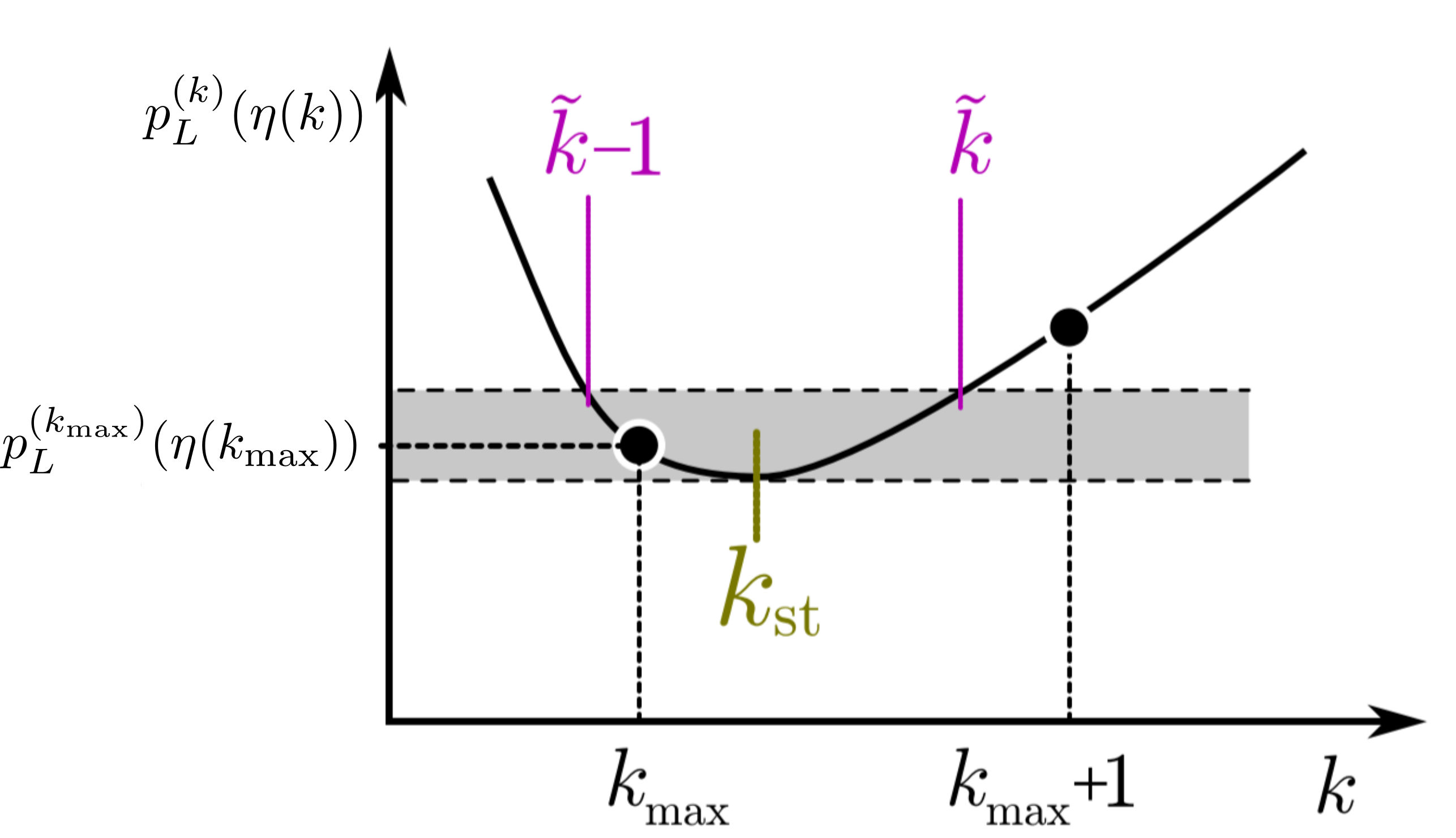}
\caption{$p_L^{(k)}(\eta(k))$ around its minimum.}
\label{fig:p_kmax_singleminimum}
\end{center}
\end{figure}
We can also plot the maximum accuracy reachable for the logical gates: $p^{(k_{\max})}(\eta(k_{\max}))$ as a function of $\eta_0$. The graph is represented on the figure \ref{fig:p_Delta_kmax} \textbf{a)} with $D=291$. It corresponds to a model in which the noise grows proportionally with the number of physical gates. Indeed, $D=291$ represents the number of fault locations inside one logical cNOT (this concept has been defined in the section \ref{sec:quantitative_estimation_first_level_protection}) \cite{aliferis2005quantum} for one level of concatenation, which is roughly\footnote{The curves presented in this chapter have been done based on approximate estimations for the number of physical components. More quantitative analysis will be done in the next chapter in the section \ref{sec:estimation_physical_gates}.} corresponding to the number of physical gates in a cNOT\footnote{The number of physical gate inside a cNOT for one concatenation level could exactly be estimated based on the table \ref{table:FT_gate_breakdown}. We would see that the number of fault locations and physical gates would be about the same order of magnitude but different. It can be explained by a various number of technical details, but, for instance, an initialization of a qubit in the state $\ket{0}$ is not considered as a gate, but as fault could occur there, it counts as a fault location. Also, the calculation leading to $291$ fault locations has been done in \cite{aliferis2005quantum}, and it assumes that measurements in the $\sigma_x$ basis can be done natively, i.e., without having to apply a Hadamard gate before doing a $\sigma_z$ measurement (which is not what is assumed in the counting leading to the table \ref{table:FT_gate_breakdown}). But those are small differences that will not change the general shape of the curves we are presenting.}. We will see in the section \ref{sec:estimation_physical_gates} of the next chapter a more accurate estimation of the number of physical gates after $k$ concatenation, but we consider it as being a valid approximation in order to access qualitative figures.

From this analysis, a striking behavior is occuring. We can wonder ourselves how low should the physical fault probability of a single isolated gate, $\eta_0$ be in order to make at least one concatenation level usefull. For this, we must find for which condition we have $k_{\max} \geq 1$ which means $\widetilde{k} >1$. We find that it implies:
\begin{align}
\eta_0 < \eta_{\text{thr}} D^{-2 \beta}
\end{align}
Considering $D = 291$, and $\beta=1$, thus if the noise grows proportionally to the number of gates, it would mean that $\eta_0$ should be more than five orders of magnitude lower than the typical threshold to make error correction useful at all. The situation $ \eta_{\text{thr}} D^{-2 \beta} < \eta_0 < \eta_{\text{thr}}$ is an example which could induce the behavior of the curve C1 that we represented on figure \ref{fig:qualitative_behavior_scale_dep}. The reason behind this is that the noise grows too fastly. Even though someone could believe error correction to be useful (because $\eta_0 < \eta_{\text{thr}}$), the high number of elements required for $k=1$ and the fact the noise grows proportionally to the number of physical elements makes error correction actually useless here.

At this point, we conceptually introduced what we need to estimate the resource a calculation will cost. We are going to see a concrete physical example in a following section. But one could argue that the reason behind our scale-dependent noise seems a little bit fictive. Indeed, why would someone \textit{impose} an arbitrary limitation in some resource? Actually, this could be the case in some experiments depending on the resource, but we would like to find an example in which it is clear that having a scale-dependent noise is inescapable. As discussed in the introduction, some physical situations such as crosstalk, for instance, due to long-range interactions between qubits, can induce this behavior. Luckily, the calculations we just did in this section can directly be mapped to some of these situations.
\subsection{Long-range correlated crosstalk}
\label{sec:long_range_correlated_crosstalk}
Crosstalk is a kind of noise that can occur when more than one qubit is inside a quantum computer. We can consider two kinds of crosstalk. The first one can be described with local noise models and is usually due to the fact that an experimentalist trying to control one given qubit to do a single qubit gate might in an undesired manner drive other qubits than the one initially targetted, at the same time \cite{heinz2021crosstalk,schutjens2013single}. In those situations, the noise can be considered as being local because the state of the qubits that have been manipulated in an undesired manner will not depend on the state of other qubits in the computer. We can understand those noise models with our approach; one has to find the law $\eta(k)$ describing such situation and inject it in \eqref{eq:logical_error_probability}. The maximum accuracy of the computer could then be found by calculating $k_{\max}$. We can also make the remark that local crosstalk models can occur experimentally because of a limited frequency bandwidth \cite{schutjens2013single}. Seeing it as a resource, it would be an example that could enter in the general approach behind the section \ref{sec:relation_scale_dep_resource_estimation}. There also exists correlated crosstalk issues (see for instance \cite{mckay2019three,takita2017experimental,proctor2019direct}). In this case, because there are parasitic qubit-qubit interactions, the state of a qubit $i$ inside the computer might be affected by the state of a qubit $j \neq i$ in an undesired manner. This would lead to a noise model that is non local and for this reason we cannot, in principle, understand its effect with \eqref{eq:logical_error_probability} (in the first paragraph of the section \ref{sec:axiomatic_FT} we explained that \eqref{eq:logical_error_probability} requires $\eta$ to describe a local noise model). Here, we will see that we can actually understand the effect of correlated crosstalk models by adapting a little bit the calculations we already presented.

The situation we consider can be described by a total Hamiltonian $H$ composed of two terms:
\begin{align}
H=H_S+H_{\text{int}}.
\end{align}
The Hamiltonian $H_S$ is the (time-dependent) Hamiltonian that implements the ideal circuit, and $H_{\text{int}}$ is a Hamiltonian making the qubits interact in an \textit{undesired} manner. Here, we assume that $H_{\text{int}}=\sum_{\langle i,j \rangle} H_{ij}$ where $H_{ij}$ is a two-qubit Hamiltonian making the qubits $i$ and $j$ interact. The sum is performed on \text{any} pair of qubits in the quantum computer. For this reason, the noise model will correspond to a long-range\footnote{Because any pair of qubits can interact in principle, this model is called "long-range".} correlated noise. Our goal here is to understand how accurate the computer can be, given the presence of $H_{\text{int}}$.

We understand from this model that we are leaving the standard assumption of fault-tolerance in which each gate can be described by a quantum channel: there are correlations between all the qubits within the computer such that the notion of fault for a physical gate that we used in the section \ref{sec:axiomatic_FT} becomes ill-defined: it is not possible to ignore those correlations to understand the dynamic. It has been shown in \cite{aharonov2006fault} that despite this fact, such noise models can still be understood with standard fault-tolerance theory.

In order to treat such noise models, the approach done in \cite{aharonov2006fault} consists in introducing the quantity 
\begin{align}
\Delta(N_{\text{qubits}})=\sup_i \left[ \sum_{j \neq i} ||H_{ij}|| \right],
\label{eq:Delta_def}
\end{align}
where $N_{\text{qubits}}$ is the number of physical qubits in the computer, and $||.||$ is the sup-operator norm \cite{treil2013linear,aharonov2006fault,todd2010operator}. $\Delta$ depends on $N_{\text{qubits}}$ because the sum involved in its definition is performed on all the qubits within the computer. Calling $t_0$ the duration of the longest physical gate in the computer (which gives the duration of one timestep of error correction), $\Delta(N_{\text{qubits}}) t_0$ can be interpreted as an error amplitude because $(\Delta(N_{\text{qubits}}) t_0)^2$ is proportional to the loss in fidelity of the state containing all the qubits in the computer, due to those parasitic two qubit interactions, during this time $t_0$, see \cite{terhal2005fault} for further details.

Now, we would like to understand how fault-tolerance works in the presence of such non-local noise. It can be done by replacing $\eta$ in \eqref{eq:logical_error_probability} by:
\begin{align}
\eta = e^{1+1/2e}\sqrt{2 \Delta(N_{\text{qubits}}) t_0}.
\label{eq:etacrosstalk}
\end{align}
The square root in \eqref{eq:etacrosstalk} is only here for mathematical reasons, allowing to map a non-local noise model on the results of fault-tolerance, which in principle applies to local noise models. The proof behind this mathematical result is in \cite{aharonov2006fault}. Now, the quantity $p_L^{(k)}=\eta_{\text{thr}}(\eta/\eta_{\text{thr}})^{2^k}$ has the same interpretation than $t_0 \Delta(N_{\text{qubits}})$ of being an error amplitude. The reason is that the principle of concatenations allows to reduce the noise in the way it has been quantified. If the noise strength has a probabilistic interpretation (which was the case until now, as explained in the last paragraphs of \ref{sec:arbitrary_accurate_quantum_computing}), then $p_L^{(k)}$ will have this interpretation. If the noise strength has an error amplitude interpretation, then $p_L^{(k)}$ will have this same interpretation. It is something to keep in mind when we show the results.

As explained in this paper, and as we can directly understand here, if $\Delta(+\infty)$ is finite, arbitrarily accurate quantum computing remains possible given the fact that $\eta < \eta_{\text{thr}}$ for $N_{\text{qubits}}=+\infty$. Indeed it would exactly correspond to the hypothesis we made in section \ref{sec:general_situation}, removing (iv) (we commented right before (iv) that having a scale-dependent noise is not an issue for arbitrarily accurate computation if it satisfies $\sup_k[\eta(k)]<\eta_{\text{thr}}$).

Now, the quantity $\Delta$ might have many possible behaviors as a function of $N_{\text{qubits}}$ because many two-qubit Hamiltonians $H_{ij}$ exist. Here, we are interested in the case in which $\eta$ might diverge and lead to a scale-dependent noise satisfying the general assumptions we made in this chapter. It is a regime that has not been studied in \cite{aharonov2006fault}. More precisely, we are going to assume that: 
\begin{align}
||H_{ij}|| = \delta/r_{ij}^z
\label{eq:Hij}
\end{align}
where $r_{ij}$ is the distance between the qubits $i$ and $j$, and $z$ a positive power describing the speed at which the interaction decreases. We will work in the regime where $\Delta(+\infty)=+\infty$ which will give some conditions on $z$. Our goal is thus to determine $\Delta(N_{\text{qubits}})$ which, as $N_{\text{qubits}}$ is a function of $k$, will give us access to the law $\eta(k)$. In order to make further connections with local noise models, we define $\epsilon=e^{1+1/2e} \sqrt{2}$, and $\Delta^{(k)}$ such that:
\begin{align}
t_0 \Delta^{(k)}=\frac{\eta_{\text{thr}}^2}{\epsilon^2} \left( \frac{\epsilon^2 t_0 \Delta(N_{\text{qubits}}(k))}{\eta_{\text{thr}}^2} \right)^{2^k}.
\label{eq:t0_times_deltak}
\end{align}
Then, we directly have $p_L^{(k)}=\epsilon \sqrt{t_0 \Delta^{(k)}}$. We notice in \eqref{eq:t0_times_deltak} that $t_0 \Delta^{(k)}$ as a function of $t_0 \Delta(N_{\text{qubits}}(k))$ is analog to $p^{(k)}_L$ as a function of $\eta(k)$ in \eqref{eq:logical_error_probability} under the transformation $\eta_{\text{thr}} \to (\eta_{\text{thr}}/\epsilon)^2$. This will be one of the reasons why the graphs a) and b) on the figure \ref{fig:p_Delta_kmax} will have strong similarities and it is what motivated us to define $\Delta^{(k)}$ this way.

Now, we need to evaluate $\Delta(N_{\text{qubits}})$. The details of the calculations are presented in the appendix \ref{app:long_range_correlated_proof}. Here we will just explain the very basic principle and detail the two scenarios we are considering. First, the exact value of $\Delta(N_{\text{qubits}})$ depends on the geometry on which the qubit are spread. Here we consider two standard scenarios: in the first one, the qubits are regularly spaced on a $d=1$ dimension grid where the distance between each qubit is called $a$. In the second one the qubits are positioned on a $\sqrt{N_{\text{qubits}}} \times \sqrt{N_{\text{qubits}}}$ square lattice (with still a lattice spacing being $a$). We then have two functions to evaluate: $\Delta_{z,d}(N_{\text{qubits}})$ ($d$ being the dimension of the lattice). It gives us the following two sums to evaluate:
\begin{align}
&\Delta_{z,d=1}(N_{\text{qubits}})=\frac{2 \delta}{a^z} \sum_{j > 0}^{N_{\text{qubits}}/2} \frac{1}{j^z}
\label{delta1}\\
&\Delta_{z,d=2}(N_{\text{qubits}})=\frac{\delta}{a^{z}} \sum^{\sqrt{N_{\text{qubits}}}/2}_{i=-\sqrt{N_{\text{qubits}}}/2} \sum_{j=-\sqrt{N(k)}/2}^{\sqrt{N_{\text{qubits}}}/2} \frac{c_{ij}}{\sqrt{i^2+j^2}^z},
\label{delta2}
\end{align}
where we defined $c_{ij}$ such that $c_{ij}=0$ iff $i=j=0$. Those sums diverge for $z \leq d$ which is thus the regime we are interested in. As shown in the appendix, using the fact that the number of qubits grow with a law being approximately $N_{\text{qubits}}(k)=Q_L D^k$, with $D=291$, $Q_L$ being the of logical qubits (which corresponds to the number of physical qubits for $k=0$)\footnote{The number of physical qubits after $k$ concatenations knowing that there are $Q_L$ logical qubits will also be calculated more precisely in \ref{sec:estimation_physical_qubits}. This "quick estimation" used here is motivated by the principle of Russian dolls behind fault-tolerance construction, that naturally leads to a growth close to $D^k$.}, upper and lower bounding those sums by appropriate integrals, we find in both cases that
\begin{align}
\Delta_{z,d}(N_{\text{qubits}}(k)) \approx \Delta^{(0)}_{z,d}D^{k(1-z/d)}
\label{eq:Delta_zd}
\end{align}
in the limit where $N_{\text{qubits}}(k) \to +\infty$, where the constant $\Delta^{(0)}_{z,d}$ satisfies depending on the value of $d$:
\begin{align}
&\Delta^{(0)}_{z<d,d=1}=\frac{\delta  Q_L^{1-z} 2^z}{a^z(1-z)}\label{eq:delta0_d1} \\
&\Delta^{(0)}_{z<d,d=2}=\frac{\delta 2^{z} C_z Q_L^{1-z/2}}{a^z} \label{eq:delta0_d2} 
\end{align}
with $C_z=(2/(2-z)) \int_{\pi/4}^{\pi/2} d \theta sin(\theta)^{z-2}$. We also treated the case $z=d$ (which gives $\Delta_{z,d}(N_{\text{qubits}}(k))$ that grows as $C \ln(k)$ for some $C$) in \cite{fellous2020limitations}, but it corresponds to an additional study we prefer to not consider here. 

To summarize, here we did the following. First, we used the results from \cite{aharonov2006fault} which allows to understand how non-local noise affects the accuracy of the logical gates by doing the replacement \eqref{eq:etacrosstalk}, where $\Delta(N_{\text{qubits}})$ is defined in \eqref{eq:Delta_def}. We studied the physics for those non-local noise in the case $||H_{ij}||$ is being given by the law \eqref{eq:Hij}, and in the regime $\Delta(+\infty)=+\infty$ (it gives a condition on how fast $||H_{ij}||$ decreases with the distance between the qubits $i$ and $j$). Then, using the fact that $p_L^{(k)}=\epsilon \sqrt{t_0 \Delta^{(k)}}$, and $t_0 \Delta^{(k)}$ given in \eqref{eq:t0_times_deltak}, we can deduce the maximum accuracy of the gates. It is represented on figure \ref{fig:p_Delta_kmax} b).
 
\begin{figure}[h!]
\begin{center}
\includegraphics[width=0.9\textwidth]{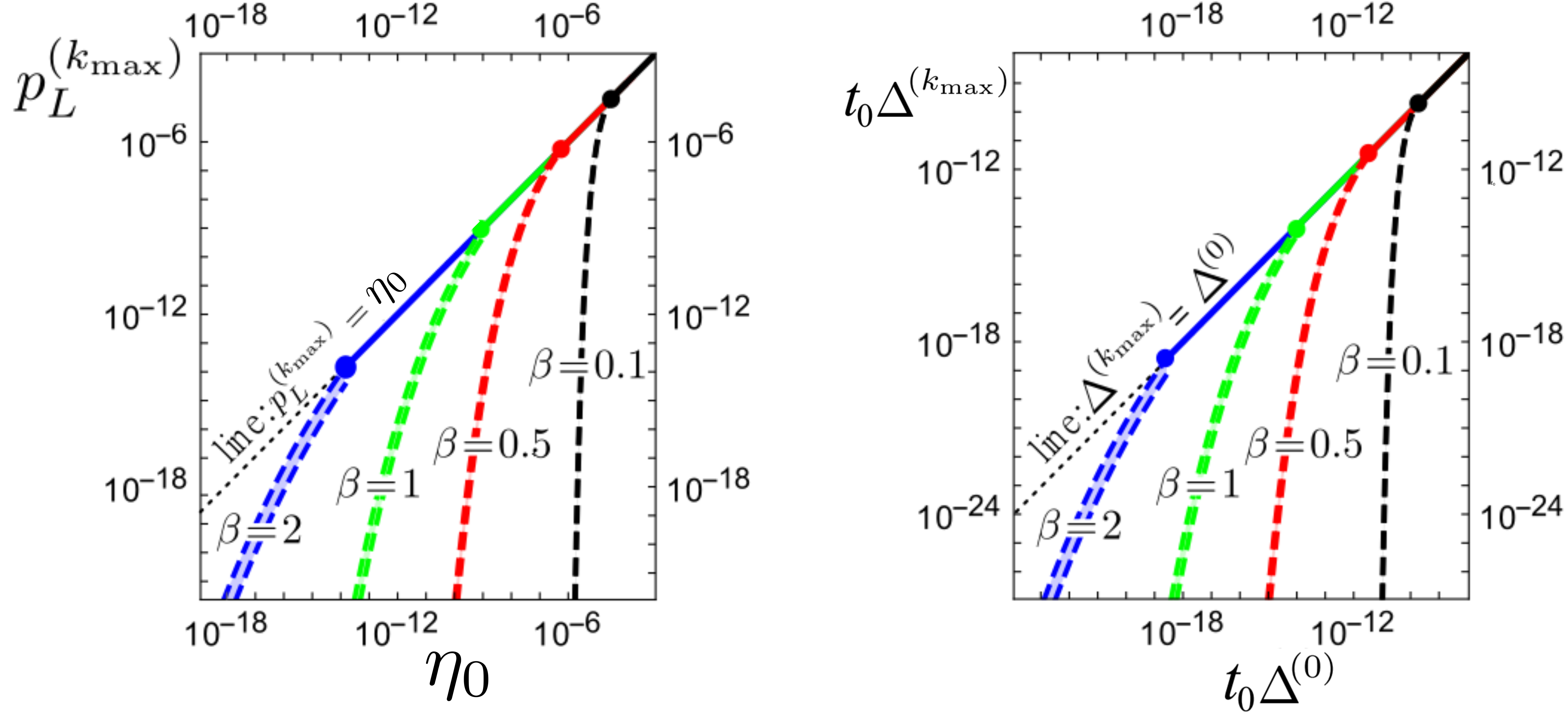}
\caption[Caption]{\textbf{a)}: Lowest probability of error of a logical gate as a function of $\eta_{0}$ for a physical error probability $\eta(k)=\eta_0 D^{\beta k}$, corresponding to a noise growing proportionally to the number of physical elements up to some power (it corresponds to the models treated in section \ref{sec:noise_growth_prop}) \protect \footnotemark. \textbf{b)}: $t_0 \Delta^{(k_{\max})}$ as a function of $t_0 \Delta^{(0)}$ for a physical error amplitude $\Delta(N_{\text{qubits}}(k))=\Delta^{(0)} D^{\beta k}$ corresponding to the non-local long range interaction model due to crosstalk. The expressions of $\Delta^{(0)}$ are given in \eqref{eq:delta0_d1} and \eqref{eq:delta0_d2}. This number allows to deduce $p_L^{(k_{\max})}=\epsilon \sqrt{t_0 \Delta^{(k_{\max})}}$ ($\epsilon \approx 4.6$) which has a (logical) error amplitude interpretation (it must be squared to get something "homogeneous" to a probability of failure). Even though error correction might help to limit correlated noise, we see that in order to be useful the noise strength must be so small that one could already perform huge calculation without needing to use correction. In both graphs $D=291$ as $291^k$ can be considered as being a rough approximation of the number of physical qubits (resp physical gates) inside one logical qubit (resp logical gate). The reason why a) and b) look similar (qualitatively, not quantitatively: the values on the axis are very different) is, first because the physical errors $t_0 \Delta(N_{\text{qubits}}(k))$ and $\eta(k)$ in a) and b) follow similar laws as a function of $k$, and then, because those errors on a logical level also follow a similar law as explained around \eqref{eq:t0_times_deltak}: surprisingly the long-range correlated noise model behaves not so differently than the resource constrained model described in \ref{sec:noise_growth_prop}.}
\label{fig:p_Delta_kmax}
\end{center}
\end{figure}
\FloatBarrier
\footnotetext{Strictly speaking, we represented the area between $p_L^{(\widetilde{k}-1)}$ and $p_L^{(\widetilde{k})}$ with $\widetilde{k}$ defined in \eqref{eq:ktilde} and the text above (this is why there are two blue lines visible for instance). We recall that $p_L^{(k_{\max})} \in [p_L^{(\widetilde{k}-1)},p_L^{(\widetilde{k})}]$. It would have been possible to directly plot $p_L^{(k_{\max})}$ but when this thesis was written the codes allowing to do such graph were no longer available. Same explanations for the graph $b)$.}

In conclusion, we understand that our approach can also allow us to consider the maximum accuracy one can get in the presence of non-local, correlated noise by doing a "remapping" of $\eta$ as indicated in \eqref{eq:etacrosstalk}. It extends our approach to more general noise models. The example we took in this section also illustrates the fact that scale-dependent noise models might be intrinsic to some architecture because of "parasitic" interactions. Surprisingly, the crosstalk model we took gives curves that have the same qualitative features as local noise induced by resource limitations, as explored in \ref{sec:noise_growth_prop} (but the quantitative values on the axis are very different). However, we also see that even though the long-range correlated crosstalk noise we studied can be reduced with error correction, the noise must be so low from the start that for all practical purposes, error correction might not be needed in those regimes. 

\section{From scale dependence to energetic estimation of a fault-tolerant algorithm}
\label{sec:from_scale_dep_to_energetic_estimation}
\subsection{The algorithm: Quantum Fourier transform}
\label{sec:algo_QFT}
Our goal now is to show the relationship between scale-dependent noise and resource estimation in a concrete example in order to find the minimum resource required to implement an algorithm. In order to be concrete, we must consider some algorithms to implement. We decided to choose an algorithm that is (i) well documented and (ii) used as a subroutine for many quantum algorithms. The Quantum Fourier Transform (QFT) belongs in this family. It is, for instance, one of the subroutines used in the Shor factoring algorithm, which allows decrypting a message encrypted with the classical RSA algorithm. The RSA algorithm allows securing messages by using the fact that a classical computer will take a very long time to be able to factorize into prime numbers some integer $P$ if $P$ is big enough \cite{nielsen2002quantum}. The Shor algorithm will use a quantum computer to speed up this calculation such that a factorization of this number will be found in a reasonable amount of time. It requires, among other things, to perform a quantum Fourier transform on $N = \log_2(P)$ qubits: this number corresponds to the number of bits required to encode the integer $P$. A way to implement this algorithm is represented in the figure \ref{fig:circuit_qft}.
\begin{figure}[h!]
\begin{center}
\includegraphics[width=0.9\textwidth]{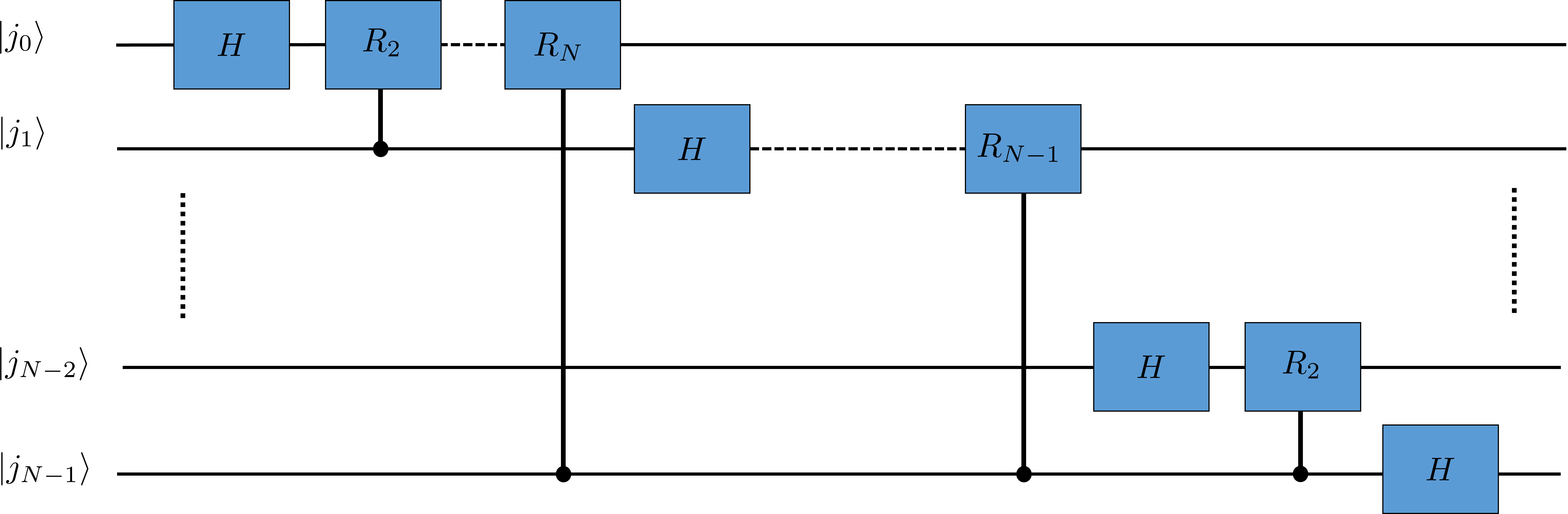}
\caption{Circuit describing a quantum Fourier transform performed on $N$ qubits. The controlled operations are controlled rotations which are implemented as shown on figure \ref{fig:controlled_R}.}
\label{fig:circuit_qft}
\end{center}
\end{figure}
\FloatBarrier
\begin{figure}[h!]
\begin{center}
\includegraphics[width=0.5\textwidth]{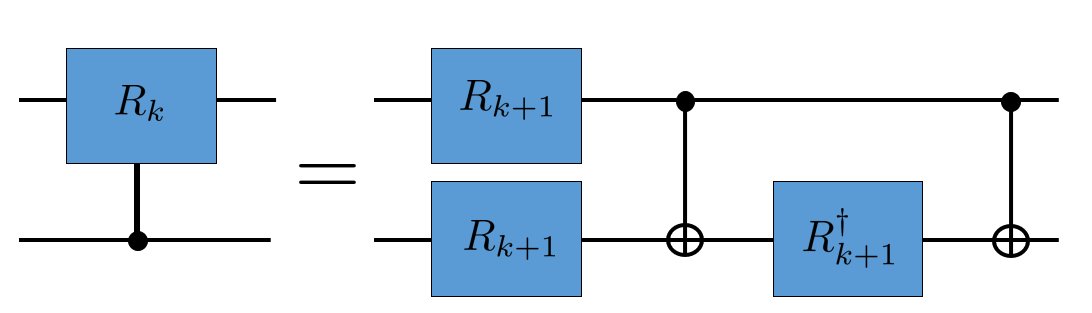}
\caption{The controlled $R_k$ operations can be decomposed by two cNOT and $R_{k+1}$ operations, where $R_k$ is a rotation of angle $\frac{\pi}{2^k}$ around the $z$ axis of the Bloch sphere as defined in \eqref{eq:rk} \cite{kim2018efficient}}
\label{fig:controlled_R}
\end{center}
\end{figure}
\FloatBarrier
It requires to implement controlled rotation around the $z$ axis, which requires to implement cNOT and $R_k$ gates where the latter is defined as:
\begin{equation}
R_k \equiv \begin{pmatrix}
1 & 0\\
0 & e^{i \frac{2 \pi}{2^k}}
\end{pmatrix}
\label{eq:rk}
\end{equation}
This is the algorithm we will base our energetic estimation upon. We can figure out from this algorithm that there are $N(N-1)/2$ controlled rotations and $N$ Hadamard gates. A typical value of $N$ we can consider is $2048$, which corresponds to the quantum Fourier transform involved in an algorithm that would decrypt a message encoded with an RSA key used today. To simplify the discussion, and as we are interested in explaining the concepts in a simple way, we will assume that all the controlled gates involved are cNOT gates having a probability of failure $p_L^{(k)}(\eta(k))$ for $k$ concatenations. The total number of logical gates is then of order $N^2$, and the probability of failure becomes:
\begin{align}
p^{(k)}_{\text{unsucessful}}=N^2 p_L^{(k)}(\eta(k))
\label{eq:p_failure_QFT}
\end{align}
At this point, we need to find the function $\eta(k)$. It will be given by the physical model describing the gate that we now introduce.
\subsection{The physical model: superconducting qubit in waveguide}
\label{sec:physical_model_chap3}
To have a model of noise for the gates, we must consider a physical system. We will assume that all the gates in the quantum computer will be as noisy as single qubit $\pi$-pulses. We will also consider that the only reason why the evolution is noisy is because of spontaneous emission, such that we will use the noise model described in the section \ref{sec:dynamic_in_presence_thermal_noise} and consider that $\overline{n}_{\text{tot}}=0$ in \eqref{eq:evolution_noisy_gate_with_thermal_noise}. This is, of course, an oversimplifying assumption: different physical gates are obviously described by different models. But we are only interested in qualitative behaviors here. This is the reason why we will consider that all the physical gates are as noisy as single qubit $\pi$-pulses which entirely defines our noise model.

Here, our goal will be to minimize the energy required for the pulses that are driving the gates. For this reason, we must express the parameters behind this problem as a function of this energy. We recall the following expressions given in the first chapter of this thesis. The power associated to a pulse resonant at the qubit frequency $\omega_0$, for a Rabi frequency $\Omega$ is: $P_g=\hbar \omega_0 \Omega^2/(4 \gamma_{\text{sp}})$ (see \eqref{eq:power_fct_rabi} in section \ref{sec:energetic_cost_to_perform_a_gate}). We deduce that the energy a $\pi$-pulse contains is simply: $E=P_g t_{\pi}=\hbar \omega_0 \pi  \Omega/(4 \gamma_{\text{sp}})$ (we used $t_{\pi}=\pi/\Omega$). It gives us the expression of the Rabi frequency as a function of the energy contained in the pulse: we obtain $\Omega=\frac{4 \gamma_{\text{sp}} E}{\pi \hbar \omega_0}$. Using this result and the fact that the evolution must be integrated for a duration $t_{\pi}=\pi/\Omega$, the dynamic is completely expressed as a function of the energy contained in the pulses. The only remaining parameters to fix are associated with the characteristics of the qubit and their coupling to the waveguide, basically the values of $\gamma_{\text{sp}}$ and $\omega_0$. We consider for them the values given in \ref{sec:state_of_the_art_qubit}.

Now that the dynamic has been described and that all the parameters we need have been provided, we must find the expression of the scale-dependent noise induced by the limited amount of energy in the pulses. In order to catch the strength of noise induced by this evolution, we define the map $\mathcal{N}$ such that:
\begin{align}
\mathcal{E}=\mathcal{U} \circ \mathcal{N},
\end{align}
where $\mathcal{E}$ is the quantum channel that is associated to the master equation \eqref{sec:dynamic_in_presence_thermal_noise} (we follow the approach described in the section \ref{sec:fault_and_errors}). Solving the dynamic, the process $\mathcal{N}$ can be computed. More precisely, it can be decomposed on the Pauli matrices basis such that: $\mathcal{N}(\rho)=\sum_{ij} \chi_{ij} \sigma_i \rho \sigma_j$, where the $4 \times 4$ matrix $\chi$ having for elements $\chi_{ij}$ for $i$ and $j$ integers between $0$ and $3$ entirely describe the process. The physical gate fault probability is finally defined as $\eta \equiv \max_{i>0} \chi_{ii}$. The reason why we consider this quantity as corresponding to the physical gate fault probability can be seen as an approximation: in some sense, this simplification means that we model the noise by considering that it corresponds to a probabilistic Pauli noise of "strength"\footnote{There are three coefficients describing a single qubit Pauli noise, here everything behaves as if two of those components where equal to $0$ and the last one to $\eta$. We could also consider that $\eta$ is the strength of a depolarizing channel which would give very similar results. What is important here is that we approximate our noise channel by a probabilistic noise model: this is the main approximation we are doing in order to be able to reason with probabilities and to avoid using norm-based estimations.} $\eta$. A completely rigorous treatment of the noise here would ask us to compute a norm for the process $\mathcal{N}$ and to use this norm as what plays the role of $\eta$, but this approach, even though more rigorous, would add un-necessary complication for what is here an example of principle. Also, it would lead to a very poor upper bound of the probability of logical error of the gate that is unnecessarily pessimistic. We can note that it is frequent in the literature to approximate noise models by Pauli noise \cite{jayashankar2021achieving,beale2018quantum}.

Performing the calculations, we find that
\begin{align}
&\eta=\chi_{11}=\frac{\pi^2}{16} \frac{\hbar \omega_0}{E}.
\label{eq:eta_phys}
\end{align}
Assuming that the energy for one logical gate is being fixed to $E_L$, we deduce that after $k$ concatenations, the physical fault probability becomes.
\begin{align}
\eta(k)=\frac{D^k \pi^2}{16} \frac{\hbar \omega_0}{E_L}
\end{align}
We are in the presence of a scale-dependent noise induced by limited resources. It allows us to estimate the minimum energy we have to spend in order to successfully run the algorithm.
\subsection{Minimum energy to perform the QFT}
On the figure \ref{fig:p_kmax_QFT}, we plot in black solid lines $\min_k [p_L^{(k)}(\eta(k))]=p_L^{(k_{\max})}(\eta(k_{\max}))$ as a function of $\overline{n}_L=E_L/\hbar \omega_0$ (the number of photons in the pulses that are being used by the logical gate). It corresponds to the maximum accuracy a logical gate can get to for a given photon budget for the logical gate. This curve admits discontinuity in its derivative. They correspond to the moment when $k_{\max}$ jumps from one integer value to another, as also illustrated in figure \ref{fig:kmax_QFT}. Typically, in the low photon regime $k_{\max}=0$ then it goes to $k_{\max}=1$ and $k_{\max}=2$. On figure \ref{fig:p_kmax_QFT} are also plotted, in gray dotted lines, the curves $p_L^{(k)}(\eta(k))$ in function of $\overline{n}_L$ for $k=0,1$ or $2$. Those curves will match $p_L^{(k_{\max})}$ on some regime, typically when $k_{\max}$ will have the same value of the $k$ associated with a gray line. The solid black line represents the "lower envelope" of all the gray dotted lines. This is expected as it corresponds to the maximum accuracy, i.e., the \textit{lowest} probability of fault as a function of the number of photons. No gray curve can thus be below the black one, by definition. Let us go back on the question of energetics. We can interpret further the exact meaning of the solid black line. It can actually be interpreted from two perspectives. For a given number of photons, it allows to deduce the maximum accuracy we can get to, as we just explained, but it also represents the \textit{minimum} number of photons one has to spend (for the logical gate) in order to reach a given accuracy target, i.e., a given value of $p_L^{k_{\max}}(\eta(k_{\max}))$. Those two different questions are actually connected. For instance, if we want the logical gate to have a probability of failure being $10^{-5}$, we would need approximately $10^5$ photons for the logical gate. And reciprocally, if we have $10^5$ photons for the logical gate, the maximum accuracy we can get to is a logical gate having a probability of failure being $10^{-5}$. Now we can interpret physically what happens when $k_{\max}$ is changing, which corresponds to a discontinuity in the derivative of the black curve. Reading the curve from low photon to high photon regime, the moment when there is a change of slope for $p_L^{(k_{\max})}(\eta(k_{\max}))$ physically means that an extra concatenation level starts helps to minimize the energetic expenses. Let us focus for instance on the change $k=0 \to 1$ occurring for $\overline{n}^0_L \approx 10^9$. It could be possible to increase the accuracy keeping $k=0$ and considering $\overline{n}_L>\overline{n}^0_L$. It corresponds to the gray dotted lines $p_L^{(k=0)}$. But doing so would not be smart in the sense that the accuracy could be reduced even more by increasing the concatenation level without spending more photons.
\begin{figure}[h!]
\begin{center}
\includegraphics[width=0.6\textwidth]{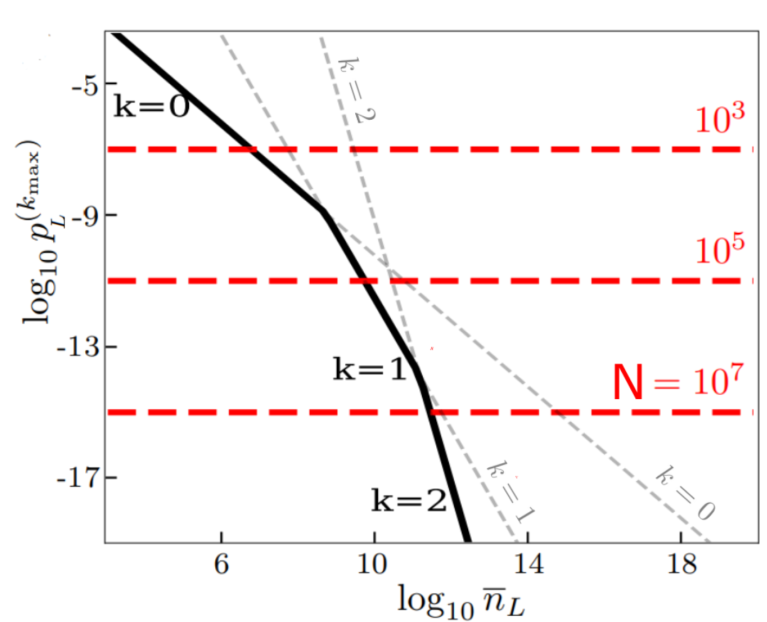}
\caption{In black solid lines: maximum accuracy one can get to for a given budget of photon ($p_L^{(k_{\max})}(\eta({k_{\max}}))$). The gray dotted lines are associated to $p_L^{(k)}(\eta({k}))$ for $k \in \{0,1,2\}$. The red horizontal lines correspond to different size of QFT, associated to Shor algorithm trying to factorize a key encoded on $N$ bits.}
\label{fig:p_kmax_QFT}
\end{center}
\end{figure}
\begin{figure}[h!]
\begin{center}
\includegraphics[width=0.6\textwidth]{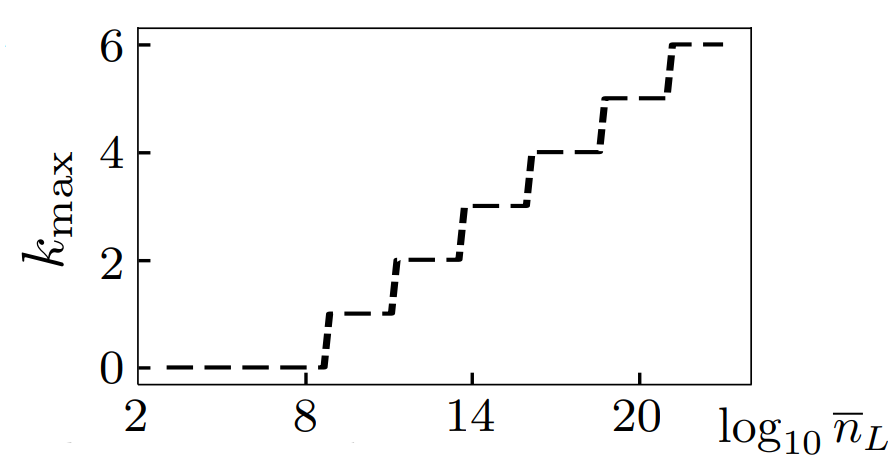}
\caption{Concatenation level allowing to reach the maximum accuracy, $k_{\max}$, as a function of $\overline{n}_L$.}
\label{fig:kmax_QFT}
\end{center}
\end{figure}

In order to be concrete, we want to know what is the accuracy we need in order to run a quantum Fourier transform and to deduce from that the minimum energy we need to run this algorithm. Assuming we want the algorithm to succeed $2/3$ times, we can deduce the fault probability each logical gate has to have. Indeed, using \eqref{eq:p_failure_QFT}, we deduce that we should have in this case: $p_L^{(k)}=\frac{1}{3 N^2}$. The remaining thing to know is thus the value of $N$. The red horizontal lines on figure \ref{fig:p_kmax_QFT} represent the fault probability we should have in order to have a QFT algorithm working $2/3$ times properly and having $N=10^3$, $10^5$ or $10^7$ qubits. To give an idea of how big $N$ can be in practice, for the QFT used within the Shor algorithm, the number of qubits $N$ required would be equal to $2048 \sim 10^3$ for today encryption protocols. But let us take $N=10^5$ for the example. Looking at the curve, we deduce that we would need a \textit{minimum} of $\overline{n}_L=10^9$ photons per logical gates. Using the fact that there would have about $N^2$ logical gates, considering $\omega_0/2 \pi \approx 6 GHz$, the \textit{minimum} total energy required would be about $E=N^2 \hbar \omega_0 \overline{n}_L  \approx 30 \mu J$. This energy is very low, and it allows us to see with a concrete example that having a scale-dependent noise is not necessarily an issue in itself. Indeed, if we assumed that we were limited to $1J$ of energy, the accuracy would have been fundamentally limited, but at a level in which huge algorithms could already be implemented. It would be interesting to adapt such calculations to detailed models of scale-dependent noise due to limited frequency bandwidth in order to see quantitatively if such problematics are really an issue for scalability (and if so, what is the largest algorithm it would be possible to implement in the presence of such noise), now that we provided a way to tackle this kind of questions.

In the end, this example illustrates the principle of resource estimation with a very simplified model. It contains the basic ingredients we are going to use in order to estimate the energetic cost in a more realistic manner in the following chapters: we will use this connection between noise and resource in order to minimize the energetic expense, which seems promising to make the quantum computer energy efficient. Now, here are some comments about the quantitative results we obtained here. First, it would seem from our calculation that no error-correction would be required for $N=10^3$ (which corresponds to what would be needed for many applications such as the Shor algorithm applied for today's encryption protocols). The reason behind this in our models is that (i) we only take into account the noise due to spontaneous emission, but mainly (ii) we assume that the physical noise can be put arbitrarily close to $0$ using enough photons. Indeed, $\eta$ in \eqref{eq:eta_phys} converges to $0$ in the limit of infinite energy. In practice, this behavior wouldn't occur because there are always other sources of noise that we neglect here. For instance, in practice, superconducting qubits are not exactly two-level systems: there are extra energy levels. Too short pulses (corresponding here to our regime where $\eta \to 0$) would induce leakage errors in the dynamic that are not caught with our model. Then, we find that the energy required is very low. This is explained by the fact we are only taking into account the energy contained in all the pulses that are driving the quantum gates. In practice, the energetic cost of quantum computing for superconducting qubits will be orders of magnitudes higher and mainly due to the cryogenic cost. In the last chapter of this thesis, we are going to consider a complete model that will include all those cryogenic costs, allowing us to make more realistic estimations. 
\section{Conclusion}

In this chapter, we have explained what has to be expected in the presence of a scale-dependent noise, i.e., a noise that grows with the size of the computer. Typically, the accuracy of the computer will be limited as soon as the scale-dependent noise grows in an unbounded manner. We provided the tools and concepts that allow estimating this accuracy. The maximum accuracy the computer can get to is not trivial to estimate in general, but we studied what happens in the presence of limited resources and for long-range correlated noise models. In some of those cases, the maximum accuracy can easily be accessed. We also showed that estimating the cost in a resource to perform a fault-tolerant computation can be phrased as a problem of finding the minimum resources required under the constraint of targetting a given accuracy, for a scale-dependent noise induced by a resource limitation. Indeed, assuming a resource to be limited, and assuming that the fewer this resource is available, the noisier the physical gates will be, the noise will grow with the computer size from resource conservation. This vision allows us to estimate what is the maximum accuracy the computer can get to, assuming that we have at disposal a given amount of the resource, and reciprocally, what is the minimum resource it costs to reach a given accuracy for a computation. We illustrated those concepts with a light-matter interaction model based on superconducting qubits embedded into waveguides. This example was a concrete illustration of a scale-dependent noise that is not a threat: fixing the total amount of energy to some "reasonable" macroscopic value induced a scale-dependent noise forbidding the quantum threshold theorem to apply (because the probability of fault for the physical gates grows with the concatenation level). But this maximum accuracy was large enough for all practical purposes. Our analysis was, however, too simple to be used for realistic estimation of the energetic cost of quantum computing: it should be understood as a first step toward that direction which establishes first basic concepts behind this problem. Typically we see that the key point is to relate the accuracy of the computation to the resource one wants to estimate and minimize. The purpose of the next chapters will be to generalize our framework and to do a more quantitative estimation of the energetic cost of quantum computing.
\begin{appendices}
\chapter{Example of non monotonous behaviors for $p_L^{(k)}$}
\label{app:non_monotonous_behavior_pL}
It might be surprizing that if $\eta(k)$ is a strictly increasing function, we can still expect multiple minima for $p^{(k)}_L(\eta(k))$. To give a basic intuition, we can design a fictive but instructive situation in which a similar behavior as on curve C5 would occur. Let us consider that $\eta_0$ is very small (i.e $\eta_0 \ll \eta_{\text{thr}}$). If $\eta(k)$ is such that it increases brutally in the range $[0,1]$ \textit{while satisfying} $\eta(1) < \eta_{\text{thr}}$, it could be the case that $p_L^{(1)}(\eta(1))>p_L^{(0)}(\eta(0))=\eta_0$. But if in addition, $f$ then \textit{increases slowly for further concatenation levels} it might exist a concatenation level $k$ such that $p_L^{(k)}(\eta(k))<p_L^{(0)}(\eta(0))$. Indeed as an extreme scenario, if $\eta(k \geq 1)$ is almost flat, this region would correspond to the standard fault-tolerance scenario and we will know that at some point concatenations would help. In this case we would first see concatenation as detrimental, then as profitable, then detrimental again. As a numerical example, we can consider what happens for  $\eta_0=10^{-8}$, $\eta_{\text{thr}}=10^{-4}$, $f(k\geq 1)=10^{3+0.21k}$ ($f(0)=1$). This function satisfies the assumption (i)-(iv) we made, and we find numerically that the situation is degraded then improved then degraded as a function of the concatenation levels as one can see on figure \ref{fig:example_nontrivial_generalcase}. The intuitive reason being the one we explained, $f(0)=1$ but $f(1)=1000$: the noise increases fastly at the beginning (while having $\eta(1)<\eta_{\text{thr}}$). Then it increases in a much slower manner allowing to make concatenation usefull at some point. And as $\eta(k)$ increases in an unbounded manner, at some point $p_L^{(k)}$ starts to diverge.
\begin{figure}[h!]
\begin{center}
\includegraphics[width=0.7\textwidth]{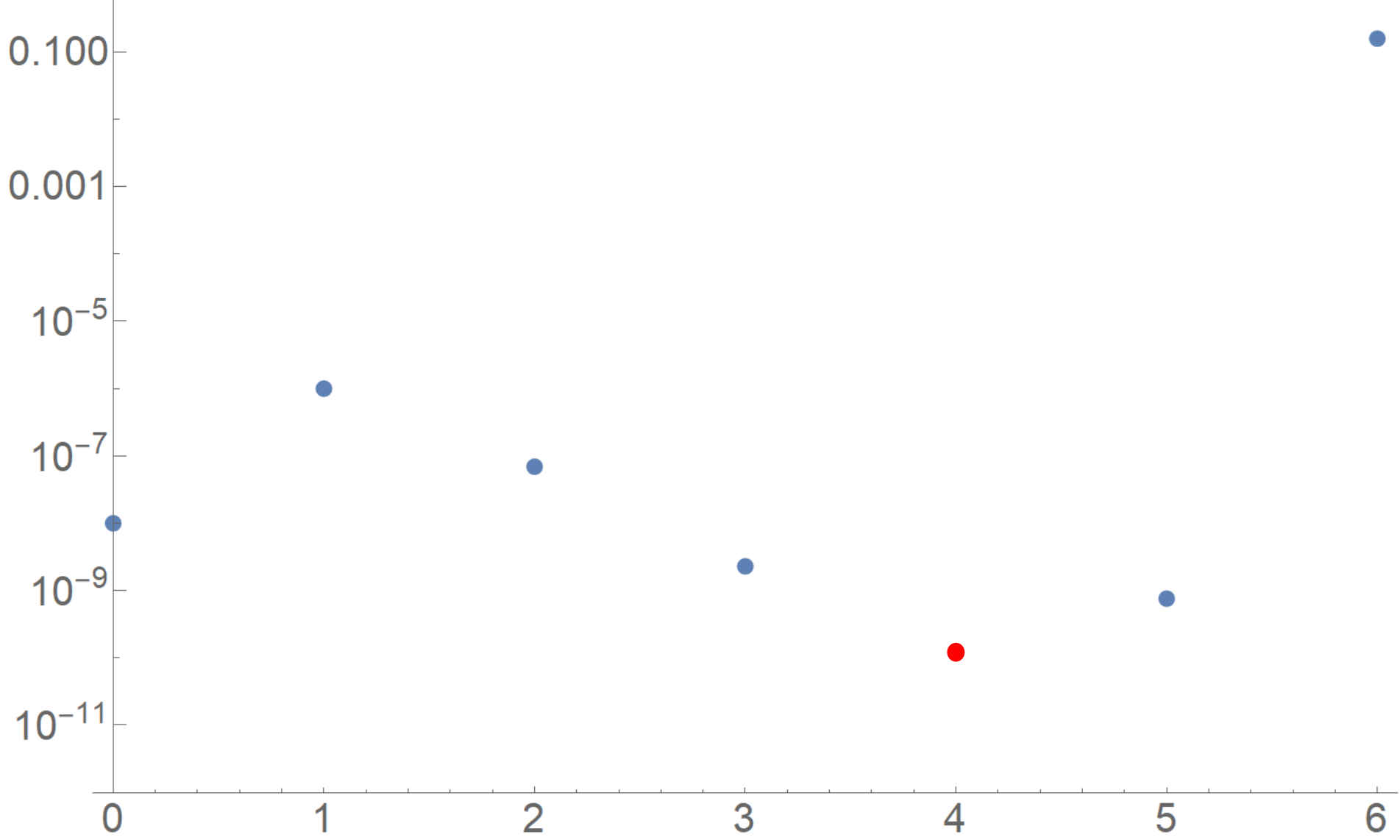}
\caption{On this figure is represented the probability of error of the logical gate after $k$ concatenations: $p_L^{(k)}(\eta(k))$ as a function of the concatenation level $k$ for the scale-dependent noise $\eta(k)$ described in the main text. It shows that concatenations might be initially detrimental before improving the protection and finally being detrimental again. The maximum accuracy is reached at $k_{\max}$ represented by the red point.}
\label{fig:example_nontrivial_generalcase}
\end{center}
\end{figure}
\FloatBarrier
\chapter{Long range correlated noise}
\label{app:long_range_correlated_proof}
As we are going to see, using the fact that $N_{\text{qubits}}=Q_L D^{k}$ where $Q_L$ is the number of logical qubits and $D$ is the number of physical qubits within a logical qubit for one concatenation level, this quantity will surprisingly give rise to a scale-dependent noise $\eta(k)=\eta_0 D^{\beta k}$ in the limit of large number of physical qubits in the computer. Thus, those crosstalks models behave mathematically the same way as the resource constraints models we studied. The exact expression of $\Delta(N_{\text{qubits}})$ depends on the way the qubits are spread. If we assume that they are spread on a one-dimensional lattice of regular spacing $a$, the sum has to be performed on a 1D topology. Using the fact that the strongest interaction will be felt by the qubit in the center of the computer, we get
\begin{align}
\Delta_{z,d=1}(N_{\text{qubits}})=\frac{2 \delta}{a^z} \sum_{j > 0}^{N_{\text{qubits}}/2} \frac{1}{j^z},
\label{delta1_app}
\end{align}
where the $2$ that multiplies $\delta$ comes from the fact that we exploited the symmetry around the central qubit to only sum on positive values of $j$. On the other hand, if the qubits are spread on a squared two-dimensional lattice of regular spacing $a$, we have, where $c_{ij}=0$ iff $i=j=0$ and $c_{ij}=1$ otherwise:
\begin{align}
\Delta_{z,d=2}(N_{\text{qubits}})=\frac{\delta}{a^{z}} \sum^{\sqrt{N_{\text{qubits}}}/2}_{i=-\sqrt{N_{\text{qubits}}}/2} \sum_{j=-\sqrt{N(k)}/2}^{\sqrt{N_{\text{qubits}}}/2} \frac{c_{ij}}{\sqrt{i^2+j^2}^z}.
\label{delta2_app}
\end{align}
Both \eqref{delta2} and \eqref{delta1} are diverging for $z \geq d$, where $d$ is the dimension on which the qubits are spread, which gives rise to a scale-dependent noise. Our goal is now to estimate $\Delta_{z,d}(N_{\text{qubits}}(k))$. In general, its expression is complicated and involves the infinite sum we wrote, but to simplify calculations and as fault-tolerance requires a high number of physical qubits, we will realize an asymptotic expansion of those quantities in the limit $N_{\text{qubits}}(k) \to +\infty$. We now give the main principles behind those derivations.

First, we can focus on the one dimensional case. We define $M=N_{\text{qubits}}/2$. We need to evaluate \eqref{delta1}. It can be easily be done by bounding $1/j^z$ by two integrals. Indeed, using the fact that $x \to 1/x^z$ is decreasing, we have:
\begin{align}
&\int_{n}^{n+1} dx \frac{1}{x^z}\leq \frac{1}{n^z} \leq \int_{n-1}^n dx \frac{1}{x^z}\\
& \int_{2}^{M+1} dx \frac{1}{x^z}\leq \sum_{n=2}^{M} \frac{1}{n^z} \leq \int_{1}^M dx \frac{1}{x^z}
\end{align}
For $z<1$, we have:
\begin{align}
\int_{1}^M dx \frac{1}{x^z} \sim \frac{M^{1-z}}{1-z}
\label{integral_1D}
\end{align}
Thus, we deduce that for large $M$,
\begin{align}
\Delta_{z,d=1} \sim \frac{2 \delta}{a^z}\sum_{n=2}^{M} \frac{1}{n^z} \sim \frac{2 \delta}{a^z} \frac{M^{1-z}}{1-z}
\end{align}
And replacing $M=N_{\text{qubits}}(k)/2=Q_L D^{k}/2$, we find the formula used in the main text:
\begin{align}
\Delta_{z<d,d=1} \sim \frac{2 \delta}{a^z}\sum_{n=2}^{M} \frac{1}{n^z} \sim \frac{ \delta 2^z}{a^z} \frac{Q_L^{1-z}}{(1-z)} \ D^{k(1-z)}
\end{align}
We can also express everything in term of the variable $\eta$, we find:
\begin{align}
&\eta(k,d=1,z<1) \sim \eta_0(1,z<1) D^{\beta_z k}\\
&\eta_0(1,z<1)=e^{1+1/2e} 2\sqrt{\frac{\delta t_0 Q_L^{1-z}}{a^z(1-z)2^{1-z}}}\\
&\beta=\frac{1-z}{2}
\end{align}
For the case $z=d=1$, the same reasonning can be performed (the integral \eqref{integral_1D} giving rise to a logarithm). And we obtain for large $N(k)=Q_L D^k$:
\begin{align}
&\eta(k,d=1,z=1)=2e^{1+1/2e}\sqrt{\frac{\delta t_0}{a}} \sqrt{\ln\left(\frac{Q_L D^k}{2}\right)}
\end{align}
Now, we can also perform the calculation in two dimensions. The principle is roughly the same but some further care must be taken to be sure to have the appropriate scaling. Now considering that $M=\sqrt{N_{\text{qubits}}}/2$, \eqref{delta2} can be expressed as:
\begin{align}
&\Delta_{z,d=2}(k)=\frac{4 \delta}{a^z}\left( \widetilde{\Delta}_{II,z}+\widetilde{\Delta}_{I,z} \right)\\
&\widetilde{\Delta}_{II,z}=\sum_{n=1}^{M} \sum_{m=1}^M \frac{1}{(m^2+n^2)^{z/2}}\\
&\widetilde{\Delta}_{I,z}=\sum_{m=1}^{M} \frac{1}{m^z},
\end{align}
where we also used the symmetries of the problem, which explains the factor $4$ in front of $\delta$. We also notice that $\widetilde{\Delta}_{I,z}$ has already been determined by the 1D case that we just treated. Thus, we only need to compute $\widetilde{\Delta}_{II,z}$. For this purpose, we can bound $U_{nm} \equiv (m^2+n^2)^{-z/2}$ by two integrals using the fact that $x \to (x^2+n^2)^{-z/2}$ and $y \to (x^2+y^2)^{-z/2}$ are decreasing functions and it gives us:
\begin{align}
\int_{2}^{M+1} dx \int_{2}^{M+1} dy \frac{1}{\sqrt{x^2+y^2}^z} \leq \sum_{n=2}^{M} \sum_{m=2}^M \frac{1}{\sqrt{m^2+n^2}^{z}} \leq \int_{1}^M dx \int_{1}^M dy \frac{1}{\sqrt{x^2+y^2}^z}
\label{encadrement_serie}
\end{align}
Now, if $\widetilde{\Delta}_{II,z}$ diverges, we have $\widetilde{\Delta}_{II,z} \sim \sum_{n=2}^{M} \sum_{m=2}^M \frac{1}{\sqrt{m^2+n^2}^{z}}$. Thus we can focus on this quantity to know the perturbative behavior of $\widetilde{\Delta}_{II,z}$. This perturbative behavior will be deduced by the perturbative behaviors of the integrals on the left and on the right. For this reason we now define $I_{a,z}(M)$:
\begin{align}
I_{a,z}(M) \equiv \int_a^M dx \int_a^M dy \frac{1}{\sqrt{x^2+y^2}^z}
\label{Ia}
\end{align}
Performing a change of variable to polar coordinates, this integral can be re expressed as a function of the variables $(r,\theta)$ and the integration area is represented by the black square on figure \ref{fig:integral_Ia}.
\begin{figure}[h!]
\begin{center}
\includegraphics[width=0.5\textwidth]{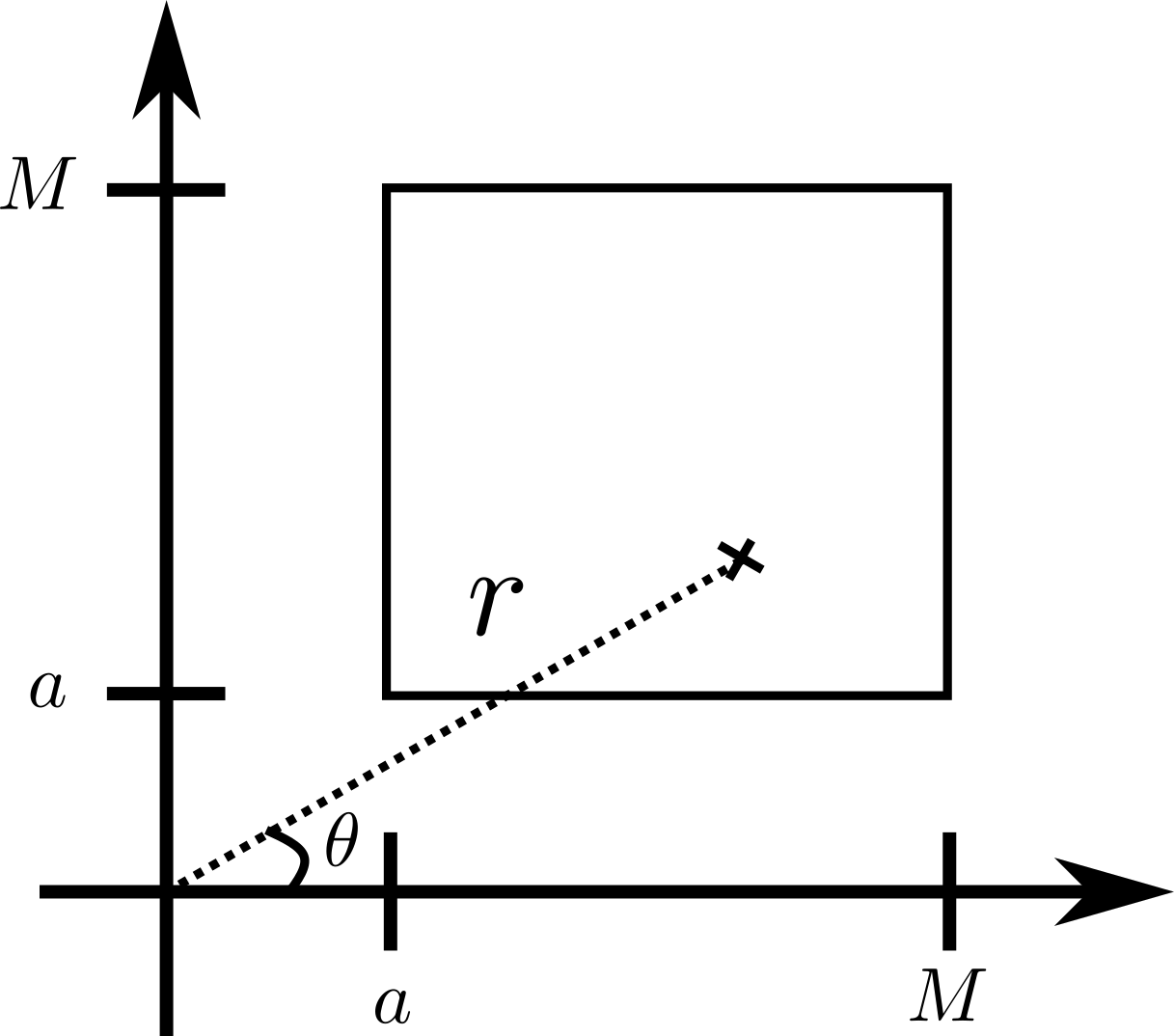}
\caption{Area of integration associated to the integral \eqref{Ia}.}
\label{fig:integral_Ia}
\end{center}
\end{figure}
Using the appropriate boundaries, we have (we only show the results for $z<2$ here, the case $z=2$ gives rise to a logarithmic dependence for which we will only write the result):
\begin{align}
I_{a,z}(M)=&\int_{\arctan(a/M)}^{\pi/4} d \theta \int_{a/\sin(\theta)}^{M/\cos(\theta)} dr \ r^{1-z}+\int_{\pi/4}^{\arctan(M/a)} d \theta \int_{a/\cos(\theta)}^{M/\sin(\theta)} dr \  r^{1-z} \notag \\
&=M^{2-z}\frac{2}{2-z} \int_{\pi/4}^{\arctan(M/a)} d \theta \frac{1}{\sin(\theta)^{2-z}}-2 \frac{a^{2-z}}{2-z} \int_{\arctan(a/M)}^{\pi/4} d \theta \frac{1}{\sin(\theta)^{2-z}}
\end{align}
Both terms in this last line can diverge for large $M$. Indeed the first one scale as $M^{2-z}$ always diverges under our assumption that $z<2$, and the second one involves an improper integral that can diverge for $M \to +\infty$. Actually, we can show that it diverges "slower" than $M^{2-z}$ such that we don't have to take it into account for the asymptotic expressions. 

To show it, we notice that there exists a constant $c$ such that $\sin(\theta) \geq c \theta$ on the range $[0,\pi/4]$, thus, as $2-z>0$, $1/\sin(\theta)^{2-z} \leq 1/(c \theta)^{2-z}$:
\begin{align}
0 \leq \int_{\arctan(a/M)}^{\pi/4} d \theta \frac{1}{\sin(\theta)^{2-z}} \leq \frac{1}{c^{2-z}} \int_{\arctan(a/M)}^{\pi/4} d \theta \frac{1}{\theta^{2-z}} 
\end{align}
If $2-z>1$, the upper bound converges. If $2-z < 1$, we have:
\begin{align}
\frac{1}{c^{2-z}} \int_{\arctan(a/M)}^{\pi/4} d \theta \frac{1}{\theta^{2-z}} \sim \epsilon \arctan(a/M)^{z-1} \sim \epsilon' M^{1-z}
\end{align}
For some constants $\epsilon, \epsilon'$, and if $2-z=1$, we have (for some other constants $\epsilon'', \epsilon'''$):
\begin{align}
\frac{1}{c} \int_{\arctan(a/M)}^{\pi/4} d \theta \frac{1}{\theta} \sim \epsilon'' \ln(\arctan(a/M)) \sim \epsilon''' \ln(M)
\end{align}
Thus for $z \leq 2$, in all cases, we find out that $\int_{\arctan(a/M)}^{\pi/4} d \theta \sin(\theta)^{z-2}=o(M^{2-z})$, and we deduce that:
\begin{align}
&I_a(M) \sim M^{2-z} C_z\\
&C_z=\frac{2}{2-z} \int_{\pi/4}^{\pi/2} d \theta \frac{1}{\sin(\theta)^{2-z}}.
\end{align}. 
Using \eqref{encadrement_serie} it implies $\widetilde{\Delta}_{II,z} \sim  M^{2-z} C_z$. On the other hand, in the 1D case our calculation showed that $\widetilde{\Delta}_{I,z}=o(M^{2-z})$ (because it will either grow as a logarithm or as a law in $M^{1-z}$). All this allow us to conclude in the end that $\Delta_{z,d=2} \sim \frac{4 \delta}{2^{2-z}a^z} C_z M^{2-z}$. 

Finally, replacing $M=\sqrt{Q_L D^{k}}/2$, we deduce:
\begin{align}
&\Delta_{z<d,d=2}(k) \sim \frac{2^z \delta Q_L^{1-z/2}}{a^z} C_z D^{k(1-z/2)}\\
& \eta = e^{1+1/2e}\sqrt{2 \Delta(k) t_0} \sim e^{1+1/2e} \sqrt{2 t_0 \frac{2^z \delta Q_L^{1-z/2}}{a^z} C_z }D^{k(1/2-z/4)}
\end{align}
\end{appendices}

\chapter{The energetic cost of quantum computing: full-stack framework}
\label{chap:Fullstack1}
In this chapter, we are going to design the general method we propose in order to do resource estimation of quantum computing. What we call resource is in principle very general (it could be any cost function), but we directly apply it to the problematics of energetics, more precisely, of power consumption. The central point on which our approach is built consists of relating the quality of the algorithm output to the power that is spent inside of the quantum computer. 

We will see that, by asking to minimize the power consumption of the computer \textit{under the constraint} that the algorithm succeeds, we will have access to the minimum power consumption required to implement successfully the algorithm \textit{as well} as the optimal architecture the computer and the algorithm should have in order to reach this minimum. What we mean is that we will be able to optimize all the tunable parameters associated with the hardware (qubit temperature, attenuation,...) but also the way the algorithm is implemented and the level of error correction needed (this last item will be done in the next chapter). Our goal here is to provide a global and unified framework that includes aspects of engineering, algorithm, quantum gate physics, as well as the idea of optimizing this architecture to minimize power consumption. 

Our method is said to be full-stack in the sense that it allows to include aspects coming from a variety of fields such as engineering, algorithms, and quantum physics, in order to do the energetic estimation. The problem of energetics being intrinsically transverse, such transverse approaches are required to understand it. The concept of full-stack is very recent in the field of quantum computing, and it consists in including in the model the different layers ("stack") required in a quantum computer. Such layers are sometimes identified as corresponding to the quantum algorithm, the compiler (the software translating the algorithm in a form that can be run on the hardware), the hardware qubit technology, the way the qubits are controlled, etc., \cite{gokhale2020full}, but the general philosophy, in the end, is to have a "multi-layer" description of the computer in the model. We can cite the following recent work based on those approaches \cite{murali2019full,rodrigo2020exploring,amy2020staq}. However, the aspects of energetics or cryogenics are usually not considered in those approaches. In order to scale up quantum computers, this is something that is important to consider in the design, especially for superconducting qubits that must be maintained at very low temperatures. This is what we are going to focus on.

On the topic of resource estimation, different works have also been done. In the context of fault-tolerant quantum computing, the resources optimized are almost always the number of physical qubits and gates required by the computation \cite{li2015resource,kim2021fault,di2020fault}. This is usually done by comparing different error-correction schemes to find the one that uses the least amount of resources while having the best efficiency in detecting and correcting errors. Outside of fault tolerance, the idea of relating a success to a resource (often power or energy) in order to minimize the latter is a concept that has been recently explored in various contexts, see \cite{abah2019energetic,ikonen2017energy,campbell2017trade,deffner2021energetic,robert2021resource} and references therein. But overall, the few energetic studies done are usually focused on algorithms implemented without error correction, and where only a very specific component of the energetic cost (typically the energy exactly required by the quantum gates) is taken into account in the final bill. It doesn't include the necessary global vision required (the energy strictly required by the qubits will usually be a very small component of the total energy needed in a quantum computer). More quantitative estimations have also been done \cite{krinner2019engineering,mcdermott2018quantum}, but they focus on the engineering aspects of the quantum computer without really including algorithm considerations, and they usually do not study what happens for fault-tolerance\footnote{When such aspects are considered, they are usually only based on a rough estimation of the number of physical qubits without taking in consideration the physical gates for instance. In the next chapter, we will do an "in-depth" study that includes such components.} where the energetic cost will be a crucial element to take into account in the design. More importantly, those quantitative approaches do not use any knowledge about the noise strength in their model in order to optimize the architecture\footnote{Of course some aspect of noise are present, for instance in \cite{krinner2019engineering} the fact that the superconducting qubits should be isolated from thermal noise is taken into account. But the qubits are forced to be at $10mK$ (which has an important influence on the energetic cost), and the attenuation is forced to be at a typical value.}. In the end, to assess the energetic cost of quantum computing and make it energy efficient, having a global and optimized vision is necessary, as a computer is, by essence, a multidisciplinary object. Our goal in this chapter is to provide such an approach.

To understand the method, we are going to consider two simple toy examples, where no error correction will be performed. In the first one, in section \ref{sec:single_qubit_example} we will find the minimum power required to implement a single-qubit gate. This example will illustrate that there are many ways to reach a given targeted fidelity associated with very different power consumptions. Minimizing the power consumption under the constraint of targeting a given fidelity allows making the gate much more energy efficient. This example is followed by the sections \ref{sec:min_pow_increases_with_accuracy} and \ref{sec:equivalence_minpow_maxacc} where we will give some general properties this minimization under constraint will usually satisfy.

The second example treated in section \ref{sec:energetic_cost_nisq_algo} will allow us to optimize the way an algorithm is implemented for the same goal of minimizing the power. In some sense, we generalize here the approach of the previous chapter as the noise will no longer have to be a function of the resource to minimize, and because we will be able to optimize the quantum computer architecture while performing the minimization. The method we propose can then, in principle, be applied both in quantum algorithms using quantum error-correction (fault-tolerant quantum computing) or not using it (however, adapting it for \text{hybrid} quantum-classical algorithms frequently used in NISQ would require further investigations as explained in the last paragraph of \ref{sec:energetic_cost_nisq_algo}). 

We finish this chapter by explaining how to adapt our framework for fault-tolerance in the section \ref{sec:adapting_framework_FT}, where we will also do all the quantitative estimations about the number of physical qubits and gates required for a large-scale calculation.

This chapter\footnote{In term of contribution, I realized the major part of the work presented in this chapter, excepted the simulations behind the figures \ref{fig:min_power_fct_metric_1qb}, \ref{fig:competition_power_noise} and \ref{fig:min_power_fct_compression} that has been done by Jing Hao Chai (but I designed the examples behind the plots).} is at the end dedicated to explain the general principle and to provide simple toy examples where it can be applied. We will also see what are the differences between energetic estimation for algorithms implemented with or without quantum error correction. The following chapter will consist in using this method in order to make a quantitative and detailed energetic estimation of implementing an algorithm on a fault-tolerant quantum computer based on superconducting qubits.
\section{The energetic cost of quantum computing: general vision}

\subsection{Formulating the question as a minimization under constraint}
The general purpose of an algorithm is to provide an answer to some computational task with a targeted success probability considered as "high enough": there is no point in having a bad answer. The probability of finding a successful answer will be related to how much some resource an experimentalist is ready to invest in the calculation. Indeed, implementing an algorithm successfully (in the sense that it provides an answer having a high probability of success) has a cost. We saw examples in the previous chapter, on the figure \ref{fig:p_kmax_QFT} with the energy that we allowed to spend in order to create the pulses driving the qubits. When we restricted too much this energy, the accuracy of the logical gate was too low, which limited the success probability of an implemented algorithm.

Without loss of generality, we will now assume that the resource we are interested in is the \textit{power} required to implement an algorithm. Because the power and the success probability of the algorithm are related, and because we are interested in using the lowest amount of power in order to make the quantum computer power efficient, the problem we are interested in solving is then 
\begin{align}
P_{\min} \equiv \min_{\bm{\delta}}(P(\bm{\delta}))_{\big | p_{\text{failure}}(\bm{\delta}) \leq p^{\text{target}}_{\text{failure}}},
\label{eq:ideal_solve}
\end{align}
where $P$ is the power spent for the calculation, $p_{\text{failure}}$ is the probability that the algorithm fails, i.e that it provides a wrong answer.
This equation thus consists in finding the \textit{minimum} power required to be sure that the algorithm fails less frequently than a given target. We see that we introduced the family of parameters $\bm{\delta}$ on which both $p_{\text{failure}}$ and $P$ depend (the fact it is a family is represented by the bold notation associated to vectors). They represent the \textit{tunable} parameters that the experimentalist can vary in order to reach the minimum. A typical example would be the qubit temperature $T_{\text{Q}}$: changing it would at the same time modify the power it costs to perform the computation (the electrical power for the cryogenic typically depends on $T_{\text{Q}}$), but the success probability as well (a too high temperature would induce a poor quality in the answer). Thus, in general, both the probability of success and the power depend on the family\footnote{It is possible that one parameter in the family of parameters $\bm{\delta}$ only affect the success and not the power cost, and reciprocally.} $\bm{\delta}$. We also see that not only the minimum power can be found performing this minimization, but the optimum set of parameters allowing to reach this minimum, $\bm{\delta}_{\text{opt}}$ will also be found\footnote{In general, we could expect more than one set of optimal parameters $\bm{\delta}_{\text{opt}}$ minimizing the power consumption. This is not really a problem for us, but it should be noted.}. This approach is thus more general from what we explained in \ref{sec:relation_scale_dep_resource_estimation}: here, the failure probability doesn't have to be a function of the power: both might be indirectly related through the set of parameters $\bm{\delta}$. 

Now, solving this problem is, in general, too complicated: it is usually not possible to know the probability that a given algorithm fails. For this reason, what we will do is to estimate the "quantity of noise" there is in the algorithm output. It can be done with some "metrics"\footnote{Of course, the only thing that really matters experimentally is to have a successful algorithm. So the metric that one chooses must have nice properties allowing to estimate or at least bound the probability that the algorithm fails.} (the worst-case infidelity or the probability of having an unsuccessful algorithm in the context of fault tolerance would be two examples of such metrics). The equation to solve then takes the following expression:
\begin{align}
P_{\min} \equiv \min_{\bm{\delta}}(P(\bm{\delta}))_{\big | \mathcal{M}(\bm{\delta}) \leq \mathcal{M}_{\text{target}}}.
\label{eq:Pmin_general}
\end{align}
Now, for the exact same reason as before, solving this equation would allow us to find the minimum power required for the computer as well as the optimal family of parameters in the architecture, allowing us to reach this minimum. This is the central equation we will use in the last two chapters of this thesis: the problem of resource estimation of quantum computing can be seen as a problem of minimization under constraint.
\subsection{General properties implied by the question}
Here, we will show that despite the fact \eqref{eq:Pmin_general} is a very general problem of minimization under constraint, making some reasonable hypotheses on the behaviors of $\mathcal{M}$ and $P$ allows to find interesting properties.
\subsubsection{The minimum power increases with the targetted accuracy}
\label{sec:min_pow_increases_with_accuracy}
Let's assume that there exists \textit{at least one} parameter $\delta_j$ in the list of parameters  $\bm{\delta}=(\delta_0,\delta_1,...,\delta_n)$ such that, when we increase it, the power decreases while the quantity of noise increase. Mathematically it means that, for any value of $\bm{\delta}$:
\begin{align}
&\partial_{\delta_j} P(\bm{\delta})<0\\
&\partial_{\delta_j} \mathcal{M}(\bm{\delta})>0
\end{align}
We also add the condition that for any values of $\delta_{i \neq j}$, we have:
\begin{align}
\mathcal{M}(\delta_j=\delta_j^{\max}) \geq \mathcal{M}_{\text{target}},
\end{align}
where $\delta_j^{\max}$ is the maximum value $\delta_j$ can reach (it could be equal to $+\infty$).

One example of such a parameter will often be the temperature of the qubits. If we increase the qubit temperature, the power it costs to maintain them cool will typically decrease, but the quantity of noise will usually increase because they will face more thermal noise. Also, typically for qubits at ambient temperature (corresponding to $\delta_j=\delta_j^{\max}$), the noise is likely to be very high and we expect $\mathcal{M}(\delta_j=\delta_j^{\max}) \geq \mathcal{M}_{\text{target}}$. In summary, a parameter behaving as $\delta_j$ will typically exist in relevant physical systems. 

These hypotheses will imply the following two properties. First, (i) the minimum power required such that the quantity of noise of the algorithm is equal \textit{or lower} to $\mathcal{M}_{\text{target}}$ is the same as the minimum power required such that the quantity of noise is \textit{equal} to $\mathcal{M}_{\text{target}}$, i.e., the following equality is true:
\begin{align}
P_{\min} \equiv \min (P(\bm{\delta}))_{\big | \mathcal{M}(\bm{\delta}) \leq \mathcal{M}_{target}} = \min (P(\bm{\delta}))_{\big | \mathcal{M}(\bm{\delta}) = \mathcal{M}_{target}}
\end{align}
Let us do a proof by contradiction of this fact. We assume that the minimum power is found when $\mathcal{M} < \mathcal{M}_{\text{target}}$. As $\partial_{\delta_j} \mathcal{M}>0$, we can increase $\delta_j$ for some range without violating the condition $\mathcal{M} \leq \mathcal{M}_{\text{target}}$. Doing so, we will also decrease the power as $\partial_{\delta_j} P<0$. Thus, as soon as $\mathcal{M}<\mathcal{M}_{\text{target}}$ we do not minimize the power consumption: the minimum is necessarily reached for $\mathcal{M}=\mathcal{M}_{\text{target}}$\footnote{It is possible to exactly reach that point because $\mathcal{M}(\delta_j=\delta_j^{\max}) \geq \mathcal{M}_{\text{target}}$}. In the same line of thoughts, we can also prove that (ii) $P_{\min}$ is a \textit{strictly decreasing} function of $\mathcal{M}_{\text{target}}$. To show this, let us consider two targets such that $\mathcal{M}^{(1)}_{\text{target}}<\mathcal{M}^{(2)}_{\text{target}}$. Then, as:
\begin{align}
\min(P)_{\big | \mathcal{M} \leq \mathcal{M}^{(2)}_{\text{target}}}=\min(P)_{\big | \mathcal{M} = \mathcal{M}^{(2)}_{\text{target}}},
\end{align}
the minimum power obtained for $\mathcal{M}<\mathcal{M}^{(2)}_{\text{target}}$ is necessarily bigger or equal to the one obtained for $\mathcal{M}=\mathcal{M}^{(2)}_{\text{target}}$. Because of that, as $\mathcal{M}^{(1)}_{\text{target}}<\mathcal{M}^{(2)}_{\text{target}}$, we deduce that: $P_{\min}(\mathcal{M}^{(1)}_{\text{target}}) \geq P_{\min}(\mathcal{M}^{(2)}_{\text{target}})$. It remains to show that they cannot be equal. We also do a proof by contradiction. Let us consider a set of parameters $\bm{\delta}$ such that: $P(\bm{\delta})=P_{\min}(\mathcal{M}^{(1)}_{\text{target}})$, and $\mathcal{M}(\bm{\delta})=\mathcal{M}^{(1)}_{\text{target}}$\footnote{We say "a" set of parameters and not "the" set of parameters as the minimum power consumption could be in principle reached for many different values of $\bm{\delta}$.}. If we increase $\delta_j$ until the moment $\mathcal{M}=\mathcal{M}^{(2)}_{\text{target}}$, as $\partial_{\delta_j} P < 0$, it has for effect to \textit{strictly} decrease the power consumption while ensuring $\mathcal{M}=\mathcal{M}^{(2)}_{\text{target}}$. Thus, the minimum power reached for $\mathcal{M}=\mathcal{M}^{(2)}_{\text{target}}$ is necessarily \textit{strictly} lower than the one reached for $\mathcal{M}=\mathcal{M}^{(1)}_{\text{target}}$.

In summary, here we proved the following property that we will frequently use.
\begin{property}{Behavior of the minimum power}

\label{prop:behavior_min_power}
If there exists at least one parameter $\delta_j$ in the family of parameters $\bm{\delta}$ such that, for any value of $\bm{\delta}$
\begin{align}
&\partial_{\delta_j} P(\bm{\delta})<0\\
&\partial_{\delta_j} \mathcal{M}(\bm{\delta})>0
\end{align}
And:
\begin{align}
\mathcal{M}(\delta_j=\delta_j^{\max}) \geq \mathcal{M}_{\text{target}},
\end{align}
where $\delta_j^{\max}$ is the maximum value $\delta_j$ can reach (it could be equal to $+\infty$). Then, we have:
\begin{align}
P_{\min} \equiv \min (P(\bm{\delta}))_{\big | \mathcal{M}(\bm{\delta}) \leq \mathcal{M}_{target}} = \min (P(\bm{\delta}))_{\big | \mathcal{M}(\bm{\delta}) = \mathcal{M}_{target}},
\end{align}
which means that the minimum power consumption is found when we ask for an accuracy that it is exactly equal to what we target: asking for a better accuracy is necessarily more costly. We also have that
$P_{\min}$ is a strictly decreasing function of $\mathcal{M}_{\text{target}}$ (for the same reason).
\end{property}
This property properly establishes the intrinsic connection between noise and power: asking for a better accuracy has an energetic cost. But again, this connection is true because of the hypotheses we made on the behavior of power and metric with respect to the parameter $\delta_j$. We believe that these hypotheses will be true for many physical systems. It is certainly true for the systems studied in the two last chapters of this thesis.
\subsubsection{Equivalence between minimum power for a targetted accuracy and maximum accuracy for a given power}
\label{sec:equivalence_minpow_maxacc}
Now, we can show that asking to minimize the power in order to reach a given accuracy is, under a precise meaning we will clarify, equivalent to asking to find the best accuracy the computer can get to for a given (i.e., fixed) amount of power available. This property is true as soon as $P_{\min}$ is a decreasing function of $\mathcal{M}_{\text{target}}$ (which will usually be the case as explained in property \ref{prop:behavior_min_power}).

Those two questions are equivalent in the sense of what is represented in figure \ref{fig:equivalence_minpow_maxacc}. On this curve is represented the minimum power required to implement the algorithm as a function of the maximum quantity of noise acceptable in its answer: $\mathcal{M}_{\text{target}}$. In this case, this graph can be read in two directions. Indeed, if we invert the meaning of the axis, i.e., we interpret it as a curve representing a value of some quantity of noise as a function of some injected power, we can also interpret it as the \textit{minimum} quantity of noise it is possible to reach as a function of this injected power. More precisely, let us consider any couple $(\mathcal{M}^0_{\text{target}},P^0 \equiv P_{\min}(\mathcal{M}^0_{\text{target}}))$. By definition of $P_{\min}$, if someone wants to have an accuracy $\mathcal{M}^0_{\text{target}}$, the minimum power required will be $P^0 \equiv P_{\min}(\mathcal{M}^0_{\text{target}}))$. But reciprocally, if this person has a given amount of power $P^0$ available, the minimum quantity of noise it can get to (thus maximum accuracy) is $\mathcal{M}^0_{\text{target}}$.
\begin{figure}[h!]
\begin{center}
\includegraphics[width=0.5\textwidth]{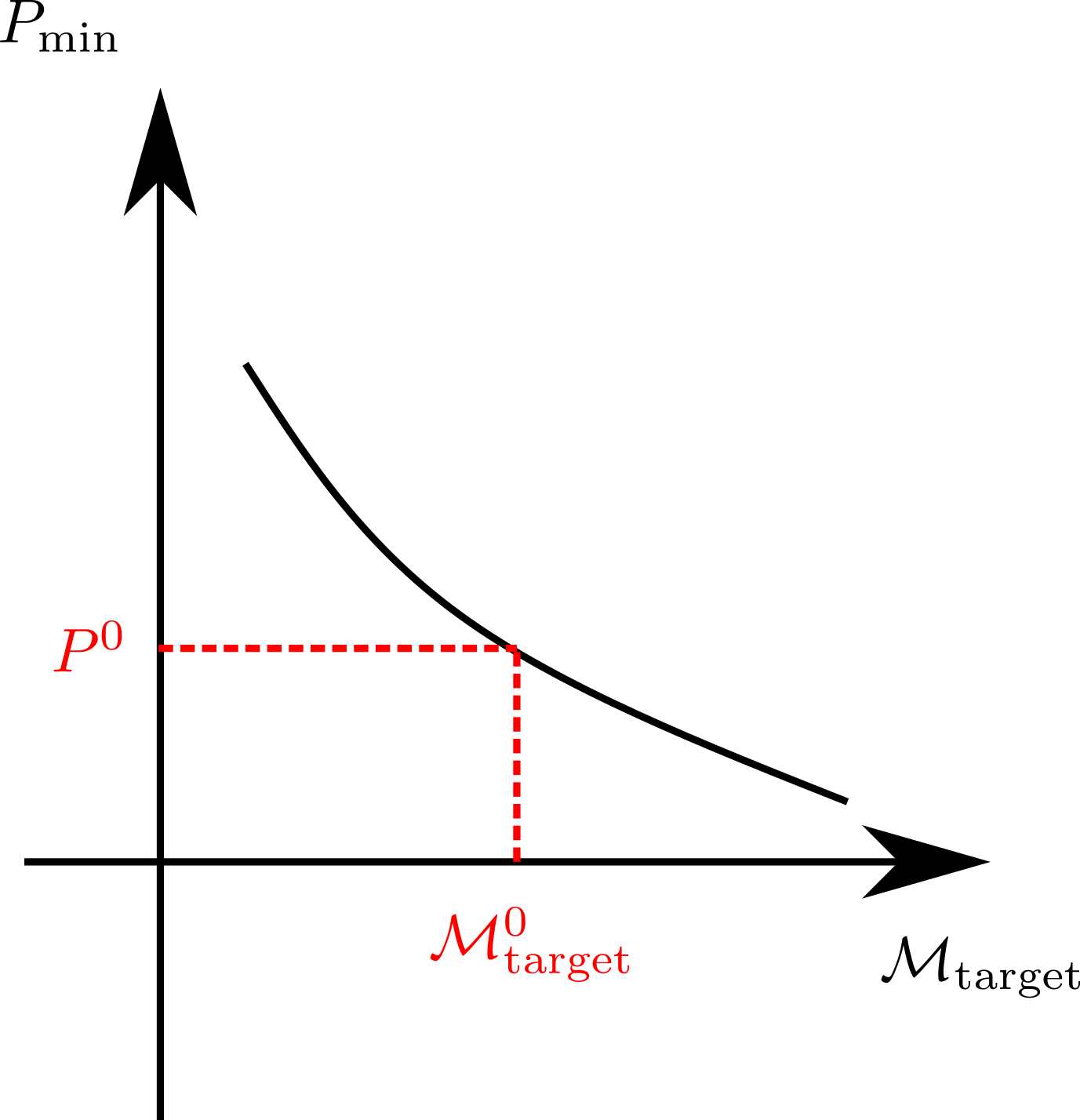}
\caption{The minimum power necessary as a function of the maximum quantity of noise acceptable targetted, $\mathcal{M}_{\text{target}}$ is represented. Because this function is monotonous, the curve can also be interpreted as the lowest quantity of noise reachable in the algorithm output (i.e maximum accuracy possible) for a given amount of power available as explained in the main text.}
\label{fig:equivalence_minpow_maxacc}
\end{center}
\end{figure}
\FloatBarrier

In order to show it, we consider again the couple $(\mathcal{M}^0_{\text{target}},P^0 \equiv P_{\min}(\mathcal{M}^0_{\text{target}}))$. We now ask: what is the maximum accuracy, i.e the lowest quantity of noise $\mathcal{M}$ the computer can get to if it has at its disposal the power $P^0$. Necessarily, it can be \textit{at least} as accurate as $\mathcal{M}^0_{\text{target}}$: $\mathcal{M}$ must satisfy $\mathcal{M} \leq \mathcal{M}^0_{\text{target}}$. Indeed $\mathcal{M}^0_{\text{target}}$ is an explicit example of target reachable for an injected power $P^0$. Now it remains to prove that the strict inequality is not possible: no better accuracy is reachable for this amount of power. This can be shown using again a proof by contradiction: let's assume that we find that the best accuracy it is possible to reach for an injected power $P^0$ satisfies $\mathcal{M} < \mathcal{M}^0_{\text{target}}$. We assume that this power is reached for the family of parameters $\bm{\delta}=\bm{\delta}_0$. We necessarily have $P^0=P(\bm{\delta}_0) \geq P_{\min}(\mathcal{M})$: indeed $P^0$ amount of power would be used to have an accuracy $\mathcal{M}$, but it might be possible to use less power to reach that accuracy. As on the other hand $P^0=P_{\min}(\mathcal{M}^0_{\text{target}})$, we deduce that $P_{\min}(\mathcal{M}^0_{\text{target}})\geq P_{\min}(\mathcal{M})$. Now, this result is absurd. Indeed, as $\mathcal{M}<\mathcal{M}^0_{\text{target}}$, and as $P_{\min}$ \textit{strictly} decreases as a function of $\mathcal{M}$, necessarily $P_{\min}(\mathcal{M}^0_{\text{target}}) < P_{\min}(\mathcal{M})$. 

All this shows that the curve can be read in two directions: \textit{there is an equivalence between asking to minimize the noise for a given amount of power available and asking to minimize the power in order to target a given quantity of noise.}
\subsection{An introductive example: energetic cost of a single-qubit gate}
\subsubsection{Description of the problem}
\label{sec:single_qubit_example}
Let us illustrate the concepts with a very simple example: we wish to minimize the power to spend in order to implement a $\pi$-pulse on a qubit. For this, we need to specify: (i) the noise model describing the physics, (ii) the way we quantify this noise, i.e., which metric we are going to choose, (iii) the expression of the power we want to minimize and (iv) the list of parameters we would like to optimize in order to minimize this power.

For (i) and (ii), we will consider that the only reason why the qubits are noisy is because of spontaneous emission, and thermal photons in the line driving the qubits. The characteristics of this single-qubit gate we will consider are provided in the section \ref{sec:state_of_the_art_1qb}, and the characteristics of the qubits are in \ref{sec:state_of_the_art_qubit}. The metric we will use will either be the worst case or average infidelity. We recall their expression, from the section \ref{sec:infidelity_operations}
\begin{align}
\mathcal{M}=(X+Y \overline{n}_{\text{tot}}) \gamma_{\text{sp}} \tau,
\end{align}
where $(X_{\text{worst}},Y_{\text{worst}})=(1,1)$ and $(X_{\text{avg}},Y_{\text{avg}})=(1/3,2/3)$ for worst-case infidelity and average one respectively. The metrics are then defined. But we cannot be satisfied of those expressions because we don't know what $\overline{n}_{\text{tot}}$, which represents the number of thermal photons in the line\footnote{It can actually represent any kind of noisy photons, but we are focusing on the thermal noise in this PhD.}, is equal to.

The value of $\overline{n}_{\text{tot}}$ will depend on how the signals that are going to drive the qubit to make the gates will be generated. In the case the qubit is at the temperature $T_{\text{Q}}$, and \textit{if the signals are generated at this same temperature}, we would simply have: $\overline{n}_{\text{tot}}=\overline{n}_{\text{BE}}(T_{\text{Q}}) \equiv 1/(e^{\beta_{\text{Q}} \hbar \omega_0}-1)$, with $\beta_{\text{Q}}=1/(k_b T_{\text{Q}})$, and $\omega_0/2 \pi$ is the qubit frequency. $\overline{n}_{\text{BE}}$ is the Bose Einstein population for bosons of frequency $\omega_0/2 \pi$, at temperature $T_{\text{Q}}$ \cite{cohen1998mecanique}. But if the signals are generated at a temperature $T_{\text{Gen}}$, then the noisy thermal photons at this temperature would affect the qubit evolution. This is why attenuators are typically put on the driving lines. Their role is to remove the noise (and in particular the thermal noise) coming from the higher temperature stages. An attenuator is composed of resistive elements that will dissipate into heat any signal injected at its input as shown on figure \ref{fig:attenuator_principle}. If $P_{\text{in}}$ is the input power before the attenuator, and $P_{\text{out}}$ the power at its output, we define the attenuation ratio by $A \equiv P_{\text{in}}/P_{\text{out}}$. It is possible to show \cite{pozar2011microwave} that for signals generated at $T_{\text{Gen}}$, an attenuator put at the temperature $T_{\text{Q}}$ will generate an amount of noisy photons $\overline{n}_{\text{tot}}$ being equal to:
\begin{align}
\overline{n}_{\text{tot}}(T_{\text{Q}},T_{\text{Gen}},A)=\frac{A-1}{A} \overline{n}_{\text{BE}}(T_{\text{Q}})+\frac{\overline{n}_{\text{BE}}(T_{\text{Gen}})}{A} \stackrel{A \gg 1}{\approx} \overline{n}_{\text{BE}}(T_{\text{Q}})+\frac{\overline{n}_{\text{BE}}(T_{\text{Gen}})}{A}
\end{align}
Let us first comment on the formula in the $A \gg 1$ regime. It means that the number of noisy photons felt by the qubits will be equal to the number of noisy photons at the attenuator temperature, on which we add the noisy photons coming from the higher temperature stage but attenuated by a factor $A$. The noise is basically attenuated in the same way as the signals sent. For infinite attenuation, the qubits will be entirely isolated from any thermal noise coming from higher temperature stages. However, this cannot be the exact expression. For instance for $T_{\text{Q}}=T_{\text{Gen}}$ (if everything is thermalized at $T_{\text{Q}}$), we would not find $\overline{n}_{\text{tot}}=\overline{n}_{\text{BE}}(T_{\text{Q}})$ with the approximated expression. This is the reason why the exact expression is actually the one of the left: it correctly describes the situation when $T_{\text{Q}}=T_{\text{Gen}}$ and $A=1$ (no attenuation: everything behave as if the qubits where thermalized at $T_{\text{Gen}}$), such formula is what is used to describe the propagation of noise in microwave circuits \cite{pozar2011microwave}. Physically we can understand the situation as follows: the signal generation stage can be thought of as some cavity which is thermalized at $T_{\text{Gen}}$. This cavity thus emits blackbody radiation in a one-dimensional waveguide: it is the thermal noise. Then this noise propagates in a ballistic manner in the waveguide (as a waveguide will by construction be a good conductor for microwave frequencies, the noise won't be attenuated there). This noise is then dissipated inside of the attenuator and "converted" into phonons: this is a kind of heat that can be evacuated by cryogenics. But the attenuator, being thermalized at $T_{\text{Q}}$, also generates thermal noise that will interact with the qubits. It corresponds to $\overline{n}_{\text{BE}}(T_{\text{Q}})$.
\begin{figure}[h!]
\begin{center}
\includegraphics[width=0.8\textwidth]{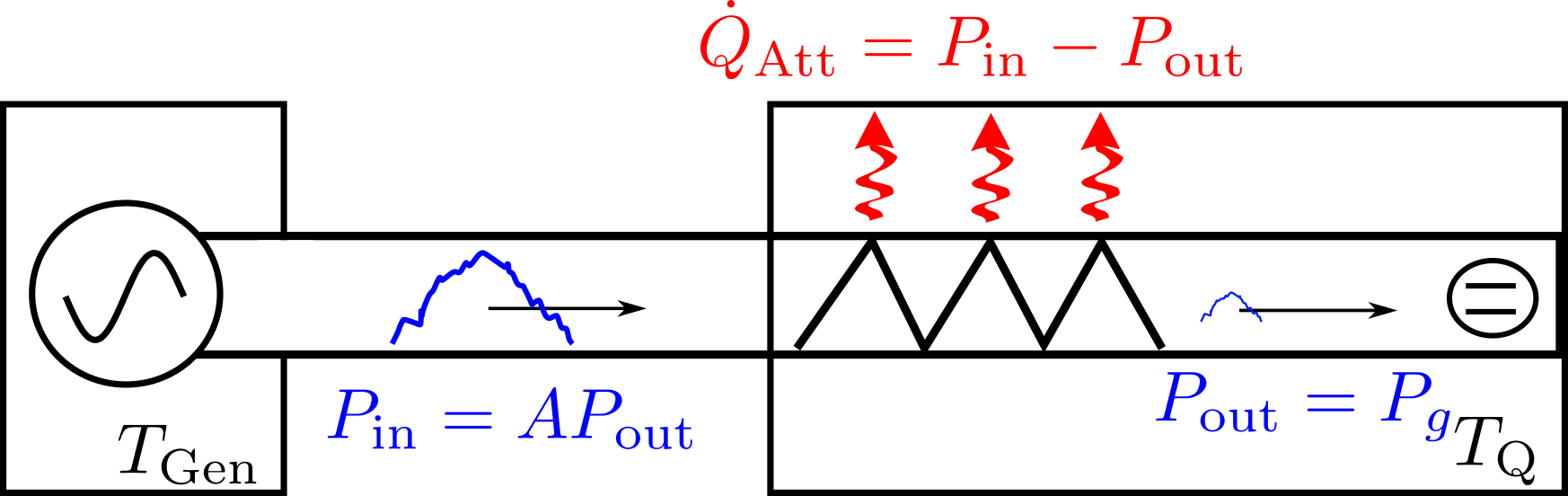}
\caption{Principle behind the attenuation. The signals are generated at a temperature $T_{\text{Gen}}$. Thus, in addition to the signals there is noise at this temperature that is propagating inside of the waveguide as represented by the noisy blue gaussian. The qubit is preceeded by an attenuator at the temperature $T_{\text{Q}}$. The attenuator will dissipate the noise before it reaches the qubit, but it dissipates the signal in the same manner. It is the fact that the heat $\dot{Q}_{\text{Att}}$ has to be evacuated by a cryogenic unit that makes the gate energetically costly (as shown by \eqref{eq:power_single_qubit_gate}).}
\label{fig:attenuator_principle}
\end{center}
\end{figure}
\FloatBarrier
Now, in practice, the effect of the thermal noise coming from the higher temperature stages will be reduced by the following strategy. First, the signal that must arrive on the qubit has a well defined value: we must have $P_{\text{out}}=P_g$, where we recall from \ref{sec:energetic_cost_to_perform_a_gate}, $P_g=\hbar \omega_0 \Omega^2/(4 \gamma_{\text{sp}})$ with $\Omega=\pi/\tau_{\text{1qb}}$ ($\tau_{\text{1qb}}$ is the single-qubit gate duration, and we want to implement a $\pi$-pulse). Then, the attenuation $A$ is \textit{usually} chosen in such a way that $\overline{n}_{\text{BE}}(T_{\text{Gen}})/A < \overline{n}_{\text{BE}}(T_{\text{Q}})$: the noise coming from the higher temperature stage does not dominate the physics. From the knowledge of $A$, the experimentalist deduces what power the generated signals should have. The general philosophy is to generate signals of high enough amplitude such that after attenuation there is a high signal over noise ratio while having a signal of appropriate power for the physics of the quantum gate. At this point, we can say that the metrics we expressed are now "well defined" because we know what $\overline{n}_{\text{tot}}$ is equal to as a function of the characteristics of the problem.

Now, we need to give the expression of the power (point (iii)). Here, we will be interested in the electrical power that we need in order to remove the heat dissipated inside of the attenuators. We will assume that our cryogenic unit has a Carnot efficiency\footnote{We will make further comments about the cryogenic efficiencies of realistic cryostat in the section \ref{sec:efficiency_is_carnot} of the next chapter.}. The heat dissipated per unit time is simply equal to $P_{\text{in}}-P_{\text{out}} \approx (A-1) P_g \approx A P_g$\footnote{We implicitly neglect the heat that will be dissipated when the signals will "go back", i.e after reflection on the qubit. Indeed, the power that would be dissipated by this reflection will be lower than $(A-1) P_g$, and we will neglect it. Furthermore, because $A \gg 1$ will typically be true, we can consider $(A-1)P_g \approx A P_g$. We also neglect the heat dissipated in the attenuators by the thermal noise, which is negligible compared to $A P_g$.}, the low-temperature stage on which the heat has to be evacuated has a temperature $T_{\text{Q}}$, and the high-temperature stage is basically the laboratory at temperature $300K$. Assuming that $A \gg 1$ (if it is not the case, the evolution would be way too noisy), we get:
\begin{align}
P(T_{\text{Q}},A)=\frac{300-T_{\text{Q}}}{T_{\text{Q}}} A P_g
\label{eq:power_single_qubit_gate}
\end{align}
Finally, for our last point (iv), we specify the variables we want to optimize. It will be $\bm{\delta}=(T_{\text{Q}},A)$ in this first example as we will assume that $T_{\text{Gen}}=300K$ (the signals are generated outside of the cryogenics here). We are now ready to solve the minimization under constraint:
\begin{align}
P_{\min} \equiv \min_{(T_{\text{Q}},A)} P(T_{\text{Q}},A)_{\big | \mathcal{M}(T_{\text{Q}},A) = \mathcal{M}_{\text{target}}}
\end{align}
\subsubsection{The competition between noise and power}
We first illustrate in figure \ref{fig:competition_power_noise} in a concrete manner the intrinsic competition there is between asking for low noise and low power consumption. In the rest of this thesis, we will also use the decibel: "dB" unit to represent the attenuation. This unit is related to the attenuation in "natural unit", i.e when $A=P_{\text{in}}/P_{\text{out}}$ as:
\begin{align}
A_{\text{dB}}=10 \log_{10}(A)
\end{align}
For instance, a signal being attenuated of a factor $10$ will correspond to $10 \text{dB}$. A signal attenuated of a factor $10^5$ will correspond to $50 \text{dB}$. 
\begin{figure}[h!]
\begin{center}
\includegraphics[width=0.8\textwidth]{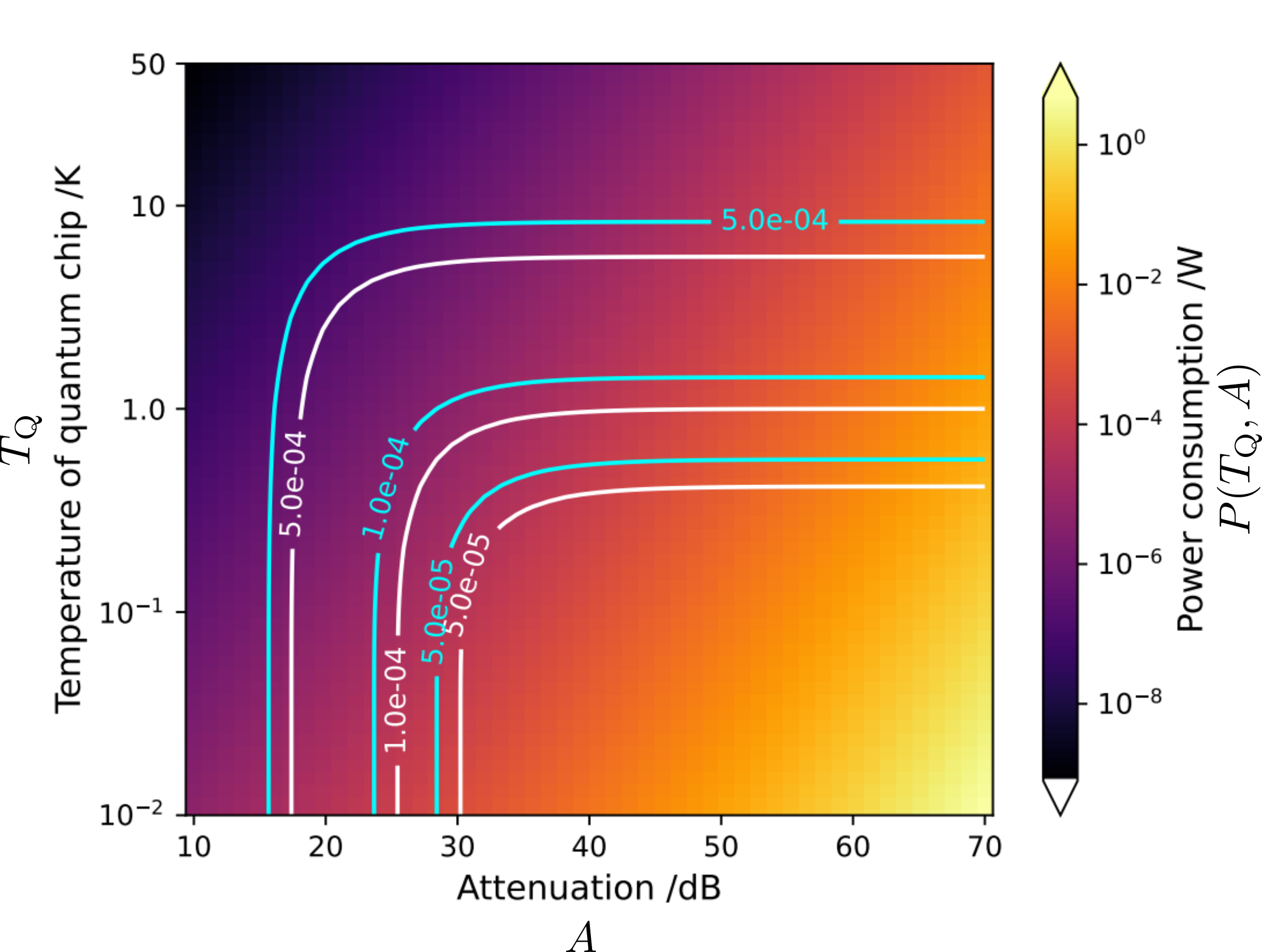}
\caption{Unoptimized power consumption as a function of the attenuation and the qubit temperature (color-plot). The cyan and white contour lines respectively represent the average and worst case infidelity.}
\label{fig:competition_power_noise}
\end{center}
\end{figure}
\FloatBarrier
On this graph, we see that putting a high attenuation and a low qubit temperature induces a high power consumption. Reciprocally a high qubit temperature and a low attenuation lead to small power consumption. We can also see that the regions of high power consumption correspond to regions of low noise, and regions of lower power consumption correspond to regions where the noise is high. Indeed the contour lines have lower values in the bottom right part of this graph which corresponds to a region of high power consumption (and higher values in the upper left part of this graph where the power is low). It illustrates the intrinsic opposite behavior between noise and power: in order to reach a high accuracy, one has to pay the bill!

Let us now focus on a given contour line: we wish to implement the single-qubit gate in order to reach a targetted accuracy. The power consumption on those lines can drastically vary. For instance, the white contour line (worst-case infidelity), which is at the bottom right of this graph, is associated with power varying in the range (approximately) $[10^{-3},10^{-1}]$. It illustrates that choosing wisely $T_{\text{Q}}$ and $A$ is a key asset to potentially make important energetic saves: there are many possible ways to reach a given accuracy, and each of those choices can be associated to a drastically different power consumption. Of course, this example is very simple as we only modeled the heat dissipated in the attenuators, but we see that it can potentially lead to quantitative saves in power consumption, and we can naturally expect that it would remain true for more complicated models\footnote{In particular because more tunable parameters will be involved, the optimum choice of parameters allowing to minimize the power consumption might be not trivial for more complicated models.}.
\subsubsection{Minimizing the power consumption}
Now, we can minimize the power consumption in order to reach a given target. This is shown on figure \ref{fig:min_power_fct_metric_1qb}.
\begin{figure}[h!]
\begin{center}
\includegraphics[width=0.8\textwidth]{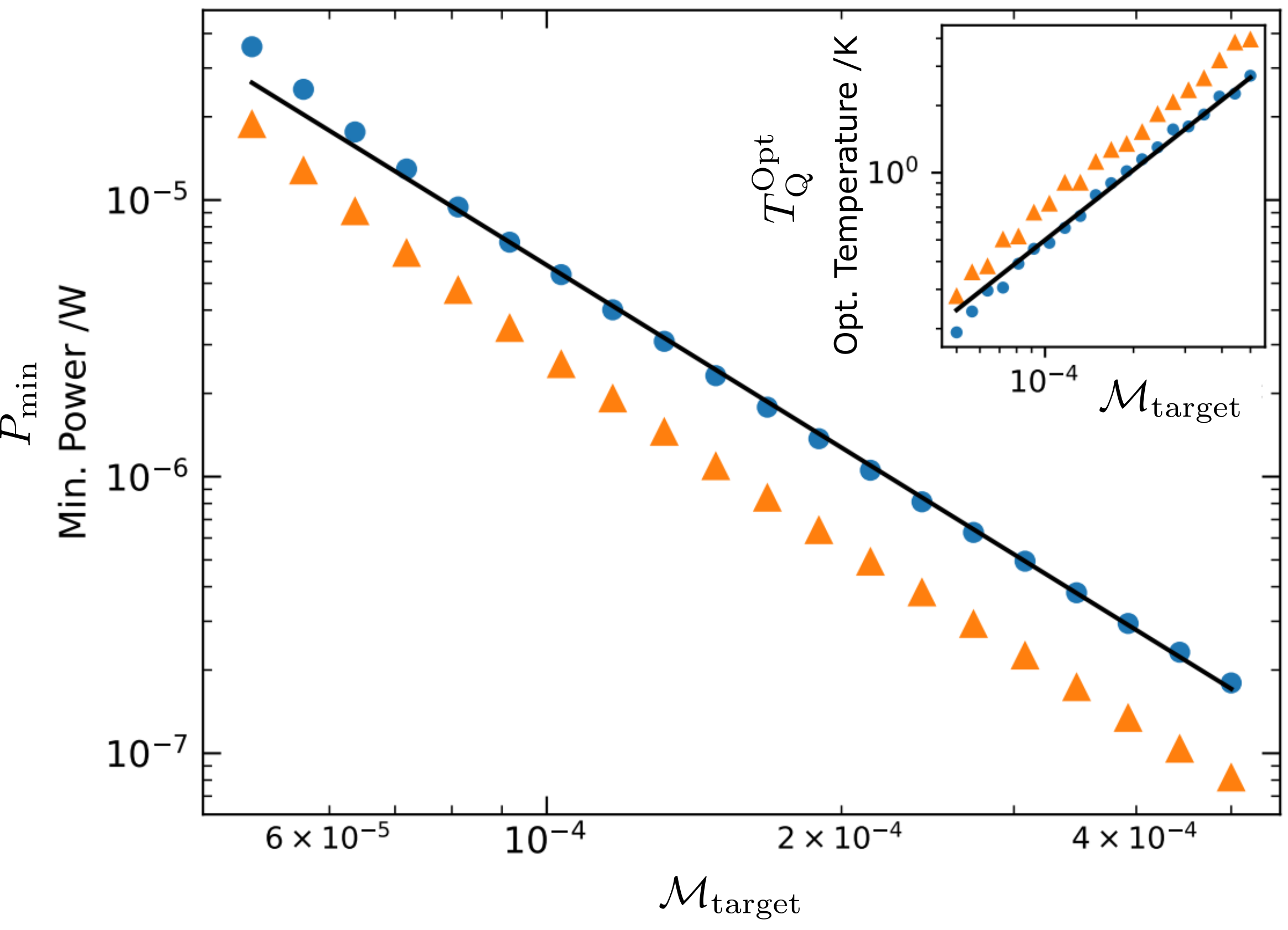}
\caption{Minimum power consumption as a function of the metric target. Inset: optimum qubit temperature $T_{\text{Q}}^{\text{Opt}}$ allowing to reach this minimum. Blue points: worst-case fidelity, orange: average. The black line is a guide for the eye. As explained in \ref{sec:equivalence_minpow_maxacc}, this curve can also be interpreted as the minimum metric target it is possible to reach as a function of the injected power.}
\label{fig:min_power_fct_metric_1qb}
\end{center}
\end{figure}
\FloatBarrier
We see on that graph that the minimum power is a strictly decreasing function of $\mathcal{M}_{\text{target}}$ (whatever the metric we choose: worst-case or average infidelity). This is an illustration of the property \ref{prop:behavior_min_power}. Indeed in the case we study, when $T_{\text{Q}}$ increases, the power decreases while the metric increases, satisfying this way the hypotheses behind the property\footnote{As a side remark, in this case, we would have the same kind of behavior for $1/A$: increasing it (thus reducing the attenuation) lowers the power but increases the amount of noise}. We also see from the inset \textit{how} the power consumption is being decreased for a higher target: the temperature of the qubits increases which increases the (Carnot) efficiency at which the heat is being removed. We did not represent it here but we would also see that the optimum attenuation would decrease as a function of $\mathcal{M}_{\text{target}}$. As explained in \ref{sec:equivalence_minpow_maxacc}, we can also switch the interpretation of this curve and realize that it also represents the maximum accuracy (i.e., the minimum value of $\mathcal{M}_{\text{target}}$) it is possible to reach as a function of some injected power. Finally, this graph also illustrates that the minimum power required depends on the metric that is used. The worst-case fidelity as being "more demanding" will ask for a bigger power consumption to reach a given target in contrary to the average fidelity. To make a connection with the previous section, the minimum power obtained for $\mathcal{M}_{\text{target}}=10^{-4}$ for the average infidelity (it is about $2.10^{-6}$) corresponds to the minimum power we would find along the upper left cyan contour line metric line in figure \ref{fig:min_power_fct_metric_1qb}.
\section{Establishing a full-stack framework: hardware, software, noise, resource approach}
\label{sec:establishing_fullstack_framework_physical}
At this point, we have shown how parameters describing the "hardware" part of the computer can, in principle, be optimized given the competition between noise and power ($A$ and $T_{\text{Q}}$ were the variables playing this role). But our example was about a single-qubit gate which is not very interesting in practice: the goal of a computer is to run algorithms composed of many gates. Furthermore, a computer is composed of much more elements than the one we described. This is why we would like to have a full-stack model for our quantum computer. What we mean by full-stack is that we would like in our approach of power minimization to include all the elements coming from the engineering part of the quantum computer (the heat conduction in the cables, the energetic cost of signal generation, the cryogenic aspects,...), from the algorithm (the shape of the implemented algorithm has a role on the power consumption), and of course, the elements coming from the quantum hardware (the physics of the quantum gates, the physics of quantum noise, etc.). Our minimization under constraint will include the physical architecture of the computer but also of the shape of the implemented algorithm through the family of parameters $\bm{\delta}$. Because of that, we will optimize the quantum computer, from software to hardware. This is what we mean in the title of this section by hardware, software, noise (the constraint consists in targeting a given value for the metric quantifying the noise strength), resource (here the resource is the power) approach. We will only consider physical qubits and gates in this section (thus, the framework is described at this point for an algorithm implemented without quantum error correction). We are going to adapt it at the logical level in the section \ref{sec:adapting_framework_FT}.
\subsection{Establishing the framework at the physical level} 
\subsubsection{General model behind the resource (power) and the quantity of noise}
\label{sec:formulation_problem_physical}
Our goal here is to find a generic expression for the power function $P(\bm{\delta})$ and for the metric $\mathcal{M}(\bm{\delta})$ quantifying the noise. The first thing to acknowledge is that there might have a constant energetic cost regardless of what the computer is doing, i.e., if it is running an algorithm or not. We call it the \textit{static} energetic cost. One typical example of that is heat conduction: all the cables necessary to control the qubits will always be inside of the computer, and they will bring heat that will have to be evacuated\footnote{We will see other examples of static consumption outside from heat conduction later on.} which costs electrical power for the cryogenic. Here, we will assume that the static consumption scales proportionally with the number of (physical) qubits $Q_P$ inside the computer\footnote{It might not always be the case. For instance, in an architecture in which the qubits are implemented on a 2D grid, such as \cite{li2018crossbar} (single-qubit gates can be controlled by sending a signal in the appropriate (line, column) coordinate), the power would also contain a term proportional to $\sqrt{Q_P}$.}. Such scaling comes from the fact that the number of cables and amplifiers (amplifier cannot usually be shut down when not used; thus, they are within the static consumption) is usually increasing proportionally to the \textit{total} number of qubits inside the computer. But in addition to static consumption, there are also the \textit{dynamic} costs. They will increase the power consumption in a manner that is time-dependent. The best example of that is the power associated with signals that are being sent to the qubits: they only have to be sent if an algorithm is running, and this power might also vary in time within an algorithm\footnote{One concrete example is the power required to remove the heat dissipated into the attenuators: this heat will have to be evacuated only if a gate is acting on the qubits.}. We have:
\begin{align}
P=Q_P a + \sum_{i \in \mathcal{G}} N^{(i)}_{P,\parallel} b^i,
\label{eq:physical_power_gen}
\end{align}
where $\mathcal{G}$ describes the primitive gateset considered for the gates (it corresponds to the physical gates on which the algorithm is decomposed). The coefficient $b^i$ represents the power it costs to implement \textit{one} physical gate of type $i$ (the type of gate corresponds to the precise gate that is being implement among the gateset), the term $N^{(i)}_{P,\parallel}$ is the number of gate of this type that are activated at a given instant in the algorithm.
The coefficient $a$ in \eqref{eq:physical_power_gen} represents the power it costs \textit{per qubit} to run the computer. From now on, we will assume that we are only interested in the \textit{average} power consumption such that $N^{(i)}_{P,\parallel}$ should be understood as an \textit{average} number of gates activated in parallel (in principle this term might be time-dependent).

We gave the expression of the power function. Now we must give the expression of the quantity of noise introduced by the algorithm, i.e., the expression of the metric. In all that follows, we assume that the noise level is low such that we can reason perturbatively. In order to be suitable for theoretical calculations, one requirement that we can ask for the metric is that from its value on each individual gate, we can bound it for the entire algorithm (this requirement implicitly assumes that the effect of the noise is local: it is possible to define a quantum channel for each quantum gate, and from that to find how much noise a given gate introduces). A typical property allowing it is to ask that:
\begin{align}
\mathcal{M}_P^{\text{exact}} \leq \mathcal{M}_P^{\text{bound}} \equiv \sum_{i \in \mathcal{G}} N_P^{(i)}\mathcal{M}^{(i)}_{P},
\label{eq:sub_additivity}
\end{align}
where $\mathcal{M}_P^{\text{exact}}$ is the \textit{exact} quantity of noise introduced by the algorithm, $\mathcal{M}^{(i)}_{P}$ is an estimation of the quantity of noise introduced by the gate $i$ and $N_P^{(i)}$ is the number of $i$'th gate in the algorithm. To clarify ideas, if we consider the metric being the worst case infidelity, $\mathcal{M}_P^{\text{exact}}$ would be the infidelity of the full algorithm and $\mathcal{M}_P^{(i)}$ the worst case infidelity of the $i$'th quantum gate in the algorithm. The worst-case infidelity of an algorithm is always lower or equal to the sum of the infidelity of each quantum gate \cite{gilchrist2005distance} (we are reasonning perturbatively: $\mathcal{M}^{(i)}_{P} \ll 1$). The way we will quantify the noise occurring in the algorithm will be based on the best estimation we can have of $\mathcal{M}_P^{\text{exact}}$: $\mathcal{M}_P^{\text{bound}}$. The property \eqref{eq:sub_additivity} is satisfied by many metrics (in addition to the worst-case fidelity, many different operator norms satisfy such property \cite{gilchrist2005distance}, as well as the logical error probability\footnote{For the logical error probability, what we mean is that the probability for an algorithm to fail is estimated with the sum of the probability that any logical gate fails, we recall \eqref{eq:probability_error_algo_k_concatenations} of the first chapter.}). Then, our metric will be defined as $\mathcal{M}_P=\mathcal{M}^{\text{bound}}_P$ and will satisfy:
\begin{align}
\mathcal{M}_P=\sum_{i \in \mathcal{G}} N_P^{(i)} \mathcal{M}^{(i)}_{P}.
\end{align}
Finally, we assume that we can define a clock speed for the algorithm. It means that the algorithm will be implemented in a finite sequence of primitive timesteps all of duration $\tau_P^{\text{timestep}}$. In all of those steps, a \textit{unique} primitive gate will be implemented per qubit (or per two qubits if the gate is a two-qubit one). The assumption that there is a clock speed is not equivalent to saying that the duration to implement all the gates is the same. It just means that if a primitive gate happens to be faster than the timestep, the qubit will then follow an identity evolution until the end of the timestep: $\tau_P^{\text{timestep}}$ thus corresponds to the longest gate in the primitive gateset. It allows us to define the depth $D_P$ of the algorithm as the number of timesteps between the preparation of the qubits and their measurements. In this thesis, we took the convention of not including the final measurements inside of the depth. This is because we will neglect the final measurements in order to simplify our discussions: as they are acting "at the boundary" of the algorithm, they won't contribute in a significant manner to the power consumption.\footnote{The measurements will only impact the power "over the surface" of the algorithm while all the rest of the gates will impact it "in the volume": the measurement cost can be neglected as soon as the depth is "not too small", which will be the case in our calculations.} It doesn't mean that we will not take into account the cost of measurement amplification: this is usually a \textit{static} consumption for which our argument doesn't apply that we are going to model (but this will come in the next chapter)\footnote{In fault-tolerance we will see in section \ref{sec:estimation_physical_qubits} that while the measurement of the logical qubits will contribute in a negligible manner to the number of gates acting in parallel, it will not be the case of the measurement of the ancillae qubits.}. Using the depth, we have, by definition of $N_{P,\parallel}^{(i)}$, $N_P^{(i)}=D_P N^{(i)}_{P,\parallel}$, and we can then write:
\begin{align}
\mathcal{M}_P=D_P \sum_{i \in \mathcal{G}} N_{P,\parallel}^{(i)}  \mathcal{M}^{(i)}_{P}
\label{eq:physical_metric_gen}
\end{align}
\subsubsection{The separate roles of hardware, software and noise in the resource estimation}
\label{sec:separate_role_hardware_software_noise}
At this point the general models for both power and metric are defined. We can summarize the optimization procedure for the full-stack model as what is represented on the figure \ref{fig:general_optimization_principle}.
\begin{figure}[h!]
\begin{center}
\includegraphics[width=0.9\textwidth]{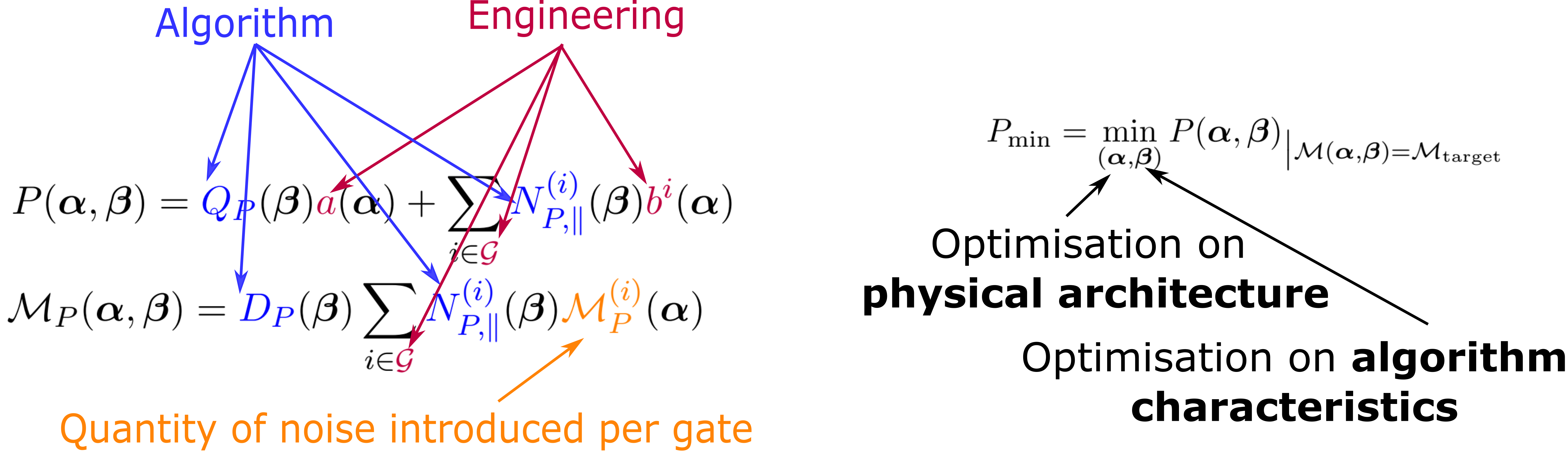}
\caption{Overall principle of minimization of the power cost. The role of the different components participating to the energetic cost: the algorithm, the engineering part of the computer, the noise model (and the way it is quantified) are represented in blue, purple, orange colors respectively. They depend on two families of parameters $\bm{\alpha}$ and $\bm{\beta}$ such that $\bm{\delta}$ in \eqref{eq:Pmin_general} satisfies $\bm{\delta}=(\bm{\alpha},\bm{\beta})$. The family $\bm{\alpha}$ represents the tunable parameters in the physical architecture of the computer. It usually represents the tunable hardware parameters (qubit temperature, attenuation etc) but not only as the clock-speed could also be optimized and it would enter in this family. The family $\bm{\beta}$ represents the tunable parameters corresponding to the manner the algorithm is implemented. We will see an example of that in the sections \ref{sec:energetic_cost_nisq_algo} and \ref{sec:charac_algo}.}
\label{fig:general_optimization_principle}
\end{center}
\end{figure}
\FloatBarrier

We see there that the algorithm, the quantity of noise, and the engineering part of the model have clearly identified roles. The algorithm characteristics are described by the functions $Q_P(\bm{\beta}), N_{P,\parallel}^{(i)}(\bm{\beta}), D_P(\bm{\beta})$ and they will play a role both in the power and the metric. Indeed if we increase the number of qubits or the number of gates that consume power\footnote{Not all gates necessarily consume power, identity gates, for instance, \textit{usually} do not require any power consumption.}, it will increase the power consumption of the computer. But it also plays a role in the metric because longer algorithms or algorithms involving a greater amount of gates in parallel will be noisier. Then, the power consumption per qubit or per gate are represented by the functions $a(\bm{\alpha})$ and $b^j(\bm{\alpha})$. Those functions will, for instance, represent the power efficiency of the cryogenic to remove the heat, how much heat the amplifiers are dissipating, what is the power of the pulse driving the quantum gates, etc. We can say that they represent the characteristics of the engineering aspect behind the computer and will thus only be involved in the expression of the power function. The quantity of noise introduced by the quantum gates will, on the other hand, only play a role in the metric as represented by the orange color. Here, we see that each expertise required to build a quantum computer have very well defined parameters to characterize; someone "at the end" has to regroup them and solve the minimization under constraint, which will provide all the experts the optimum tunable parameters they should choose in their respective field of expertise allowing to minimize the power consumption. What is interesting here is that even though some components in the computer might not seem to be related at all to "the quantum world", they will actually indirectly be related to it through the minimization under constraint. For instance, if we imagine that the cables used in the computer have a very high thermal conductivity which introduces a lot of heat in the computer, because we minimize the power consumption, they might "force" the qubit temperature, which has a direct influence on the noise, to be higher\footnote{They would do it in order to reduce the impact of heat conduction: the lower the temperature of the qubit is and the more costly it is to evacuate the heat conduction on the qubit stage (we can think about Carnot efficiency to understand that for instance).}. Doing so, they will have an influence on the way the noise is acting on the qubits (but as $\mathcal{M}=\mathcal{M}_{\text{target}}$ is imposed, the net quantity of noise is fixed, and does not depend on those cables\footnote{For instance, if $T_{\text{Q}}$ is increased, a bigger attenuation on the lines might be added to compensate this increase in $T_{\text{Q}}$ allowing to maintain the condition $\mathcal{M}=\mathcal{M}_{\text{target}}$.}).

\subsection{Example: application to the estimation of the energetic cost of an algorithm}
\label{sec:energetic_cost_nisq_algo}
Now, we can design an example in which we will optimize the shape of an algorithm that is running, in addition to the family of tunable parameters $\bm{\alpha}$ that play a role for the metric and the "engineering" part of the model. 

We will stay with the same example of power consumption we considered in the single-qubit gate example of section \ref{sec:single_qubit_example}: we will only take into account the power required to remove the heat dissipated by the driving signals in the attenuators. It means that we will neglect any \textit{static} consumption involved in the computation. 

The algorithm we are interested in implementing will only be composed of two-qubit cNOT gates and identity gates. The cNOT will be modeled as explained in section \ref{sec:state_of_the_art_2qb}. The metric we will consider using will be the worst-case infidelity (the average infidelity wouldn't satisfy the property \eqref{eq:sub_additivity} \cite{carignan2019bounding}). The identity gates will be applied at moments where some qubits are waiting on a physical timestep. The duration of a timestep corresponds to the duration of the longest physical gate used in the algorithm (thus, here $\tau^{\text{timestep}}_{P}=\tau_{\text{cNOT}}=100ns$). The identity gates will thus have a worst-case infidelity corresponding to \eqref{eq:worst_case_IF} applied for $\tau^{\text{timestep}}_P$. We assume that no power is required to apply an identity gate. The characteristics of the qubits we will use will follow the description provided in \ref{sec:state_of_the_art_qubit}.

In in this simple example we will thus have $P=Q_P a +N^{(\text{cNOT})}_{P,\parallel} b^{\text{cNOT}}+N^{(\text{Id})}_{P,\parallel} b^{\text{Id}}$ with:
\begin{align}
&a=0\\
&b^{\text{cNOT}}=\frac{300-T_{\text{Q}}}{T_{\text{Q}}} A P_g\\
& b^{\text{Id}}=0,
\label{eq:power_nisq}
\end{align}
and $\mathcal{M}_P=D_P \left(N_{P,\parallel}^{(\text{cNOT})} \mathcal{M}^{(\text{cNOT})}_{P}+N_{P,\parallel}^{(\text{Id})} \mathcal{M}^{(Id)}_{P}\right) $, with:
\begin{align}
 \mathcal{M}^{(\text{Id})}_{P}=\mathcal{M}^{(\text{cNOT})}_{P}/2=\left(1+\overline{n}_{\text{BE}}(T_{\text{Q}})+\frac{\overline{n}_{\text{BE}}(300)}{A}\right) \gamma_{\text{sp}} \tau_P^{\text{timestep}}.
\end{align}
We recognize $\overline{n}_{\text{tot}}=\overline{n}_{\text{BE}}(T_{\text{Q}})+\overline{n}_{\text{BE}}(300)/A$ in this expression. We also see that we assumed $A \gg 1$ (as it will anyway be necessary to reach $\mathcal{M}_P=\mathcal{M}_{\text{target}}$ in what follows). We will ask to have a target being $\mathcal{M}_{\text{target}}=1/3$ which implies that the fidelity of the algorithm output will be better than $2/3$ (because the metric is a pessimistic estimation: the "real" output fidelity will be better than the worst case one). 

Now, we need to describe the algorithm which will give us access to the average number of gates acting in parallel and $D_P$. For this example, we will consider running the quantum Fourier transform that we already introduced in the section \ref{sec:algo_QFT} of the previous chapter. It will be directly implemented on physical qubits (no error correction is being used yet). In order to simplify the discussion, we will remove all the Hadamard gates that are required in this algorithm (they wouldn't change the general behavior we are going to show). The algorithm is represented for $5$ qubits on the figure \ref{fig:qft_5qubits}. We take this example for pedagogic purposes.
\begin{figure}[h!]
\begin{center}
\includegraphics[width=1\textwidth]{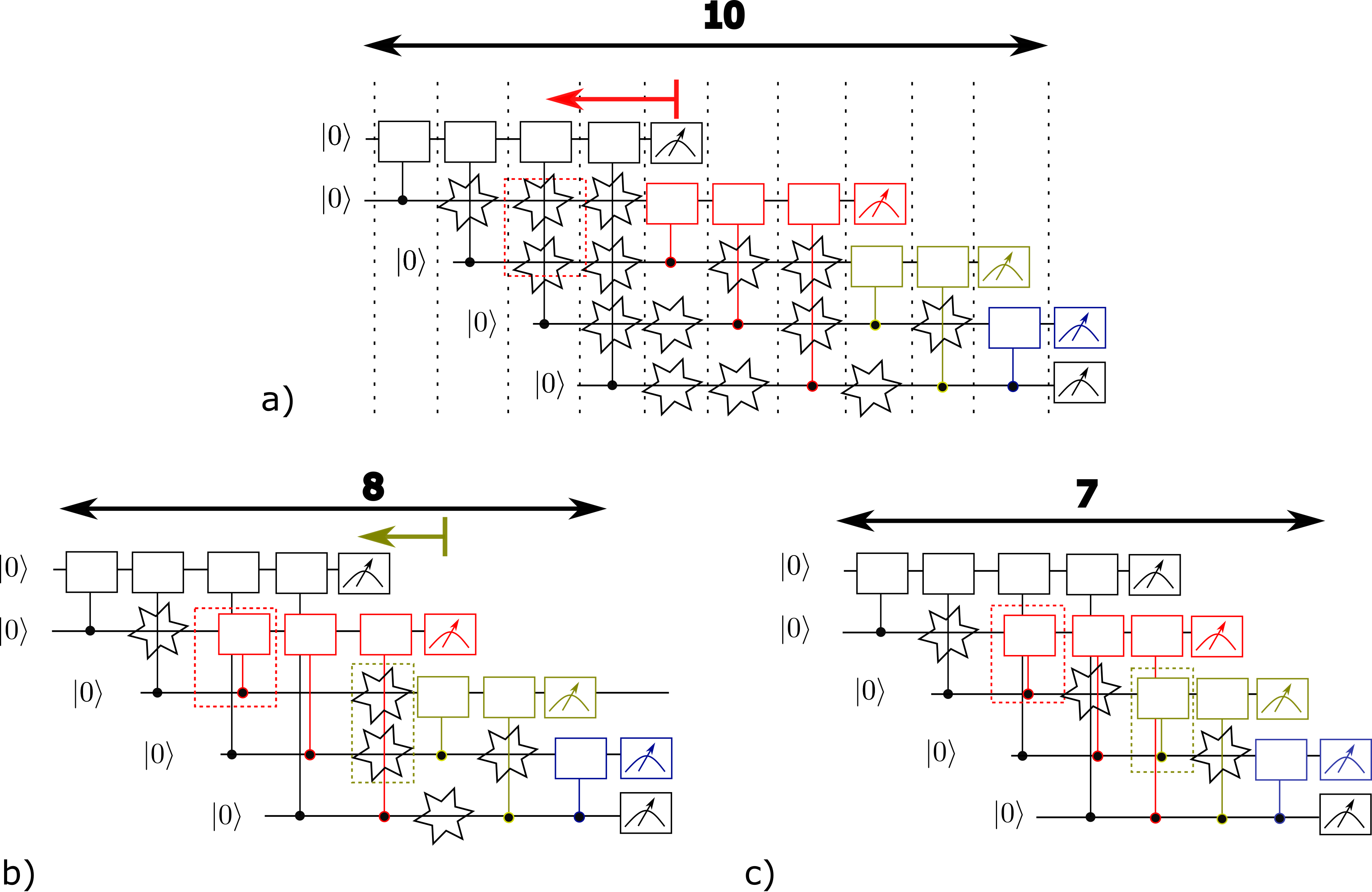}
\caption{\textbf{a)}: The QFT, introduced in \ref{sec:algo_QFT} (where Hadamard gates have been removed for simplicity), performed on $5$ qubits with no compression. The vertical dotted lines represent the different timesteps of the circuits. On some of those timesteps, some qubits are waiting as represented by the stars indicating the associated identity gates. Some other qubits are participating in a two-qubit controlled-phase gate that we modeled by a cNOT in our energetic estimation. To reduce the level of noise, we can compress the circuit. It can be done by moving toward the left the different subcomponents represented by the colors (black, green, yellow, and blue). The red arrow represents how to do the first level of compression represented on b). The red, yellow, and blue subcomponents are pushed toward the left for two timesteps to put the first red two-qubit gate inside the red dotted rectangle (removing the identity gates that were there). \textbf{b)}: First level of compression. The red, yellow, and blue circuits have been pushed toward the left such that the first red two-qubit gate is done during the third timestep in parallel of the two-qubit gate performed between the first and fourth qubit. \textbf{c)}: Second level of compression. The yellow subcomponent has also been pushed on the left (inside the yellow dotted rectangle) to reduce even more the algorithm duration and, thus, the noise. No further compressions are possible; the circuit reached the minimum depth possible.}
\label{fig:qft_5qubits}
\end{center}
\end{figure}
As we see, there are different ways to actually implement it. Either we run the algorithm in an "uncompressed" manner as on the figure \ref{fig:qft_5qubits} \textbf{a)}, which corresponds to the maximum depth, either we try to compress its length in such a way that the moments where qubits are waiting is minimized (on \textbf{c)}). The maximum compression corresponds to a circuit implemented with the smallest depth, where the average number of gates acting in parallel reaches its maximum. The lowest compression corresponds to the longest depth where the average number of gates acting in parallel reaches its minimum. One standard vision to implement algorithms is to make them the shortest possible. Indeed doing things this way, the quantity of noise that is introduced is then minimal. But it might not be the best way to do things in order to minimize power consumption \textit{under the constraint} of aiming a targetted fidelity. This is what is clearly illustrated on the figure \ref{fig:min_power_fct_compression} where the minimum power consumption (minimized on $T_{\text{Q}},A$) as a function of the algorithm depth is represented.
 
\begin{figure}[h!]
\begin{center}
\includegraphics[width=0.7\textwidth]{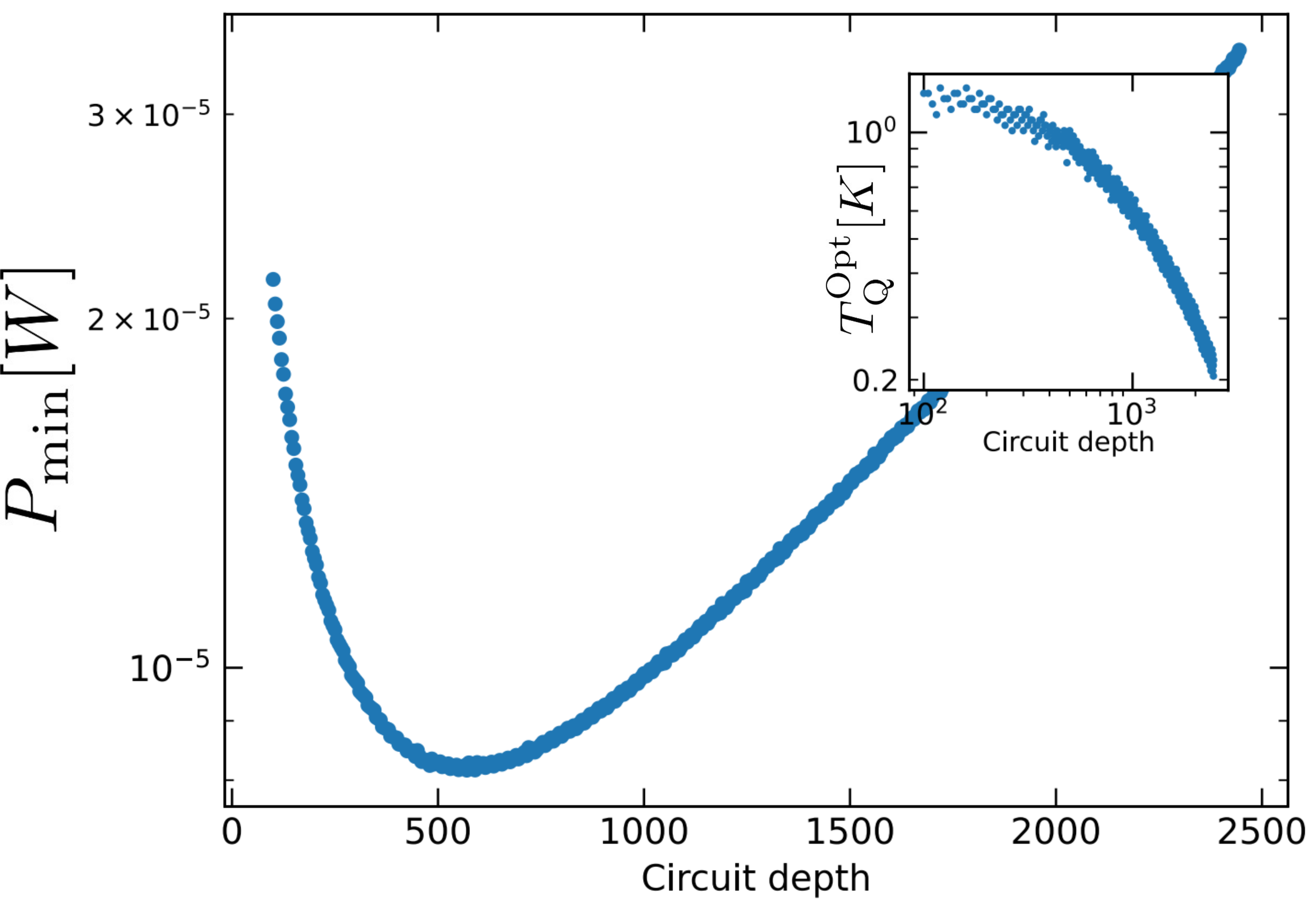}
\caption{Minimum power (minimized on $T_{\text{Q}},A$) for a QFT implemented on $30$ qubits as a function of the depth. Inset: optimal temperature $T_{\text{Q}}^{\text{Opt}}$ as a function of the compression factor. A non trivial optimal depth allows to minimize the power consumption.}
\label{fig:min_power_fct_compression}
\end{center}
\end{figure}
\FloatBarrier
On this curve, we see that the minimum power consumption is reached for a non-trivial value of the depth, neither corresponding to the minimum nor to the maximum. We also see that the optimal qubit temperature is a decreasing function of the depth. These curves illustrate in a very concrete manner the opposite behaviors between noise and power that we now explain.

We first explain the behavior of the temperature as shown on the inset. In order to understand it, we emphasize on the fact that the total number of cNOT is fixed in the algorithm, and it doesn't depend on the compression. However, the number of identity gates depends on the compression: in the low depth regime, the circuit is maximally compressed, and there are few identity gates, while in the high depth regime, the circuit is minimally compressed, and there are more identity gates. To understand the figure, we can rewrite the metric quantifying the noise as $\mathcal{M}_P=N_P^{(\text{cNOT})}\mathcal{M}_P^{(\text{cNOT})}+N_P^{(\text{Id})}\mathcal{M}_P^{(\text{Id})}$ where $N_P^{(i)}=N_{P,\parallel}^{(i)} D_P$ is the total number of physical gates of type $i$. Here, $N_P^{(\text{Id})}$ is the only variable that depends on the depth (and it is an increasing function of the depth). Because of that, in the low depth regime, there are fewer identity gates, thus fewer places where the noise can enter in the circuit. As we know that the minimum power is necessarily reached for $\mathcal{M}_P=\mathcal{M}_{\text{target}}$, the temperature of the qubits can be high in this regime, as shown in the inset. In the maximum depth regime, more identity gates are present, more places where the noise can enter exist, and the temperature of the qubits has to be lower in order to satisfy $\mathcal{M}_P=\mathcal{M}_{\text{target}}$. Now, we can study the behavior of the power. In the low depth regime, the Carnot efficiency is higher (because $T_{\text{Q}}$ is high from the previous discussion), and the Carnot efficiency decreases when the depth increases (because $T_{\text{Q}}$ decreases with the depth). Only reasoning with this efficiency we could believe that it is better to implement the circuit with a smaller depth. However, the lower the depth is, the more gates that consume power (i.e the cNOT) are acting in parallel: in the expression of the power above \eqref{eq:power_nisq} $N_{P,\parallel}^{(\text{cNOT})}$ increases when the depth decreases. Because of that, there is more heat dissipated in the low depth regime and less heat dissipated in the high depth regime. In the end, we understand that competition is occurring: it is good to reduce the depth to the benefit of a higher Carnot efficiency, but it is bad because more heat is dissipated in the attenuators. On the other hand, it is bad to increase the depth as the Carnot efficiency will be lower, but it is good because less heat will be dissipated in the attenuators. This competition implies that a non-trivial optimal depth exists in order to minimize power consumption. This is what is shown on the graph. This competition occurs because we ask to \textit{minimize} the power consumption \textit{under the constraint} that the noise strength should be equal to a target. 

In the end, this simple example illustrates that including considerations of power consumption in the algorithms can lead to non-trivial behaviors resulting from a "competition" between noise and power. This effect is a typical example of behavior that can only be understood from a \textit{transverse} approach as it relates algorithmic consideration (the depth) to engineering ones (the temperature of the qubit stage in the cryostat and the power consumption of this cryostat). In general, we can expect that the way to implement algorithms can have an important impact on power consumption, and we showed that we could, in principle, find the most energy-efficient way to implement such algorithms. 

In practice, our method can be applied for algorithms implemented without error correction, but to say that it could be applied to NISQ algorithms is a different affirmation. Indeed, most NISQ algorithms are variational algorithms that are "hybrid" in the sense that they require to implement a sequence of classical algorithms followed by a quantum algorithm (see \cite{bharti2021noisy,peruzzo2014variational} for instance). Because those kinds of algorithms are not "single shot", i.e., it is not a fixed and known set of gates that are implemented on some initial quantum state, the quantity of noise in the final density matrix might be more complicated to estimate such that what is presented here is not directly applicable\footnote{In some cases, it is believed that the noise can actually help those algorithms to converge to the good value \cite{gentini2019noise}. Because in our approach, the noise is seen as detrimental, some adaptations would be required.} (and some further investigations would be required to know to which extent our analysis could be extended to these kinds of algorithms).
\section{Adapting the full-stack framework for fault-tolerance}
\label{sec:adapting_framework_FT}
Now that the concepts have been provided for a computer not using quantum error-correction, we would like to adapt them to fault-tolerant quantum computing. This is the role of this section. We recall that in fault tolerance, physical qubits are regrouped together in order to create a logical qubit. Many physical gates are also regrouped together to create error-protected logical gates. One of the interests behind fault-tolerance construction is that this framework allows to some extent to ignore all the technical details that are occurring on the physical level and to do "as if" the logical qubits and gates were exactly like physical qubits and gates, excepted that they are less noisy. Following the same philosophy, we would like to keep the same interpretations for the physical metric and power we had in section \ref{sec:establishing_fullstack_framework_physical} and replace all the physical quantities that appeared there, such as physical depth, number of qubits, number of parallel gates by their logical equivalent (i.e., logical depth, number of logical qubits, number of logical gates acting in parallel...): we wish to adapt the full-stack framework we did on the physical level, at the logical level. We will see that with some very reasonable approximations, it will be possible, at least for the concatenated construction we are using. In this part, we will also do all the quantitative estimations about the number of physical qubits and gates that are required to perform a fault-tolerant calculation.
\subsection{Metrics at the logical level}
\label{sec:adapting_framework_FT_metric_logical}
Let us talk about the metric first. For that, we will use the tools that we already introduced: the probability of error of a logical gate. Calling it $p_L$, we recall that the probability of having an algorithm that fails corresponds to the probability any of the $N_L$ logical gates fails, which is simply $N_L p_L$ (see section \ref{sec:arbitrary_accurate_quantum_computing}). Because of that, the general description we adopted for the metric can directly be adapted to the logical level, and we have a particular case of \eqref{eq:physical_metric_gen}:
\begin{align}
&\mathcal{M}_L=N_L p_L=D_L \sum_{i \in \mathcal{G}_L} \overline{N}_{L,\parallel}^{(i)} \mathcal{M}^{(i)}_{L},
\end{align}
where for all\footnote{This is a standard approximation behind fault-tolerance: the logical error probability considered in the calculations corresponds to the probability of error of the "noisiest" logical gate: the cNOT.} $i$, $\mathcal{M}^{(i)}_{L} = p_L$, and where $\mathcal{G}_L=\{Id,H,S,X,Y,Z,cNOT\}$ is the logical gateset (the reason for this gateset is explained in the second paragraph of \ref{sec:FT_implementation_Steane_method}). Here we write $\overline{N}_{L,\parallel}^{(i)}$ instead of simply $N_{L,\parallel}^{(i)}$ to insist on the fact we are talking about the \textit{average} number of logical gates acting in parallel. We only considered an average number of gates in the previous sections, but as in the following sections, $N_{L,\parallel}^{(i)}$ will be used to represent a number of gates acting in parallel that does not correspond to an average, we add this "bar" from now on, to avoid confusions. We recall at this point that this gateset does not allow to implement an arbitrary logical gate: an extra gate such as the $T$ gate would be necessary. As also discussed in \ref{sec:FT_implementation_Steane_method}, those gates require a different procedure to be implemented, which goes beyond the scope of this Ph.D., in some sense, it means that we will approximate the energetic cost of an algorithm by only taking into account the gates in $\mathcal{G}_L$ it involves. In the end, we notice that the metric we used before is simply switched at the logical level. Everything behaves in the exact same manner.
\subsection{Power at the logical level}
Now, we can talk about the power part of the problem. One question is to wonder if it is possible to define it in an "effective manner" at the logical level. Can we just replace all the physical quantities entering in this expression with their logical counterpart and keep the same interpretation? To answer this question, we need to find the expression of the power (thus breaking down the fault-tolerant algorithm into its physical components) to try to then "recognize" the logical components with the hope that it can lead to a nice expression, analog to \eqref{eq:physical_power_gen}. The first step will thus be to estimate the number of physical components (qubits and gates) that are required in a fault-tolerant algorithm. This is the role of this section. 

\subsubsection{Estimation of the number of physical gates acting in parallel}
\label{sec:estimation_physical_gates}
We explained in the section \ref{sec:FT_implementation_Steane_method} the principle behind the construction of any of the logical gates in the gateset $\mathcal{G}_L$. To briefly summarize, the "recipe" to apply when concatenations are performed is to replace each physical gate with a 1-Rec. A 1-Rec for a single-qubit logical gate corresponds to the circuit represented in figure \ref{fig:steane_method_full} (a cNOT will have the same kind of structure excepted that error-correction will be implemented on the two logical qubits it is applied on). A proper counting of the number of physical gates in this 1-Rec, for each kind of logical gate implemented is shown again in the table \ref{table:FT_gate_breakdown_ch4}.
\begin{table*}[h!]
\begin{center}
\begin{tabular}{|c|c|c|c|c|}
	\hline
	& lvl-0 cNOT & lvl-0 Single & lvl-0 Identity & lvl-0 Measurement\\
	\hline
   lvl-1 cNOT & $135$ & $56$ & $72$ & $56$ \\
   \hline
   lvl-1 Single & $64$ & $35$ & $36$ & $28$ \\
   \hline
   lvl-1 Identity & $64$ & $28$ & $43$ & $28$\\
   \hline
   lvl-1 Measurement & $0$ & $0$ & $0$ & $7$\\
   \hline
\end{tabular}
\caption{Each row lists a FT logical gate and tabulates the lower level components required for the listed gate as columns. Those level-0 components are thus 0-Ga physical gates. The single-qubit gates here are either Hadamard, Pauli or $S$ gates. This table includes the gates required to prepare the ancillae and verifier and assumes that the verification always succeeds (there is no need to prepare more than one $Z$ or $X$ syndrome ancilla for instance). We discuss explain why we think it is a good assumption to make in \ref{app:ancilla_rejected}.}
\label{table:FT_gate_breakdown_ch4}
\end{center}
\end{table*}
\FloatBarrier
This table will allow us to deduce, based on a recursive reasoning the number of physical qubits and gates of each type that are required for a level-k logical gate. To be precise about what we call a level-$k$ logical gate here, we mean the physical elements that result from a 0-Ga that has been concatenated $k$ times\footnote{The notion of logical gate is "context dependent" as it could mean the logical implementation of the gate without error-correction (for $k=1$ it is the 1-Ga), its implementation followed by error-correction (for $k=1$ it is the 1-Rec), the appropriate set of gates that allows to define the probability of error (for $k=1$ it is the 1-exRec). Here we take the natural definition if we want to estimate the number of physical resources, which corresponds to the resulting circuit from a 0-Ga concatenated $k$ times (for $k=1$ it would correspond to the 1-Rec).}. We assume that we have $N_L^{\text{cNOT}}$, $N_L^{\text{1qb}}$, $N_L^{\text{Id}}$,$N_L^{\text{Meas}}$ logical cNOT, single-qubit, identity gates and measurement gates in the algorithm. We put those elements in a vector:$X_0=(N_L^{\text{cNOT}},N_L^{\text{1qb}},N_L^{\text{Id}},N_L^{\text{Meas}} )^T$ ($T$ means transposition). The number of \textit{physical} cNOT, single-qubit, identity and measurement gates after $k$ concatenations can be put in the vector: $X_k=(N_P^{\text{cNOT}}(k),N_P^{\text{1qb}}(k),N_P^{\text{Id}}(k),N_P^{\text{Meas}}(k) )^T$. Such vector satisfies $X_k=A^k X_0$, where:
\begin{align}
& A=\begin{pmatrix} 135 & 64 & 64 & 0\\ 56 & 35 & 28 & 0\\ 72 & 36 & 43 & 0 \\ 56 & 28 & 28 & 7\end{pmatrix}.
\end{align}
Indeed, for $k=1$, we just have to apply the table \ref{table:FT_gate_breakdown_ch4} to the number of gates implemented in the algorithm: it allows us to find that $X_1=A X_0$ which provides the number of physical gates after the first concatenation level. For the second concatenation level, each of those physical gates will \textit{again} be replaced by a 1-Rec; thus, we just have to "reuse" this table another time but applied this time on $X_1$, it gives $X_2=A X_1 = A^2 X_0$. The same principle applies recursively for further concatenation levels, which explains why $X_k=A^k X_0$. In the end, from the knowledge of the number of gates required by the algorithm, we can access the number of physical gates required in the algorithm for an arbitrary concatenation level. We deduce that we have the following \textit{total} number of physical gates required in the algorithm:
\begin{align}
X_k = \left(
\begin{array}{c}
 \frac{1}{3} 199^k (2 N_L^{\text{cNOT}}+N_L^{\text{Id}}+N_L^{\text{1qb}})+\frac{7^k}{3} \left(N_L^{\text{cNOT}} -N_L^{\text{Id}} - N_L^{\text{1qb}}\right) \\
 \frac{7}{48} 199^k (2 N_L^{\text{cNOT}}+ N_L^{\text{Id}}+ N_L^{\text{1qb}})+\frac{7^k}{48} \left(-14 N_L^{\text{cNOT}} -7 N_L^{\text{Id}} +41\  N_L^{\text{1qb}}\right) \\
 \frac{3}{16} 199^k (2 N_L^{\text{cNOT}}+ N_L^{\text{Id}}+ N_L^{\text{1qb}})+\frac{7^k}{16} \left(-6 N_L^{\text{cNOT}} +13 N_L^{\text{Id}} -3\ N_L^{\text{1qb}}\right) \\
 \frac{7}{48} 199^k (2 N_L^{\text{cNOT}}+ N_L^{\text{Id}}+ N_L^{\text{1qb}})+\frac{7^k}{48}\left(-14 N_L^{\text{cNOT}} -7 N_L^{\text{Id}}+48N_L^{\text{Meas}}-7N_L^{\text{1qb}}\right)\\
\end{array}
\right)
\label{eq:number_physical_gate_k_concat}
\end{align}
We notice from this expression that two numbers seem to have particular importance in this formula: $199$ and $7$ (they are the only two numbers elevated to a power $k$). They correspond to the two eigenvalues of the matrix $A$. In order to interpret their signification, let us first neglect the term proportional to $7^k$. Doing this, we realize that $X_k \propto (2 N_L^{\text{cNOT}}+N_L^{\text{Id}}+N_L^{\text{1qb}})$. It means that whatever the physical gate we are looking at, the exact type of gates implemented by the algorithm does not really matter; what matters is the \textit{total number} of logical gates involved (counting twice the cNOT and removing the measurement gates at the logical level for a reason given in a few lines). For instance, if we had implemented $1$ logical identity, or $1$ single-qubit logical gate (and no other gates), the number of physical components of each type wouldn't change. It symbolizes the fact that the fault-tolerant construction of the gates contains a dominant part that does not really depend on the gate that is being implemented, and we can understand that from the figure \ref{fig:steane_method_full}: whatever the 1-Ga has to be protected, the error-correction structure (i.e., the 1-Ec) is the same for any of the logical gates to implement. But this has two exceptions: first, we need two implementations of error-correction for a cNOT because it involves two logical qubits, it explains the $\times 2$ in front of $N_L^{\text{cNOT}}$. Then, the measurement gates will not be followed by error correction (because the information becomes classical when qubits are being measured). It explains why $N_L^{\text{Meas}}$ does not appear in the term proportional to $199^k$: the number of logical measurement gates does not play a significant role in the number of physical elements when concatenations are being performed. But the error correction is not everything; there is also the transversal implementation of the gate (the 1-Ga). And this part will differ depending on the exact logical gate we want to implement. For instance, a logical identity gate would require $7$ less single-qubit gate and $7$ more identity gates to be implemented. Those little differences are expressed by the much smaller eigenvalue $7$. And we see that the terms proportional to this eigenvalue are not the same for each physical gate (because here, there are variations depending on the logical gates that have been implemented). 

Now, we are here interested in the \text{average} number of physical gates that are acting in parallel. In order to access it, we need to know the average number of gates acting in parallel for the ideal algorithm (for each type: single-qubit, identity, and two-qubit gates). Putting those numbers in a vector $X_0=(\overline{N}_{L,\parallel}^{\text{cNOT}},\overline{N}_{L,\parallel}^{\text{1qb}},\overline{N}_{L,\parallel}^{\text{Id}},\overline{N}_{L,\parallel}^{\text{Meas}})^T$, we can deduce the total number of physical gate inside an average algorithm logical timestep (one algorithm timestep corresponds to the time it takes to execute any logical gate\footnote{Some logical gates can in principle be implemented faster than others. For instance, a logical identity could only last for two timesteps as the transversal implementation of the gate could be done instantaneously (as an identity operation can be performed instantaneously). But as we need to keep all the qubits "synchronized", the longest logical gate to execute should give the duration of a logical timestep of the algorithm (even though "smart" optimizations could be performed for specific algorithms in order to reduce the algorithm duration slightly).}). It naturally corresponds to \eqref{eq:number_physical_gate_k_concat} where we replace the number of logical gates by the average number of logical gates acting in parallel ($N_{L}^{x} \to \overline{N}_{L,\parallel}^{x}$ for $x \in \{\text{cNOT}, \text{Id}, \text{1qb},\text{Meas}\}$). Now, this is not exactly the average number of physical gates that are acting in parallel. Indeed, a logical timestep is composed of a given number of physical timesteps. More precisely, calling $\tau_P^{\text{timestep}}$ the duration of a physical timestep (which corresponds to the duration of the longest physical gate used in the FT construction: we will assume it corresponds to both the duration of a cNOT and measurement gates in what follows), we have after $k$ concatenations a logical gate that lasts for $\tau_{L}^{\text{timestep}}=3^k \tau_P^{\text{timestep}}$. Indeed, as explained in \ref{sec:number_timesteps_for_logical_and_ancilla}, the logical data qubits are going through three timesteps: the first one corresponds to the transversal implementation of the gate, followed by two timesteps used to perform the syndrome measurement\footnote{We recall that we assume the correction is not physically performed with our assumptions of "keeping track" of the syndrome with classical processing.}. Increasing the concatenation level, each of those timesteps will again be multiplied by three. This is how we deduce $\tau_{L}^{\text{timestep}}=3^k \tau_P^{\text{timestep}}$. The ancillae must, however, be prepared in advance: including their preparation and the moment they can be "plugged out", the duration of this gate would be longer. But as explained on the figure \ref{fig:timestep_logical_gate}, the good timescale to consider to estimate the average number of gates would still be $\tau_{L}^{\text{timestep}}=3^k \tau_P^{\text{timestep}}$.
\begin{figure}[h!]
\begin{center}
\includegraphics[width=1\textwidth]{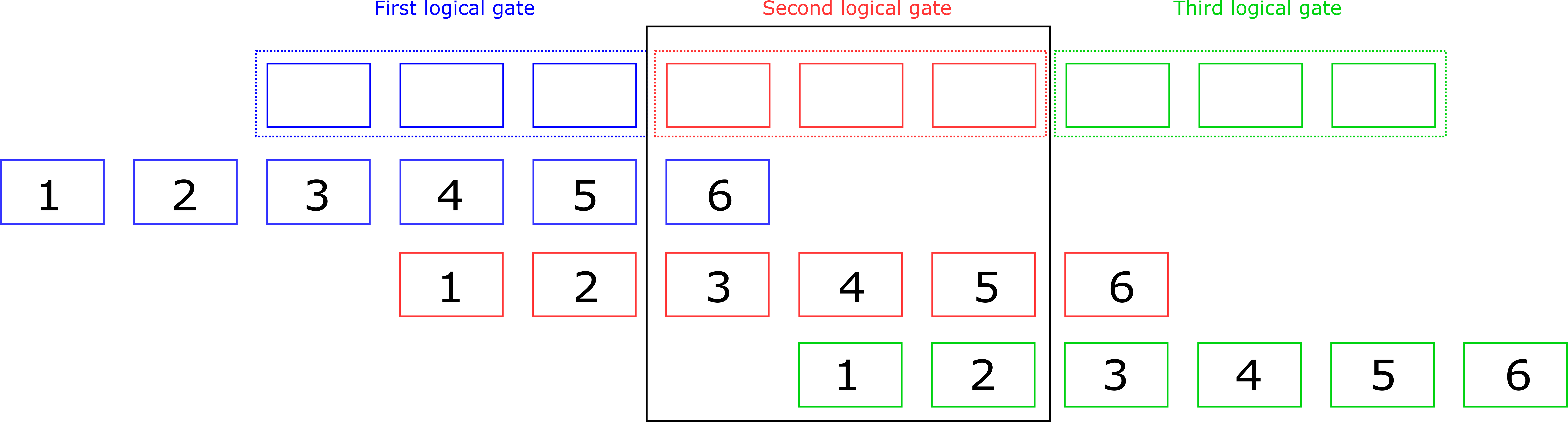}
\caption{Determination of the average number of physical gates acting on an average logical timestep. On the first line of this figure are represented a sequence of three logical gates that are composed of three physical timesteps as indicated by the three red rectangles (the transversal implementation and the measurement of the $X$ and $Z$ syndromes). Those logical gates require ancillae, which are represented on the line below by the matching colors. Between the moment the ancillae are prepared and measured, there are more than three timesteps (we represented $6$ here for simplicity, but the exact number would be different and would actually depend on the ancilla we are looking at). However, in order to compute the average number of physical gates acting in one logical gate, everything will behave \textit{as if} those gates were actually implemented in parallel of the $3$ timesteps. Indeed, if we focus on the red logical gate, the ancillae might have gates applied on them before ($1$ and $2$ red boxes) and after the logical gate ($6$ red box), but the preceding and following logical gates would exactly compensate those missing gates. More precisely, this is true if we neglect the boundary of the algorithm and if we are interested in the average number of gates acting in parallel (which is the case).}
\label{fig:timestep_logical_gate}
\end{center}
\end{figure}
\FloatBarrier
Finally, the average number of $x \in \{\text{cNOT}, \text{Id}, \text{1qb},\text{Meas}\}$ physical gates in this logical timestep is:
\begin{align}
&\overline{N}_{P,\parallel}^{x}(k)=N_{\text{during } \tau_L}^{x}(k) /3^k
\label{eq:average_number_physical_gate_timestep}
\end{align}
Where $N_{\text{during } \tau_L}^{x}(k)$ is the number of $x$ type physical gates that are inside an average logical timestep, i.e the number of $x$ physical gates there are assuming that $\overline{N}_{L,\parallel}^{\text{cNOT}},\overline{N}_{L,\parallel}^{\text{1qb}},\overline{N}_{L,\parallel}^{\text{Id}},\overline{N}_{L,\parallel}^{\text{Meas}}$ logical gates are acting in parallel in the algorithm. We find (those numbers are simply obtained by replacing $N_{L}^{x} \to \overline{N}_{L,\parallel}^{x}$ for $x \in \{\text{cNOT}, \text{Id}, \text{1qb}\}$ in \eqref{eq:number_physical_gate_k_concat}):
\begin{align}
&N_{\text{during } \tau_L}^{\text{cNOT}}(k)= \frac{1}{3} 199^k (2 \overline{N}_{L,\parallel}^{\text{cNOT}}+\overline{N}_{L,\parallel}^{\text{Id}}+\overline{N}_{L,\parallel}^{\text{1qb}})+\frac{7^k}{3} \left(\overline{N}_{L,\parallel}^{\text{cNOT}} -\overline{N}_{L,\parallel}^{\text{Id}} - \overline{N}_{L,\parallel}^{\text{1qb}}\right)
\label{eq:tot_number_physical_gate_timestep1}\\
&N_{\text{during } \tau_L}^{\text{1qb}}(k)=\frac{7}{48} 199^k (2 \overline{N}_{L,\parallel}^{\text{cNOT}}+ \overline{N}_{L,\parallel}^{\text{Id}}+ \overline{N}_{L,\parallel}^{\text{1qb}})+\frac{7^k}{48} \left(-14 \overline{N}_{L,\parallel}^{\text{cNOT}} -7 \overline{N}_{L,\parallel}^{\text{Id}} +41\  \overline{N}_{L,\parallel}^{\text{1qb}}\right)
\label{eq:tot_number_physical_gate_timestep2}\\
&N_{\text{during } \tau_L}^{\text{Id}}(k)= \frac{3}{16} 199^k (2 \overline{N}_{L,\parallel}^{\text{cNOT}}+\overline{N}_{L,\parallel}^{\text{Id}}+\overline{N}_{L,\parallel}^{\text{1qb}})+\frac{7^k}{16} \left(-6\overline{N}_{L,\parallel}^{\text{cNOT}} +13\overline{N}_{L,\parallel}^{\text{Id}} - 3 \overline{N}_{L,\parallel}^{\text{1qb}}\right) 
\label{eq:tot_number_physical_gate_timestep3}\\
&N_{\text{during } \tau_L}^{\text{Meas}}(k)= \frac{7}{48} 199^k (2 \overline{N}_{L,\parallel}^{\text{cNOT}}+\overline{N}_{L,\parallel}^{\text{Id}}+\overline{N}_{L,\parallel}^{\text{1qb}})+\frac{7^k}{48} \left(-14 \overline{N}_{L,\parallel}^{\text{cNOT}} -7 \overline{N}_{L,\parallel}^{\text{Id}}+48 \overline{N}_{L,\parallel}^{\text{Meas}} - 7 \overline{N}_{L,\parallel}^{\text{1qb}}\right) 
\label{eq:tot_number_physical_gate_timestep4}
\end{align}
The second parenthesis present in those expressions, as associated with the eigenvalue $7$, will typically be small enough to be ignored compared to the first parenthesis associated with the eigenvalue $199$. We can justify our claim with what follows. In practice, they would be the closest for $k=1$ (when $k$ increases further, the first parenthesis dominates even more). Excluding nonphysical cases where there are only logical measurements acting in parallel (the number of measurements in parallel cannot dominate the other logical gates as logical qubits are manipulated before being measured), a very pessimistic scenario would correspond to a case where $\overline{N}_L^{\text{1qb}}$ would dominate the number of gate acting in parallel. For example in the case $k=1$, if $\overline{N}_L^{\text{1qb}}=1$ (and the other logical gates acting in parallel are present in negligible number), the ratio of the first parenthesis divided by the second one in \eqref{eq:tot_number_physical_gate_timestep2} would give a number close to $5$. Not taking into account the smallest eigenvalue in this \textit{very pessimistic} case would imply that we are making an error in our estimation for the number of physical gates about $20 \%$. This is small enough to be ignored as our purpose will be to estimate the energetic cost of running a fault-tolerant quantum algorithm in order of magnitudes. For this reason, we will approximate the average number of physical single, cNOT, and measurement gates by\footnote{We don't need to find the average number of physical identity gates as they do not consume any power with our hypotheses.}:
\begin{align}
&\overline{N}_{P,\parallel}^{\text{cNOT}}(k) \approx \frac{1}{3} \left(\frac{199}{3}\right)^k (2 \overline{N}_{L,\parallel}^{\text{cNOT}}+\overline{N}_{L,\parallel}^{\text{Id}}+\overline{N}_{L,\parallel}^{\text{1qb}})
\label{eq:average_number_physical_gate_timestep1}\\
&\overline{N}_{P,\parallel}^{\text{1qb}}(k) \approx \frac{7}{48} \left(\frac{199}{3}\right)^k (2 \overline{N}_{L,\parallel}^{\text{cNOT}}+ \overline{N}_{L,\parallel}^{\text{Id}}+ \overline{N}_{L,\parallel}^{\text{1qb}})
\label{eq:average_number_physical_gate_timestep2}\\
&\overline{N}_{P,\parallel}^{\text{Meas}}(k) \approx \frac{7}{48} \left(\frac{199}{3}\right)^k (2 \overline{N}_{L,\parallel}^{\text{cNOT}}+\overline{N}_{L,\parallel}^{\text{Id}}+\overline{N}_{L,\parallel}^{\text{1qb}})
\label{eq:average_number_physical_gate_timestep3}
\end{align}
We notice that the number of physical measurements is roughly the same as the number of the other physical gates (actually, it is exactly equal to the number of single-qubit gates, for instance). This is an important difference with the case where no error correction is being performed. Indeed when no error correction is performed, it is fine to neglect the dynamic cost associated with measurements because they only affect the boundaries of the algorithm. Here, because error correction is being performed, measurements are performed on each logical timestep, and their dynamic energetic cost might not be negligible (they act "in the volume" of the algorithm). The dynamic cost of measurement corresponds to the readout pulses that are sent to read the qubit state as opposed to the static costs of measurement that are associated with the amplifiers (amplifiers cannot \textit{usually} be turned off when not used).
\subsubsection{Estimation of the number of physical qubits used in the computer}
\label{sec:estimation_physical_qubits}
The number of physical gates being estimated for a given concatenation level $k$, it remains to estimate the number of physical qubits to be able to finally write down the expression of the power. For the physical gates, we were interested in finding their average number because the associated power consumption was of a dynamic type. For the physical qubits, we need to find the bottleneck of the algorithm. Indeed physical qubits cannot be removed from the computer when they are not being used, so we need to find the moment when the maximum number of physical qubits will be required to estimate their number. 

In order to do the estimation, we notice that the number of ancillae qubits dominates the total number of physical qubits required. Indeed, as one can see from figure \ref{fig:steane_method_full}, for $k=1$ for instance, we would need $28$ (resp $28 \times 2$) ancillae qubits for a single or identity (resp cNOT) logical qubit gate\footnote{It assumes that the syndrome ancillae are always being accepted. As explained in the appendix \ref{app:ancilla_rejected}, the typical number of ancillae rejected can be neglected compared to the total number of ancillae accepted.}. Compared to the $7$ to $14$ physical data qubits, as we are only interested in orders of magnitude, this is a fair approximation to make. Furthermore, increasing the concatenation level would only make the number of physical ancillae qubits dominate even more. For this reason, we will only try to estimate the number of ancillae qubits. Then, we recall that the $k$'th concatenation level is obtained from the $k-1$ level in which each physical gate has been replaced by a 1-Rec. We can use that to estimate the number of physical qubits. We first assume that the algorithm requires a unique logical timestep composed of $(N_{L,\parallel}^{\text{cNOT}},N_{L,\parallel}^{\text{1qb}},N_{L,\parallel}^{\text{Id}},N_{L,\parallel}^{\text{Meas}})$ cNOT, single, identity and measurement logical gates acting in parallel. From this knowledge, we can find the total number of physical gates required on the $k-1$ level. Multiplying this number by $28$ for the single-qubit and identity gates and $28 \times 2$ for the cNOT, we deduce the total number of ancillae qubits required. We notice that we don't multiply measurement gates by $28$ because they do not require ancillae (because they are not followed by a quantum error-correction procedure). It would give us access to the total number of ancillae qubits after $k$ concatenation, and thus roughly the total number of physical qubits. We notice that this calculation assumes that no ancillae qubits can be recycled after $k$ concatenation (i.e., each physical gate in the $k-1$ level will need "fresh" ancillae when replaced by a 1-Rec to implement the $k$ concatenation level). Both the $X$ and $Z$ syndrome ancillae are participating to $9$ physical timesteps as explained in \ref{sec:number_timesteps_for_logical_and_ancilla}. Thus, in principle, we can reuse them once they have done those $9$ timesteps. In a similar way, the verifier ancilla can be reused (even earlier because they are used for a smaller amount of timesteps). Properly counting the number of ancillae that could be reused would, however, be a complicated task for an arbitrary concatenation level. Furthermore, in practice, we will mainly be interested in the first $3$ concatenations levels in the next chapter (more concatenation levels would imply an unrealistic power demand for all practical purposes for the models used there). Because of that, the number of timesteps for the logical gates on the $k-1$ level is also equal to $9$: $\tau^{\text{timestep}}_L/\tau^{\text{timestep}}_P=3^{k-1}=9$ which means that anyway, the level of recycling of the syndrome ancilla qubit would be "poor" (but we might be able "in principle" to recycle some of the verifier ancilla qubits, which represents the half of the total number of ancillae, as they are used for a smaller amount of timesteps). We assumed in the end that the ancillae qubits are not recycled within this logical gate. We note that it can actually be seen as the way the computer is designed and not as an assumption simplifying calculations. In summary, using \eqref{eq:tot_number_physical_gate_timestep1},\eqref{eq:tot_number_physical_gate_timestep2} and \eqref{eq:tot_number_physical_gate_timestep3}, we can estimate that the total number of physical qubits required in an algorithm composed of a unique logical timestep having $(N_{L,\parallel}^{\text{cNOT}},N_{L,\parallel}^{\text{1qb}},N_{L,\parallel}^{\text{Id}},N_{L,\parallel}^{\text{Meas}})$ cNOT, single, identity and measurement logical gates acting in parallel is $Q_P(k) \approx \frac{28}{199} 199^k (2 N_{L,\parallel}^{\text{cNOT}} + N_{L,\parallel}^{\text{1qb}}+N_{L,\parallel}^{\text{Id}})$. 

Now, an algorithm is composed of many logical timesteps. With our assumption where the ancillae qubits are not being recycled within a logical gate, the number of physical qubits required is given by the moment in the algorithm where $2 N_{L,\parallel}^{\text{cNOT}} + N_{L,\parallel}^{\text{1qb}}+N_{L,\parallel}^{\text{Id}}$ will be \textit{maximum} (which is in some way the "bottleneck" in term of physical resources required). Thus, the number of physical qubits can be written as:
\begin{align}
Q_P(k) \approx \frac{28}{199} 199^k \max_i \left[2 N_{L,\parallel,i}^{\text{cNOT}} + N_{L,\parallel,i}^{\text{1qb}}+N_{L,\parallel,i}^{\text{Id}}\right]
\label{eq:number_physical_qubits_k_concat}
\end{align}
where $i$ denotes the logical timestep on which there are $N_{L,\parallel,i}^{\text{cNOT}}$, $N_{L,\parallel,i}^{\text{1qb}}$, $N_{L,\parallel,i}^{\text{Id}}$ logical cNOT, single-qubit and identity gates acting in parallel. However this is not exactly correct (we need to refine one last time the calculation). In principle, for the same reason that consecutive gates on the $k-1$ level had to use fresh ancillae, we cannot use the same ancillae for consecutive logical gates because the ancilla preparation takes longer than the duration of the logical gate (physical ancillae within consecutive logical gates might overlap). Here, we will assume that the ancillae used inside a logical gate can be reused (recycled) when all the ancillae of this logical gate have finished to be used (it means that we recycle ancillae qubits inside a logical gate only when none of those qubits are used anymore). As explained in the appendix \ref{app:recycling_ancilla_logical}, modifying $Q_P(k) \to 4 Q_P(k)$ in \eqref{eq:number_physical_qubits_k_concat} would lead to an appropriate counting of such overlap. We will take into account this coefficient, such that the number of physical qubits used after $k$ concatenations is:
\begin{align}
Q_P(k) \approx \frac{112}{199} 199^k \max_i \left[2 N_{L,\parallel,i}^{\text{cNOT}} + N_{L,\parallel,i}^{\text{1qb}}+N_{L,\parallel,i}^{\text{Id}}\right]
\label{eq:number_physical_qubits_k_concat_no_recycling}
\end{align}
However, to define the power at the logical level, we would like to "make appear" the number of logical qubits $Q_L$. In order to do it, we can notice that we have\footnote{It assumes that the algorithm uses the minimal number of logical qubits it conceptually need. Indeed, this equation means that the moment the maximum number of logical qubits are doing something "together" is equal to the number of logical qubits. But we could imagine inefficient implementation where this wouldn't be true.}
\begin{align}
\max_i \left[2 N_{L,\parallel,i}^{\text{cNOT}} + N_{L,\parallel,i}^{\text{1qb}}+N_{L,\parallel,i}^{\text{Id}}+N_{L,\parallel,i}^{\text{Meas}} \right]=Q_L,
\label{eq:number_logical_qubits}
\end{align} 
where $N_{L,\parallel,i}^{\text{Meas}}$ is the number of logical measurement gates at the timestep $i$. The reason behind this equation can intuitively be understood from the figure \ref{fig:bottleneck_algo_physical_qubits}. We recall that $N_{L,\parallel,i}^{\text{Meas}}$ did not appear in \eqref{eq:number_physical_qubits_k_concat} because measurement gates do not require ancillae, and our estimation of the number of physical qubits was given by the number of ancillae (as they dominate). Now, for all practical purpose we will usually have $\max_i \left[2 N_{L,\parallel,i}^{\text{cNOT}} + N_{L,\parallel,i}^{\text{1qb}}+N_{L,\parallel,i}^{\text{Id}}\right] \approx \max_i \left[2 N_{L,\parallel,i}^{\text{cNOT}} + N_{L,\parallel,i}^{\text{1qb}}+N_{L,\parallel,i}^{\text{Id}}+N_{L,\parallel,i}^{\text{Meas}}\right]$, which allows us to write:
\begin{align}
Q_{P}(k) \approx \frac{112}{199} 199^k Q_L,
\label{eq:number_physical_qubits_from_logical}
\end{align}
We explain why in a few lines. This formula does not apply for $k=0$ as it is based on estimating the total number of physical qubits from the number of ancillae qubits, and there are no ancillae qubits when no concatenations are performed. We can understand why our approximation will usually be valid with the following analysis. First, we call $i_0$ the timestep on which \eqref{eq:number_logical_qubits} is satisfied (the index satisfying the maximization). If on this timestep, $N_{L,\parallel,i_0}^{\text{Meas}}$ is small compared to the other gates, then our approximation is directly valid. If not, at the logical timestep right before (i.e for the timestep $i_0-1$), all those measurement gates were either logical identity, single-qubit or cNOT gates. It implies that $2 N_{L,\parallel,i_0-1}^{\text{cNOT}} + N_{L,\parallel,i_0-1}^{\text{1qb}}+N_{L,\parallel,i_0-1}^{\text{Id}} \geq N_{L,\parallel,i_0}^{\text{Meas}}$ (this is an inequality because on the timestep $i_0-1$, the measurements are replaced by other logical gates, but qubits that were not measured in $i_0$ might be affected by some other logical gates). Thus, $\max_i \left[2 N_{L,\parallel,i}^{\text{cNOT}} + N_{L,\parallel,i}^{\text{1qb}}+N_{L,\parallel,i}^{\text{Id}}\right] \geq N_{L,\parallel,i_0}^{\text{Meas}}$. And as $N_{L,\parallel,i_0}^{\text{Meas}}$ was the dominant number of gates on the timestep $i_0$, it shows that $\max_i \left[2 N_{L,\parallel,i}^{\text{cNOT}} + N_{L,\parallel,i}^{\text{1qb}}+N_{L,\parallel,i}^{\text{Id}}\right]$ gives a good estimation to $\max_i \left[2 N_{L,\parallel,i}^{\text{cNOT}} + N_{L,\parallel,i}^{\text{1qb}}+N_{L,\parallel,i}^{\text{Id}}+N_{L,\parallel,i}^{\text{Meas}}\right]$, and thus that \eqref{eq:number_physical_qubits_from_logical} is satisfied. In what follows, we will always assume it to be the case as we only care about order of magnitudes estimations, but also mainly because the examples we will consider in the next chapter will clearly satisfy this hypothesis (it will not be an approximation but an exact result in the examples we will take). We illustrate all this on figure \ref{fig:bottleneck_algo_physical_qubits} b).
\begin{figure}[h!]
\begin{center}
\includegraphics[width=1\textwidth]{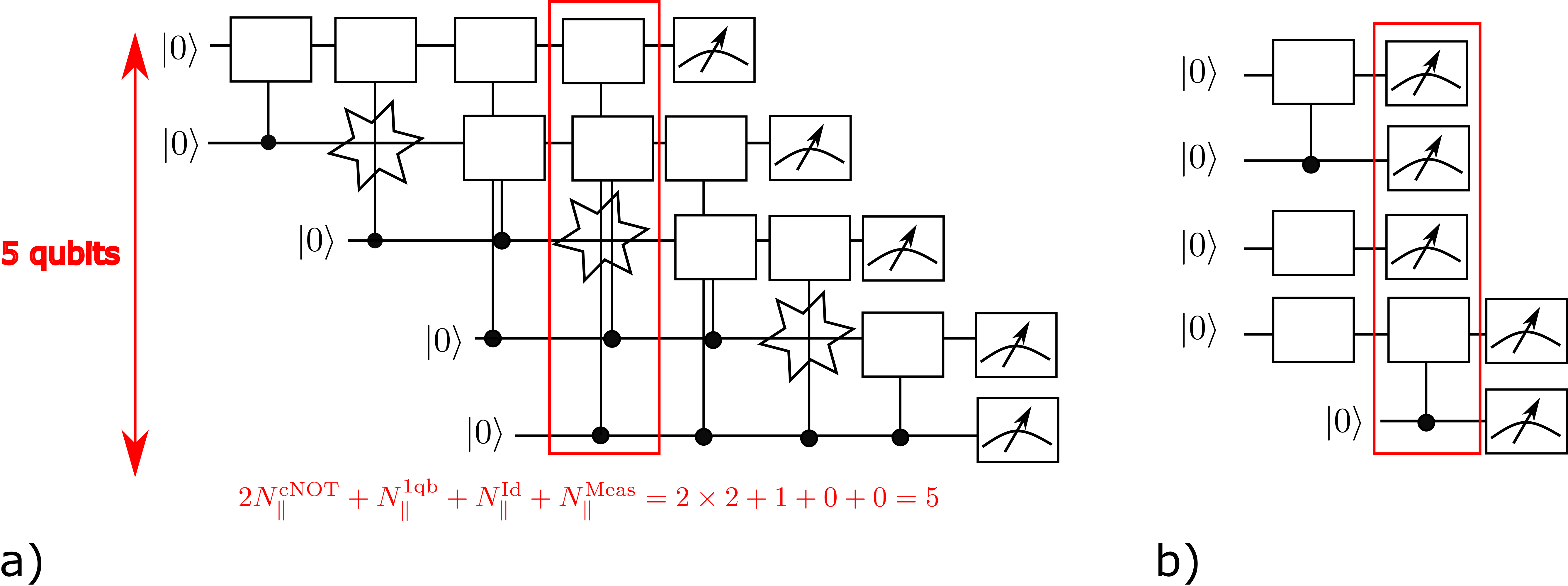}
\caption{\textbf{a)} How the number of qubits is related to the maximum number of gates acting in parallel. On this figure, we illustrate an algorithm composed of two-qubit and identity gates. We notice that the number of qubits required by the algorithm (here they appear to be logical ones) is equal to $\max_i[2 N_{L,\parallel,i}^{\text{cNOT}} + N_{L,\parallel,i}^{\text{1qb}}+N_{L,\parallel,i}^{\text{Id}}+N_{L,\parallel,i}^{\text{Meas}}]=5$, it happens on the red rectangle (but also on the following timestep). \textbf{b)} Intuitive understanding of why we can expect that $Q_L \approx \max_i \left[2 N_{L,\parallel,i}^{\text{cNOT}} + N_{L,\parallel,i}^{\text{1qb}}+N_{L,\parallel,i}^{\text{Id}} \right]$. The red box rectangle satisfies $2 N_{L,\parallel}^{\text{cNOT}} + N_{L,\parallel}^{\text{1qb}}+N_{L,\parallel}^{\text{Id}}+N_{L,\parallel}^{\text{Meas}}=Q_L$ exactly. But because the measurement gates dominate on this timestep, and that they are identity, single or two qubit gates on the timestep before (we assume that a qubit will never directly be measured as it would be useless for an algorithm), the result is close to $\max_i \left[2 N_{L,\parallel,i}^{\text{cNOT}} + N_{L,\parallel,i}^{\text{1qb}}+N_{L,\parallel,i}^{\text{Id}}\right]=4$.}
\label{fig:bottleneck_algo_physical_qubits}
\end{center}
\end{figure}
\FloatBarrier
\subsubsection{Expression of the power}
\label{sec:expression_power_ch4}
Now, we can regroup the different things we calculated to express the power at the logical level. On the physical level, it is a function: $P=Q_P a + \sum_{i \in \mathcal{G}_{P}} N^{(i)}_{P,\parallel} b^i$. Assuming that all single-qubit gates consume the same amount of power, using \eqref{eq:average_number_physical_gate_timestep1}, \eqref{eq:average_number_physical_gate_timestep2}, \eqref{eq:average_number_physical_gate_timestep3} and \eqref{eq:number_physical_qubits_from_logical}, this expression can be "translated" at the logical level:
\begin{align}
&P=Q_L t(k) a + \left( b^{\text{1qb}} u(k) + b^{\text{cNOT}} v(k)+ b^{\text{Meas}} w(k)\right) (\overline{N}_{L,\parallel}^{\text{1qb}}+\overline{N}_{L,\parallel}^{\text{Id}}+2 \overline{N}_{L,\parallel}^{\text{cNOT}}) \label{eq:logical_power}\\
& t(k) \equiv \frac{112}{199} 199^k\\
& u(k) \equiv \frac{7}{48}\left(\frac{199}{3}\right)^k\\
& v(k) \equiv \frac{1}{3} \left(\frac{199}{3}\right)^k\\
& w(k) \equiv \frac{7}{48} \left(\frac{199}{3}\right)^k
\end{align}
We recall that we only focused on the transversal logical gates. An algorithm would, in principle, require other gates than cNOT, what we called single-qubit gates here (which were Hadamard or Pauli), identity, or measurement gates. Those other gates, such as the $T$-gate, have to be implemented with a completely different protocol which goes beyond the scope of this thesis. Taking them into account would, in principle, add additional terms in this power function. We also made an implicit assumption here: all the single-qubit physical gates will consume the same amount of power consumption (we did not make any distinction in our calculation between Hadamard or Pauli gates, for instance). 

All those precisions being said, we see that we were able to define the power at the logical level with some adaptations. We have to keep in mind that each term involved in this equation could vary from a coefficient roughly between $1$ and $10$ as we only estimated quantities in their order of magnitudes (but we believe that our estimations are probably closer to the $1$ than $10$). Now, we can recall the most important hypothesis that led us here. For the static consumption, we only took into account the ancillae qubits because they give the dominant contribution to the total number of physical qubits. We also assumed that no recycling is being performed for all the ancillae qubits that are inside a given logical gate and that any ancilla used in a logical gate can be reused only when \textit{all} the others ancillae of this logical gate have also "finished their job". We also assumed that the number of physical qubits is proportional to $Q_L$. This will be true for some algorithms and only an approximation for some others (it depends if the moment when the maximum number of logical qubits are involved at the same time corresponds to a moment some logical qubits are being measured). If this happens to not be true, this approximation should still be quite fair (see the discussions on this subject in section \ref{sec:estimation_physical_qubits}). For the dynamic costs, the main assumption that we did is to neglect the "variation" in the concatenated construction between the different kinds of logical gates. It allowed us to have a number of physical gates proportional to $(\overline{N}_{L,\parallel}^{\text{1qb}}+\overline{N}_{L,\parallel}^{\text{Id}}+2 \overline{N}_{L,\parallel}^{\text{cNOT}})$. And we also did not make any distinction in the power consumption for each type of physical single-qubit gate. Now, we believe that the most "critical" assumptions, the ones that would change by more than $100 \%$ some of the numbers in \eqref{eq:logical_power} are (i) the assumption that the number of physical qubits is proportional to $Q_L$ (the fact it a very good, or "rough" approximation is algorithmic dependent), (ii) the assumptions behind the qubits recycling (which can actually be seen as an assumption on how the computer works rather than an approximation of the physics). The other considerations should not change our results quantitatively. 

Now, we comment on the expressions we obtained. We notice here the functions $t$,$u$,$v$,$w$. They represent the increasing power cost of the logical elements when more and more error-correction is being performed (because more and more physical components are there). They are the "translation" of the energetic cost from the logical to the physical world (where the power is "really" being spent). In the section \ref{sec:formulation_problem_physical}, we said that we would neglect the dynamic cost related to measurement. This was a fair assumption to do because physical measurements only occurred on the boundary of the algorithm (and will thus have a negligible "weight" in terms of parallel operation compared to all the other gates occupying the "volume" of the algorithm). Here, as we are doing quantum error correction, physical measurements are done "very frequently in time" (thus inside the "volume") to detect errors. In principle, we can no longer neglect them\footnote{Actually as in practice measurement require much less dynamic power than single-qubit gate we will still be able to neglect them, but we wanted to emphasize on this difference anyway.}, this is why there is a cost per physical measurement gate represented by $b^{\text{Meas}}$ in \eqref{eq:logical_power}. We also notice that $t(k)/u(k)$, $t(k)/v(k)$ and $t(k)/w(k)$ diverge in the infinite concatenation level. This is based on our assumption that the ancillae qubits are not being recycled within a logical gate. Because of that, the static consumption will dominate in the $k \to \infty$ limit. In practice, we can expect that with smart recycling scenarios, those ratios would actually converge toward some constant. We can also see that the dynamic consumption is proportional to the total quantity $(\overline{N}_{L,\parallel}^{\text{1qb}}+\overline{N}_{L,\parallel}^{\text{Id}}+2 \overline{N}_{L,\parallel}^{\text{cNOT}})$. It means that the exact logical gate that is being implemented in the algorithm doesn't matter; what matters is the average number of logical qubits that are "actively" participating on an information processing level. It also allows us to notice that in the opposite of what usually happens without error correction, logical identity gates require power consumption. The reason is that it costs energy to protect the information because of the error-correction subroutines.
 
\section{Conclusion}
In this chapter, we proposed a formulation for the question of the resource cost of quantum computing. A resource is, in principle, any cost function, but here we used the power consumption as our resource of interest. This formulation is based on a minimization under constraint: we ask to minimize the power it costs to run an algorithm \textit{under the constraint} that the noise contained in its output, as quantified by some metric, is not higher than a given target chosen by the experimentalist. It allows to indirectly assess that the success of the algorithm will be good enough.

Formulating the question this way allows optimizing all the parameters the experimentalist can tune in order to reach this minimum. It includes parameters related to the hardware part of a quantum computer (qubit temperature, attenuation on the line), but also to the software part in principle (how the algorithm is implemented). For this reason, in the section \ref{sec:establishing_fullstack_framework_physical} we proposed a full-stack framework in which the engineering, algorithmic and quantum physics parts involved in the quantum computer are all represented. Solving the problem of minimization under constraint, the various competition phenomenon that might be occurring, possibly between the different fields of expertise involved in the design, are automatically taken into account, and the most energy-efficient architecture of the computer can, in principle, be found. We illustrated it with a concrete example in which the most energy-efficient implementation of an algorithm can be found in the section \ref{sec:energetic_cost_nisq_algo}. This optimal implementation could be found because we related "at the same time" algorithmic characteristics to "engineering" ones (here represented by the temperature $T_{\text{Q}}$ which plays a role in the noise but also in the cryogenic Carnot efficiency). This optimal implementation results from the intrinsic competition there exist between the noise and the power: both cannot be low at the same time, which is at the heart of the physics behind the optimization we propose.

In the section \ref{sec:min_pow_increases_with_accuracy}, we also saw that for some reasonable hypotheses (that we expect to be true for many physical systems, and that are at least true for the ones we consider), the minimum power consumption is an increasing function of the accuracy targetted for the algorithm: the lowest the noise in the algorithm answer is desired, the more power it will cost. We then showed in section \ref{sec:equivalence_minpow_maxacc} that under those same hypotheses, the questions of asking the maximum accuracy one could get for a given resource and asking what is the minimum power one has to spend in order to reach a given accuracy are mathematically equivalent.

Finally, in section \ref{sec:adapting_framework_FT}, we provided the tools that are required to phrase our approach in the context of fault-tolerant quantum computing, where both the power function and the metric quantifying the noise could be expressed on the logical level. We also did quantitative estimations of the number of physical qubits and gates required. The elements provided in this section, and more generally in this chapter, will be used in the next chapter, where we will finally try to estimate the energetic cost of a realistic full-stack model of a superconducting quantum computer.
\begin{appendices}
\chapter{Recycling of the ancillae on the logical level}
\label{app:recycling_ancilla_logical}
Here, we justify why multiplying \eqref{eq:number_physical_qubits_k_concat} by $4$ should give a good estimation of the number of ancillae required after $k$ concatenations. We recall that the hypotheses behind the calculation are that (i) the ancillae required to implement a logical gate are not recycled within this logical gate, and (ii) on the logical level, the recycling is done such that the physical ancillae inside a given logical gate can only be reused when \textit{all} those physical ancillae have finished their work (i.e., they have all been finally measured).

For this, we recall that the ancillae taking the longest time to be used are the $X$ and $Z$-syndrome ancillae. For $k=1$, they last for $9$ physical timesteps. From now on, we only reason with the $X$-syndrome ancillae as the reasoning for the $Z$ ones would be similar. In practice, $5$ physical timesteps are before the logical gate implemented, $3$ in parallel (a logical gate lasts for $3$ timesteps for $k=1$, there is the 1-Ga implemented then the cNOT for the X-syndrome and the cNOT for the Z-syndrome, see figure \ref{fig:steane_method_full} and \ref{sec:number_timesteps_for_logical_and_ancilla}), and $1$ after\footnote{For the X-syndrome we must do an $X$ measurement. But as we can only do $Z$ measurements by hypothesis, we must apply a Hadamard and then do a $Z$ measurement. Looking at the circuit on the figure \ref{fig:steane_method_full}, we would see that the final $Z$ measurement is done "after" the logical gate has finished being implemented on the data qubits.}  which means that the ancillae are overlapping with $4$ consecutive logical gates in principle. This is represented in the figure \ref{fig:logical_ancilla_recycling}, and this is the number we took in the main text.
\begin{figure}[h!]
\begin{center}
\includegraphics[width=1.0\textwidth]{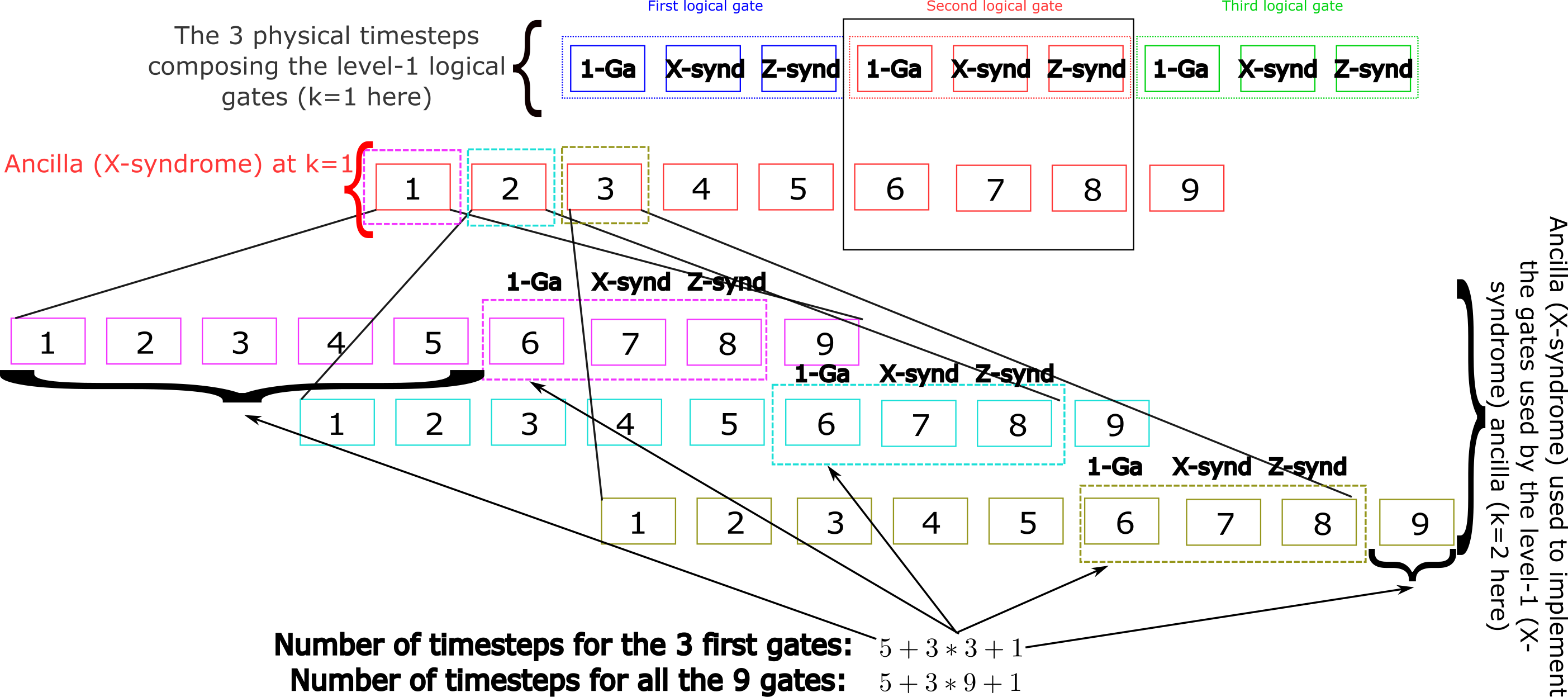}
\caption{Estimation of an upper bound on the number of physical ancillae required in an algorithm, taking into account the overlap of preparation with preceding and following logical gates. On the first line, we represented the three physical timesteps composing the different blue, red, and green logical gates (3 timesteps because $k=1$ at this point: the 1-Ga is implemented followed by the cNOT of the X-syndrome and the cNOT of the Z-syndrome, see the figure \ref{fig:steane_method_full}). The second line represents the $9$ physical timesteps required by the ancilla within the red logical gate of the first line. Those ancillae actually correspond to the $X$-syndrome ancillae, which are the ones lasting for the longest amount of time. $5$ timesteps are occurring before the gate and $1$ after (they are the reason for the potential overlap with the ancillae of the preceding and following logical gates). Concatenating once again, each physical gate of the second line will be replaced by a 1-Rec and will thus be error corrected, which requires ancillae. The first of those gates (the purple dotted box on the second line) will then need ancillae, represented in purple in the third line. They will be implemented $5$ timesteps before the transversal implementation of the gate (1-5 purple boxes, the transversal 1-Ga implementation is the 6'th box) and $1$ after the two syndromes have finished (9'th purple box). The second of those gates (cyan) also need ancillae when protected by error-correction, which will be prepared $5$ timesteps before the transversal implementation (the cyan 6'th box), and $1$ timestep after the two syndromes have finished. Same principle for the yellow, and in principle for the rest of the level-1 gates acting on the ancillae (the red gates from 4 to 9 on the second line). The number we need to estimate is the number of timesteps between the moment when the ancillae used to protect the first gate of the second line is "injected" in the calculation and the moment when the ancillae used to protect the last gate of the second line has finished working. More precisely, we need this analog value for further concatenations. For $k=2$, the number is equal to $5+3*9+1$ as shown in the image, and for an arbitrary $k$ it is expressed in \eqref{eq:N_timestep_X_syndrome_result}.}
\label{fig:logical_ancilla_recycling}
\end{center}
\end{figure}
Now, our goal is to justify that this $\times 4$ multiplication is still a good estimation, even for further concatenation levels. If we imagine concatenating another time, each of the physical gates applied on the ancillae on those $9$ physical timesteps will be protected, which means replaced by a 1-Rec, and will involve new physical $X$-syndrome ancillae lasting for $9$ physical timestep. This is represented on the figure \ref{fig:logical_ancilla_recycling}. For instance, the purple dotted box of the second line corresponds to a gate acting on the ancilla, on the first timestep in which it is being manipulated. This gate, when protected by error-correction, will also contain $X$-syndrome ancillae (represented by the third line containing $9$ purple boxes), which will also be prepared $5$ timesteps before the gate protected is implemented and will be measured $1$ timestep after the protected gate has finished its implementation. Analog explanations for the dotted cyan, yellow gates acting on the level-1 ancillae, and it would be the same explanation for the level-1 gates between 4 and 9.

In the end, we can estimate that the number of timesteps required by the physical ancillae satisfies, after $k$ level of concatenations:
\begin{align}
&N^{(k)}_{\text{Timestep}}=3 N^{(k-1)}_{\text{Timestep}}+6
\label{eq:N_timestep_X_syndrome} \\
&N^{(1)}_{\text{Timestep}}=9
\end{align}
For instance, for $k=2$, the total number of timesteps for the ancillae is $3$ times the number of timesteps the ancillae used on the $k=1$ level, on which we must add "boundary conditions": the ancillae used to protect the first gate of the $k=1$ level will last for $5$ timesteps before this gate, and the ancillae used to protect the last gate of the $k=1$ for $1$ timestep after. It gives the additional $6=5+1$. The reasoning remains when concatenating on further levels, and we thus have \eqref{eq:N_timestep_X_syndrome} that is satisfied. Solving this recursive relation, we obtain:
\begin{align}
N^{(k)}_{\text{Timestep}}=3^{k}*4-3
\label{eq:N_timestep_X_syndrome_result}
\end{align}
Using the fact that a logical gate concatenated $k$ times lasts for $3^k$ physical timestep, we obtain that the X-syndrome ancilla of a given logical gate will overlap with $\approx N^{(k)}_{\text{Timestep}}/3^k \approx 4$ consecutive logical gates.
\chapter{Energetic cost of the rejected ancillae}
\label{app:ancilla_rejected}
Here, we justify why we neglect the number of rejected ancillae and the associated gates applied on them. They are rejected in the case the verification failed, see figure \ref{fig:steane_method_full}. To say that we neglect the number of rejected ancillae means for us that we can do the energetic estimations assuming that the ancillae prepared are always accepted by the verifiers and that no ancillae are required in case of a failure event. We recall that as one syndrome ancilla must necessarily "pass" the verifier test at the moment it is needed, a given number of syndrome and verifier ancillae must be prepared \textit{in parallel} to be sure that a syndrome ancilla will be accepted "on the appropriate moment"\footnote{What we mean is that if it happens that one syndrome ancilla is being rejected, because qubits have finite lifetime we cannot wait the time to prepare another syndrome and verifier ancillae and to make them pass the test: we need to be sure to have a syndrome ancilla ready to perform the syndrome measurement right when needed, which could induce a large overhead in principle.}. It could, in principle, increase the cost of the physical resources required in a significant manner. 

We will first show that the \textit{average} number of the extra ancillae needed in the case the first verification fails counts for a negligible part in the total counting of ancillae. The intuitive reason is that the ancillae will be more likely to be accepted than rejected. Because of that, the number of extra ancillae required in case of one or multiple rejects is expected to be lower than the number of ancillae anyway required (i.e., the number of ancillae we would need under the assumption the verification always succeeds). But reasoning with an average is not enough as we cannot put an average number of ancillae in a computer. Then, we will give the reasons that allow us to think that because different logical qubits can share a common reservoir of extra ancillae, the "real" number of ancillae physically present in the computer in the case the verification fails can also be neglected for our purpose of order of magnitude estimations.
\section{The average number of extra ancillae required in case of rejection is negligible}
Let us call $p$ the probability to reject a syndrome ancilla. The probability to have an ancilla accepted on the $n+1$'th try is: $p^n(1-p)$. 

We call "syndrome-spot" a place where an ancilla is either accepted or rejected. At each syndrome-spot, the average number of extra ancillae required in case of rejection is:
\begin{align}
\overline{N}_{\text{extra}}=\sum_{n=1}^{+\infty} n \times p^n(1-p)
\end{align}
Indeed, if the ancilla is accepted after one failure (probability $p(1-p)$ that it occurs), only $1$ additional ancilla would have been needed, if it has been accepted after two tries (probability $p^2(1-p)$ that it occurs), $2$ additional ancillae would have been needed etc. Using the fact $\sum_{n=0}^{+\infty} n p^n = p/(1-p)^2$, we deduce:
\begin{align}
\overline{N}_{\text{extra}}=\frac{p}{1-p} \leq 1 \text{ for } p \leq 1/2
\end{align}
Thus, at each syndrome-spot, there are fewer extra ancillae required in case of rejection than the number of ancillae that must be present under the assumption the verification always succeeds. Actually, $p$ is expected to be very small in practice (because physical failure rates are below the threshold). A pessimistic estimation for $p$ can then be taken equal to $\approx 10 \eta_{\text{thr}} \approx 10^{-3}$, which makes $\overline{N}_{\text{extra}} \approx 10^{-3} \ll 1$. Indeed, there are about $\approx 10$ physical gates involved in the ancilla preparation. Assuming that the failure probability of those gates is about $\eta_{\text{thr}}$ (worst-case scenario), we would have $p \approx 10 \eta_{\text{thr}}$ being the probability that verification fails for the physical ancilla. An implicit assumption behind this calculation is that we only did one concatenation level (because there are different levels to consider for the ancillae when $k>1$). We will consider a more general case (and we will do a more accurate estimation) in what follows.
\section{The overhead required in case ancillae are rejected can be neglected}
Up to this point, we reasoned with a statistical average approach. However, we do not have an "average" number of ancillae in the computer but a fixed value of them. Reasoning with an average is in principle not enough because what could happen is that even if it is "very unlikely" to need, for instance, $10$ extra ancillae on a given syndrome spot, those ancillae might anyway be here "in the case" this verification fails, even if it is a rare event. We need a more accurate estimation to be sure that we can neglect the number of extra ancillae here in case the verification fails. This is the goal of what follows. One of the strategies behind what we present is to use a common reservoir of extra ancillae for multiple logical qubits. Because it is unlikely that \textit{all} the logical qubits would have their ancillae rejected at the same time, and given the typical value of the probability of error of a logical gate in all the examples we will consider, the number of ancillae in this reservoir is expected to be small compared to the number of ancillae that would strictly be required by the logical qubits sharing this reservoir, in the case the verification always succeeds. We would, however, like to emphasize on the fact that this part is still a work in progress (we are not sure if what we present is entirely rigorous, we need some further verifications), but we hope that the reader will be convinced enough by what is already written.

For our explanations, we first consider only doing one level of concatenation; we generalize after. We consider that a group of $Q_L^0 \leq Q_L$ logical qubits are \textit{sharing} $M$ \textit{extra} ancillae that constitute a reservoir used only in the rare cases the verifications are failing. We assume that $Q_L^0 p \ll 1$ (this is why $Q_L^0 \leq Q_L$, we took it in order to be able to reason in a perturbative manner in what follows). The number of syndrome spots inside those $Q_L^0$ logical qubits being roughly equal to the number of logical qubits (up to a factor of $2$), we will assume that they are identical in what follows. If we have $M < Q_L^0$, then we can, for our order of magnitude estimations, neglect those $M$ extra ancillae against the number of ancillae that are needed anyway in case the verification always succeeds. It would allow the calculations made in the main text to be valid. Our goal is then to justify that $M < Q_L^0$ which requires finding $M$.

In order to find it, what we can do is to calculate the value of $M$ such that the probability that we actually needed $M+1$ ancillae (i.e., we do not have enough ancillae "in reserve") is lower than the probability $p_L$ that a logical gate fails. If it is the case, it would mean that the dominant reason why a logical gate fails is not because of a lack of ancilla, and then the calculation will succeed often enough with only having those $M$ extra ancillae in the reservoir.

With $1$ extra ancilla, the calculation can resist to \textit{one} logical gate applied on one of those $Q_L^0$ logical qubits that failed because it didn't have enough ancillae (this extra ancilla will be used by this logical gate). Thus, having $M=1$ extra ancilla makes the probability that a logical gate fails \textit{because not enough ancillae were in the reservoir} being equal\footnote{Actually, dominated by (because it cannot resist to more errors but those events have a much smaller probability of occurring).} to the probability to either have two different logical gates applied on two different logical qubits that failed because they didn't have enough ancillae, either one logical gate applied on a logical qubit that failed twice (i.e., had two ancillae rejected) which is equal to $Q_L^0 p^2+ \binom{Q_L^0}{2} p^2 \approx (Q_L^0 p)^2$. If $2$ extra ancillae are being used, we can resist to two logical gates that failed because they didn't have enough ancillae\footnote{Or to one logical gate that failed because it had two ancillae rejected.}. Thus, the probability that a gate fails because it lacks ancilla is of order $(Q_L^0 p)^3$. In conclusion, with $M$ \textit{extra} ancillae, a logical gate applied on any one of those $Q_L^0$ logical qubits will fail because not enough ancillae are in the computer with a probability that is about $(Q_L^0 p)^{M+1}$. We need to make sure that this probability is lower than the probability of error of the logical gate that we aim, thus:
\begin{align}
(Q_L^0 p)^{M+1} < p_L \Leftrightarrow M > \frac{\ln(p_L)}{\ln(Q_L^0 p)}-1
\end{align}
Considering $p_L=10^{-12}$ (this is the lowest probability of error we will need to have in all the next chapter), for $Q_L^0=100$ (it allows us to have $Q_L^0 p = 10^{-1}$ which we consider being valid\footnote{Again this part is constituted of quick estimations that should be reinforced.} for our perturbative approach), we have, $M>11$. Thus $M<Q_L^0$, the extra ancillae do not dominate the number of ancillae that would be required under the assumption the verification always succeeds. The reason why is (i) because the probability of errors of the logical gate we want is not "too low", (ii) different logical qubits are sharing the same reservoir of extra ancillae, (iii) the probability that an ancilla is being rejected is actually low.

Now, this is only true for the first concatenation level. We would like to generalize to more levels. In order to do so, we can apply this same reasoning recursively. There is about $\epsilon=M/Q_L^0 \approx 0.1$ (taking $M=11$) extra ancillae required per ancilla needed in the case the verification always succeeds. The number of extra ancillae required in the reservoir after $k$ concatenation for $k=3$ (we take this example to illustrate the calculation) is equal to the number of extra ancillae required assuming the verification failed on either the first, second, or third concatenation level (or if it failed on multiple levels at the same time). Following the tree given in figure \ref{fig:neglecting_extra_ancilla_tree}, we can estimate this number of extra ancillae being equal to: $\epsilon^3+3\epsilon^2 + 3\epsilon \approx 0.3 < 1$. Thus for one ancilla required under the assumption the verification always succeeds, there is $0.3$ extra ancillae required for the "real" scenario. Thus, we can neglect this overhead in our calculations (both from the number of ancillae and the gates that would be applied to them). Considering that we will never do more than $5$ concatenations in the next chapter, the same reasoning can, in principle, be applied to show that we can neglect this overhead even for $5$ concatenations level (the extra number of ancillae would be $\approx 5 \epsilon < 1$). We recall again that this part needs to be re-checked properly (we did not have the time to do all the verifications at the moment this thesis was written).
\begin{figure}[h!]
\begin{center}
\includegraphics[width=0.7\textwidth]{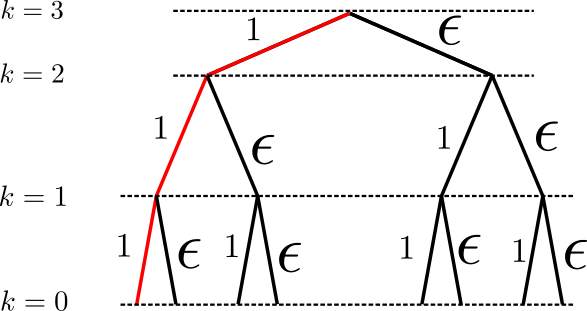}
\caption{Why we can neglect the rejected ancillae even up to $k=3$. The variable $\epsilon=M/Q_L^0 \approx 0.1$ represents the number of extra ancillae per ancilla "anyway required" (i.e the number of extra ancillae here in case some verifications are failing, per ancilla that are here under the assumption the verification always succeeds). The red line represents the situation in which the ancillae always pass the verification (what we considered in our calculations). At each level, there is a probability that the ancilla fail the verification. In this case, about $\epsilon \approx 0.1$ extra ancillae are required.}
\label{fig:neglecting_extra_ancilla_tree}
\end{center}
\end{figure}
\end{appendices}

\chapter{Full-stack approach to the energetic cost of quantum computing: application to superconducting qubit quantum computer}
In the previous chapter, we introduced the principle behind the energetic estimation of fault-tolerant quantum computing. We used examples to illustrate the method but we stayed on the level of toy models. In this chapter, we are going to go toward some more realistic models with the hope to get some orders of magnitude of power consumption we might expect in a fault-tolerant quantum computer\footnote{In this chapter, I made the major part of the work presented.}. 

The quantum computer we will model will be based on superconducting qubits. Such quantum computers are the ones that are currently being developed by Google (we can think about the Sycamore processor which has been used in a recent claim of quantum supremacy \cite{arute2019quantum} and Bristlecone \cite{villalonga2019flexible}), IBM that actually has a large number of quantum processors \cite{tornow2020non,garcia2017five}. Many other companies are also working (or are going to work) with superconducting qubits quantum computers such as Intel \cite{intelquantumcomputing}, Rigetti \cite{motta2020determining}, Alice \& Bob \cite{aliceandbob} but also the academic sector. QuTech has, for instance, two small-sized superconducting quantum computers: Starmon 5 and Spin 2 \cite{starmon5,spin2}. Among all those examples, we can say that our models are closer to the quantum processors used by IBM. Indeed, this company mainly uses the cross-resonance scheme \cite{sheldon2016procedure,malekakhlagh2020first,chow2011simple,kirchhoff2018optimized} in order to implement the two-qubit gates. This model of gates allows implementing two-qubit gates between qubits that are not necessarily nearest-neighbor (this is a requirement for the concatenated code construction) by connecting them with a quantum bus. Also, those gates can last for $160ns$  \cite{sheldon2016procedure}, which is close to the $100ns$ of duration we took for the two-qubit gates in our models. And finally, those gates can work with fixed-frequency qubits which is what we are going to consider as well. See \ref{sec:state_of_the_art_2qb} to see how we model such gates here.

In this chapter, we will develop a complete full-stack model for the quantum computer. It will include the classical electronics that is necessary to either generate signals, communicate with the laboratory, or amplify the measurement signals. We will also include the heat conduction of all the necessary cables. We will keep considering the heat dissipated in the attenuators. The model is detailed in the sections \ref{sec:hardware_model} and \ref{sec:noise_quantity}.

The minimization under constraint described in \ref{sec:separate_role_hardware_software_noise} will allow us to understand which characteristics of a quantum computer play the most important role in its power consumption. Between the temperature at which the qubits are, the total attenuation there is on the lines, and the temperature at which the signals are generated, what is the most critical parameter to optimize? Is it really necessary to put the qubits at very low temperature (i.e., close to $10mK$) in order to run a fault-tolerant calculation, and what is a low enough temperature quantitatively? Such questions will be studied in the sections \ref{sec:quantitative_interest} and \ref{sec:charac_algo}. In those two sections and in the section \ref{sec:making_energy_efficient}, we will also see which elements in the computer are participating in the most significant manner in the energetic cost: is it the dissipation in the attenuators, the electronics generating the signals or the heat conduction? It will allow us to establish some first strategies that seem promising in order to make quantum computing more energy efficient. In this work, we will also estimate quantitatively how much power consumption we can hope to save with our optimization compared to a non-optimized scenario in which "typical" values are provided. We will see that in regimes of high power consumption (i.e., bigger than the megawatt), our optimization can allow saving more than two orders of magnitude in power consumption (we discuss this in the section \ref{sec:quantitative_interest}).

Finally, in section \ref{sec:charac_algo} we will see how the shape of the algorithm (i.e., its logical depth and the number of logical qubits it contains) has an influence on the power consumption.
\section{Hardware engineering model}
\label{sec:hardware_model}
In the section \ref{sec:expression_power_ch4} of the previous chapter, we gave the generic expression for the power consumption that we recall:
\begin{align}
&P=Q_L t(k) a + \left( b^{\text{1qb}} u(k) + b^{\text{cNOT}} v(k)+ b^{\text{Meas}} w(k)\right) (\overline{N}_{L,\parallel}^{\text{1qb}}+\overline{N}_{L,\parallel}^{\text{Id}}+2 \overline{N}_{L,\parallel}^{\text{cNOT}}) \label{eq:logical_power_ch5}\\
& t(k) \equiv \frac{112}{199} 199^k\\
& u(k) \equiv \frac{7}{48}\left(\frac{199}{3}\right)^k\\
& v(k) \equiv \frac{1}{3} \left(\frac{199}{3}\right)^k\\
& w(k) \equiv \frac{7}{48} \left(\frac{199}{3}\right)^k
\end{align}
The coefficients $a$, $b^{\text{1qb}}$, $b^{\text{cNOT}}$, $b^{\text{Meas}}$ are characterizing the engineering aspects behind the quantum computer. They will describe how the power will depend on the number of \textit{physical} qubits (for $a$), on the number of \textit{physical} single, cNOT and measurement gates active in parallel (for $b^{\text{1qb}}$, $b^{\text{cNOT}}$, $b^{\text{Meas}}$). The coefficients $t(k), u(k), v(k)$ and $w(k)$ describe how the power cost will grow as a function of the quantity of error correction (quantified by the concatenation level $k$) that is being done. $t(k)$ represents how the number of physical qubits increases with $k$, and the other coefficients represent how the number of physical single, cNOT, and measurement gates increase as a function of $k$. Using the terminology introduced in the section \ref{sec:formulation_problem_physical}, $t(k)$ is scaling the static costs, i.e., the power that has to be spent whatever the algorithm implemented is doing. The coefficients $u(k), v(k)$ and $w(k)$ are scaling the dynamic costs, which depend on the logical gates implemented in the algorithm. The dynamic costs will only consume power when quantum gates are implemented on the qubits. We see that because of the error correction performed, the ratio of static over dynamic power consumption is fixed by the level of concatenation. This difference between dynamic and static power is something that has to be taken into account as it can significantly impact the energetic bill. For instance for $k=3$,  $t(k)/u(k) \approx 100$: there are two orders of magnitudes more physical qubits than physical single-qubit gates!\footnote{We implicitly assume $Q_L \approx (\overline{N}_{L,\parallel}^{\text{1qb}}+\overline{N}_{L,\parallel}^{\text{Id}}+2 \overline{N}_{L,\parallel}^{\text{cNOT}})$ in this statement which means that all the logical qubits are "doing something" for all the timesteps.} But for $k=1$, $t(k)/v(k) \approx 5$: there are only $5$ more physical qubits than two-qubit cNOT gates for one level of concatenation, we see that the conclusions we can extract are very dependent on the regime of parameters we are looking at. What this discussion illustrates is that in order to properly study the energetic cost of quantum computing to then make it energy-efficient, it is important to characterize both the static and dynamic costs as they can be significantly different.

Now, the only thing we need to determine from the engineering model are the values of the coefficients $a$, $b^{\text{1qb}}$, $b^{\text{cNOT}}$, $b^{\text{Meas}}$, and this is what we are going to estimate now.
\subsection{The global picture}
\label{sec:global_picture}
In this chapter, we are going to follow the same philosophy as what we did in the previous chapter: we will only consider the energetic cost required to remove the heat introduced within the cryostat; we call it the cryogenic cost here. Thus, any energy spent outside (i.e., at $300K$) will be considered as "free" in our estimation. The reason why we focus on the cryogenic cost is because any power that is spent inside the cryostat is usually transformed into heat that has to be evacuated. Then, because removing $1W$ of heat will cost more than $1W$ of work (i.e., electricity) at cryogenic temperatures (we can think about Carnot efficiency to understand that), focusing on the cryogenic cost would give the dominant power consumption\footnote{Of course if some elements at $300K$ consume a lot of power, they might dominate the cryogenic cost. But in practice, we believe that the dominant power consumption will anyway be related to cryogenics.}. Now, many different architectures of quantum computers could exist; we need to make some choices, we cannot be entirely general. One thing to have in mind is that a quantum computer should be seen as a hybrid quantum/classical computer. Indeed, classical electronics is required to generate the signals that will reach the qubit but also to manage the execution of the quantum algorithm. For instance, it has to keep in memory all the sequence of gates that have to be applied on the physical qubits. It must also perform some calculations that allow deducing the value of the syndrome from the measurement outcomes. One open question is to know where it is optimal to put the electronics, in terms of simplicity for the design but also in terms of energy efficiency \cite{bardin201929,patra2020scalable}. Two extreme scenarios can be considered. First, we could imagine merging all the classical electronics with the quantum core. In this scenario, the engineering would be much more simple (as it would remove a great number of cables). It would also remove the major part of heat conduction because we would only need one cable between $300K$ and the quantum core giving the instruction about which algorithm to perform, another cable extracting the measurement outcomes provided at the very end of the algorithm, and one cable providing the power to the electronics\footnote{This is probably a little bit extreme but it is to give the idea: the number of cables required would be very small.}. However, classical electronics dissipates a large amount of heat, and if we assume putting it at the qubit temperature (thus in the $10-100mK$ range), the quantity of heat to remove would be phenomenal: it would be an inefficient architecture. The other extreme is to put all the electronics at $300K$. In this case, we face the opposite problem: we would need to bring a large number of cables from the laboratory to the quantum core. They would bring a very large amount of heat to evacuate, and it might not be very energy efficient either. All this explains why there are strategies considered as being better for the scalability, which consists in putting the electronics that is generating the signals at some intermediate temperature between $300K$ and $T_{\text{Q}}$ (the qubit temperature). We call this temperature $T_{\text{Gen}}$ and it is typically considered being around $4K$ \cite{bardin201929,patra2020scalable}; the aspects of information processing being kept at $300K$. Doing this strategy, a large number of cables would be required between $T_{\text{Gen}}$ and $T_{\text{Q}}$, but the associated temperature gradient being lower, the heat flow would be less a problem\footnote{Typical heat conduction models involve a heat flow proportional to $T_a^i-T_b^i$ where $T_a$ and $T_b$ are the temperature of the two ends of the cable, and $i$ some positive power. Thus a low temperature gradient involves a low amount of heat conduction.}. Furthermore, as soon as the cable is below $\approx 10K$, superconducting microwave cables can be used, which are associated with a very low heat conduction flow \cite{krinner2019engineering,mcdermott2018quantum}. On the other hand, a much smaller amount of cables would be required between $300K$ and $T_{\text{Gen}}$ because those wires would only bring digitized information about which signal should be generated for which qubit (and those cables can be optical fibers which are insulators and thus also associated to a low heat flow). 

Outside of heat conduction and signal generation, an important question to answer is what is the optimal qubit temperature and attenuation level that should be chosen in the architecture. Choosing an attenuation close to $T_{\text{Gen}}/T_{\text{Q}}$ \cite{johnson2012optimization,reed2014entanglement} and putting the qubits at $T_{\text{Q}}=10mK$ is a choice that is usually considered. But is it necessarily the best one in terms of power consumption? If it happened to be possible to put the qubit at some higher temperature, would we be able to save a large amount of power consumption? All those questions are important to answer in order to know how the computer should be designed for the very goal of reducing power consumption, but outside of pure energetic considerations, it is also important to have access to this information as it can be used for other practical reasons. To give an example, different cooling technologies are associated with different temperature ranges \cite{weisend2016cryostat,enss2011tieftemperaturphysik}, and having the information that putting the qubits close to $10mK$ is not strictly necessary can be important information for the design.

In summary, here, we will consider a quantum computer where the electronics generating the signals will be \textit{inside} the cryostat, at a temperature $T_{\text{Gen}}$. Our goal will then be to optimize this temperature along with the qubit temperature $T_{\text{Q}}$, the total attenuation $A$, and the concatenation level $k$ in order to minimize the power consumption through our formulation of minimization under constraint. We will also try to understand what are the most important characteristics to optimize in the quantum computer and which elements are contributing to the energetic cost in the most dominant manner.
\subsection{The architecture we consider}
Now, we give more precisely the elements we are going to consider in our energetic estimations. As explained before, as we reason with the cryogenic cost, we have to know how much heat they dissipate, which will tell us how much electrical power will be required to evacuate that heat. We will then perform some approximations that will allow us to neglect some of those costs. The architecture of the computer we consider is represented in figure \ref{fig:quantum_computer_archi_cable_with_lab}, and we now list the sources of heat dissipation it contains. We should add that this list is extrapolated from typical quantum computer architectures (and what is planned for the near future) behind superconducting qubit platforms, such as the computers from Google or IBM, but many of the sources of heat dissipation given below will have analogs in other platforms.
\begin{description}
\item [Heat dissipation from signal attenuation:] We introduced it in section \ref{sec:single_qubit_example}. The signals must be attenuated in order to remove the thermal noise which introduces heat.
\item [Heat conduction in all the cables connecting to the qubits:] All the cables within the quantum computer are conducting heat. The associated heat load must also be evacuated by the cryostat. In our model, we use conventional coaxial microwave cables for $T>10K$ and superconducting cables (microstrip lines to be precise) below as they are associated with a lower heat load; the precise model is described in appendix \ref{app:cable_models}.
\item [Heat dissipation from signal generation and digitization:] Generating signals has an energetic cost that is usually higher than the energy contained in the generated signal. If the signals are generated inside of the cryogenics, the heat the electronics dissipates has to be evacuated. It corresponds to the DAC represented on figure \ref{fig:quantum_computer_archi_cable_with_lab}. We must also digitize the readout signals in order to read them in the laboratory. It corresponds to the ADC represented on this same figure.
\item [Heat dissipation from signal amplification:] The typical amplitude of the signals containing the measurement results about the state of the qubits is low such that the thermal noise might be an issue when those signals are read. In order for the signals to remain with a high signal over noise ratio, they must be amplified at low temperatures. This process dissipates heat.
\item [Heat dissipation for multiplexing and demultiplexing:] To reduce the heat introduced by the thermal conduction, one typical strategy is to perform multiplexing and demultiplexing. It consists in putting information that should be addressed to different components into a single cable. For digital signals, it requires some electronics that will "encode" (multiplex) data so that a lot of information can be sent into one cable. It also requires after to "decode" (demultiplex) this data. The electronics performing those operations will cost energy and thus dissipates heat. It corresponds to DEMUX-MUX on the figure \ref{fig:quantum_computer_archi_cable_with_lab}.
\item [Heat dissipation from Joule effect in DC cables:] The classical electronics that is put inside of the cryogenic requires a DC power, provided by some cables, in order to work. If a large amount of power has to be provided, the Joule effect induced by the resistance of those cables might introduce heat that we will need to evacuate.
\end{description}
\begin{figure}[h!]
\begin{center}
\includegraphics[width=1\textwidth]{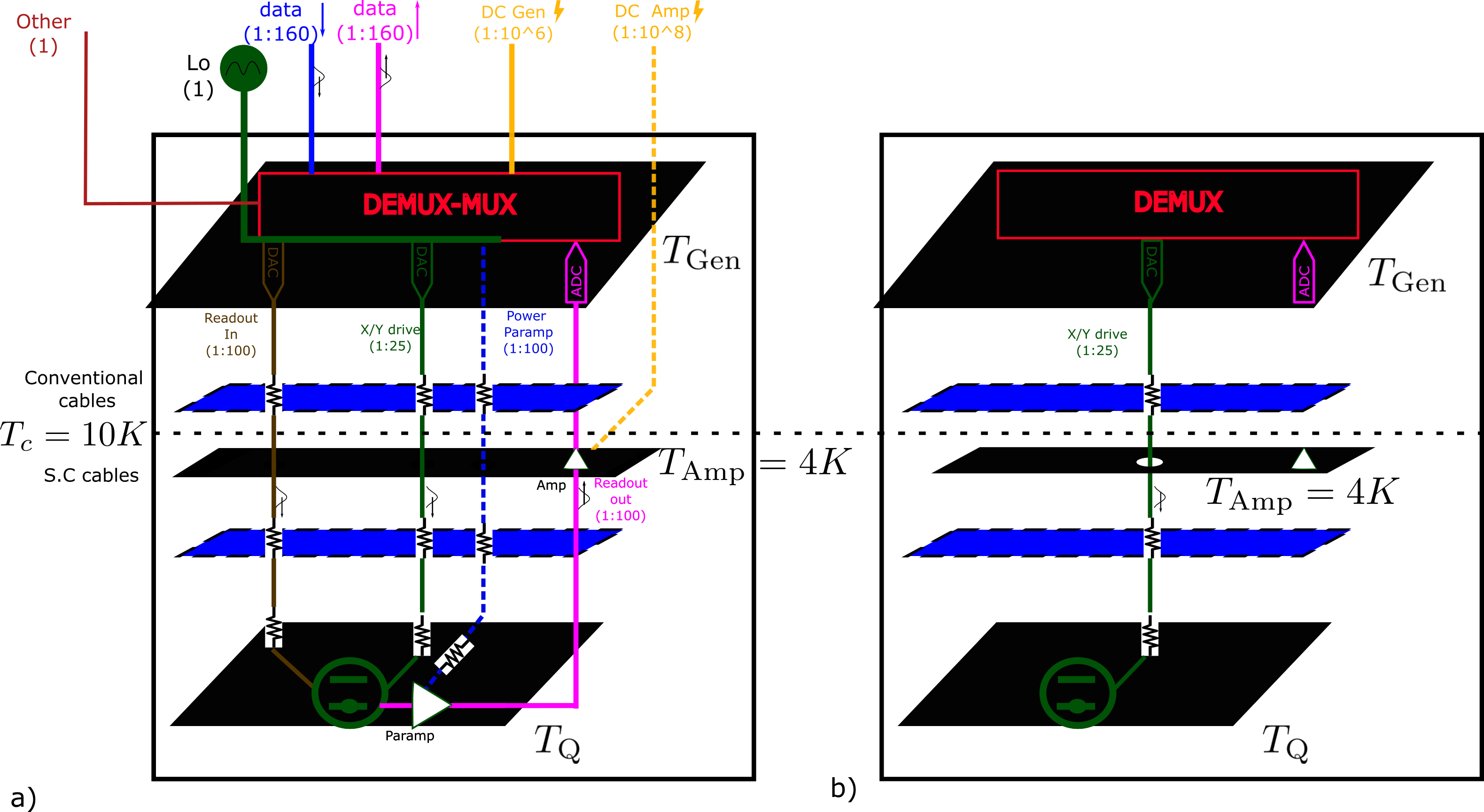}
\caption{\textbf{a)} Architecture of the computer our engineering model is based on. \textbf{b)} The simplified model where we retain only the dominant sources of heat. The numbers (1:X) on the cables give the level of multiplexing: it means that there are X cables per physical qubits. If (1) is written below a cable, it means that the associated cables do not scale with the number of physical qubits (and the associated number is typically low). Even though only one DAC and ADC per cable are represented, there is, in principle, one per physical qubit. Below $10K$, superconducting cables, which are associated with a lower heat flow, can be used. In our model, we then "switch" the material of the cable when the temperature goes below $10K$ (a cable between $300K$ and $5K$ would then be composed of a first portion in $[300K,10K]$ using a conventional metallic conductor and a portion in $[10K,5K]$ composed of a superconducting material, see appendix \ref{app:cable_models}. The hole in b) on the amplification stage represents the fact that the X/Y cables are not thermalized on this stage (to simplify our model as explained in the main text).}
\label{fig:quantum_computer_archi_cable_with_lab}
\end{center}
\end{figure}
\FloatBarrier
The left part of the figure \ref{fig:quantum_computer_archi_cable_with_lab} 
contains all the elements there are inside the computer that are dissipating heat; they exactly correspond to the elements we listed. The right part is what we modeled in our calculation. For reasons we are going to give, we think that the elements we removed from the left graph (to build the right one) can be neglected in the order of magnitude calculation that is our goal here. Here, we describe the architecture before simplification (i.e., the left part of the figure \ref{fig:quantum_computer_archi_cable_with_lab}), we will talk about the approximations we made after.

\subsubsection{Signal generation stage}
We first focus on the stage of temperature $T_{\text{Gen}}$ where signals are being generated. This stage is composed of various electronics components. For the quantitative values we are going to consider, we will consider that they are based on CMOS technology where experiments have shown that it is possible to put them at low temperatures \cite{patra2020scalable,le202019}. We first recognize elements that are called DAC (for Digital to Analog Circuits). Even though a unique DAC is represented here, one per physical qubit will be necessary. The role of a DAC is to generate the pulses that are going to drive the physical qubits. The way it works is that it receives digital information describing the shape of the pulse to generate. This information will be received as a bit-string describing a discretized approximation of the exact envelope we wish to generate. This approximated discretized version of the signal is generated by the electronics and filtered to obtain a "smooth" envelope. Finally, this envelope is multiplied with a signal being at the qubit frequency: the local oscillator, and sent to the qubits, the pulse necessary to drive the quantum gate has now been entirely created. To be a little bit more precise, two envelopes are actually generated; they are then multiplied by the local oscillator and a $\pi/2$ dephased version of the local oscillator. Those two envelopes allow making possible to control the amplitude of the signal but also its phase, which is necessary for single-qubit operations where the phase of the signal plays a role in the axis of rotation of the gate implemented. Also, as all qubits are not necessarily at the same frequency, the DAC can also use the local oscillator and "shift" its frequency in order to create a signal having the appropriate frequency of the qubit to be driven: it has been shown to be possible in a recent proposal \cite{patra2020scalable} on which our order of magnitude of heat dissipation will be based on. In the end, we understand why it is called a DAC: from a digitized description of the signal, it generates an analog one corresponding to the pulse that is going to drive the quantum gates. The ADC (Analog to Digital Circuit) does the opposite thing. It will "convert" the analog signals coming from the measurement outcomes of the qubits into digital ones. Indeed, the information about the qubit state after a measurement will be contained in a propagating analog pulse. The ADC is able to interpret the characteristics of this pulse and convert it into digital information. There is also one ADC per physical qubit that is required. 

We also recognize one element that we call DEMUX-MUX (which means demultiplexer-multiplexer). Its role is to (i) demultiplex the instructions coming from the laboratory about which gate has to be performed on which qubit and to transmit this instruction to the components that will generate the signals: the DAC. It is also (ii) multiplexing the data received by the multiple ADC. The goal of multiplexing and demultiplexing is to reduce the number of cables required between the laboratory and the signal generation stage. 

The digital data transiting between the laboratory and the stage of temperature $T_{\text{Gen}}$ are propagating into optical fibers represented by the blue and pink cables on the image. 

In addition to those optical fibers, we need wires containing DC voltages that power the electronics at $T_{\text{Gen}}$. They correspond to the yellow cable of the figure.

There is also one cable that will contain the local oscillator signal. This signal is monochromatic, at the "central" frequency of the qubits. What we mean by central frequency is that the qubits might not all be at the exact same frequencies: this is actually something to avoid in practice to avoid issues such as crosstalk \cite{theis2016simultaneous,schutjens2013single}. Transmons qubits can, for instance, be in the range of frequency $[3GHz,6GHz]$ \cite{krantz2019quantum}. In this case, the central frequency would be about $4.5 GHz$, which would correspond to the local oscillator frequency. The DAC would use this oscillator signal in order to generate the signals at the appropriate frequencies for all the qubits in the considered bandwidth.

Finally, a small number of cables (that does not scale with the number of physical qubits) is being represented under the name "Other". One example would be the clock signal required by the electronics \cite{bardin201929}.
\subsubsection{Measurement amplification}
We also recognize two amplification stages in black colors: one is done at $T_{\text{Q}}$ with parametric amplifiers \cite{aumentado2020superconducting}, while another one is done at $T_{\text{Amp}}$ with HEMT amplifiers\cite{korolev2011note} based on classical electronics (we will call them "classical amplifiers" in what follows). Classical amplifiers also require a DC power cable represented in yellow. The parametric amplifiers require a microwave signal. Those kinds of amplifiers amplify the signal while introducing a minimal amount of noise during the process, at the limit of the fundamental quantum noise \cite{krantz2019quantum}. They are made of superconducting components and do not dissipate heat by themselves. However, the pump signal they require needs to be attenuated and introduces heat dissipation. Then, we can notice all the microwave cables: some ("Readout In") are used to inject the readout signal, which is interacting with the qubit to be measured and then leaves the quantum core ("Readout Out"). It is on this "leaving part" that the signals are being amplified. 

\subsubsection{Driving cables}
We also have the microwave cables here to inject the resonant pulses that will drive the single-qubit gate ("X/Y drive"). In our approach, we assumed that the two-qubit gates would also be performed by injecting appropriate signals in the X/Y drive as we consider the cross resonance gate scheme. We refer to \ref{sec:state_of_the_art_2qb} and the introduction of this chapter as we already explained there the principle of such gates. We notice the absence of "Z" line, which corresponds to the line where signals are usually sent to change the frequencies of the qubits. Because we assume using the cross-resonance scheme (or a similar one that would only require to send signals on the X/Y lines to perform two-qubit gates), we can consider that all the qubits inside the computer are at fixed frequencies so that we don't need such line.

\subsubsection{Additional cryogenic stages}
Finally, we notice that in addition to the three black stages doing signal generation, containing the qubits, and amplifying the signals (thus the stages "fundamentally" required for the computer to work), we added extra stages in blue. They are here to help evacuate in a more progressive manner the heat dissipated into attenuators and the heat conduction. What we mean is that if we considered an extreme scenario where we would only have the qubits and the signal generation stage (let us forget about amplification), there would have a big attenuator at $T_{\text{Q}}$ dissipating a lot of heat there. The heat conduction would also have to entirely be evacuated at $T_{\text{Q}}$. It would be very energy inefficient as the total quantity of heat to evacuate would entirely be removed at a temperature having a very poor Carnot efficiency. It explains why we added those "extra" blue stages: they will help to "benefit" from higher Carnot efficiencies: the heat is evacuated more progressively. In practice, we will consider in what follows three additional stages to $T_{\text{Q}}$ and $T_{\text{Gen}}$ and their temperature will be determined as a function of $T_{\text{Q}}$ and $T_{\text{Gen}}$ as we explain in \ref{sec:reducing_number_tunable_parameters}. 
\subsection{Approximations leading to our model}
Now, we will make simplifications on the model represented on the left of the figure \ref{fig:quantum_computer_archi_cable_with_lab}. The simplifications we are going to do will not be valid for all the range of parameters we will consider varying in the curves we are going to present, but they will at least be valid for the most important points on those curves that we will mainly discuss. We will anyway provide in this section the regime of parameters where we believe our model is physical and justify why we think so.
\subsubsection{The efficiency of the cryostat is the Carnot efficiency}
\label{sec:efficiency_is_carnot}
This assumption is very important as it will scale the full energetic study. In our model, we assume that to remove $1W$ of heat at a temperature $T$ will cost $(300-T)/T$ of Watts of electrical power, i.e., we assume the best possible efficiency for heat removal: Carnot efficiency.

The cryostats used in the typical quantum computing experiments of today have efficiencies much lower than Carnot efficiency; our assumption seems unrealistic at first view. Actually, those poor efficiencies are explained because the cryostats used in quantum computing today are (i) of small size (i.e., they do not have to evacuate a very large amount of heat), and (ii) the cryostats are not fully optimized: current experiments are small prototypes. The cryostat that will have to be used in a large-scale quantum computer will have to evacuate a large amount of heat, and it should be well optimized (in particular for this reason).

Reaching cryostat efficiencies close to the Carnot efficiency is actually reachable for large-scale cryogenics. For instance, the cryogenic architecture used in the CERN has an efficiency that is about $30 \%$ to Carnot, and it evacuates heat at $4K$ \cite{parma2015cryostat}. More generally, a \textit{well-designed} cryostat can be expected to have an efficiency being more than $10 \%$ of Carnot in the full range $[10mK,300K]$ \cite{martin2021energy,green2019helium,irds}. This is because an important concept for cryogenics is the quantity of heat that has to be evacuated: there is a distinction between chip-scale and large-scale cryostats, the latter being much more power-efficient \cite{martin2021energy,green2019helium,irds,weisend2016cryostat}. For those reasons, in our calculations, we are taking a Carnot efficiency for all the stages. 
\subsubsection{We can neglect the heat conduction, and the Joule effect in all the cables between $300K$ and $T_{\text{Gen}}$}
\label{sec:neglecting_heat_cond_joules}
Let us show first that the cables between the laboratory and the signal generation stage introduce a negligible heat flow on the signal generation stage compared to the heat dissipated by the classical electronics there (DEMUX-MUX, and DAC/ADC). In order to simplify the discussion, we put aside the heat dissipated by DEMUX-MUX; we will consider it in a few paragraphs. The typical values of heat dissipation \textit{per physical qubit} for the DAC and ADC elements are in the milliwatt range \cite{bardin201929,patra2020scalable}, which is what we will consider here. The fact it is a consumption per qubit (and not per active gate) is the current state of the art of electronics. It is based on the fact that those elements will \textit{always} generate a signal to a given physical qubit (and the signal will be of zero amplitude if the qubit doesn't have to be driven). But it could be possible in principle to turn off the electronics when it is not being used, such that the consumption would be related to dynamic costs. We will study the consequences of that in the section \ref{sec:making_energy_efficient}.

Now, we need to estimate if the heat flow associated with the cables going toward the laboratory will be negligible compared to the heat this electronics dissipates. We start by discussing the optical fiber (blue and pink data cables). We give some reference values to keep in mind. A typical coaxial cable composed of a mixture of conventional conductor and insulator typically conducts $\sim 1mW$ of heat if it is put between $0K$ and $300K$ (see appendix \ref{app:typical_heat_flows}). Optical fibers (blue and pink data cables on the figure) will be composed of insulator material so that they will conduct less heat than that. But to simplify, we will reason with this worst-case scenario where one data cable would conduct $1mW$ of heat between the laboratory and the signal generation stage. The question is then: how many of those data cables do we need? The answer is related to how much information per unit time would have to transit in them and how much a typical optical fiber can transmit data per unit time. We treat the first question. Every $\tau_P^{\textit{fastest}}=25ns$, which corresponds to the duration of the fastest physical gate inside the computer (in our model, they correspond to single-qubit gates, see \ref{sec:state_of_the_art}), the DAC has to generate a new signal for a physical qubit. A reasonable value is based on the fact that a "good" pulse can be described with $16$ points in time \cite{bardin201929}, where on each time we can attribute a voltage amplitude encoded on four bits. It gives a total amount of information to describe this pulse being $4 \times 16=64$ bits. In the end, we obtain a quantity of information for each physical qubit that is: $64/(25ns)\approx 2.5 Gb/s$. Then, a good optical fiber can convey $400 Gb/s$ of information \cite{wade2018bandwidth}. It tells us that about one optical fiber can convey the information for $160$ physical qubits, and it would conduct less than $1mW$ of heat over a temperature gradient of $[300K,0K]$. So, the heat conduction \textit{per qubit} would be less than $10 \mu W$ which is lower than the typical consumption of the DAC/ADC: we can neglect it. We estimated the amount of information "going down" to create the pulses, but we also need to estimate the information "going up" toward the laboratory. This will be the information associated with the measurement outcomes. In an extremely worst-case scenario, all the physical qubits are being measured at once. We thus get one bit of information per physical qubit to provide to the laboratory on this timescale of $25 ns$. This will be negligible to what we already estimated, so the critical thing to estimate is the information used to drive the gates rather than the information coming from the qubits measurement. 

Later on, we will also try to reduce the consumption of the DAC/ADC in order to see how it affects the total power consumption. For this reason, we would like to see until when we can reasonably neglect the heat conduction of the optical fibers. We believe that as soon as the DAC/ADC wouldn't consume more than $1 \mu W$, we could neglect the heat conduction of the optical fibers. Typically, we can roughly consider that an optical fiber conducts $\times 10$ less heat than a typical microwave cable. To give closer ideas, at $300K$, silicon dioxide (SiO2), which is used in optical fibers, has a thermal conductivity of $\approx 1 W/m.K$ \cite{shackelford2000crc} while the thermal conductivity of 304 stainless steel (which is widely used for microwave cables \cite{krinner2019engineering}) is more than $10 W/m.K$ at this temperature. And the $\times 10$ factor of difference would remain for lower temperatures \cite{marquardt2002cryogenic,shackelford2000crc}. And even assuming this worst-case, simply adding an extra cryogenic stage at $T \approx (300-T_{\text{Gen}})/2$ would drastically reduce the impact of the heat conduction from those wires thanks to the good Carnot efficiency of this stage, such that their impact in power consumption would be negligible compared to the power required to remove the heat dissipated by the electronics. But we will remain with this "safer" assumption that the electronics should consume more than $1 \mu W$ per qubit to keep our model valid.

There are now DC wires that are providing power to all the electronics. Those wires will bring heat conduction but also introduce Joule effect. In order to understand what we mean, let us consider the typical resistance of a $1m$ length cable: $R = 10m \Omega$. We assume that the electronics is working under the typical voltage of $U=5V$. We call $\dot{q}_{\text{Gen}}$ the power consumed \textit{per physical qubit} by the electronics, and $N_{\text{DC}}$ the number of DC wires we use. There is interest in lowering the number of DC wires because we want to minimize heat conduction, but at the same time, all the electronics will require a given amount of DC current in order to work. The more DC wires are being used; the fewer Joule effect would be a problem as this current would "spread" on more wires and reduce the heat dissipated. In a more quantitative manner, the total Joule effect dissipated in the DC wires satisfies: $P_J=N_{\text{DC}} R (I/N_{\text{DC}})^2$ where $I=\dot{q}_{\text{Gen}} Q_P/U$ is the current that is consumed by all the electronics (we recognize on the numerator the \textit{total} power consumed by all the electronics). What we want is $P_J/Q_P < \dot{q}_{\text{Gen}}$: we want the Joule effect \textit{per qubit} in those cable to not dominate the heat dissipated \textit{per qubit} of all the electronics. It gives us the relation: $N_{\text{DC}}>R \dot{q}_{\text{Gen}}/U^2  Q_P$. We have $R \dot{q}_{\text{Gen}}/U^2 \approx 10^{-6}$ for $\dot{q}_{\text{Gen}} \approx 1mW$: to have the Joule effect negligible we need more than one wire every million of physical qubits. Considering this as being our choice, the heat conduction from the DC wire will also naturally be negligible given the typical values of heat conduction we gave before. Thus, the heat load associated with DC cables can be entirely neglected. And obviously, for $\dot{q}_{\text{Gen}}=1 \mu W$, our conclusion would remain (Joule effect and heat conduction would be even less a problem).

We can also naturally neglect the heat load associated with the local oscillator: there is a unique coaxial cable for this signal. For the same reason, we can neglect the few amounts of cables written under the name "Other" which correspond to the clock signal, trigger, etc.

At this point, we showed that the heat conduction and Joule effect would be negligible compared to the heat dissipated by the ADC and DAC as soon as it is in the "order of magnitude" range $[1 \mu W, 1mW]$. But, there is one last point to discuss: how much the DEMUX-MUX would dissipate? A typical value to demultiplex or multiplex data for the optical fibers we considered is $1pJ$ per bit (see \cite{wade2018bandwidth})\footnote{The cost depends on the optical fiber. Indeed optical fibers able to transmit more than $400 Gb/s$ of information do exist, but the information would be encoded in a way that may require a bigger consumption for the multiplexing/demultiplexing tasks than the one we are considering here.}. We know that we have $2.5 Gb/s$ of information transiting \textit{per physical qubit} in order to manipulate them (the information entering in the DAC). The power required to demultiplex the data is then about $2.5mW$, which is comparable to the consumption of the DAC/ADC. There is also a power required to multiplex the information coming from the measurement outcomes after it has been digitized by the ADC. But as we already explained, the amount of information involved will be much smaller than the one required to drive the qubits such that the multiplexing cost would be negligible.

In the very end, we will consider $\dot{q}_{\text{Gen}}=5mW$ as a standard value. It will include the cost of demultiplexing, multiplexing (this one is actually negligible, as we just showed), and signal generation per physical qubit. Then, in order to make predictions about the progress in CMOS technology, or other technologies such as adiabatic circuits \cite{debenedictis2020adiabatic} \footnote{This technology is based on reversible logic, which allows making the electronics much more energy-efficient, putting its energy efficiency closer to the fundamental Landauer limit.} and single flux quantum logic \cite{mukhanov2019scalable} \footnote{This is a technology based on realizing classical electronics with superconducting circuits. The performance is not equivalent to CMOS technology, but it is expected to be much more energy-efficient as superconducting circuits are not resistive.}, we will make it vary, imposing $\dot{q}_{\text{Gen}}=\epsilon \times 5mW$, for $\epsilon \in [10^{-3},10^0]$. Those other technologies are less mature than CMOS, but the hope is that they would consume significantly less power. We consider $\epsilon \geq 10^{-3}$ because it corresponds to the range of validity of our model\footnote{Actually, it is very likely that $\epsilon$ could be even more reduced while keeping a physical validity for the model. It would be the case if we consider that instead of transmitting all the waveforms of the signal to implement, we would just transmit a few bits of information describing the gate to implement. The DAC would then read in memory the waveform to generate. Such approaches are likely to be much more energy-efficient.}, but we notice that the lower bound could also correspond to a very rough\footnote{Up to our knowledge, there is no detailed analysis of how much power would single flux quantum logic consume for a specific task such as signal generation. The value $\epsilon=10^{-3}$ seems, however, plausible to hope for the future of this technology as one can see in \cite{mcdermott2018quantum}.} estimation of how much single flux quantum logic would consume power \cite{mcdermott2018quantum}.
\subsubsection{We can neglect the energetic cost associated to parametric amplifiers}
\label{sec:neglecting_param_amp}
We can see in figure \ref{fig:quantum_computer_archi_cable_with_lab} b) that we also removed all the elements associated with the parametric amplifiers. We recall that those components are made of superconducting materials and do not dissipate heat intrinsically. However, they require a pump microwave signal (as represented by the blue dotted line) which needs to be attenuated. The typical total attenuation required between where the pump is being generated and the parametric amplifier is comparable to the one that would be put on the driving lines. For this reason, assuming that the pump signal must be activated for a duration comparable to the duration of the pulses driving the quantum gates, it will be enough for us to compare the power of the pump signal to the typical power to drive single-qubit gates. The pump signal must typically be at least $100 \times$ bigger than the total power of the signal it has amplified \cite{benjaminPrivate} (thus the signal at the output of the parametric amplifier). We now estimate this value.

First, we will consider that measurements last for a duration being $\tau_{P}^{\text{Meas}}=100 ns$ (it is in the typical range of fast readout techniques \cite{dassonneville2020fast,heinsoo2018rapid,benjaminPrivate}). From this and the assumption that the qubits can be set to different frequencies over a bandwidth of about $\Delta f \sim 1GHz$, we deduce that we can, in principle, multiplex the readout signals by putting $100$ of them on the same line (we can relax a bit this number as we see later). Indeed, we can roughly estimate that one measurement will occupy $1/100ns$ of spectral bandwidth (around the frequency of the associated qubit). We deduce that we can multiplex about $\Delta f*\tau_{P}^{\text{Meas}}=100$ signals on the readout measurement lines: each measurement on this line will occupy a well defined spectral band such that the spectral resolution for the measurements will be good enough. Also, $1GHz$ of bandwidth is reachable with traveling-wave parametric amplifiers \cite{esposito2021perspective}. Now, one parametric amplifier will need a pump signal $100 \times$ bigger (in power) than the amplified signal associated with $100$ measurements (parametric amplifiers typically need to amplify $100$ times the measurement signals \cite{krantz2019quantum,esposito2021perspective}). It now remains to estimate the power of a typical measurement. One measurement releases one photon in the duration of the measurement, which gives $\hbar \omega_0/\tau_{P}^{\text{Meas}} \approx 10^{-17} W$. Then, using the fact that $100$ of those measurements will be amplified $100$ times and that the pump signal must typically be $100$ times bigger than the resulting amplified signal, we deduce that the pump power should be about $10^{-11} W$, a value close to what can be found experimentally \cite{esposito2021perspective,ranadive2021reversed}.

We need to compare this value to the typical power required to implement single and two-qubit gates to know if we can neglect the energetic cost of the parametric amplifiers. For single-qubit gates, we recall from section \ref{sec:energetic_cost_to_perform_a_gate} that we have a typical power $P_g \approx \hbar \omega_0 \pi^2/(4 \gamma_{\text{sp}} \tau^2)$. For the $25 ns$ single-qubit gates we consider in our models, it gives $P_g = 10^{-11} W$ if $\gamma_{\text{sp}}=1kHz$ (it will be the maximum value of $\gamma_{\text{sp}}$ we are going to consider in the plots that follows, the power $P_g$ here is thus the lowest one to expect\footnote{We also recall that we took the same typical power for the two-qubit gates, see \ref{sec:state_of_the_art_2qb}}). This value is comparable to the power required for the pump signal. Thus not taking into account the pump power will not change our estimations based on orders of magnitudes. We notice that it assumes that the pump can be turned off when not used (because the power of the single-qubit and two-qubit gates $P_g$ is only on when gates are active). It is in principle possible with parametric amplifiers \cite{LucaPrivate}. If not, given the fact $t(3)/(u(3)+v(3)) \approx 30$ (this is the ratio of physical qubits divided by the number of active single and two-qubit gates\footnote{For an algorithm satisfying $Q_L \approx (\overline{N}_{L,\parallel}^{\text{1qb}}+\overline{N}_{L,\parallel}^{\text{Id}}+2 \overline{N}_{L,\parallel}^{\text{cNOT}})$ which will be the case in what follows.}) it would mean that for $\gamma_{\text{sp}}=1kHz$, the pump signal cannot be neglected, and we would leave the regime of validity of our model: in this case, we should only focus on points where $\gamma_{\text{sp}}\lesssim 100Hz$. In what follows, we will assume that the pump can be disabled when the parametric amplifier is not used.
\subsubsection{We can neglect the heat conduction of the remaining cables, excepted X/Y drive, and the Joule effect for DC cables associated to classical amplifiers}
\label{sec:we_can_neglect_heat_conduction_remaining_cables_DC_Joules}
The remaining cables we must study are the readout-in, readout-out, pump cables for the parametric amplifiers, and DC cables for the classical HEMT amplifiers. All those cables, except the DC one, will be microwave coaxial cables for $T>10K$ and superconducting microstrip lines for $T<10K$ (see appendix \ref{app:cable_models}), and they will have the same physical properties. In order to know what we can neglect, we then just have to compare the number of each type of cable used. It is given by the quantity of multiplexing we can do. 

The single-qubit gates are the fastest, and they will be the "bottleneck" in terms of multiplexing. Indeed, following the same reasoning as the one above, we can drive about $1GHz*25 ns=25$ qubits with a single wire; otherwise, the spectral resolution of the signals is too poor, and issues such as crosstalk could occur. 

Now, we can study the cables associated with measurements. As explained in the section \ref{sec:neglecting_param_amp}, the measurements can have a greater level of multiplexing so that the associated cables "readout in" and "readout out" will be smaller in number. An analog justification holds for the cables associated with the pump signal of the parametric amplifiers. Thus, we deduce that the heat conduction associated with measurement and pump cables can entirely be neglected compared to the heat conduction of the X/Y lines.

Finally, there are also classical amplifiers at the typical temperature of $4K$. One of those amplifiers consumes about $5mW$ of power \cite{benjaminPrivate,LucaPrivate} and will amplify the signal contained in one readout line (that measures 100 qubits), thus the amplifiers consume $50 \mu W$ per physical qubit. Those amplifiers need a DC wire to receive power. But those wires can be regrouped such that one wire can provide power to many amplifiers. But it also induces Joule effect. Following the same reasoning as we did for the signal generation stage, we can easily show that both the Joule effect and heat conduction will be negligible compared to the heat dissipated by the amplifier. Indeed it is exactly the same calculation excepted that the heat \textit{per qubit} dissipated by the amplifier would become $5mW/100 =50 \mu W$. Replacing  $\dot{q}_{\text{Gen}} \to \dot{q}_{\text{Gen}}/100$ in the calculation we did previously would show that we would need at least $1$ wire every $10^8$ physical qubits in order to keep the Joule effect negligible. One wire for $10^8$ qubits would also give totally negligible heat conduction compared to the $50 \mu W$ dissipated \textit{per qubit} by the amplifiers. 

Finally, in the same line of thought as what we did for the signal generation stage, we will also assume that the power consumption of the amplifiers can be reduced such that the typical power consumption per physical qubits on the classical amplification stage will be $\dot{q}_{\text{Amp}}=\epsilon*50 \mu W$, where $\epsilon$ will take different values in our plots ($\epsilon=1$ would correspond to today state of the art).

\subsubsection{We can neglect the heat dissipated in the attenuators for readout-in}
\label{sec:neglecting_measurement}
The typical power of the signals that have to be injected to perform a measurement is much lower than the power to drive the quantum gates. Also, the average number of measurement gates is the same as the average number of single-qubit gates (this is represented by the values of $u(k) = w(k)$ in \eqref{eq:logical_power_ch5}). Thus, we can neglect the heat dissipated for the attenuators of the readout-in lines in comparison to the attenuators of the X/Y lines. 
\subsubsection{Brief summary of approximations}
In conclusion here, we believe that we have good reasons to approximate the architecture by the model of figure \ref{fig:quantum_computer_archi_cable_with_lab} b), at the very least in the regime $\gamma_{\text{sp}} \in [10Hz,1kHz]$, and\footnote{If the travelling wave parametric amplifier cannot be turned on only when needed, then this range has to be relaxed to roughly $[10Hz,100 Hz]$, see the discussion in the section \ref{sec:neglecting_param_amp}} for $\epsilon \in [10^{-3},10^0]$, where $\epsilon$ is such that the heat dissipated per qubit on the signal generation stage and classical amplification stage both satisfy $\dot{q}_{\text{Gen}}=\epsilon*5mW$, $\dot{q}_{\text{Amp}}=\epsilon*50 \mu W$. Many of the approximations we did here could actually be relaxed (nothing prevents us from taking into consideration the heat dissipated by the pump signal for instance), but the reason we wanted to do this study of what we can or cannot neglect in detail is also because it is instructive in itself as it will allow us to know what are the best strategies to make quantum computing energy-efficient.
\subsection{Expression of the power function for our model}
We are now ready to write down an explicit expression for the power consumption. Our goal is to find the values of the coefficients $a$, $b^{\text{1qb}}$, $b^{\text{cNOT}}$, $b^{\text{Meas}}$ in \eqref{eq:logical_power_ch5}. In order to write in a compact way the formulas, we will call $T_1 \equiv T_{\text{Q}}$, $T_K \equiv T_{\text{Gen}}$. We will consider adding $K-2$ intermediate stages between $T_1$ and $T_K$ where only heat conduction and dissipation into attenuators will be evacuated there (those stages won't contain any electronics). In addition to those stages there is also the classical amplification stage at the temperature $T_{\text{Amp}}$. We can write down the expression of the power.
\begin{align}
P&=\frac{300-T_K}{T_K}\left(\dot{q}_{\text{Gen}}-\frac{1}{25}\dot{q}_{\text{Cond}}(T_{K-1},T_K) \right) Q_P \notag \\
&+\sum_{1<i<K-1}\frac{300-T_i}{T_i} \left( \frac{1}{25}(\dot{q}_{\text{Cond}}(T_{i},T_{i+1})-\dot{q}_{\text{Cond}}(T_{i-1},T_i))\right) Q_P+ (\widetilde{A}_i-\widetilde{A}_{i-1})P_g\left(\frac{\tau_P^{\text{1qb}}}{\tau_P^{\text{timestep}}}\overline{N}^{\text{1qb}}_{P,\parallel}+\overline{N}^{\text{cNOT}}_{P,\parallel}\right) \notag \\
&+ \frac{300-T_{\text{Amp}}}{T_{\text{Amp}}} \dot{q}_{\text{Amp}} Q_P\\
&+\frac{300-T_1}{T_1} \left( \frac{1}{25}\dot{q}_{\text{Cond}}(T_{i},T_{i+1})\right) Q_P+ (\widetilde{A}_1-1) P_g\left(\frac{\tau_P^{\text{1qb}}}{\tau_P^{\text{timestep}}}\overline{N}^{\text{1qb}}_{P,\parallel}+\overline{N}^{\text{cNOT}}_{P,\parallel}\right),
\end{align}
The first line represents the stage where signals are generated. Here only the heat due to signal generation has to be evacuated. There is also heat conduction in the X/Y lines that is going to leave this stage to go to the lower stages (but it will be negligible compared to $\dot{q}_{\text{Gen}}$). The heat conduction \textit{per cable} for the X/Y lines is represented by the function $\dot{q}_{\text{Cond}}(T_{i},T_{i+1})$. The first temperature is the temperature of the coldest stage (where the heat is moving toward), and the second temperature is the one of the hottest stage. The coefficients $(300-T)/T$ represent the Carnot efficiency of the stage.

The second line is associated to the blue intermediate stages of figure \ref{fig:quantum_computer_archi_cable_with_lab}. Their role is to help evacuating heat conduction and dissipation into attenuators. We recognize a first term representing the incoming and outgoing heat conduction. The second term is associated to the heat that is dissipated into the attenuators. On this stage, the attenuation is $A_i=\widetilde{A}_i/\widetilde{A}_{i-1}$, where $\widetilde{A}_i=\prod_{n=1}^{i} A_n$ is the \textit{total} attenuation from $T_{1}$ up to the stage of temperature $T_i$ (it is represented on figure \ref{fig:attenuators} for more clarity).  The signal before attenuation on this stage thus has the power amplitude $\widetilde{A}_i P_g$. And after being attenuated this amplitude becomes $\widetilde{A}_{i-1} P_g$. Thus the total heat dissipated reads $(\widetilde{A}_i-\widetilde{A}_{i-1})P_g$. We notice that this power is then multiplied by $\tau_P^{\text{1qb}}/\tau_P^{\text{timestep}}\overline{N}^{\text{1qb}}_{P,\parallel}+\overline{N}^{\text{cNOT}}_{P,\parallel}$, where $\overline{N}^{\text{1qb}}_{P,\parallel}$ and $\overline{N}^{\text{cNOT}}_{P,\parallel}$ are representing the number of \textit{physical} single and two-qubit gate active in parallel. It comes from our assumption that (i) a two-qubit gate requires the same amount of power than a single-qubit gate (ii) a single-qubit gate is only acting for a portion of the timestep (while the cNOT last for the whole timestep). Thus the power spent to drive the single-qubit gates must be weightened by $\tau_P^{\text{1qb}}/\tau_P^{\text{timestep}}$. In principle, given the assumptions we already did we could roughly estimate that $\tau_P^{\text{1qb}}/\tau_P^{\text{timestep}}\overline{N}^{\text{1qb}}_{P,\parallel}+\overline{N}^{\text{cNOT}}_{P,\parallel} \approx \overline{N}^{\text{cNOT}}_{P,\parallel}$, but the full formula has been considered in our models.

We notice that on the amplification stage, only the heat dissipated by the amplifiers is evacuated (no heat dissipated from attenuators and no heat from conduction). We considered this choice in our model as it simplifies the expression of the power. Thus there are no attenuators on this stage (as already represented by figure \ref{fig:quantum_computer_archi_cable_with_lab}), but the cables are also not thermalized here (we could thermalize them, but given the fact we will anyway consider $K=5$, i.e., there are three intermediate stages between $T_{\text{Q}}$ and $T_{\text{Gen}}$, deciding to evacuate heat conduction on the amplifier stage or not will not change significantly the power we obtain\footnote{This is actually something that we checked numerically: no quantitative differences occurred when we removed or not the heat conduction of the cables on the signal amplification stage.}).

Now, we can access the coefficients $a$, $b^{\text{1qb}}$, $b^{\text{cNOT}}$ from \eqref{eq:logical_power_ch5}. We have:
\begin{align}
a&=\frac{300-T_K}{T_K}\left(\dot{q}_{\text{Gen}}-\frac{1}{25}\dot{q}_{\text{Cond}}(T_{K-1},T_K) \right) \notag \\
&+\sum_{1<i<K-1}\frac{300-T_i}{T_i} \left( \frac{1}{25}(\dot{q}_{\text{Cond}}(T_{i},T_{i+1})-\dot{q}_{\text{Cond}}(T_{i-1},T_i))\right)\notag \\
&+ \frac{300-T_{\text{Amp}}}{T_{\text{Amp}}}  \dot{q}_{\text{Amp}}  \notag \\
&+\frac{300-T_1}{T_1} \left( \frac{1}{25}\dot{q}_{\text{Cond}}(T_{i},T_{i+1})\right),
\label{eq:a}
\end{align}

\begin{align}
b^{\text{1qb}}&=\frac{\tau_P^{\text{1qb}}}{\tau_P^{\text{timestep}}}  P_g  \left( \sum_{1<i<K-1}\frac{300-T_i}{T_i} (\widetilde{A}_i-\widetilde{A}_{i-1}) +\frac{300-T_1}{T_1} (\widetilde{A}_1-1)\right),
\label{eq:b1qb}
\end{align}
\begin{align}
b^{\text{cNOT}}&=P_g  \left( \sum_{1<i<K-1}\frac{300-T_i}{T_i} (\widetilde{A}_i-\widetilde{A}_{i-1}) +\frac{300-T_1}{T_1} (\widetilde{A}_1-1)\right)
\label{eq:bcNOT}
\end{align}
\begin{align}
b^{\text{Meas}}&=0
\label{eq:bMeas}
\end{align}
The reason why $b^{\text{Meas}}=0$ has already been explained in the section \ref{sec:neglecting_measurement}: the dissipation in the readout-in lines is negligible compared to the one in the X/Y lines. At this point, the engineering model has entirely been described as all the coefficients $a$, $b^{\text{1qb}}$, $b^{\text{cNOT}}$ and $b^{\text{Meas}}$ have been specified. The next step for us is to characterize the quantity of noise that will be contained in the answer of the algorithm, i.e., we need to specify the metric.
\section{Expressing the quantity of noise}
\label{sec:noise_quantity}
\subsection{Expression as a function of the quantity of noisy photons}
As explained in the section \ref{sec:adapting_framework_FT_metric_logical} of the previous chapter, the metric $\mathcal{M}_L$ we consider will be the probability than any logical gate in the circuit fails. Calling $p_L$ the probability that one logical gate fails, we thus have
\begin{align}
&\mathcal{M}_L=N_L p_L \label{eq:metric_logical}\\
&p_L=\eta_{\text{thr}}\left(\frac{\eta}{\eta_{\text{thr}}}\right)^{2^k},
\end{align}
with $\eta_{\text{thr}} \approx 10^{-4}$ for probabilistic noise (see the section \ref{sec:arbitrary_accurate_quantum_computing} in the second chapter). The number of logical gates $N_L$ can be deduced from the depth and the average number of logical gate of each type that are acting in parallel: $N_L=D_L(\overline{N}_{L,\parallel}^{\text{1qb}}+\overline{N}_{L,\parallel}^{\text{cNOT}}+\overline{N}_{L,\parallel}^{\text{Id}})$. Now, we need to better express the physical noise. We will proceed in a similar way as the one done in the section \ref{sec:physical_model_chap3} of the third chapter of this thesis: we will consider the noise to be a probabilistic Pauli noise of strength  $\eta=\max_i \chi_{ii}$, where $\chi$ is the $\chi$ matrix associated to the noise map $\mathcal{N}$ (see also in the section \ref{sec:physical_model_chap3} what motivates us to do this simplifying approximation). 

We find that, if the noise map is acting for a duration $\tau$, we have:
\begin{align}
\eta=\frac{(1+2\overline{n}_{\text{tot}}^{(1)})\gamma_{\text{sp}}\tau}{4},
\label{eq:eta}
\end{align}
where $\overline{n}_{\text{tot}}^{(1)}$ is the total number of thermal photons interacting with the qubits (the exponent $(1)$ is associated to our numbering where $T_{\text{Q}}=T_1$). It is related to their temperature but also to the temperature of all the other cryogenic stages. The duration $\tau$ we need to consider is the duration associated with the longest physical gate: it corresponds to the physical timestep $\tau_P^{\text{physical}}$ and is equal to both cNOT and measurement gates. The next step is to find the expression of $\overline{n}^{(1)}_{\text{tot}}$.
\subsection{Expressing the noisy photons number as a function of the temperatures and attenuations}
As we now have $K$ temperature stages where $T_1=T_{\text{Q}}$ and $T_{K}=T_{\text{Gen}}$, we then have a family of attenuations: $\{A_i\}_{i=1,...,K-1}$ to put on all the stages below $T_K$ as shown on the Figure \ref{fig:attenuators}. Our goal is to determine $\overline{n}_{\text{tot}}^{(1)}$ as a function of all those parameters.
\begin{figure}[h!]
\centering 
\includegraphics[width=0.4\textwidth]{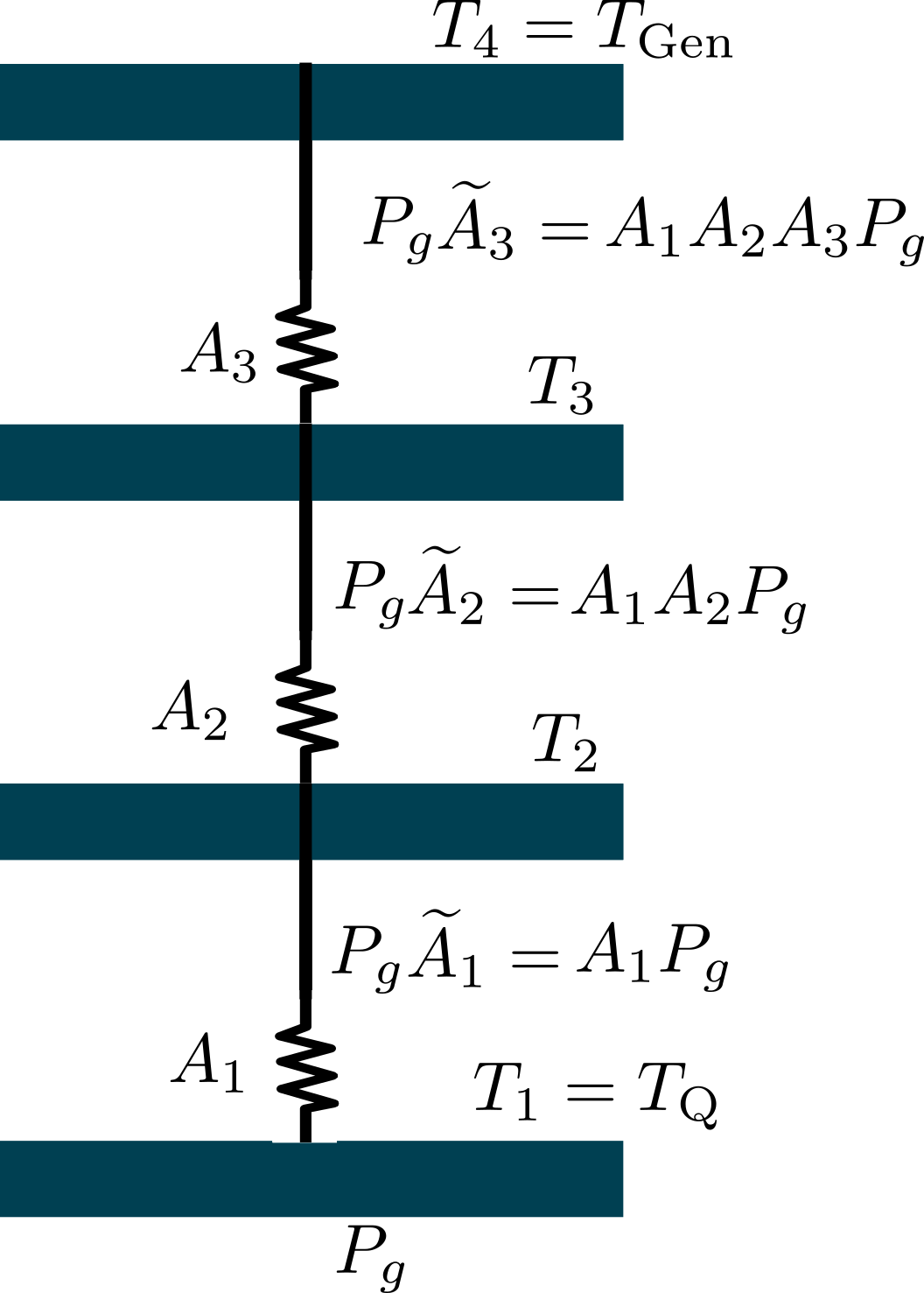}
\caption{Notations we use for attenuations and temperatures of the different stages when having $K=4$. $P_g \widetilde{A}_i$ is the power in the cable on the stage of temperature $T_i$ \textit{before} the attenuator.}
\label{fig:attenuators}
\end{figure}
\FloatBarrier
\subsubsection{Total number of thermal photons at $T_1$}
In order to determine it, we use the relationship between the total number of noisy photons on a stage of temperature $T_i$ and the number of noisy photons on the higher temperature stage at $T_{i+1}$. The relationship is simply \cite{pozar2011microwave}:
\begin{align}
\overline{n}^{(i)}_{\text{tot}}=n_{BE}(T_i) \frac{A_i-1}{A_i}+\frac{n_{\text{tot}}^{(i+1)}}{A_i}
\label{eq:ntot_recursive}
\end{align}
We also have the boundary condition that on the stage of temperature $T_K$, the total amount of noise will be given by the thermal noise at this temperature:
\begin{align}
\overline{n}_{\text{tot}}^{(K)}=n_{BE}(T_K),
\label{eq:ntot_Cl}
\end{align}
where $n_{BE}(T_K)$ is the Bose Einstein population, defined in the section \ref{sec:dynamic_in_presence_thermal_noise}.
At this point, it is important to notice that we are making the assumption that the highest temperature participating in the thermal noise corresponds to the temperature at which the signals are being generated, and not $300K$. However, as we can see from the figure \ref{fig:quantum_computer_archi_cable_with_lab}, no cables have been attenuated for temperatures higher than $T_{\text{Gen}}$. At the moment this thesis is written, we are not entirely sure if it is a good assumption to make, and we need further investigations. If it appears to not be valid, the model could easily be adapted. One choice would for instance be to replace $n_{BE}(T_K)$ by $n_{BE}(300)$ in \eqref{eq:ntot_Cl}. In the end, we can deduce the total number of photons on the qubit temperature stage: using \eqref{eq:ntot_recursive} and \eqref{eq:ntot_Cl}, we deduce:
\begin{align}
\overline{n}_{\text{tot}}^{(1)}=\sum_{i=1}^{K-1} n_{BE}(T_i) \frac{A_i-1}{A_1...A_i}+\frac{n_{BE}(T_K)}{A_1...A_{K-1}}
\label{eq:total_noisy_photons}
\end{align}
\subsection{Reducing the number of tunable parameters}
\label{sec:reducing_number_tunable_parameters}
In principle, our work to determine $\overline{n}_{\text{tot}}^{(1)}$ is over and it involves the family $\{T_i\}$ and $\{A_i\}$ as tunable parameters. We could imagine optimizing all of them, but this would lead to way too many variables to consider such that the computation would not be doable. In order to simplify the problem, we will give some "reasonable" constraints on the attenuations and temperatures such that the optimization will be parametrized only by the total attenuation there is on the lines and the qubits and signal generation temperatures $T_1$ and $T_K$.

For the temperatures, we will ask them to be regularly spaced in order of magnitudes. For instance, for $K=5$, if $T_1=10mK$ and $T_K=100K$, we would like that the rest of the temperature would satisfy: $T_2=100mK$, $T_3=1K$, $T_4=10K$. This is what is usually considered in cryostats \cite{krinner2019engineering}, and it allows to have a temperature spreading that does not isolate too much one stage from the other. For instance, if we assume $T_1=10mK$ and $T_K=300K$, choosing temperature linearly in the range $[T_1,T_K]$ would isolate a lot the lowest temperature stage where the efficiency is very poor (the closest temperature stage, $T_2$, for $K=5$, would be at $75K$ for example). It would force to evacuate a large amount of heat where the efficiency is very low, which would be extremely energy inefficient. Of course, our choice is probably not the best; the best optimization for the intermediate stage would, for instance, depend on the exact heat conduction model behind the cables. In the end, we can mathematically phrase our choice; the temperatures are chosen such as:
\begin{align}
&T_{1 \leq i \leq K}=10^{x_0+(i-1)x_1},
\end{align}
where $x_0$ and $x_1$ are chosen by the boundary conditions $T_1$ and $T_K$. We deduce:
\begin{align}
&x_1=(\log(T_K)-\log(T_1))/(K-1)\\
&x_0=\log(T_1)
\end{align}
Which gives:
\begin{align}
T_i=T_1 \left(\frac{T_K}{T_1}\right)^{\frac{i-1}{K-1}}
\label{eq:choice_temp}
\end{align}
In order to reduce the number of variables used to describe the attenuations, we will consider the following. The total attenuation $A$ will be the parameter we will optimize. We will consider that the attenuation on a given stage will be proportional to the ratio between the temperature of the stage and the temperature of the closest higher stage. The motivation behind such assumption is that in the high temperature regime ($k_b T \gg \hbar \omega_0$), we have $n_{BE}(T) \approx k_b T/\hbar \omega_0$. The number of noisy photons is proportional to the temperature. Thus to make the number of noisy photons coming from higher temperature stage negligible to the number of noisy photons where the attenuator has been put, we should attenuate of an amount typically given by the ratio of those temperatures (this is also what is frequently done experimentally, see \cite{johnson2012optimization,reed2014entanglement}). Now it is not clear how much attenuation should exactly be put (exactly this ratio? Two orders of magnitude bigger "for safety"?). It is also not so clear as in the low temperature regime, where the qubits are usually put, the approximation $n_{BE}(T) \approx k_b T/\hbar \omega_0$ is no longer correct. This is why we believe that optimizing the total attenuation is something that is important to do. In the end, our choice is to consider:
\begin{align}
&A_i=\epsilon \frac{T_{i+1}}{T_i}
\end{align}
where $\epsilon$ is the proportionality coefficient that we need to determine through the constraint $\prod_{i=1}^{K-1} A_i = A$. Now, using the fact that our temperature are chosen following \eqref{eq:choice_temp}, we deduce:
\begin{align}
A_i=A^{1/(K-1)}
\label{eq:att_per_stage}
\end{align}
The same value of attenuation is being put on each temperature stage. And this value is deduced from the total attenuation value $A$ that we are going to find. In the end, here, we gave some reasonable constraints on the tunable parameters to optimize. Better constraints could probably be chosen, but we will see that even with such constraints, we will be able to save large amounts of power consumption with our optimization.

\subsubsection{Exploiting the condition $\mathcal{M}_L=\mathcal{M}_{\text{target}}$}

As shown in the property \ref{prop:behavior_min_power} of the previous chapter, because increasing $T_{\text{Q}}$ implies that the metric increases while the power decreases, the minimum power consumption will be obtained when the metric satisfies $\mathcal{M}_L=\mathcal{M}_{\text{target}}$. This property is useful for our purposes. Indeed we can use it to remove one parameter from the optimization procedure, which will allow us to speed up our simulation in a significant manner (without this, the simulations we want to do would be a little bit too long for the laptop on which they were implemented).

Here, we have: $\mathcal{M}_L=N_L p_L$. If we want our algorithm to provide the correct answer with a probability $1-\mathcal{M}_{\text{target}}$, $\overline{n}_{\text{tot}}^{(1)}$ must then satisfy:
\begin{align}
\overline{n}_{\text{tot}}^{(1)}=\overline{n}_{\text{tot}}^{\text{target}} \equiv \frac{1}{2}\left(\frac{4 \eta_{\text{thr}}}{\gamma_{\text{sp}} \tau_{P}^{\text{timestep}}}\left( \frac{\mathcal{M}_{\text{target}}}{\eta_{\text{thr}}N_L} \right)^{2^{-k}}-1 \right)
\label{eq:ntot_fct_target}
\end{align}
We can exploit this condition to express the total attenuation as a function of $T_1, T_K$ and $k$. Indeed, injecting \eqref{eq:att_per_stage} and \eqref{eq:choice_temp} in \eqref{eq:total_noisy_photons}, we find the following polynomial equation as a function of the variable $y=A_i=A^{1/(K-1)}$:
\begin{align}
(\overline{n}_{\text{tot}}^{(1)}-n_{BE}(T_1))y^{K-1}+\sum_{p=1}^{K-2} (n_{BE}(T_{K-1-p})-n_{BE}(T_{K-p}))y^p+n_{BE}(T_{K-1})-n_{BE}(T_{K})=0,
\label{eq:polynomial_to_solve_for_A}
\end{align}
where $\overline{n}_{\text{tot}}^{(1)}$ has been provided in \eqref{eq:ntot_fct_target}.
For $K \leq 5$, there are analytic solution of this equation for the variable $y$. Keeping only the physical ones (i.e $y \geq 1$: the attenuation must be bigger than one), we deduce possible values of $A$ as a function of $T_1$, $T_K$ and $k$ that will allow to satisfy $\mathcal{M}_L=\mathcal{M}_{\text{target}}$. In principle, for $K=5$, $4$ possible values for $A$ can be found for each values of $T_1$, $T_K$ and $k$ when $\mathcal{M}_L=\mathcal{M}_{\text{target}}$.

At this point, we can now solve the minimization under constraint:
\begin{align}
P_{\min} \equiv \min (P(\bm{\delta}))_{\big | \mathcal{M}_L(\bm{\delta}) = \mathcal{M}_{target}},
\end{align}
where $\bm{\delta}=(T_{1},T_{K},k,A)$. The power function $P(\bm{\delta})$ satisfies \eqref{eq:logical_power_ch5} with the coefficients $a(\bm{\delta})$, $b^{\text{1qb}}(\bm{\delta})$, $b^{\text{cNOT}}(\bm{\delta})$ ,$b^{\text{Meas}}(\bm{\delta})$ that have been provided in \eqref{eq:a}, \eqref{eq:b1qb},\eqref{eq:bcNOT},\eqref{eq:bMeas} and where the values of all the temperatures and attenuations involved in this expression are related to $T_{1}$ and $T_{K}$ through \eqref{eq:choice_temp} and \eqref{eq:att_per_stage}. As we did not optimize the temperature of the amplification stage, we considered $T_{\text{Amp}}=4K$ for $T_1<T_{\text{Amp}}<T_{K}$, $T_{\text{Amp}}=T_K$ if $T_K<4K$ and $T_{\text{Amp}}=T_{1}$ if $T_1 > 4K$. The optimization of the temperature stage would require to put in our model how much noise is added in the measurements by amplifying at a too high temperature, this goes beyond the scope of our study. 

Finally, in practice, this optimization is being performed by a "brute force" approach. We will basically calculate $P(\bm{\delta})$ in a discretized grid where $10mK<T_1<300K$, $T_1<T_K<300K$, $k$ being an integer lower than $6$ (the power consumption rapidly increases with $k$, $k \leq 6$ is a good "upper bound" for the discretization in $k$). For each value of the parameters, we calculate all the possible values of $A$ satisfying \eqref{eq:polynomial_to_solve_for_A}. We remove the unphysical ones, and we compute the power for the remaining physically valid attenuation found (and for the sampling parameters $T_1,T_K,k$). Then, we construct a list containing all the power consumptions we find with this sampling, and we then simply find the minimum of this list numerically. It gives us access to $P_{\min}$ as well as to the optimum parameters $(T^{\text{Opt}}_{1},T^{\text{Opt}}_{K},k^{\text{Opt}},A^{\text{Opt}})$ allowing to reach this minimum. In all the examples we are going to show in what follows, we verified numerically that there is a unique well-identified minimum in the problem (the configuration space appears to be very smooth with a unique minimum well-identified). We are now ready to study the results.
\section{Expressing the software part of the problem: fault-tolerant QFT}
\label{sec:software_part}
Here, we are going to estimate the energetic cost of a fault-tolerant quantum Fourier transform which will have the typical size of the one used within the Shor algorithm. It will thus be composed of $Q_L=2048$ logical qubits, and the total number of gates it requires is $\approx 2048^2$, which mainly contains logical two-qubit controlled phase gates. To simplify the discussion, for our energetic estimation, we will only consider those two-qubit gates and replace each of them by logical cNOT gates. This is, of course, a simplifying assumption as those controlled phase gates must be decomposed properly on the gateset that it is possible to implement fault-tolerantly (and it will contain more than one cNOT). We discuss in the appendix \ref{app:gateset_decomposition} what would change if such decomposition were taken into account.

Now, we recall that the only characteristics of the algorithm we need are $Q_L$, $\overline{N}_{L,\parallel}^{\text{1qb}},\overline{N}_{L,\parallel}^{\text{Id}}, \overline{N}_{L,\parallel}^{\text{cNOT}}$, and $D_L$ (or equivalently $N_L$ as $D_L$ is actually only used to find from the average number of logical gate acting in parallel the total number of logical gates). As we already know that $N_L=Q_L^2=2048^2$, it remains for us to estimate $\overline{N}_{L,\parallel}^{\text{cNOT}}$. This estimation is intrinsically related to how we decide to implement the algorithm. Here we will not optimize the best way to implement it; we will assume that it is implemented in the compressed manner shown in figure \ref{fig:compressed_QFT} (following the same kind of approach as the one used in the previous chapter, we could also imagine optimizing the way to implement this algorithm). The depth of this circuit satisfies $D_L=2(Q_L-2)+1 \approx 2 Q_L$. To show it, we can notice that the first two-qubit gates represented in red are the first occurrence of a repeating pattern occurring until the end of the algorithm. This pattern is repeated $Q_L-2$ times and is composed of two consecutive gates. At the very end of the algorithm, an extra two-qubit gate must be implemented, which adds the final $+1$. Thus, $D_L=2(Q_L-2)+1 \approx 2 Q_L$ is the depth of this algorithm\footnote{We recall that by convention, we do not take into account the final measurement in the depth.}. We deduce from that $\overline{N}_{L,\parallel}^{\text{cNOT}}=N_{L}^{\text{cNOT}}/D_L \approx Q_L/2$. The other logical gates (logical identity, or Hadamard that we actually did not represent them on figure \ref{fig:compressed_QFT}) would give a negligible contribution to the total number of gates acting in parallel, compared to the number of two-qubit gates; this is why we neglect them.
\begin{figure}[h!]
\begin{center}
\includegraphics[width=0.7\textwidth]{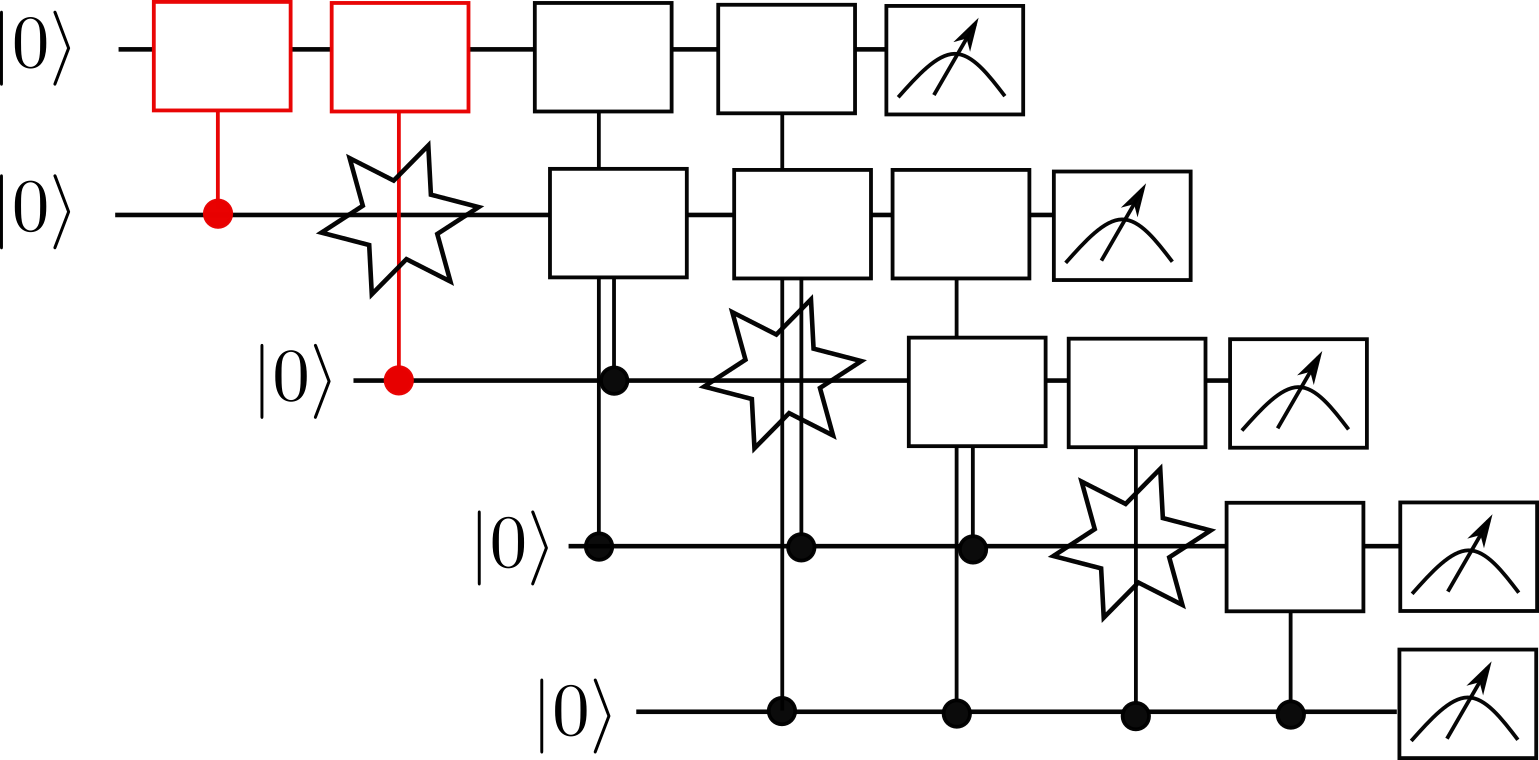}
\caption{QFT performed on $5$ qubits in a compressed way, removing the single-qubit gates, which are present in a negligible number compared to the two-qubit controlled phase gates. The two-qubit gates represented in red are the first occurrence of a repeating pattern occurring until the end of the algorithm. It will help us to estimate the depth of this algorithm, as explained in the main text.}
\label{fig:compressed_QFT}
\end{center}
\end{figure}
\FloatBarrier
At this point, the parameters describing the algorithm have been specified. We are now ready to estimate the energetic cost through the procedure of minimization of the power under the constraint of aiming a targeted accuracy. 
\section{Application of the hardware-software-noise-resource framework: estimation of the minimum power required to implement a QFT}
\label{sec:application_to_QFT}
\subsection{Power consumption as a function of qubit lifetime: }
\subsubsection{Getting a general intuition}
\label{sec:getting_intuition}
We start by considering the power consumption as a function of the lifetime of the qubits. We will consider $\gamma_{\text{sp}}^{-1} \in [1ms,100ms]$ which, as $\gamma_{\text{sp}}^{-1}$ roughly represent the qubit lifetime at $0$ temperature, corresponds to a qubit lifetime being between the millisecond and the hundred milliseconds. A qubit lifetime of one millisecond is, at the date this thesis is written, close to what it is possible to do in the state-of-the-art superconducting qubits \cite{kjaergaard2020superconducting}. A higher lifetime than that should be considered as what we could hopefully expect in the next years. We will also consider different values for the heat dissipated by the electronics. We recall that we introduced the parameter $\epsilon$ such that the heat dissipated by the electronics generating the signals and performing multiplexing/demultiplexing dissipates $\dot{q}_{\text{Gen}}=\epsilon 5 mW$ of heat per physical qubits, while the classical amplifiers are dissipating $\dot{q}_{\text{Amp}}=\epsilon 50 \mu W$ heat per physical qubit (thanks to multiplexing, one amplifier can amplify readout signals for $100$ of physical qubits). Also, strictly speaking, our engineering model is only valid if $\epsilon \geq 10^{-3}$ as explained in the section \ref{sec:neglecting_heat_cond_joules}. However, we still represent the points that are slightly outside of the strict range of validity of our model, keeping in mind that they should be interpreted with some care. The results are shown on figure \ref{fig:optimal_function_gamma_with_attenuators_consoperqubit}.
\begin{figure}[h!]
\begin{center}
\includegraphics[width=1.0\textwidth]{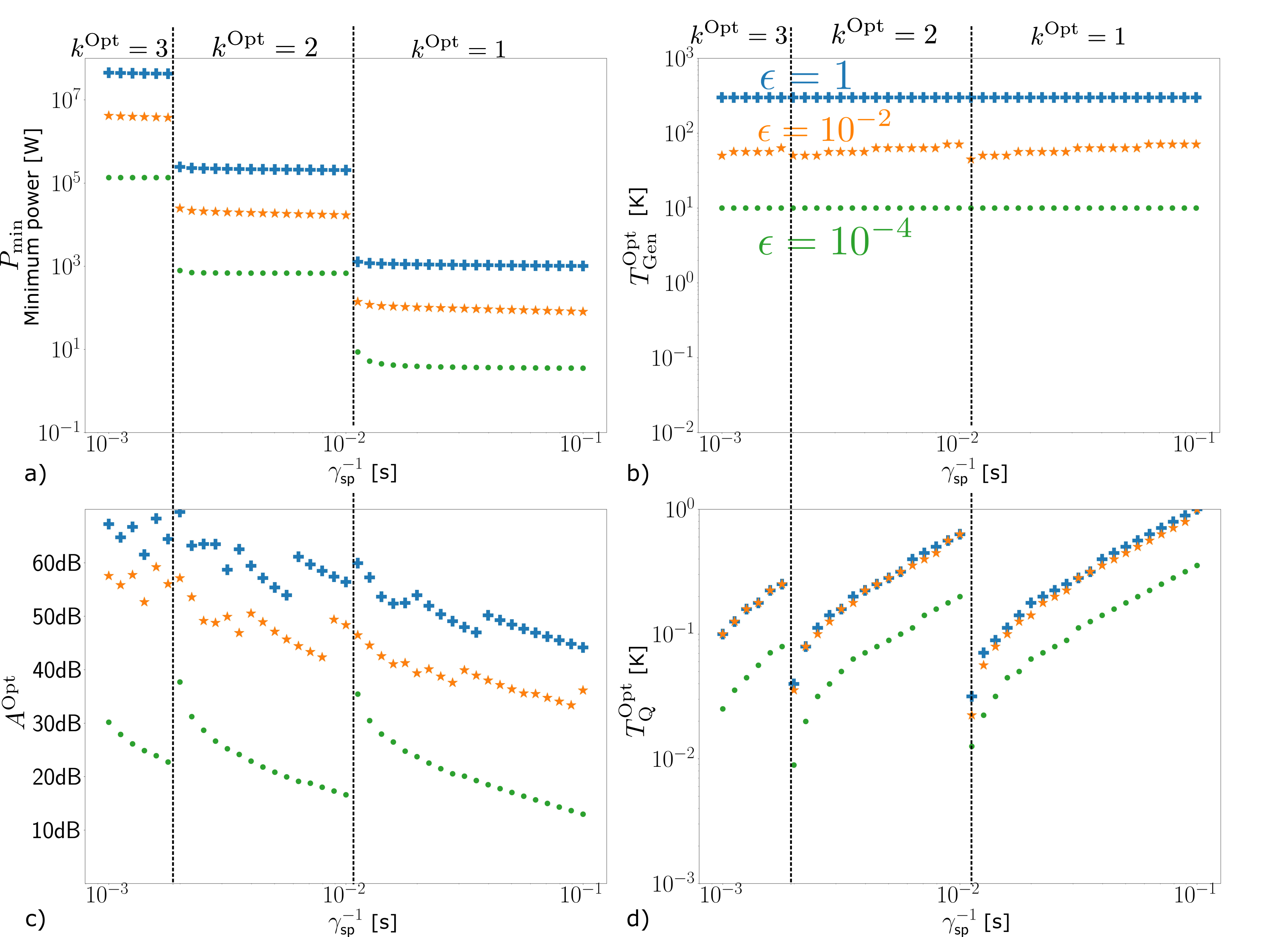}
\caption{Minimum power and optimum parameters allowing to reach this minimum as a function of the qubit lifetime $\gamma_{\text{sp}}^{-1}$ (in seconds) for different values of heat dissipated for the electronics represented by the different curves associated to different values of $\epsilon$. The blue curve for $\epsilon=1$, for $\gamma_{\text{sp}}^{-1}=1ms$ is close to today state of the art qubit lifetime. A qubit lifetime getting close to $\gamma_{\text{sp}}^{-1}=10ms$ corresponds to an optimistic long-term vision for the qubits and $\epsilon=10^{-2}$ a very optimistic value for the classical electronics ($\epsilon=10^{-1}$, not represented, is probably more plausible \cite{bardin201929}). \textbf{a)} Minimum power (in Watts), \textbf{b)} Optimum temperature for the signal generation stage (in Kelvin), \textbf{c)} Optimum attenuation level, \textbf{d)} Optimum qubit temperature (in Kelvin). The vertical black dotted lines delimit the regions associated to different values for the concatenation level.}
\label{fig:optimal_function_gamma_with_attenuators_consoperqubit}
\end{center}
\end{figure}
\FloatBarrier
Before explaining what is happening physically on those graphs, we can first comment that the power consumption varies between $1W$ and $10MW$ here (as said before, we only care about orders of magnitudes because of the approximations we did, so $80 MW$ or $13MW$ for instance, will for us be "mapped" to $10 MW$ in all our comments). Also, we are assuming a Carnot efficiency for the cryogenics, which, as explained in \ref{sec:efficiency_is_carnot} is likely to be between $1$ and $10$ times more efficient than plausible \textit{large-scale} cryogenics (the power obtained should be multiplied by a factor between $1$ and $10$ to get an estimation closer to a realistic but well-designed cryostat). A first general message we can give at this point is thus that fault-tolerant quantum computing might be demanding in terms of power consumption, and this is an aspect that should be considered in the design. However, we emphasize on the fact we are talking about power and not energy. This power would be consumed in a very short amount of time. For $k$ level of concatenations, from the discussions before the figure \ref{fig:timestep_logical_gate} we have a total duration of the algorithm being $D_L*\tau_L^{\text{timestep}}=D_L*3^k*\tau_P$. It is in the millisecond range for $k=3$ here (the energy required for the algorithm, when there is $10MW$ of power consumption would then be comparable to the one spent by a typical heater turned on for $10s$). This is something to keep in mind. Then, we can also notice that all the optimal parameters are greatly varying as a function of the qubit lifetime and energetic performance of the electronics (as encoded in the parameter $\epsilon$). Typically, the temperature of the qubits varies between $10mK$ and $1K$, the signal generation temperature between $10K$ and $300K$, and the attenuation between $10 dB$ and $70dB$ on all those graphs. It gives us a second general message: the optimum parameter to choose in the design are strongly dependent on the performance of the different components used inside the computer, and enforcing the qubits to be exactly at $10mK$ is, for instance, not something absolutely required (and probably not even desirable in order to make the computer energy efficient, we will discuss this point in further details).

That being said, we can now give overall intuitions behind the behavior of those curves. We notice that the better the qubit lifetime is, the lower the power consumption is, and the tunable parameters $T_{\text{Q}},T_{\text{Gen}},k,A$ seem to have well-defined variations. The concatenation level has the tendency to be reduced with better qubit lifetime and when $k$ is constant, $T_{\text{Q}},1/A$ and to some extent $T_{\text{Gen}}$ are all increasing with $\gamma_{\text{sp}}^{-1}$. In order to understand why we need to recall that the minimum of power consumption such that the algorithm succeeds with at least a targetted success rate is necessarily reached when $\mathcal{M}_L=\mathcal{M}_{\text{target}}$. Let's assume that we were on a point of minimum power consumption, but we increased $\gamma_{\text{sp}}^{-1}$ for all other parameters fixed. Because of that, $\mathcal{M}_L$ will be modified (the quantity of noise decreases), and we will then have $\mathcal{M}_L < \mathcal{M}_{\text{target}}$, which from the property \ref{prop:behavior_min_power} in the previous chapter cannot correspond to the point of minimum power. In order to find the minimum power consumption, the tunable parameters $T_{\text{Q}}$, $A$, $T_{\text{Gen}}$ and $k$ will be tuned to make the equality $\mathcal{M}_L=\mathcal{M}_{\text{target}}$ satisfied again. The way it can do this is by either increasing $T_{\text{Q}}$, $1/A$, $T_{\text{Gen}}$ or lowering the concatenation level as such modifications are increasing the quantity of noise as one can see from \eqref{eq:total_noisy_photons} (if we inject the expressions  \eqref{eq:choice_temp} and \eqref{eq:att_per_stage} in \eqref{eq:total_noisy_photons}, we see that increasing $1/A,T_{\text{Q}}$ or $T_{\text{Gen}}$ will increase the number of noisy photons at $T_{\text{Q}}$, and thus will increase $\mathcal{M}_L$ through \eqref{eq:eta} and \eqref{eq:metric_logical}). And, of course, a combination of those changes can also occur. Now, what will exactly occur will depend on how each of those tunable parameters affect the power consumption. The optimization will, in priority, change the parameters which possibly increase the noise but decrease in a "large" manner the power consumption. Let us go a little bit deeper in the understanding. First, the parameter which will have the biggest influence on the power consumption is clearly the concatenation level. To give ideas, the number of physical qubits for $k=1,2,3$ is respectively $\approx 10^5, 10^7, 10^{10}$. Reducing the concatenation level thus has a huge impact on power consumption. But such reduction can only be made for a low enough noise. If it is not possible to reduce the concatenation level, there are other ways to reduce the power consumption: one can increase $T_{\text{Q}}$ or $1/A$. This is what is happening for fixed value of $k$ on the curves of figure \ref{fig:optimal_function_gamma_with_attenuators_consoperqubit}.  For instance when $k=2$, $T_{\text{Q}}$ grows with $\gamma_{\text{sp}}^{-1}$ between $30mK$ and almost $1K$ (for $\epsilon=1$ or $10^{-2}$). A similar behavior seems to occur for the attenuation even though it is only really clear for $\epsilon=10^{-4}$: the curves seem to be "noisy" for $\epsilon=1$ or $10^{-2}$. We will comment on why in a few paragraphs. It is also possible to adapt $T_{\text{Gen}}$ as a function of $\gamma_{\text{sp}}^{-1}$, we see for instance that the curve for $\epsilon=10^{-2}$ slightly varies. But the variations for this variable are less intuitive to guess. Indeed while it is clear that increasing $T_{\text{Q}}$ or $1/A$ increases the quantity of noise while reducing the power, for $T_{\text{Gen}}$ the power function itself, i.e., without considering constraining the tunable parameters by the success condition might have non-trivial variations as a function of $T_{\text{Gen}}$ (for all other parameter fixed). It comes from the competition between the desire to choose $T_{\text{Gen}}$ low in order to reduce the heat conduction and to choose it being high (close to $300K$) to increase the efficiency at which the heat dissipated by the electronics is evacuated. 

Now, we can also explain why the curves representing the attenuation are "blurred" in the low $\gamma_{\text{sp}}^{-1}$, high $\epsilon$ regime. The reason is that the minimum is being found by sampling over the variables $T_{\text{Q}},T_{\text{Gen}},k$, and $A$ is found from the knowledge of those variables by calculating the roots of the polynomial equation \eqref{eq:polynomial_to_solve_for_A}. Our sampling having a limited accuracy, we will be limited in the accuracy to find those roots, but it cannot be the only explanation because on other regimes, we use the same sampling, and the curves are much more smooth. The additional element explaining this phenomenon is that the regime where those curves seem blurred correspond to where the heat conduction is important (because $T_{\text{Gen}}\approx 300K$), and the heat dissipated by the electronics as well (we recall that $\epsilon$ affects both the signal generation stage but also the amplifier stage at $T_{\text{Amp}}=4K$; thus even though $T_{\text{Gen}}=300K$ $\epsilon$ has an influence). In this regime, the attenuation is not what plays a dominant role, and finding the best attenuation is no longer critical\footnote{What we say here is not a proof, it is something that we checked numerically. Indeed when $T_{\text{Gen}}$ is closer to $300K$, the heat conduction might increase but the heat dissipated in the attenuator as well because a greater attenuation would be required. After numerical verifications (that we will partially see later on), we saw that the attenuation is indeed not playing a major role in those regions.}. For this reason, the sensitivity on the optimal value of $A$ will be reduced, which also participates in increasing the noise of this curve in this regime. We will justify better that the attenuation is indeed not what plays the dominant role in this regime in a next series of plots.

Now, we can be a little bit more precise about the behavior of the parameters when the optimal concatenation level is changing. Basically, when the qubit lifetime gets better and better, at some point, one can relax the level of concatenation. It is occurring for $\gamma_{\text{sp}}^{-1} \approx 2ms$ and $10ms$. It drastically reduces the power consumption, but it comes with a cost: in order to allow this change of concatenation to occur, the experimentalist will have to drastically reduce the physical noise, which is done by reducing $T_{\text{Q}}$, $1/A$ and (a little bit) $T_{\text{Gen}}$. This is what explains the kind of "reset" in all the tunable parameters that occur when the concatenation level is being changed. 

We can also comment that it is not obvious to have a power being reduced when the qubit lifetime gets better. Indeed, when $\gamma_{\text{sp}}^{-1}$ increases, the noise is being reduced, but the coupling between the qubit and the waveguide as well. It will have for effect (for all other parameters fixed) to increase the quantity of power required in the pulse to drive the qubits, which would increase the heat dissipated in the attenuators. The fact that the power \textit{is} actually been reduced with an increase of $\gamma_{\text{sp}}^{-1}$ shows indirectly that the behavior we just explained does not play a big enough role, but it was important to notice it.

Now that we provided some general messages about the power consumption and gave intuition about the variation of the different parameters, we are going to give messages about (i) how much power we are saving with our optimization, (ii) what is the most "power expensive" between the electronics the heat conduction and the dissipation in the attenuators, and (iii) what should be done in order to reduce the power consumption. For (iii), we can, of course, partially answer this last question by saying that the consumption of the electronics should be reduced and the qubit lifetime should be better, but 
those are "expected" results, we would like to give more \textit{quantitative} messages, explain what should be done in \textit{priority}, and how much we can expect to save quantitatively by doing those appropriate change.
\subsubsection{Quantitative interest of the optimization}
\label{sec:quantitative_interest}
Here, we are giving quantitative values about how much power we can hope to save with our optimization. In order to do this, we need to consider some "reference" situations to compare our optimization with. We consider here two main scenarios. The first one will consist in putting the qubits at a temperature $T_{\text{Q}}=10mK$, the signal generation at a temperature $T_{\text{Gen}}=4K$, choosing an attenuation given by the ratio $A=T_{\text{Gen}}/T_{\text{Q}}=400 \approx 26dB$. This choice for $T_{\text{Gen}}$ is motivated by recent proposals \cite{bardin201929,patra2020scalable} suggesting to put the electronics that is generating the signals driving the qubits close to $4K$, we refer to the discussion made in \ref{sec:global_picture} for the motivations behind this temperature. We choose a level of attenuation given by the ratio of temperatures as it corresponds to what is typically put experimentally; see the section \ref{sec:global_picture}. The concatenation level will then be chosen by estimating the minimum value it should have in order to have an algorithm that will succeed at least as often as the chosen target. For instance, $\mathcal{M}_L$ is calculated for the specified values of $T_{\text{Q}}, T_{\text{Gen}}, A$ for $k=0,1,2...$. Then the minimum value of $k$ allowing to have $\mathcal{M}_L \leq \mathcal{M}_{\text{target}}$ is the one that will be chosen. This choice for the concatenation level is being made in order to have the minimum number of physical qubits required to implement the algorithm successfully. It corresponds to the "standard way" to find the appropriate concatenation level to use.

The second scenario we consider is basically the same excepted that the signal are generated in the laboratory at $T_{\text{Gen}}=300K$. Thus we will also have $T_{\text{Q}}=10mK$, the attenuation also satisfies $A=T_{\text{Gen}}/T_{\text{Q}}$ (but it is now equal to $30000\approx 45 dB$), and the concatenation level is also being chosen by finding its minimum value allowing to have $\mathcal{M}_{L} \leq \mathcal{M}_{\text{target}}$. On the figure \ref{fig:gain_in_consumption_function_gamma_with_attenuators_consoperqubit}, we represent on a) the ratio between the power that would be obtained for the scenario in which $T_{\text{Gen}}=300K$ compared to our optimization, thus the quantity $P_{300K}/P_{\min}$. On b) we plotted the ratio associated to the other scenario: $P_{4K}/P_{\min}$.
\begin{figure}[h!]
\begin{center}
\includegraphics[width=1.0\textwidth]{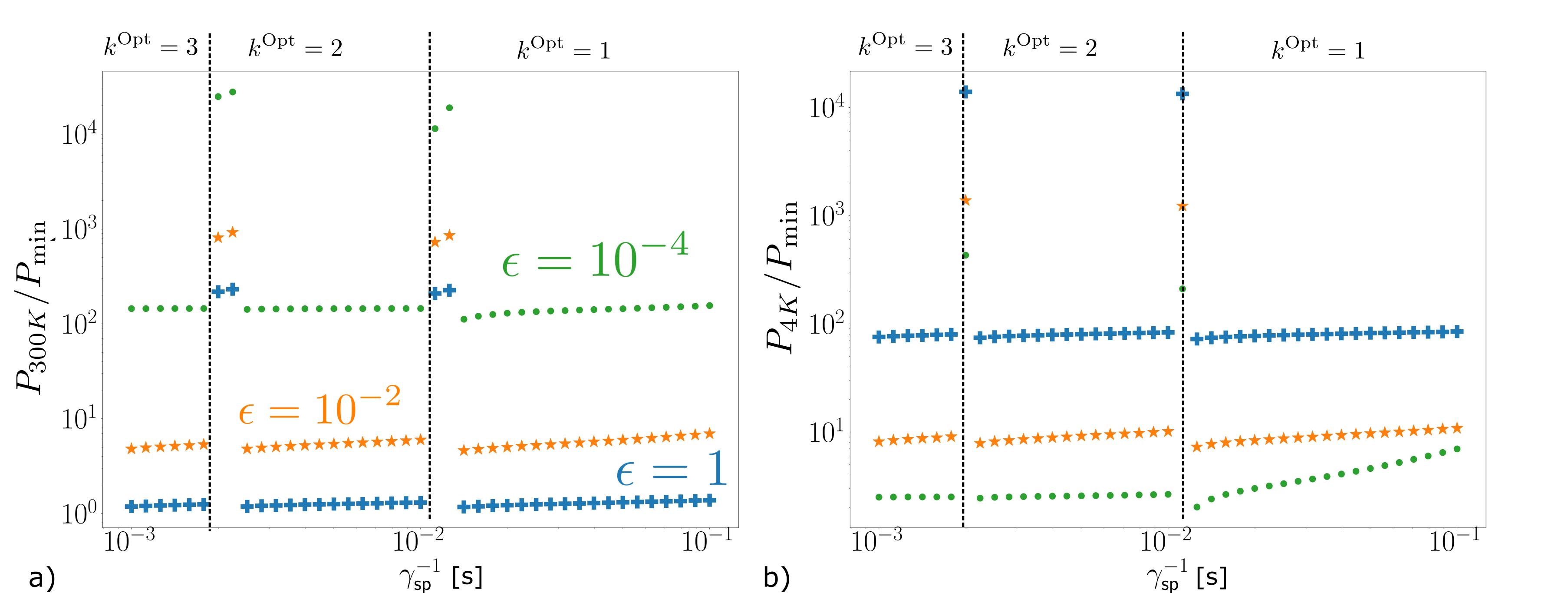}
\caption{How much power is being saved with the optimization compared to "typical" choice of parameters. \textbf{a)} Ratio between the power consumption of an unoptimized scenario where the signals are generated at $300K$ divided by the optimal power consumption we found. \textbf{b)} Ratio between the power consumption of an unoptimized scenario where the signals are generated at $4K$ divided by the optimal power consumption we found. For the unoptimized scenario, the choice of the other parameters ($T_{\text{Q}},k$ and $A$ is described in the main text). The simulations done here for the optimized scenario are the exact same as the ones shown in the figure \ref{fig:optimal_function_gamma_with_attenuators_consoperqubit}. The concatenations level written on this image and separated by the vertical black dotted lines are the ones that correspond to the curves $P_{\min}$. The concatenation levels associated with the curves $P_{300K}$ and $P_{4K}$ are not necessarily the same (see comments in the main text).}.
\label{fig:gain_in_consumption_function_gamma_with_attenuators_consoperqubit}
\end{center}
\end{figure}
\FloatBarrier
From those graphs, we see that the quantity of power that can be saved can vary between almost nothing (a ratio slightly bigger than $1$) and $4$ orders of magnitudes. Let us comment on the graph a) first. We can notice there that some values are much higher than others, close to $\gamma_{\text{sp}}^{-1} \approx 2ms$ and $\gamma_{\text{sp}}^{-1} \approx 10ms$ indicating that the optimization is saving a greater amount of power there. They are exactly matching the moment when our optimization is changing the value of the concatenations. What is happening there is that the optimization is saving us power by reducing the value of the concatenation level. For instance, if we focus on the blue curve (i.e $\epsilon=1$), for $\gamma_{\text{sp}}^{-1} \approx 2ms$, the concatenation associated to $P_{\min}$ changes from $k=3$ to $k=2$. On those points, $P_{300K}$ remains at three levels of concatenations (this cannot be understood from the graph as we did not represent the value of the concatenation level for $P_{300K}$, but we would find that $k=3$ for this curve on the two "peaked value" around $\gamma_{\text{sp}}^{-1}=2ms$). Because we save one level of concatenation compared to $P_{300K}$, a great amount of power can be saved as many physical qubits and gates can be removed from the computer; this explains the peaks observed. The legitimate question to ask is then: what is our optimization doing to "save" this extra level of concatenation? This question can be answered by looking at the blue curve (i.e $\epsilon=1$) on the figure \ref{fig:optimal_function_gamma_with_attenuators_consoperqubit} c). We see on this graph that much more attenuation than what is usually prescribed is being put: we have about or more than $60 dB$ where the blue curve of figure \ref{fig:gain_in_consumption_function_gamma_with_attenuators_consoperqubit} a) is being peaked. By putting more attenuation, the probability of error per physical gate is reduced enough to avoid doing an extra concatenation, and it thus saves a large amount of power. This will be the same principle for the other "peaked values" that we can see on either figure \ref{fig:gain_in_consumption_function_gamma_with_attenuators_consoperqubit} a) or b) (we will not comment on them too much). Now those peaked values are very specific; they are the ones in which the qubit lifetime (and also the characteristics of the algorithm) are close to a critical value where the concatenation level is about to change, the interest is thus in some sense limited to very particular scenarios. This is why we now study the behavior outside of those peaks.

Let us focus again on the blue curve ($\epsilon=1$) of figure \ref{fig:optimal_function_gamma_with_attenuators_consoperqubit} a). Outside of the peaked values, compared to the scenario $P_{300K}$, we do not save much power. One of the reasons is because the optimum temperature we found for $\epsilon=1$ is actually $T_{\text{Gen}}=300K$ (same temperature than for $P_{300K}$), and also because $P_{300K}$ and our optimum are associated to the same concatenation level (again outside of the peaked values). The only way to save power consumption is then to tune the attenuation and qubit temperatures appropriately. We can see on figure \ref{fig:optimal_function_gamma_with_attenuators_consoperqubit} c) and d) that $A \sim 60 dB$ and $T_{\text{Q}} \sim 100mK$ which is quite different from the typical $40dB$ and $10mK$ that are behind the curve $P_{300K}$. But then, even with such different values, the power consumption for the optimized scenario for $\epsilon=1$ is pretty similar to the one found for $P_{300K}$. It allows us to deduce that the most critical parameter to fix, outside from $k$, is actually $T_{\text{Gen}}$. A negligible amount of power will be saved by appropriately choosing the attenuation or qubit temperature: it is not \textit{in this specific example} what will play the most significant role. In the appendix \ref{app:appendix_example}, we show an example in which optimizing the qubit temperature as a function of the qubit lifetime (when the optimal concatenation is not varying) is something that is important. The remark that wisely choosing $T_{\text{Gen}}$ is an important requirement is also something that we can notice on the other graph (figure \ref{fig:gain_in_consumption_function_gamma_with_attenuators_consoperqubit}
b), still for the curve where $\epsilon=1$). Forcing the signal generation stage to be at $4K$, as it is one of the solutions considered in some recent studies \cite{patra2020scalable,bardin201929} would make us spend $\times 100$ more power than what is necessary. If we keep in mind that for $\gamma_{\text{sp}}^{-1}=1ms$ the minimized power consumption is in the $10 MW$ range (see the figure \ref{fig:optimal_function_gamma_with_attenuators_consoperqubit} a)), we see how important fixing $T_{\text{Gen}}$ is (otherwise we would spend power in the $GW$ range).

Now, we mainly commented the curve where $\epsilon=1$, but we can also look at $\epsilon=10^{-2}$, the optimum temperature we found there is non trivial: $T_{\text{Gen}} \approx 50K$ as we can see from the figure \ref{fig:optimal_function_gamma_with_attenuators_consoperqubit} b), and we save one order of magnitude compared to the two unoptimized scenarios $P_{300K}$ and $P_{4K}$ allowing us to remain in a consumption about $10 kW$ instead of $100 kW$ for $\gamma_{\text{sp}}^{-1}$ below (but close) to $10ms$ for instance.  All this illustrates that wisely choosing $T_{\text{Gen}}$ is an important requirement, and the optimum value for this temperature is not easy to guess in advance. 

In the end, those results illustrate that considering a transversal and optimized vision, relating the noise to the resource spent, in the design of a quantum computer can greatly improve the potential in terms of scalability: for $\gamma_{\text{sp}}^{-1}=1ms$, we see in our examples that we can reduce the power consumption from the gigawatt range (which corresponds to a scenario where the electronics is at $4K$ which is sometimes considered as a good choice for scalability \cite{bardin201929,patra2020scalable}) to $10 MW$, making the architecture much more scalable\footnote{Of course this quantitative conclusion is associated to the concatenated code we use. For an analog regime of parameters, for another code, we could eventually find that putting the electronics at $4K$ is a good choice. This is something that would be interesting to study.}.
\subsubsection{Power consumption as a function of classical electronics performances} 
\label{sec:power_consumption_function_electronics_performances}
So far, we studied the behavior as a function of the qubit lifetime for different values of the dissipation for the electronics. Now, we wish to study the power consumption as a function of the performance of the electronics. The graphs are represented on the figure \ref{fig:optimal_function_epsilon_with_attenuators_consoperqubit}. 
\begin{figure}[h!]
\begin{center}
\includegraphics[width=1.0\textwidth]{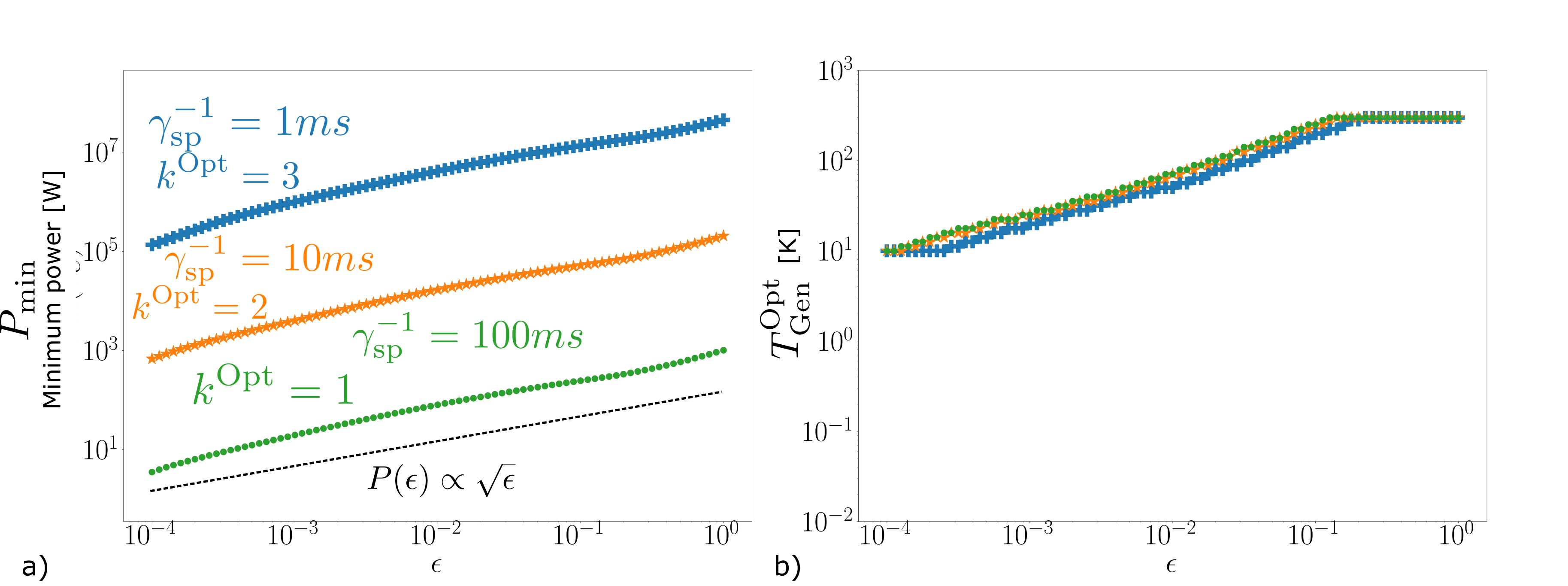}
\caption{How the performances of the electronics influence the overall energetic cost. The blue, orange and green curves are associated to qubit lifetime $\gamma_{\text{sp}}^{-1}$ being respectively $1ms$, $10ms$, $100ms$. The optimum concatenation level found for a curve having a given value for $\gamma_{\text{sp}}^{-1}$ has been found constant (the values are written on the graph \textbf{a)}). \textbf{a)} Minimum power as a function of $\epsilon$. The black dotted line is a guide for the eye showing a dependance of the minimum power in $\sqrt{\epsilon}$. \textbf{b)} Optimum signal generation temperature $T^{\text{Opt}}_{\text{Gen}}$ as a function of $\epsilon$.}.
\label{fig:optimal_function_epsilon_with_attenuators_consoperqubit}
\end{center}
\end{figure}
\FloatBarrier
First, we see two expected behaviors: the more the electronics dissipates heat, the bigger the power consumption is (on figure \ref{fig:optimal_function_epsilon_with_attenuators_consoperqubit} a) ), and the higher the temperature of the signal generation stage should be (on b)). We also see that the optimal concatenation level $k^{\text{Opt}}$ found by the optimization (they are written on the figure \ref{fig:optimal_function_epsilon_with_attenuators_consoperqubit} a)) does not vary with $\epsilon$. There is no fundamental reason for that that we found, in principle, $k^{\text{Opt}}$ should depend on $\epsilon$ because both parameters are related through the minimization under constraint. Without being rigorous, we can understand that it is something that might be unlikely to happen. Indeed, $\epsilon$ does not play a role in the expression of the metric. Then, in order to maintain the targeted accuracy, the value of $k^{\text{Opt}}$ should not change. Now, this is not entirely accurate because changing $\epsilon$ could "in principle" impact in an indirect manner the noise. For instance, if we imagine that $\epsilon$ becomes low enough in order to make $T_{\text{Gen}}^{\text{Opt}}$ low enough such that it removes almost entirely the influence of the thermal noise coming from this temperature, then we could imagine that changing $\epsilon$ would have an impact on the optimal concatenation level $k^{\text{Opt}}$. Indeed, the noise being lower, one level of concatenation could be removed, for instance. But we believe this behavior to be unlikely because $A^{\text{Opt}}$ is also adapting accordingly to $T_{\text{Gen}}^{\text{Opt}}$ in order to keep the thermal noise low enough. Thus it might be possible that because everything is being optimized, changing $\epsilon$ would, in the end, not impact in a significant manner the thermal noise on the qubit, and thus, not change the optimal concatenation level $k^{\text{Opt}}$.

Now, we can also notice that the optimum temperature for the signal generation stage does not depend very much on $\gamma_{\text{sp}}^{-1}$ for a whole set of $\epsilon$ (this is actually something we could already understand from figure \ref{fig:optimal_function_gamma_with_attenuators_consoperqubit} b)), and we can see that the minimum power consumption grows proportionally with the square root of the heat dissipated by the electronics. This is represented by the fact that the slope of the various curves on figure \ref{fig:optimal_function_epsilon_with_attenuators_consoperqubit} a) is close to a curve growing as $\sqrt{\epsilon}$. However, we did not find a fundamental reason why; we only noticed it to be true with our model.
\subsubsection{How to make the computer more energy efficient}
\label{sec:making_energy_efficient}
Here, we will briefly discuss different approaches we can choose to make the computer more energy efficient. In our discussions, we will take some degree of liberty from the strict regime of validity of our model, and we are going to take some freedom to what is strictly doable with today's technology. What we are going to do first will be to remove the attenuators and assume they can be replaced by non-dissipative, reflective filters. Basically, instead of removing the noise based on a process that dissipates heat, we will use filters that will reflect the thermal noise. This is not something that is frequently used in the context of quantum computing, but such devices can, in principle, be realized. We will also assume that the electronics can be enabled and disabled "on-demand". Basically, with this approach, the DAC, ADC, amplifiers, and multiplexing/demultiplexing units are turned off when they are not required. In principle, CMOS electronics has two components in its power consumption. One component is called static consumption, and it corresponds to leakage currents that might be between the ground and the positive DC voltage. Whatever the electronics is actually doing, those leakage currents are always here, and they are inducing heat dissipation. There is another component, called dynamic consumption (it is the same vocabulary as the one we are using for the power consumption of the quantum computer). This one is only here when the logic states of the transistors are switching from one value to another, and it induces currents that will dissipate heat on the switching events. The dynamic consumption is almost constant as a function of the temperature while the static is dominant at high temperature but rapidly decreasing with it \cite{rabaey2003digital}. Here, we will assume that we can entirely neglect the static consumption such that the electronics will strictly dissipate no heat when it is not being used. This is a possibly too idealistic scenario, but we would like to see in this extreme scenario how knowing the fact that the number of gates acting in parallel is different from the number of physical qubits can allow lowering the power consumption. The results are represented on the figure \ref{fig:optimal_function_gamma_without_attenuators_consoperqubitorpergate}. 
On a) is represented the power consumption removing any heat dissipated in the attenuators, and on b) we also do not have any attenuator, but additionally, the electronics doesn't consume any power when it is not being used.
\begin{figure}[h!]
\begin{center}
\includegraphics[width=1.0\textwidth]{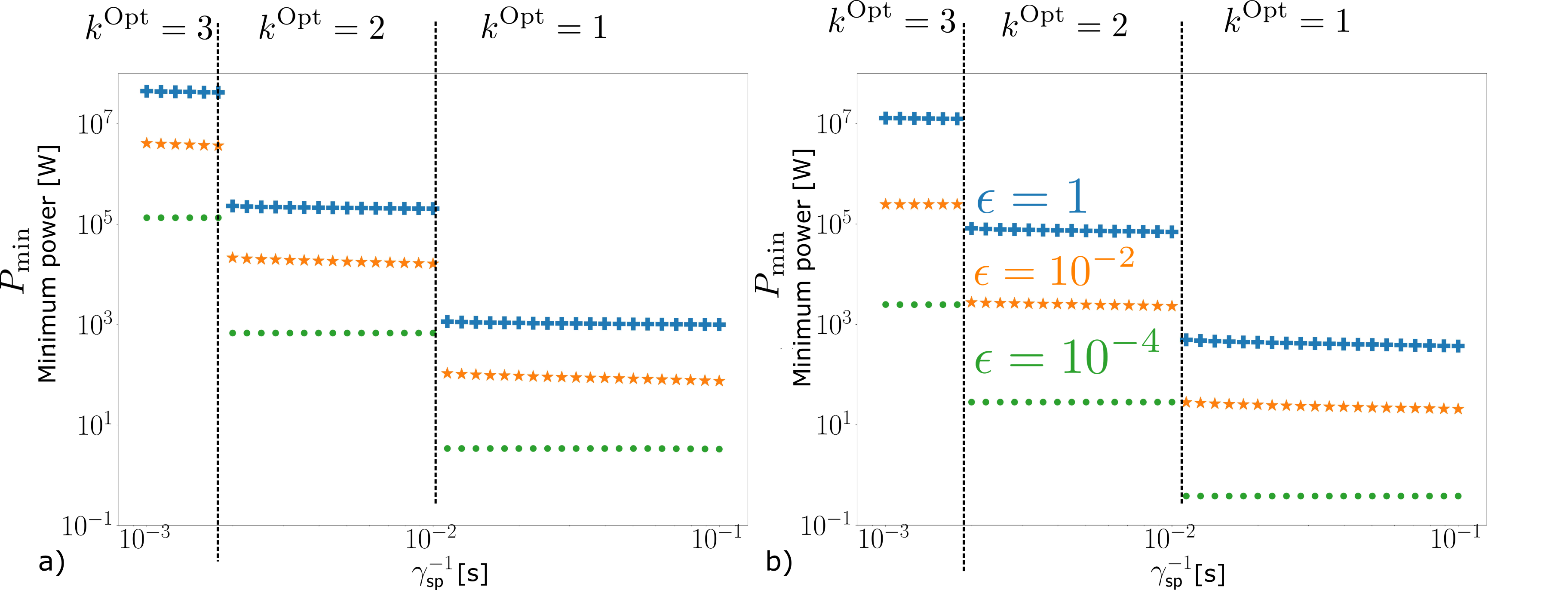}
\caption{\textbf{a)} Minimum consumption assuming no dissipation in the attenuators, i.e., they have been replaced by non-dissipating filters. \textbf{b)} Minimum consumption assuming no dissipation in the attenuators \textit{and} the electronics can be turned off dynamically when it is not being used (and assuming that it doesn't dissipate any heat when it is turned off).}.
\label{fig:optimal_function_gamma_without_attenuators_consoperqubitorpergate}
\end{center}
\end{figure}
\FloatBarrier
Comparing the figure \ref{fig:optimal_function_gamma_without_attenuators_consoperqubitorpergate} a) with \ref{fig:optimal_function_gamma_with_attenuators_consoperqubit} a), we see that the attenuators are not the main responsible for the power consumption of the computer. Indeed, removing them doesn't change the minimum power consumption very much quantitatively. We can, however, see some small changes in the curves in the high qubit lifetime, low electronics consumption (it can be seen, for instance, with the first points of the green curve ($\epsilon=10^{-4}$) occurring when $k^{\text{Opt}}$ switches from $2$ to $1$). The fact that the attenuators are playing a more important role in this region makes sense: when $\epsilon$ is getting lower, the dissipation from the electronics will have a lower impact on the power consumption than the heat dissipated in attenuators for instance\footnote{To be very precise, this is not really a proof because when the electronics dissipates a fewer amount of heat, $T_{\text{Gen}}^{\text{Opt}}$ is likely to be colder (thus closer to $T_{\text{Q}}^{\text{Opt}}$). In this case, a fewer amount of attenuation would be required, reducing as well its role in power consumption. This previous comment must then be taken with a bit of care. However, we believe that what we said is still true on a "qualitative" aspect because of the values we find for $T_{\text{Gen}}^{\text{Opt}}$ in the regimes of low $\epsilon$.}. Additionally, when the qubit lifetime is high, the coupling between the qubit and the waveguide gets lower, so pulses of higher amplitude have to be generated (and then dissipated). On the figure \ref{fig:optimal_function_gamma_without_attenuators_consoperqubitorpergate} b), we can see that having the ability to turn off the electronics when it is not being used, and assuming that when turned off, it dissipates no heat has, however, an important influence on the power consumption. We can see, for instance, that for $\epsilon=10^{-2}$, the consumption can be reduced by almost two orders of magnitudes in the regime $\gamma_{\text{sp}}^{-1}$ is between $2ms$ and $10ms$. The reason why the consumption is however not reduced that much for $\epsilon=1$ is because the power consumption of the electronics is still too high such that $T_{\text{Gen}}^{\text{Opt}}$ cannot be reduced "too far" from $300K$ (the graph is represented on the appendix \ref{app:dynamic_electronics})\footnote{To understand it simply, we can imagine that $T_{\text{Gen}}^{\text{Opt}}$ would remain at $300K$ with a reduction of $\epsilon$. In this case, the power consumption wouldn't vary as evacuating heat at ambient temperature is energetically free. What is happening here is in the same spirit ($T_{\text{Gen}}^{\text{Opt}}$ is reduced a bit but not "that much" from $300K$).}. It shows that the gain in power consumption by using the fact electronics can be dynamically turned on or off might only be effective if the electronics already dissipates a low amount of heat from the beginning. The gain observed is then mainly due to the fact that the classical amplifiers dissipate a fewer amount of heat.  

As a conclusion here, we see with the model we considered that the attenuators are not playing the most important role in the energetic balance. What plays the most important role is the heat conduction and the heat dissipation of the electronics. Allowing to turn off the consumption of the electronics when it is not being used is something that shows a clear advantage in terms of power consumption, but it can only show this advantage if the electronics consumed a low amount of power from the beginning (because otherwise, $T_{\text{Gen}}^{\text{Opt}}$ would be close from $300K$, and there wouldn't have a big gain).
\section{Energetic cost as a function of algorithm characteristics}
\label{sec:charac_algo}
In this section, we would like to study how the power consumption depends on the shape of the algorithm that has been implemented. For this reason, we are going to study the energetic cost of a quantum memory as a function of the number of logical qubits $Q_L$ and the number of logical timesteps (i.e., logical depth) $D_L$. Indeed, as we are going to explain, the energetic cost of a quantum memory characterized by $Q_L$ and $D_L$ can give a good idea of the energetic cost for any algorithm which will be based on the same number of logical qubit and depth, and that uses all its logical qubit for all the logical timesteps (i.e., such that all the logical qubits will for any logical timestep "do something"). It corresponds to the figure \ref{fig:algo_approx_memory} b) for instance (but not to the algorithm represented on a)). In order to understand why we need to recall a few elements. 

First, as we showed in the previous chapter and as it is illustrated in \eqref{eq:logical_power_ch5}, \eqref{eq:metric_logical} and the few lines of explanation following, the only characteristics of the algorithm that we need to estimate the power consumption are $\overline{N}_{L,\parallel}^{x}$ where $x$ is the type of logical gate to implement ($x \in \{\text{1qb},\text{cNOT},\text{Id}\}$), the logical depth $D_L$ and the number of logical qubits $Q_L$. To be a little bit more precise, the expression of the power function $P$ only requires to know the quantity $\overline{N}_{L,\parallel}^{\text{tot}} \equiv 2\overline{N}_{L,\parallel}^{\text{cNOT}}+\overline{N}_{L,\parallel}^{\text{1qb}}+\overline{N}_{L,\parallel}^{\text{Id}}$ and $Q_L$, while the metric only needs to access $N_L=D_L(\overline{N}_{L,\parallel}^{\text{cNOT}}+\overline{N}_{L,\parallel}^{\text{1qb}}+\overline{N}_{L,\parallel}^{\text{Id}})$. For a quantum memory, we will have $\overline{N}_{L,\parallel}^{\text{tot}}=\overline{N}_{L,\parallel}^{\text{Id}}=Q_L$: all the logical qubits are affected by a logical identity. For an algorithm using all its logical qubit for any timestep, we will also have $\overline{N}_{L,\parallel}^{\text{tot}}=Q_L$, but with this time $\overline{N}_{L,\parallel}^{\text{tot}}=2\overline{N}_{L,\parallel}^{\text{cNOT}}+\overline{N}_{L,\parallel}^{\text{1qb}}+\overline{N}_{L,\parallel}^{\text{Id}}$ (it is not only composed of logical identity). Thus we see here that at least the power function would not be modified between a quantum memory or another algorithm using all its logical qubit at the same time. But the power function is not the only thing to consider we must also study how the metric would change between a quantum memory and another algorithm. We recall that the metric can be estimated as $\mathcal{M}_L=N_L p_L$ where $N_L$ is the total number of logical gate used in the algorithm. For a quantum memory, we have $N_L=D_L \overline{N}_{L,\parallel}^{\text{Id}}=D_L Q_L$, but for an algorithm using all its logical qubits at all time, we have $N_L=D_L (\overline{N}_{L,\parallel}^{\text{Id}}+\overline{N}_{L,\parallel}^{\text{1qb}}+\overline{N}_{L,\parallel}^{\text{cNOT}})=D_L (Q_L-\overline{N}_{L,\parallel}^{\text{cNOT}})$. Thus, the metric will not be exactly the same for both situations. Fortunately, the mistake we would do by considering $N_L \approx D_L Q_L$ for the algorithm will not be very important for all practical purposes. Indeed, as by the assumption we made on the "shape" of the algorithm, we have $Q_L=2\overline{N}_{L,\parallel}^{\text{cNOT}}+\overline{N}_{L,\parallel}^{\text{1qb}}+\overline{N}_{L,\parallel}^{\text{Id}}$, and as we also necessarily have $\overline{N}_{L,\parallel}^{\text{cNOT}} \leq Q_L/2$, we deduce $Q_L/2 \leq Q_L-\overline{N}_{L,\parallel}^{\text{cNOT}} \leq Q_L$. The "mistake" we would do in the estimation of $\mathcal{M}_L$ by considering $Q_L-\overline{N}_{L,\parallel}^{\text{cNOT}} \approx Q_L$ will only be to consider the probability that the algorithm fails to be twice bigger as what it really is (in a worst-case scenario). This over-estimation of the probability of failure of the algorithm by only a factor $2$ is not expected to change dramatically the energetic estimation, as we will clearly see on the figure \ref{fig:quantum_memory} (excepted for some critical points, for instance if the optimum concatenation level found with this approximation would change under the modification $Q_L \to Q_L/2$).
\begin{figure}[h!]
\begin{center}
\includegraphics[width=1.0\textwidth]{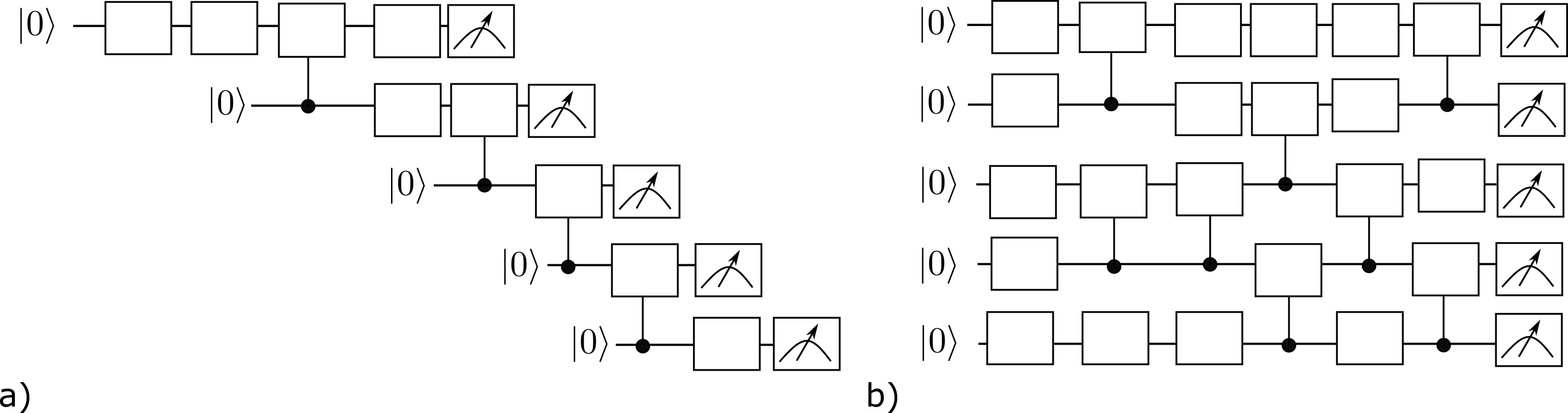}
\caption{\textbf{a)} An algorithm for which the energetic cost might be quite different from a quantum memory having the same number of logical qubits and depth. The reason is that not all logical qubits are doing something for any logical timestep: $Q_L$ will be very different from $2 \overline{N}_{L,\parallel}^{\text{cNOT}}+\overline{N}_{L,\parallel}^{\text{Id}}+\overline{N}_{L,\parallel}^{\text{1qb}}$. \textbf{b)} An algorithm for which the quantum memory approximation would give a good idea of its consumption, this algorithm has a more "rectangular shape" indicating the fact that all the logical qubits are participating in the algorithm for all the timesteps. In this example, $Q_L=2 \overline{N}_{L,\parallel}^{\text{cNOT}}+\overline{N}_{L,\parallel}^{\text{Id}}+\overline{N}_{L,\parallel}^{\text{1qb}}$ is satisfied.}
\label{fig:algo_approx_memory}
\end{center}
\end{figure}
\FloatBarrier
Thus, we believe that a quantum memory can give a good idea of the energetic cost of any algorithm\footnote{We recall that by any algorithm we consider an algorithm in which the energetic cost to implement the non-transversal gate would be negligible because the study of their energetic cost goes beyond our work.} that is making use of all its logical qubit at all times. We will show this affirmation in the following graphs. On figure \ref{fig:quantum_memory} a) is represented the power cost of a quantum memory as a function of $Q_L$ and $D_L$ varying between $100$ and $10^6$ logical qubits, where $\gamma_{\text{sp}}^{-1}=1ms$ and $\epsilon=1$ (for a case where the electronics is always turned on). On b) it is the associated optimal temperature for the qubits. The portion of those graphs having $k^{\text{Opt}} \geq 4$ is associated with a power consumption that is beyond what would be reasonable to use (to give an idea, $1 GW$ is approximately the power consumption of a nuclear reactor), but we wanted to plot values on a wide range of parameter to have intuition about how quantities are varying.
\begin{figure}[h!]
\begin{center}
\includegraphics[width=1.0\textwidth]{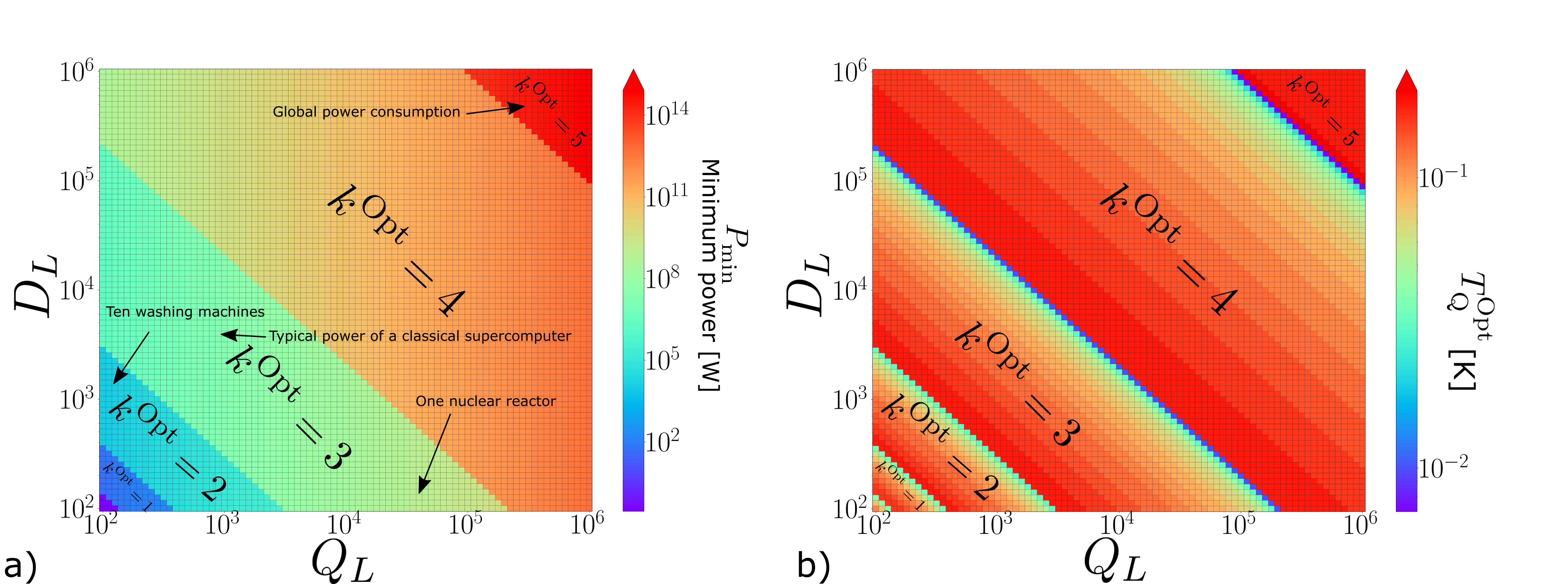}
\caption{\textbf{a)} Minimum power consumption for a quantum memory characterized acting on $Q_L$ logical qubits for a logical depth being $D_L$, $\gamma_{\text{sp}}^{-1}=1ms$ and $\epsilon=1$.  \textbf{b)} Optimum temperature $T_{\text{Q}}^{\text{Opt}}$ associated to the figure a).}.
\label{fig:quantum_memory}
\end{center}
\end{figure}
\FloatBarrier
We see on those graphs that the change of concatenation level occurs for some lines of equations $N_L=D_L Q_L = C$ for some constant $C$ (we recall that the graph is a log-log). We can interpret it as corresponding to the fact that the concatenation level changes accordingly with the size of the algorithm that is running. We also see that, for a fixed value of the concatenation level, the power consumption does not seem to vary much with the logical depth of the algorithm. However, it varies significantly with the number of logical qubits. For instance, on figure \ref{fig:quantum_memory} a), we can see that it would cost $10^{12}$ Watts to work with $D_L=10^5$ and $Q_L$ slightly below $10^6$ (to remain in the region where $k^{\text{Opt}}=4$), but about $1 GW$ for $Q_L$ close to $300$. The fact that the power does not vary very much with the logical depth (for a fixed value of the concatenation level) is not necessarily expected as when the depth increases, the error per logical gate should be reduced, which will ask to reduce the qubit temperature (we see it clearly on the figure \ref{fig:quantum_memory} b) ). Here we see that this reduction in qubit temperature is not inducing a significant variation of power consumption in the logical depth. In some sense, this is the similar behavior that occurred for the power as a function of $\gamma_{\text{sp}}^{-1}$ in the plots on figure \ref{fig:optimal_function_gamma_with_attenuators_consoperqubit}. We saw there that when the qubit lifetime is decreased (over a fixed value of $k^{\text{Opt}}$), the temperature of the qubits decreased accordingly to reduce the amount of thermal noise (because the qubits got noisier). But this did not impact very significantly the power consumption, which showed indirectly that putting the qubits at a very low temperature did not impact the power consumption significantly. Here it is the analog discussion where the role of the qubit lifetime (or, more precisely, its inverse) is replaced by the role of the depth. However, the fact that within a fixed concatenation level, adapting precisely the temperature of the qubits as a function of the qubit lifetime (or as a function of the algorithm depth) does not play an important role is a conclusion that is very "engineering architecture" dependent. In the appendix \ref{app:appendix_example}, we see an example (outside of the strict regime of validity of our model, and for a poor cryogenic efficiency) in which adapting precisely the temperature of the qubits is important and allows to make important energetic save in regime of high power consumption.

Finally, we see that the power consumption obtained on the point $Q_L=2048$, $D_L=2 Q_L$ is close to the power consumption we found in the figure \ref{fig:optimal_function_gamma_with_attenuators_consoperqubit} for $\epsilon=1$ and $\gamma_{\text{sp}}^{-1}=1ms$, it is in the $10 MW$ range. It is consistent with the fact that a quantum memory is a good approximation to the energetic cost of an algorithm.

\section{Conclusion}
\label{sec:conclusion_ch5}
In this chapter, we proposed an engineering model describing the quantum computer. We then used this engineering model in our method, which consists in finding the minimum power consumption in order to have a successful answer in the algorithm, under the constraint that the quantity of noise in this answer is below a targetted value. We applied this method in order to find the minimum power required to implement a quantum Fourier transform that has been energetically modeled by replacing all the controlled-phase gates inside with cNOTs. We saw that our approach allows to potentially save orders of magnitude in power consumption in regimes where the power consumption might be high (it can sometimes reduce the bill from the gigawatt to the tenth of megawatt). We found that the most important parameters there were to optimize were the concatenation level and the temperature at which the signals are generated. Indeed, we could deduce from our examples that a well-optimized choice for the temperature of the qubits and the total attenuation was not entirely necessary: if we had put the qubits always at $10mK$ for instance, the power consumption wouldn't be that much higher than in the regime where our optimization suggest to put them at $100mK$. We were also able to deduce from our examples that the dissipation in the attenuators is not what plays a dominant role in the power consumption; it is much more important to optimize the temperature at which the signals are being generated. We provided some possible strategies to minimize the energetic bill, which were based on allowing to dynamically turn on and off the electronics. With such strategies, as there is a difference between the number of physical qubits and the average number of gates active in parallel, a large energetic save can be done. However, they can only be gained if the electronics does not dissipate a too large amount of heat from the beginning; otherwise, the gain wouldn't be very important as their optimal temperature would be close to $T_{\text{Gen}}^{\text{Opt}} \approx 300K$. We also justified that it is possible to estimate the energetic cost as a function of the algorithm that is running by approximating the algorithm by a quantum memory having the same number of logical qubits,  the same depth, and making use of all its logical qubits at all time (in all rigor, the algorithm should not have non-transversal gates as we did not model them, we could reasonably consider that as long as it does not contain "too much" of such gates acting in parallel, the analyses done here could still hold, but this is clearly something to investigate further). We also saw in appendix \ref{app:appendix_example} that the conclusion that optimizing the temperature of the qubits properly is not what will play a major role in the reduction of power consumption is something that is not true in general. It is very architecture-dependent, and we designed an example illustrating that on a regime of constant concatenation, for a low value of electronics consumption, for a cryostat having a lower efficiency than Carnot, there is a clear interest in optimizing the qubit temperature properly (and the attenuation): orders of magnitude of power consumption can be saved showing the importance of connecting the aspect of noise to the aspect of energetic through a minimization under constraint.

As a general conclusion, we see that a superconducting qubit quantum computer might consume a large amount of power for "close to" state-of-the-art technology. More precisely, assuming very good qubits (lifetime about the $ms$), state of the art CMOS electronics to generate signals (about $5mW$ of heat dissipated per physical qubit), considering a cryostat based on Carnot efficiency (this is justified for large scale cryogenics, see the discussions in \ref{sec:efficiency_is_carnot}), and assuming that the fidelity of all the gates is only limited because of the qubit limited lifetime (we recall that this is not yet the case for the two-qubit cross resonance gates, see section \ref{sec:state_of_the_art_2qb}, it is however toward what the field is going to, and this is the "bet" for the future we are making\footnote{Anyway fault-tolerance with the code we used wouldn't be possible today with the quality of the two-qubit gates, we are forced to assume better two-qubit gates than what exist today to do our analysis. We considered that assuming that all the gates will only be limited by the qubit lifetime is a reasonable hypothesis to make for the future.}), we find that the power consumption using Steane code and its concatenated construction would be around $10MW$ for a quantum Fourier transform implemented on $2048$ logical qubits. Strictly speaking, it is the power consumption assuming that controlled-phase gates can be replaced by cNOTs in the energetic model. In some sense, because of this approximation, it is fairer to say that it would be the energetic cost of a quantum memory having the same size as the quantum Fourier transform (this is because the energetic cost of a quantum memory will be close to the energetic cost of an algorithm of similar size but not containing non-transversal gates such as the $T$ gate, see discussion in \ref{sec:charac_algo}). While it seems to be a high power consumption, it should be noticed that (i) this power consumption rapidly decreases with technological improvements, as we saw with our graphs, and (ii), it should be put in comparison of the consumption of a classical supercomputer which is in the same order of magnitude \cite{bartolini2014unveiling}. Also, if we want to compare this power to large scientific experiments such as the CERN, it would be a lower power consumption (when being used, the CERN consumes about $200 MW$ of power \cite{cernpowerconsumption}). It shouldn't be surprising that the first versions of fault-tolerant quantum computers would consume a large amount of power\footnote{Again, this conclusion of power consumption are intrinsically related to the concatenated fault-tolerant construction. To know how the energetic cost would scale for other strategies is an open question.}. Also, here we are talking about power consumption. In terms of energy, those $10MW$ would be consumed during about a millisecond as briefly explained in \ref{sec:getting_intuition}, which induces a low energy consumption. Finally, keeping those state-of-the-art values, many more optimizations could be performed than the ones we did here. For instance, we believe that a large reduction in power consumption could be gained by optimizing the temperatures of the amplification chain for the measurement outcomes and by considering a better choice for the temperatures of the intermediate stages than the one we proposed in section \ref{sec:reducing_number_tunable_parameters}. Also, more efficient recycling strategies in the ancilla used to perform error correction could be done (see all the discussions in section \ref{sec:estimation_physical_gates}).

In the very end, we believe that this work can give a first idea about the power consumption that may be required for large-scale quantum computers, in rough order of magnitudes, at least for concatenated codes. More accurate estimations could be done by refining the models; we believe that the most important part for that would be to improve the accuracy of the models describing the noise occurring during the quantum gates; here, we assumed that they are only noisy because of spontaneous emission and thermal noise, and we approximated such noise by a probabilistic model. This work also illustrates how important transversal optimizations techniques are, especially when the resource to minimize is related to the quantity of noise contained in the output of the algorithm. In our examples, we saw that $2$ orders of magnitude of power consumption could be saved in a regime where without optimizing, the power consumption would be in the gigawatt range. Those approaches can then clearly improve the potential in terms of the scalability of quantum computers.

\begin{appendices}
\chapter{Cable models}
\label{app:cable_models}
Here we provide the models of heat conduction behind the microwave cables we use. 
\section{Fourier law}
We give the basics behind the Fourier heat conduction law allowing to calculate the heat flow in a cable. Knowing the cross section $A$ of a cable, its length $L$ and the thermal conductivity $\lambda(T)$ of its constituent, we have a heat flow between $T_1 \leq T_2$, going from the high temperature to the low temperature stage being:
\begin{align}
    \dot{q}(T_1,T_2)=\frac{A}{L}\int_{T_1}^{T_2}dT \lambda(T) 
    \label{eq:fourier_heat_law}
\end{align}
In this thesis, we will take a typical length $L=20cm$ for all the cables considered between two consecutive temperature stages of the cryostat. The other characteristics of the cable (i.e., the thermal conductivity and the cross-section $A$) will depend on the cable considered.
\section{The model of cable}
In our model, we consider that for $T>10K$, conventional microwave coaxial cables must be used (i.e., they are composed of conventional conductors), and below, we can use microstrip superconducting lines. We call the cables we will consider using for $T>10K$ conventional coaxial, and for $T<10K$ superconducting microstrips. We now give their characteristics.
\subsection{Conventional coaxial}

We are going to take the typical dimension of the microwave coaxial cable called ULT-23, provided in \cite{cryoconventional}. It is composed of a metallic conductor being SUS304. This material has very close thermal properties to the 304SS, said 304 stainless steel, and the latter has a well-characterized thermal conductivity in the range $4K-300K$ \cite{marquardt2002cryogenic}, so we will consider it as being the metallic conductor\footnote{This material is anyway used in some papers such as \cite{krinner2019engineering} for a coaxial cable having a radius close from the one we consider here.}. The dimensions of this coaxial cable are the following: 
\begin{itemize}
\item inner conductor outer diameter: $0.2mm$
\item dielectric insulator outer diameter: $0.66mm$
\item outer conductor diameter: $0.86mm$
\end{itemize}
As metals conduct more heat  than insulators, and as the cross section of the insulator is comparable to the one of the metal in this cable we will only model the heat conduction associated to the stainless-steel. On the range $4K$-$300K$ its thermal conductivity $\lambda(T)$ has been characterized \cite{marquardt2002cryogenic}:
\begin{align}
    \log(\lambda_{SS}(T))=\sum_{i=0}^8 a_i \log(T)^i
\end{align}
With: $a_0=-1.4087, ~ a_1=1.3982, ~ a_2=0.2543, ~ a_3=-0.6260, ~ a_4=0.2334, ~ a_5=0.4256, ~ a_6=-0.4658, ~ a_7=0.1650, ~ a_8=-0.0199$ for $T\geq 4K$. Taking those values would provide $\lambda_{SS}(T)$ in $W/Km$ units. For $T < 4K$, we complete this model by a linear fitting (where $\lambda_{SS}(0K)=0$). This is justified as in the low temperature regime the thermal conductivity of metals depends linearly on temperature. For $T_1$ and $T_2$ greater than $10K$, we deduce $\dot{q}(T_1,T_2)$ for this model of wire.
\subsection{NbTi striplines}

For large-scale quantum computing, it would be inefficient to keep using the same conventional coaxial cables in the very low-temperature regime. And also, simply because for millions to billions of qubits, the cross-section of the microwave cables would simply be too large to connect them on each qubit. For this reason, NbTi superconductor can be used to drive the signals. The interest in using superconductors is that because of the absence of electrical resistivity, a very small cross-section can be considered. Also, superconductors have the advantage of being poor thermal conductors, which is an advantage here to reduce the energy cost. We will consider using a model of stripline as proposed in \cite{mcdermott2018quantum}. The cross-section \textit{per single line} of NbTi superconductor used is $1.5\ 10^{-11}m^2$ while the insulator around being Kapton is about $1.3\ 10^{-9}m^2$.\footnote{In \cite{tuckerman2016flexible} even smaller cross-sections than the ones we considered here have been tried and realized experimentally. They showed good properties allowing them to be used in quantum computing experiments.} The thermal conductivity of such materials is (we have two different thermal conductivity depending on the temperature range for the Kapton, this is why we have two exponents $0$ and $1$ for the thermal conductivity of the Kapton $\lambda_{Kap^{HN}}(T)$):
\begin{align}
&\lambda^0_{Kap^{HN}}(T)=4.6*10^{-3} T^{0.56} W/(K.m) \text{ ($0.5K<T<5K$)}\\
&\lambda^1_{Kap^{HN}}(T)=2.996*10^{-3}*T^{0.9794} W/(K.m) \text{ ($4K<T<300K$)}\\
&\lambda_{NbTi}(T)=27*10^{-3} T^{2} W/(K.m) \text{ ($0.05K<T<2K$) }
\end{align}
In the range NbTi and Kapton thermal conductivity are both defined (thus taking $\lambda_{Kap^{HN}}(T)=\lambda^0_{Kap^{HN}}$ for the Kapton), taking into account the respective cross-section of the materials we provided for a single line, we would see that the heat flow is dominated by the Kapton. For this reason, we will only consider its thermal conductivity in our model. In our model, we considered $\lambda_{Kap^{HN}}(T<4K)=\lambda^0_{Kap^{HN}}(T)$\footnote{It means that we extended the range of validity for the Kapton thermal conductivity below $0.5K$. This is an assumption we are making, but as the optimal temperature of the qubits found is anyway almost always bigger than $100mK$ in our work, we don't believe we would make a very big mistake with this simplification.} and $\lambda_{Kap^{HN}}(T \geq 4K)=\lambda^1_{Kap^{HN}}(T)$ (the discontinuity at $4K$ because we change of law is found to be small). At this point, we know the cross-section of the lines and their thermal conductivity, we can now compute a heat flow $\dot{q}(T_1,T_2)$ if both $T_1$ and $T_2$ are below $10K$.
\subsection{Typical heat flow values}
\label{app:typical_heat_flows}
We provide typical heat flow associated with single lines/cables.
\subsubsection{Coaxial cable model}
Here we use the coaxial cable model we described (on the full range $[0K,300K]$).
We find:
\begin{itemize}
\item $\dot{q}(0K,300K) \approx \dot{q}(10,300) \approx 4mW$
\item $\dot{q}(0K,10K) \approx 5 \mu W$
\end{itemize}
\subsubsection{NbTi striplines}
\begin{itemize}
\item $\dot{q}(0K,10K) \approx 1 nW$
\end{itemize}
We see that below $10K$, the superconducting lines are associated with a flow $\times 1000$ lower than the conventional conductor (it is one of the reasons why using superconducting lines is interesting).
\subsection{The full cable model: mixing the two}
So far, we have expressed the thermal conductivity of the two models of cables. The NbTi striplines are used below $10K$, and the conventional coaxial is used for a temperature higher than $10K$. Defining $A(T)$ as being the cross-section of the cable used (which depends on the temperature used as we have two different models of cables), we have the heat flow for one cable that is in general:
\begin{align}
\dot{q}(T_1,T_2)=\frac{1}{L}\int_{T_1}^{T_2} A(T) \lambda(T) dT,
\end{align}
where $\lambda(T<10K)=\lambda_{Kap^{HN}}(T)$ and $\lambda(T>10K)=\lambda_{SS}(T)$, and $A(T<10K)=1.3\ 10^{-9}m^2$. $A(T>10K)$ is deduced by calculating the total cross-section of the metallic conductor for the coaxial cable described there. This is what is being used in our models: this is how we consider "switching" the model of cables at $10K$. 
\chapter{Optimal temperature $T_{\text{Gen}}^{\text{Opt}}$ when the electronics is dynamically enabled.}
\label{app:dynamic_electronics}
Here, we represent the curve of $T_{\text{Gen}}^{\text{Opt}}$ associated to the figure \ref{fig:optimal_function_gamma_without_attenuators_consoperqubitorpergate}
\begin{figure}[h!]
\begin{center}
\includegraphics[width=0.5\textwidth]{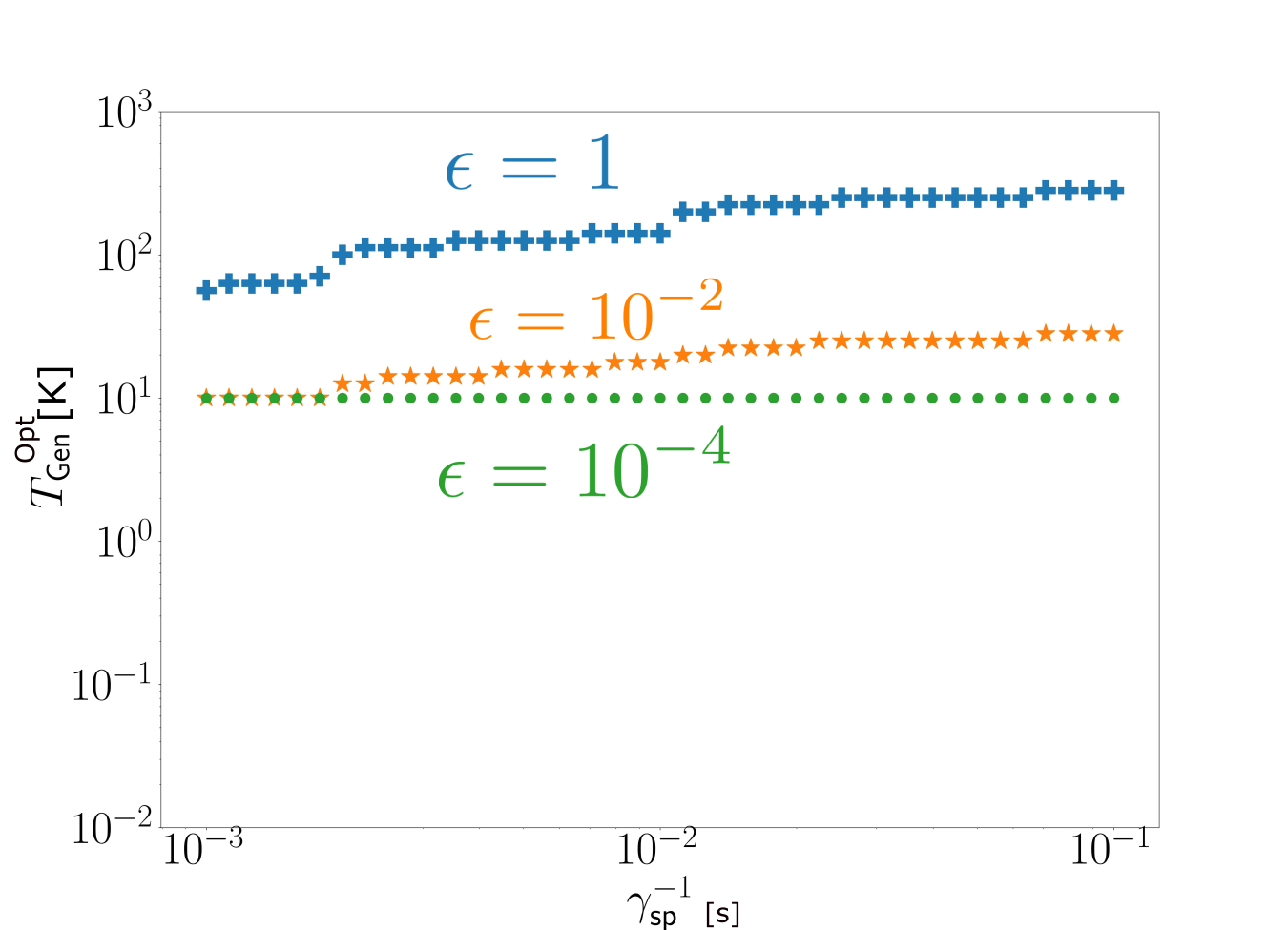}
\caption{Optimal temperature $T_{\text{Gen}}^{\text{Opt}}$ as a function of the qubit lifetime for the curve represented on the figure \ref{fig:optimal_function_gamma_without_attenuators_consoperqubitorpergate}
 b). The blue, orange and green curve correspond to $\epsilon=1$, $10^{-2}$, $10^{-3}$ respectively. We see that $T_{\text{Gen}}^{\text{Opt}}$ is higher, and not very far from $300K$ for $\epsilon=1$ as the electronics dissipates too much heat. This is why the reduction in power consumption when the electronics can be dynamically turned on is not as important in this regime as when $\epsilon$ is lower.}
\end{center}
\end{figure}
\FloatBarrier
\chapter{An example where fine-tuning of the qubits temperature and the total attenuation is crucial}
\label{app:appendix_example}
Here we give an example in which we see that the exact value $T_{\text{Q}}$ and $A$ should have in order to minimize the power consumption is important. In the main text, we saw that, of course, the qubit temperature cannot be too high, and the attenuation cannot be too low (otherwise the condition $\mathcal{M}_L=\mathcal{M}_{\text{target}}$ would either be violated, either an extra concatenation level would be required to compensate), but we did not see a very big interested into optimizing precisely those values for a fixed value of $k^{\text{Opt}}$. It showed us indirectly that in the regime we considered, the heat dissipated in the attenuator (and the exact qubit temperature) did not play the most important role in the physics as we commented in \ref{sec:making_energy_efficient}.

On the figure \ref{fig:appendix_example}, we can see the curves associated with an optimization in which we replaced the Carnot efficiency by its square. The motivation behind this choice is that it can correspond to an efficiency of cryostats that are not optimized to evacuate large quantity of heat \cite{irds,martin2021energy}; this is the kind of efficiency we can find in some "chip-scale" cryostats. This new efficiency is greatly increasing the cost of removing heat at low temperatures. We see that because of that; there is now a much more important requirement to optimize the qubit temperature and attenuation values. Indeed, even when $k^{\text{Opt}}$ is being fixed to some value, the power consumption varies significantly as a function of the qubit lifetime. This is an example in which we see that minimizing the power consumption is not only a matter of choosing the appropriate number of physical qubits per logical qubit (and the appropriate temperature $T_{\text{Gen}}^{\text{Opt}}$), even for a fixed number of physical qubits it might be crucial to optimize the qubit temperature and the attenuations. The value of $\epsilon$ we considered is, however, quite low (we reduced the consumption of the electronics by about $6$ order of magnitude compared to state of the art in CMOS technology)\footnote{In the main text, we explained that $\epsilon>10^{-3}$ is necessary to have our model valid. But it could be possible to increase much more the range of validity by considering that the waveform is saved in memory at the same stage of the ADC/DAC: this would reduce by an important amount the number of cables between $300K$ and $T_{\text{Gen}}$.}, but it could eventually correspond to other technologies such as adiabatic computing or single flux quantum logic. In summary, this example shows that in some cases, optimizing the qubits and attenuation temperature might be important to save quantitative amounts of power consumption.
\begin{figure}[h!]
\begin{center}
\includegraphics[width=0.8\textwidth]{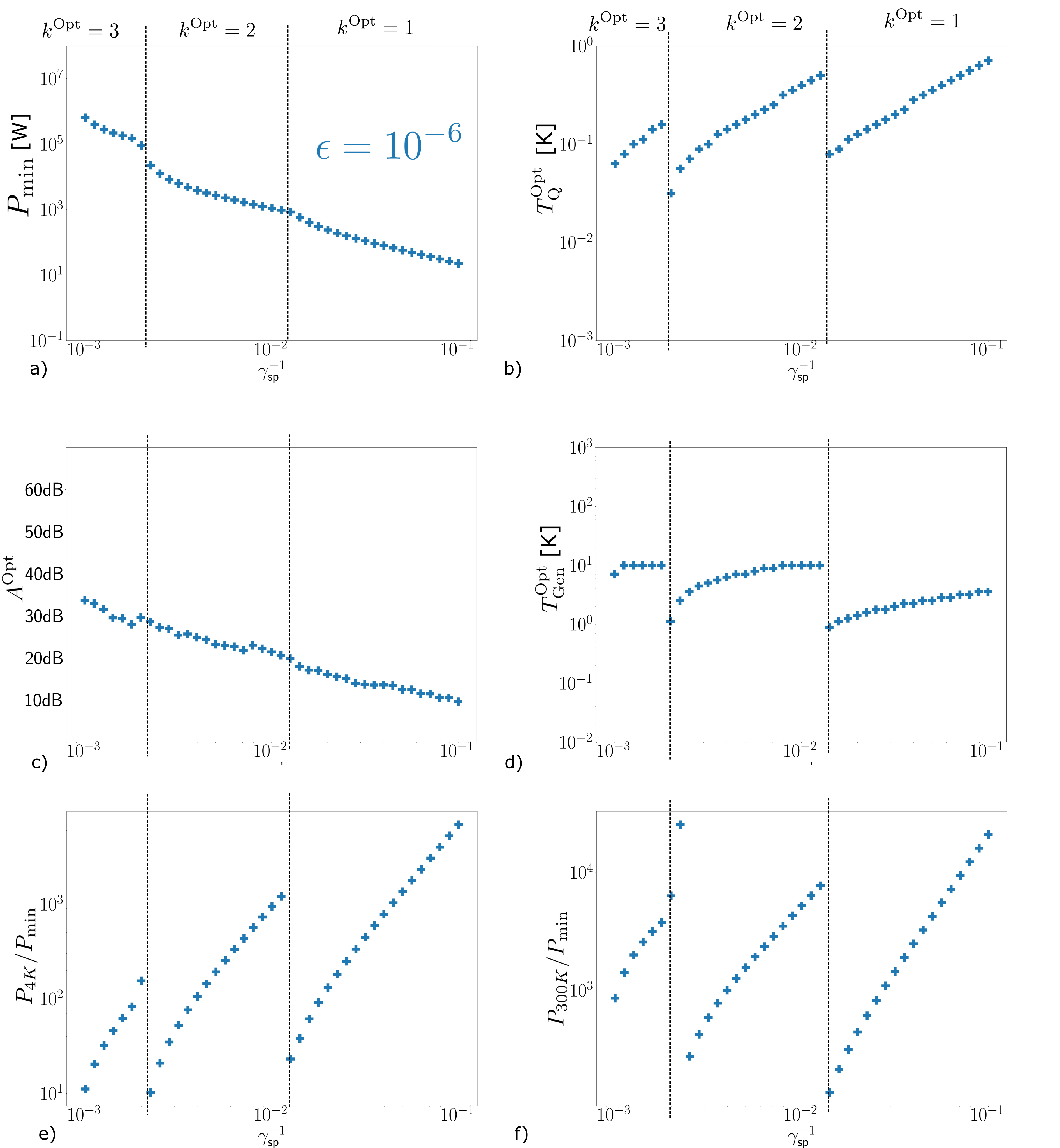}
\caption{Those graph represent the minimum power consumption and associated parameters in a case where we replaced the Carnot efficiencies appearing in \eqref{eq:a},\eqref{eq:b1qb},\eqref{eq:bcNOT} by their square, i.e: $(300-T)/T)\to ((300-T)/T)^2$. Motivation behind this choice are in the main text.}.
\label{fig:appendix_example}
\end{center}
\end{figure}
\FloatBarrier
\chapter{Decomposing the QFT on an appropriate gateset}
\label{app:gateset_decomposition}
Here, we give brief comments about what would change if we wanted to consider the exact decomposition of the quantum Fourier transform on a gateset that it is possible to implement with the Steane method. In practice, the controlled-phase gates would have to be decomposed on this fault-tolerant gateset (which is composed of Pauli, $H$, $S$, $T$, and $cNOT$ gates). A controlled-phase gate can be implemented with two cNOTs and a few phase gates $R_n$ as explained in the third chapter, in figure \ref{fig:controlled_R}, or in \cite{kim2018efficient}. But this wouldn't change in a very quantitative manner the results because the number of logical gates would be multiplied by some small number. However, the $R_n$ gates must be decomposed as well (because we cannot implement fault-tolerantly those gates with the Steane method). Using a gateset composed of $H$,$S$, Pauli (those gates are implementable with the Steane method), and $T$ gates (those gates require another procedure, for instance, magic state distillation), we would only be able to approximate the gate $R_n$ up to some deterministic error \cite{kim2018efficient}. Fortunately, as shown in \cite{kim2018efficient}, the number of gates required to approximate a gate $R_n$ grows very slowly with the accuracy desired. Also, to give an idea, in order to reach a deterministic error lower than $10^{-10}$, about $20$ gates (and an extra logical qubit) would be required. 

We did not investigate what the maximum deterministic error we could accept is (it is necessary in order to precisely estimate the total number of gates required after such decomposition), but from the results given in \ref{sec:charac_algo}, it is easy to see how the power consumption would change because of such decomposition. As the quantum Fourier transform without such decomposition admits a number of logical qubits and a logical depth being $(Q_L=2048,D_L=2 Q_L)$, in order to take into account those extra gates, we would simply have to move in this graph of an appropriate translation. For instance, if it appears that $10$ gates are required inside a controlled-phase gate, we would have to change $D_L \to 10 D_L$ to deduce the power consumption of such algorithm (we recall that as the $T$ gates are not modeled in our work, it would be the energetic cost of this algorithm neglecting the cost of the $T$ gates). A detailed analysis of this decomposition can be an outlook to consider.
\end{appendices}

\chapter*{Conclusion and perspectives}
\label{sec:conclusion}
In this thesis, we studied the question of the scalability of fault-tolerant quantum computing, mainly in the context of limited resources. The first approach we considered has been described in the third chapter. It consisted in studying what happens for fault-tolerance in the presence of a scale-dependent noise, which frequently comes from resources constraints. We showed that for some conditions on how the noise grows with the computer size, the maximum accuracy the computer can get to is intrinsically limited. We provided the tools allowing to estimate this maximum accuracy. In the case this scale-dependent noise is induced by a resource limitation, we provided a way that allows to (i) estimate the minimum resource required allowing to implement a fault-tolerant algorithm, (ii) estimate what is the maximum accuracy the computer can get to for a given amount of resource available\footnote{For (i) and (ii), the analysis is done under some hypotheses on the relationship between the noise and the resource, given in the chapter.}. We saw that having a scale-dependent noise is not necessarily an issue by itself; it depends on how fast the noise grows. It is why characterizing this dependence is essential to assess whether an architecture is scalable or not. Those first analyses provided a first approach to the problem of resource estimation of quantum computing by relating the resource to minimize to the noise felt by the qubits.

In the fourth chapter of this thesis, we generalized this approach and proposed a formulation for the problem of the resource cost of quantum computing. By asking to minimize a resource used for a calculation under the constraint that the calculation succeeds with a targeted probability, the entire architecture of a quantum computer can in principle be optimized, which includes aspects coming from fault-tolerance, algorithmic, and engineering. The principle behind this is to acknowledge that many of the elements inside a quantum computer are here in order to make sure that the calculation succeeds, possibly in a very indirect manner. Hence, by establishing the connection between the target accuracy and the architecture of the quantum computer, it is possible to optimize the whole design to ensure that the algorithm will succeed and minimize a resource at the same time.

In the last chapter, we used this method in a complete model of quantum computer based on superconducting qubits. Our goal was to find the minimum power consumption required to implement large-scale algorithms on at least thousands of logical qubits. In this part, we saw that more than two orders of magnitude of power consumption can be saved in regimes where without this optimization, the power consumption could be larger than the gigawatt. Our work seems to indicate that despite their high overhead in the number of physical qubits per logical qubit, the concatenated construction can reasonably be considered to create large-scale quantum computers if optimizations of the architecture are performed. This work allowed us to see which characteristics in the design play an important role in power consumption. For instance, we identified that the optimal temperature of the stage generating the signals is a critical parameter to fix if one wants to save power. In the majority of the examples, we also saw that the dominant source of power consumption comes from the heat that has to be removed from the electronics used inside of the cryostat and the heat conduction in the cables: the heat dissipated in the attenuators played a minor role. We have also seen that finely tuning the temperature of the quantum core is not necessary as the power consumption does not depend very much on this temperature. However, we also emphasized on the fact that this conclusion might be very architecture-dependent: we gave one concrete example of that in the appendix \ref{app:appendix_example} where the optimal attenuation level and the optimal temperature of the qubits also played a significant role in the power consumption if another efficiency than Carnot is considered for the cryostat. To give first directions to take in order to make quantum computing more energy-efficient (outside from the fact that optimization such as the one we did should be performed), we saw that being able to turn off the electronics when it is not being used is something that can reduce the power consumption of the computer if the heat it dissipates is not too large (otherwise the gain would be pretty limited for reasons explained in the lines following the figure \ref{fig:optimal_function_gamma_without_attenuators_consoperqubitorpergate}). In the same line of thought, increasing the level at which the ancillae qubits can be reused in the calculation, i.e., the level of recycling, seems to be another interesting approach that could help minimize power consumption. Also, replacing attenuators with non-dissipative filters (to isolate the qubits from thermal noise) might improve the power consumption in some regimes. However, as we explained in the last chapter, this was not the most important thing to optimize in the specific examples we took. 

Our work seems to indicate that the power consumption required by a large-scale quantum computer can be large, at least for superconducting qubit technologies error protected with Steane concatenated code. In our examples, we saw that \textit{after optimization}, about $10MW$ of power consumption could be expected to implement some "typical" large scale algorithm (see the section \ref{sec:conclusion_ch5} to understand more precisely what is the size and type of algorithm we are talking about), considering close to state of the art values for the different elements involved in the quantum computer (but assuming that the quality of all the gates is only limited by the qubit lifetime\footnote{See \ref{sec:conclusion_ch5} to know why we took this hypothesis that we believe is reasonable for the future.}). However, it has to be compared to the consumption of a classical supercomputer: it is in the same range of power consumption, but it would not be able to simulate the algorithms we studied. This consumption is lower than large scientific experiments such as the CERN, which consumes about $200MW$ of power. Also, those algorithms will be implemented in a very short amount of time such that the energy required would be small (the examples in which we find $10MW$ of power consumption would run in a few ms). We also noticed that the power bill rapidly decreases with technological improvements.

As a concluding remark, our work illustrates that having an optimized vision in the design of a quantum computer, using, for instance, the approach we are proposing, can significantly improve the potential in terms of scalability. It also illustrates that the energetic cost of quantum computing should be a figure of merit by itself on the scorecard of qubits technologies to assess their potential for scalability.

Our goal is now to provide some outlooks that look promising to make quantum computing more energy efficient. Indeed, further optimizations than the one we did could be performed, and there are ways to increase the level of details we considered in our models.
\section*{Some possible outlooks}
\label{sec:outlook_ft}
We can first describe possible explorations to do in the context of fault-tolerant quantum computing. The first outlook to consider would be to include in the modeling the energetic cost of the gates that we could not implement transversally through the Steane method: non-transversal gates such as the $T$ gate. Indeed, in our work, the energetic cost of such gate has not been modeled (we recall, however, that as a quantum memory does not require such gates, the energetic study done in \ref{sec:charac_algo} would be unchanged by such consideration). Including those gates in our modeling would allow us to see if the energetic cost would be significantly different from the one we estimated\footnote{We actually already started investigating on this question and it seems that if appropriate optimizations are performed, the energetic cost of $T$ gates is not likely to be much more significant than for the Clifford operations}, and it depends on the exact way such gates are implemented; many different proposals exist for that \cite{chamberland2020very,litinski2019magic,bravyi2005universal,yang2020covariant,paetznick2013universal}. Then, the next logical step would be to compare the energy efficiency of different quantum error correction code and their fault-tolerant implementation. Indeed, all the results in the energetic estimation we obtained are very closely connected to the fact that we used the Steane code and its concatenated fault-tolerant implementation. We could expect the quantitative estimation we did to drastically vary when other codes are considered. Also, to make quantum error correction more energy-efficient, it would be interesting to see if it is possible to implement it in an autonomous manner, i.e., without having to exchange information between the classical computer that manages the quantum algorithm's execution the quantum core. Doing things this way, all the amplification chains would have a significantly reduced energetic cost: the signals would only have to be amplified at the very end of the algorithm when the final answer would be given. Another essential aspect to investigate is the exact connectivity (i.e., with how many physical qubits one physical qubit has to interact) required by the different fault-tolerant construction. Some first investigations\footnote{This study has not been explained in the chapters of this thesis, but we started investigating it.} for the concatenated construction we used seem to indicate that a connectivity proportional to the number of logical qubits should be considered (the connectivity does not seem to increase with the level of protection). Considering the quantum Fourier transform within the Shor algorithm, it would then be about $2000$ physical qubits that should be able to interact "pairwise". This is another essential aspect for scalability and energetics as it can give constraints on how two-qubit gates can be performed. This aspect is also related to frequency overcrowding issues; including them in the modelization would be an interesting analysis to do.

In the context of quantum algorithms, we gave some intuitions about how the shape of an algorithm can influence the energetic cost. However, there are also large amounts of energy or power to save here. For example, in the figure \ref{fig:quantum_memory} of the last chapter, we saw that the number of logical qubits $Q_L$ has a more significant influence on the power cost than the logical depth $D_L$. However, would that be the case if we switched to another resource such as energy? In order to save energy, is it better to implement an algorithm in a compact way (small $D_L$) but with more logical qubits, or the other way around? Furthermore, how would the difference be quantitatively? This question raises the more general question of how algorithms should be compiled in order to make them energy, power, or more generally \textit{resource} efficient (for any resource of interest). For instance, if the non-transversal gates require a large dynamic power consumption, algorithm compilation procedures should minimize their acting in parallel.

More globally, this work allows benchmarking the energetic cost of different kinds of quantum computers and their associated architectures. How different would the power consumption of a superconducting quantum computer be compared to a spin qubit one for instance? From the way we formulated the question of the energetic cost, we now have a well-defined method allowing to do this benchmark. Then, it could be interesting to enrich the optimizations that have been performed. For instance, one "bottleneck" forbidding us to save more power consumption was that we forced the design to have an amplification stage at $4K$ (in order to make the thermal noise negligible for the readout). However, it would be surprising if this temperature is precisely the optimal one. In order to optimize this stage, we would need to have a proper modeling of how amplifying at a higher temperature can degrade the quality of the readout-signals quantitatively and thus increase the probability of error of the logical gates (if the measurements are too noisy then the quantum error correction would be wrongly performed).

All those propositions are exciting outlooks that would deserve to be investigated further to make quantum computing energy-efficient and naturally make it more scalable.

\bibliography{bibliography}

\begin{thebibliography}{100}

\bibitem{fuentes2005alice}
I.~Fuentes-Schuller and R.~B. Mann, ``Alice falls into a black hole:
  entanglement in noninertial frames,'' {\em Physical review letters}, vol.~95,
  no.~12, p.~120404, 2005.

\bibitem{dowling2003quantum}
J.~P. Dowling and G.~J. Milburn, ``Quantum technology: the second quantum
  revolution,'' {\em Philosophical Transactions of the Royal Society of London.
  Series A: Mathematical, Physical and Engineering Sciences}, vol.~361,
  no.~1809, pp.~1655--1674, 2003.

\bibitem{jaeger2018second}
L.~Jaeger, ``The second quantum revolution,'' {\em Switzerland: Springer},
  2018.

\bibitem{rauschenbeutel1999coherent}
A.~Rauschenbeutel, G.~Nogues, S.~Osnaghi, P.~Bertet, M.~Brune, J.-M. Raimond,
  and S.~Haroche, ``Coherent operation of a tunable quantum phase gate in
  cavity qed,'' {\em Physical Review Letters}, vol.~83, no.~24, p.~5166, 1999.

\bibitem{nakamura2001rabi}
Y.~Nakamura, Y.~A. Pashkin, and J.~S. Tsai, ``Rabi oscillations in a
  josephson-junction charge two-level system,'' {\em Physical Review Letters},
  vol.~87, no.~24, p.~246601, 2001.

\bibitem{haroche2006exploring}
S.~Haroche and J.-M. Raimond, {\em Exploring the quantum: atoms, cavities, and
  photons}.
\newblock Oxford university press, 2006.

\bibitem{degen2017quantum}
C.~L. Degen, F.~Reinhard, and P.~Cappellaro, ``Quantum sensing,'' {\em Reviews
  of modern physics}, vol.~89, no.~3, p.~035002, 2017.

\bibitem{pirandola2018advances}
S.~Pirandola, B.~R. Bardhan, T.~Gehring, C.~Weedbrook, and S.~Lloyd, ``Advances
  in photonic quantum sensing,'' {\em Nature Photonics}, vol.~12, no.~12,
  pp.~724--733, 2018.

\bibitem{pirandola2020advances}
S.~Pirandola, U.~L. Andersen, L.~Banchi, M.~Berta, D.~Bunandar, R.~Colbeck,
  D.~Englund, T.~Gehring, C.~Lupo, C.~Ottaviani, {\em et~al.}, ``Advances in
  quantum cryptography,'' {\em Advances in Optics and Photonics}, vol.~12,
  no.~4, pp.~1012--1236, 2020.

\bibitem{shenoy2017quantum}
A.~Shenoy-Hejamadi, A.~Pathak, and S.~Radhakrishna, ``Quantum cryptography: key
  distribution and beyond,'' {\em Quanta}, vol.~6, no.~1, pp.~1--47, 2017.

\bibitem{chen2021review}
J.~Chen, ``Review on quantum communication and quantum computation,'' in {\em
  Journal of Physics: Conference Series}, vol.~1865, p.~022008, IOP Publishing,
  2021.

\bibitem{gisin2007quantum}
N.~Gisin and R.~Thew, ``Quantum communication,'' {\em Nature photonics},
  vol.~1, no.~3, pp.~165--171, 2007.

\bibitem{moerel2021reflections}
L.~Moerel and P.~Timmers, ``Reflections on digital sovereignty,'' {\em EU Cyber
  Direct, Research in Focus series}, 2021.

\bibitem{bernstein2017post}
D.~J. Bernstein and T.~Lange, ``Post-quantum cryptography,'' {\em Nature},
  vol.~549, no.~7671, pp.~188--194, 2017.

\bibitem{gidney2021factor}
C.~Gidney and M.~Eker{\aa}, ``How to factor 2048 bit rsa integers in 8 hours
  using 20 million noisy qubits,'' {\em Quantum}, vol.~5, p.~433, 2021.

\bibitem{arute2019quantum}
F.~Arute, K.~Arya, R.~Babbush, D.~Bacon, J.~C. Bardin, R.~Barends, R.~Biswas,
  S.~Boixo, F.~G. Brandao, D.~A. Buell, {\em et~al.}, ``Quantum supremacy using
  a programmable superconducting processor,'' {\em Nature}, vol.~574, no.~7779,
  pp.~505--510, 2019.

\bibitem{tornow2020non}
S.~Tornow, W.~Gehrke, and U.~Helmbrecht, ``Non-equilibrium dynamics of a
  dissipative two-site hubbard model simulated on the ibm quantum computer,''
  {\em arXiv preprint arXiv:2011.11059}, 2020.

\bibitem{garcia2017five}
D.~Garc{\'\i}a-Mart{\'\i}n and G.~Sierra, ``Five experimental tests on the
  5-qubit ibm quantum computer,'' {\em arXiv preprint arXiv:1712.05642}, 2017.

\bibitem{intelquantumcomputing}
Intel.
  \url{https://www.intel.fr/content/www/fr/fr/research/quantum-computing.html}.

\bibitem{motta2020determining}
M.~Motta, C.~Sun, A.~T. Tan, M.~J. O’Rourke, E.~Ye, A.~J. Minnich, F.~G.
  Brand{\~a}o, and G.~K.-L. Chan, ``Determining eigenstates and thermal states
  on a quantum computer using quantum imaginary time evolution,'' {\em Nature
  Physics}, vol.~16, no.~2, pp.~205--210, 2020.

\bibitem{jones1998implementation}
J.~A. Jones and M.~Mosca, ``Implementation of a quantum algorithm on a nuclear
  magnetic resonance quantum computer,'' {\em The Journal of chemical physics},
  vol.~109, no.~5, pp.~1648--1653, 1998.

\bibitem{raimond1996quantum}
J.~Raimond and S.~Haroche, ``Quantum computing: dream or nightmare,'' {\em
  Phys. Today}, vol.~49, no.~8, pp.~51--52, 1996.

\bibitem{aliferis2005quantum}
P.~Aliferis, D.~Gottesman, and J.~Preskill, ``Quantum accuracy threshold for
  concatenated distance-3 codes,'' {\em arXiv preprint quant-ph/0504218}, 2005.

\bibitem{aliferis2007accuracy}
P.~Aliferis, D.~Gottesman, and J.~Preskill, ``Accuracy threshold for
  postselected quantum computation,'' {\em arXiv preprint quant-ph/0703264},
  2007.

\bibitem{chamberland2017overhead}
C.~Chamberland, T.~Jochym-O'Connor, and R.~Laflamme, ``Overhead analysis of
  universal concatenated quantum codes,'' {\em Physical Review A}, vol.~95,
  no.~2, p.~022313, 2017.

\bibitem{suchara2013comparing}
M.~Suchara, A.~Faruque, C.-Y. Lai, G.~Paz, F.~T. Chong, and J.~Kubiatowicz,
  ``Comparing the overhead of topological and concatenated quantum error
  correction,'' {\em arXiv preprint arXiv:1312.2316}, 2013.

\bibitem{schutjens2013single}
R.~Schutjens, F.~A. Dagga, D.~Egger, and F.~Wilhelm, ``Single-qubit gates in
  frequency-crowded transmon systems,'' {\em Physical Review A}, vol.~88,
  no.~5, p.~052330, 2013.

\bibitem{theis2016simultaneous}
L.~Theis, F.~Motzoi, and F.~Wilhelm, ``Simultaneous gates in frequency-crowded
  multilevel systems using fast, robust, analytic control shapes,'' {\em
  Physical Review A}, vol.~93, no.~1, p.~012324, 2016.

\bibitem{rosenberg2020solid}
D.~Rosenberg, S.~J. Weber, D.~Conway, D.-R.~W. Yost, J.~Mallek, G.~Calusine,
  R.~Das, D.~Kim, M.~E. Schwartz, W.~Woods, {\em et~al.}, ``Solid-state qubits:
  3d integration and packaging,'' {\em IEEE Microwave Magazine}, vol.~21,
  no.~8, pp.~72--85, 2020.

\bibitem{monroe2013scaling}
C.~Monroe and J.~Kim, ``Scaling the ion trap quantum processor,'' {\em
  Science}, vol.~339, no.~6124, pp.~1164--1169, 2013.

\bibitem{bruzewicz2019trapped}
C.~D. Bruzewicz, J.~Chiaverini, R.~McConnell, and J.~M. Sage, ``Trapped-ion
  quantum computing: Progress and challenges,'' {\em Applied Physics Reviews},
  vol.~6, no.~2, p.~021314, 2019.

\bibitem{wang2021single}
P.~Wang, C.-Y. Luan, M.~Qiao, M.~Um, J.~Zhang, Y.~Wang, X.~Yuan, M.~Gu,
  J.~Zhang, and K.~Kim, ``Single ion qubit with estimated coherence time
  exceeding one hour,'' {\em Nature communications}, vol.~12, no.~1, pp.~1--8,
  2021.

\bibitem{kjaergaard2020superconducting}
M.~Kjaergaard, M.~E. Schwartz, J.~Braum{\"u}ller, P.~Krantz, J.~I.-J. Wang,
  S.~Gustavsson, and W.~D. Oliver, ``Superconducting qubits: Current state of
  play,'' {\em Annual Review of Condensed Matter Physics}, vol.~11,
  pp.~369--395, 2020.

\bibitem{hanson2007spins}
R.~Hanson, L.~P. Kouwenhoven, J.~R. Petta, S.~Tarucha, and L.~M. Vandersypen,
  ``Spins in few-electron quantum dots,'' {\em Reviews of modern physics},
  vol.~79, no.~4, p.~1217, 2007.

\bibitem{krantz2019quantum}
P.~Krantz, M.~Kjaergaard, F.~Yan, T.~P. Orlando, S.~Gustavsson, and W.~D.
  Oliver, ``A quantum engineer's guide to superconducting qubits,'' {\em
  Applied Physics Reviews}, vol.~6, no.~2, p.~021318, 2019.

\bibitem{slussarenko2019photonic}
S.~Slussarenko and G.~J. Pryde, ``Photonic quantum information processing: A
  concise review,'' {\em Applied Physics Reviews}, vol.~6, no.~4, p.~041303,
  2019.

\bibitem{kok2007linear}
P.~Kok, W.~J. Munro, K.~Nemoto, T.~C. Ralph, J.~P. Dowling, and G.~J. Milburn,
  ``Linear optical quantum computing with photonic qubits,'' {\em Reviews of
  modern physics}, vol.~79, no.~1, p.~135, 2007.

\bibitem{criger2012recent}
B.~Criger, G.~Passante, D.~Park, and R.~Laflamme, ``Recent advances in nuclear
  magnetic resonance quantum information processing,'' {\em Philosophical
  Transactions of the Royal Society A: Mathematical, Physical and Engineering
  Sciences}, vol.~370, no.~1976, pp.~4620--4635, 2012.

\bibitem{maurand2016cmos}
R.~Maurand, X.~Jehl, D.~Kotekar-Patil, A.~Corna, H.~Bohuslavskyi,
  R.~Lavi{\'e}ville, L.~Hutin, S.~Barraud, M.~Vinet, M.~Sanquer, {\em et~al.},
  ``A cmos silicon spin qubit,'' {\em Nature communications}, vol.~7, no.~1,
  pp.~1--6, 2016.

\bibitem{vinet2018towards}
M.~Vinet, L.~Hutin, B.~Bertrand, S.~Barraud, J.-M. Hartmann, Y.-J. Kim,
  V.~Mazzocchi, A.~Amisse, H.~Bohuslavskyi, L.~Bourdet, {\em et~al.}, ``Towards
  scalable silicon quantum computing,'' in {\em 2018 IEEE International
  Electron Devices Meeting (IEDM)}, pp.~6--5, IEEE, 2018.

\bibitem{bardin201929}
J.~C. Bardin, E.~Jeffrey, E.~Lucero, T.~Huang, O.~Naaman, R.~Barends, T.~White,
  M.~Giustina, D.~Sank, P.~Roushan, {\em et~al.}, ``29.1 a 28nm bulk-cmos
  4-to-8ghz!` 2mw cryogenic pulse modulator for scalable quantum computing,''
  in {\em 2019 IEEE International Solid-State Circuits Conference-(ISSCC)},
  pp.~456--458, IEEE, 2019.

\bibitem{patra2020scalable}
B.~Patra, J.~P. Van~Dijk, A.~Corna, X.~Xue, N.~Samkharadze, A.~Sammak,
  G.~Scappucci, M.~Veldhorst, L.~M. Vandersypen, M.~Babaie, {\em et~al.}, ``A
  scalable cryo-cmos 2-to-20ghz digitally intensive controller for 4$\times$ 32
  frequency multiplexed spin qubits/transmons in 22nm finfet technology for
  quantum computers,'' in {\em 2020 IEEE International Solid-State Circuits
  Conference, ISSCC 2020}, pp.~304--306, Institute of Electrical and
  Electronics Engineers (IEEE), 2020.

\bibitem{mcdermott2018quantum}
R.~McDermott, M.~Vavilov, B.~Plourde, F.~Wilhelm, P.~Liebermann, O.~Mukhanov,
  and T.~Ohki, ``Quantum--classical interface based on single flux quantum
  digital logic,'' {\em Quantum science and technology}, vol.~3, no.~2,
  p.~024004, 2018.

\bibitem{semenov2003sfq}
V.~K. Semenov and D.~V. Averin, ``Sfq control circuits for josephson junction
  qubits,'' {\em IEEE transactions on applied superconductivity}, vol.~13,
  no.~2, pp.~960--965, 2003.

\bibitem{debenedictis2020adiabatic}
E.~P. DeBenedictis, ``Adiabatic circuits for quantum computer control,'' in
  {\em 2020 International Conference on Rebooting Computing (ICRC)},
  pp.~42--49, IEEE, 2020.

\bibitem{fowler2012surface}
A.~G. Fowler, M.~Mariantoni, J.~M. Martinis, and A.~N. Cleland, ``Surface
  codes: Towards practical large-scale quantum computation,'' {\em Physical
  Review A}, vol.~86, no.~3, p.~032324, 2012.

\bibitem{dennis2002topological}
E.~Dennis, A.~Kitaev, A.~Landahl, and J.~Preskill, ``Topological quantum
  memory,'' {\em Journal of Mathematical Physics}, vol.~43, no.~9,
  pp.~4452--4505, 2002.

\bibitem{terhal2020towards}
B.~M. Terhal, J.~Conrad, and C.~Vuillot, ``Towards scalable bosonic quantum
  error correction,'' {\em Quantum Science and Technology}, vol.~5, no.~4,
  p.~043001, 2020.

\bibitem{ofek2016extending}
N.~Ofek, A.~Petrenko, R.~Heeres, P.~Reinhold, Z.~Leghtas, B.~Vlastakis, Y.~Liu,
  L.~Frunzio, S.~Girvin, L.~Jiang, {\em et~al.}, ``Extending the lifetime of a
  quantum bit with error correction in superconducting circuits,'' {\em
  Nature}, vol.~536, no.~7617, pp.~441--445, 2016.

\bibitem{noh2020fault}
K.~Noh and C.~Chamberland, ``Fault-tolerant bosonic quantum error correction
  with the surface--gottesman-kitaev-preskill code,'' {\em Physical Review A},
  vol.~101, no.~1, p.~012316, 2020.

\bibitem{aharonov2008fault}
D.~Aharonov and M.~Ben-Or, ``Fault-tolerant quantum computation with constant
  error rate,'' {\em SIAM Journal on Computing}, 2008.

\bibitem{vool2017introduction}
U.~Vool and M.~Devoret, ``Introduction to quantum electromagnetic circuits,''
  {\em International Journal of Circuit Theory and Applications}, vol.~45,
  no.~7, pp.~897--934, 2017.

\bibitem{gu2017microwave}
X.~Gu, A.~F. Kockum, A.~Miranowicz, Y.-x. Liu, and F.~Nori, ``Microwave
  photonics with superconducting quantum circuits,'' {\em Physics Reports},
  vol.~718, pp.~1--102, 2017.

\bibitem{wendin2005superconducting}
G.~Wendin and V.~Shumeiko, ``Superconducting quantum circuits, qubits and
  computing,'' {\em arXiv preprint cond-mat/0508729}, 2005.

\bibitem{nigg2012black}
S.~E. Nigg, H.~Paik, B.~Vlastakis, G.~Kirchmair, S.~Shankar, L.~Frunzio,
  M.~Devoret, R.~Schoelkopf, and S.~Girvin, ``Black-box superconducting circuit
  quantization,'' {\em Physical Review Letters}, vol.~108, no.~24, p.~240502,
  2012.

\bibitem{basdevant2014principes}
J.-L. Basdevant, ``Les principes variationnels en physique,'' {\em Vuibert,
  Paris}, 2014.

\bibitem{goldstein2002classical}
H.~Goldstein, C.~Poole, and J.~Safko, ``Classical mechanics,'' 2002.

\bibitem{cohen1998mecanique}
C.~Cohen-Tannoudji, B.~Diu, and F.~Lalo{\"e}, {\em M{\'e}canique quantique}.
\newblock Hermann, EDP Sciences, 1998.

\bibitem{pozar2011microwave}
D.~M. Pozar, {\em Microwave engineering}.
\newblock John wiley \& sons, 2011.

\bibitem{gardiner1985input}
C.~W. Gardiner and M.~J. Collett, ``Input and output in damped quantum systems:
  Quantum stochastic differential equations and the master equation,'' {\em
  Physical Review A}, vol.~31, no.~6, p.~3761, 1985.

\bibitem{clerk2010introduction}
A.~A. Clerk, M.~H. Devoret, S.~M. Girvin, F.~ardt, and R.~J. Schoelkopf,
  ``Introduction to quantum noise, measurement, and amplification,'' {\em
  Reviews of Modern Physics}, vol.~82, no.~2, p.~1155, 2010.

\bibitem{wiegand2020semiclassical}
E.~Wiegand, B.~Rousseaux, and G.~Johansson, ``Semiclassical analysis of
  dark-state transient dynamics in waveguide circuit qed,'' {\em Physical
  Review A}, vol.~101, no.~3, p.~033801, 2020.

\bibitem{peropadre2013scattering}
B.~Peropadre, J.~Lindkvist, I.-C. Hoi, C.~Wilson, J.~J. Garcia-Ripoll,
  P.~Delsing, and G.~Johansson, ``Scattering of coherent states on a single
  artificial atom,'' {\em New Journal of Physics}, vol.~15, no.~3, p.~035009,
  2013.

\bibitem{walls2007quantum}
D.~F. Walls and G.~J. Milburn, {\em Quantum optics}.
\newblock Springer Science \& Business Media, 2007.

\bibitem{cohen1998atom}
C.~Cohen-Tannoudji, J.~Dupont-Roc, and G.~Grynberg, {\em Atom-photon
  interactions: basic processes and applications}.
\newblock 1998.

\bibitem{nielsen2002quantum}
M.~A. Nielsen and I.~Chuang, ``Quantum computation and quantum information,''
  2002.

\bibitem{cottet2017observing}
N.~Cottet, S.~Jezouin, L.~Bretheau, P.~Campagne-Ibarcq, Q.~Ficheux, J.~Anders,
  A.~Auff{\`e}ves, R.~Azouit, P.~Rouchon, and B.~Huard, ``Observing a quantum
  maxwell demon at work,'' {\em Proceedings of the National Academy of
  Sciences}, vol.~114, no.~29, pp.~7561--7564, 2017.

\bibitem{monsel2020energetic}
J.~Monsel, M.~Fellous-Asiani, B.~Huard, and A.~Auff{\`e}ves, ``The energetic
  cost of work extraction,'' {\em Physical review letters}, vol.~124, no.~13,
  p.~130601, 2020.

\bibitem{place2021new}
A.~P. Place, L.~V. Rodgers, P.~Mundada, B.~M. Smitham, M.~Fitzpatrick, Z.~Leng,
  A.~Premkumar, J.~Bryon, A.~Vrajitoarea, S.~Sussman, {\em et~al.}, ``New
  material platform for superconducting transmon qubits with coherence times
  exceeding 0.3 milliseconds,'' {\em Nature communications}, vol.~12, no.~1,
  pp.~1--6, 2021.

\bibitem{somoroff2021millisecond}
A.~Somoroff, Q.~Ficheux, R.~A. Mencia, H.~Xiong, R.~V. Kuzmin, and V.~E.
  Manucharyan, ``Millisecond coherence in a superconducting qubit,'' {\em arXiv
  preprint arXiv:2103.08578}, 2021.

\bibitem{almudever2017engineering}
C.~G. Almudever, L.~Lao, X.~Fu, N.~Khammassi, I.~Ashraf, D.~Iorga,
  S.~Varsamopoulos, C.~Eichler, A.~Wallraff, L.~Geck, {\em et~al.}, ``The
  engineering challenges in quantum computing,'' in {\em Design, Automation \&
  Test in Europe Conference \& Exhibition (DATE), 2017}, pp.~836--845, IEEE,
  2017.

\bibitem{werninghaus2021leakage}
M.~Werninghaus, D.~J. Egger, F.~Roy, S.~Machnes, F.~K. Wilhelm, and S.~Filipp,
  ``Leakage reduction in fast superconducting qubit gates via optimal
  control,'' {\em npj Quantum Information}, vol.~7, no.~1, pp.~1--6, 2021.

\bibitem{huang2020superconducting}
H.-L. Huang, D.~Wu, D.~Fan, and X.~Zhu, ``Superconducting quantum computing: a
  review,'' {\em Science China Information Sciences}, vol.~63, no.~8,
  pp.~1--32, 2020.

\bibitem{yan2018tunable}
F.~Yan, P.~Krantz, Y.~Sung, M.~Kjaergaard, D.~L. Campbell, T.~P. Orlando,
  S.~Gustavsson, and W.~D. Oliver, ``Tunable coupling scheme for implementing
  high-fidelity two-qubit gates,'' {\em Physical Review Applied}, vol.~10,
  no.~5, p.~054062, 2018.

\bibitem{xu2020high}
Y.~Xu, J.~Chu, J.~Yuan, J.~Qiu, Y.~Zhou, L.~Zhang, X.~Tan, Y.~Yu, S.~Liu,
  J.~Li, {\em et~al.}, ``High-fidelity, high-scalability two-qubit gate scheme
  for superconducting qubits,'' {\em Physical Review Letters}, vol.~125,
  no.~24, p.~240503, 2020.

\bibitem{chow2011simple}
J.~M. Chow, A.~D. C{\'o}rcoles, J.~M. Gambetta, C.~Rigetti, B.~R. Johnson,
  J.~A. Smolin, J.~R. Rozen, G.~A. Keefe, M.~B. Rothwell, M.~B. Ketchen, {\em
  et~al.}, ``Simple all-microwave entangling gate for fixed-frequency
  superconducting qubits,'' {\em Physical review letters}, vol.~107, no.~8,
  p.~080502, 2011.

\bibitem{sheldon2016procedure}
S.~Sheldon, E.~Magesan, J.~M. Chow, and J.~M. Gambetta, ``Procedure for
  systematically tuning up cross-talk in the cross-resonance gate,'' {\em
  Physical Review A}, vol.~93, no.~6, p.~060302, 2016.

\bibitem{bertet2006parametric}
P.~Bertet, C.~Harmans, and J.~Mooij, ``Parametric coupling for superconducting
  qubits,'' {\em Physical Review B}, vol.~73, no.~6, p.~064512, 2006.

\bibitem{mckay2016universal}
D.~C. McKay, S.~Filipp, A.~Mezzacapo, E.~Magesan, J.~M. Chow, and J.~M.
  Gambetta, ``Universal gate for fixed-frequency qubits via a tunable bus,''
  {\em Physical Review Applied}, vol.~6, no.~6, p.~064007, 2016.

\bibitem{schwartz2014quantum}
M.~D. Schwartz, {\em Quantum field theory and the standard model}.
\newblock Cambridge University Press, 2014.

\bibitem{preskill1998lecture}
J.~Preskill, ``Lecture notes for physics 229: Quantum information and
  computation,'' {\em California Institute of Technology}, vol.~16, p.~10,
  1998.

\bibitem{cohen2021mecanique}
C.~Cohen-Tannoudji, B.~Diu, and F.~Lalo{\"e}, {\em M{\'e}canique quantique-Tome
  3}.
\newblock EDP sciences, 2021.

\bibitem{breuer2002theory}
H.-P. Breuer, F.~Petruccione, {\em et~al.}, {\em The theory of open quantum
  systems}.
\newblock Oxford University Press on Demand, 2002.

\bibitem{calderbank1996good}
A.~R. Calderbank and P.~W. Shor, ``Good quantum error-correcting codes exist,''
  {\em Physical Review A}, vol.~54, no.~2, p.~1098, 1996.

\bibitem{steane1996multiple}
A.~Steane, ``Multiple-particle interference and quantum error correction,''
  {\em Proceedings of the Royal Society of London. Series A: Mathematical,
  Physical and Engineering Sciences}, vol.~452, no.~1954, pp.~2551--2577, 1996.

\bibitem{zheng2018efficient}
Y.-C. Zheng, C.-Y. Lai, and T.~A. Brun, ``Efficient preparation of
  large-block-code ancilla states for fault-tolerant quantum computation,''
  {\em Physical Review A}, vol.~97, no.~3, p.~032331, 2018.

\bibitem{brun2015teleportation}
T.~A. Brun, Y.-C. Zheng, K.-C. Hsu, J.~Job, and C.-Y. Lai,
  ``Teleportation-based fault-tolerant quantum computation in multi-qubit large
  block codes,'' {\em arXiv preprint arXiv:1504.03913}, 2015.

\bibitem{gottesman1997stabilizer}
D.~Gottesman, {\em Stabilizer codes and quantum error correction}.
\newblock California Institute of Technology, 1997.

\bibitem{gottesman2010introduction}
D.~Gottesman, ``An introduction to quantum error correction and fault-tolerant
  quantum computation,'' in {\em Quantum information science and its
  contributions to mathematics, Proceedings of Symposia in Applied
  Mathematics}, vol.~68, pp.~13--58, 2010.

\bibitem{knill1998resilient}
E.~Knill, R.~Laflamme, and W.~H. Zurek, ``Resilient quantum computation: error
  models and thresholds,'' {\em Proceedings of the Royal Society of London.
  Series A: Mathematical, Physical and Engineering Sciences}, vol.~454,
  no.~1969, pp.~365--384, 1998.

\bibitem{preskill1998reliable}
J.~Preskill, ``Reliable quantum computers,'' {\em Proceedings of the Royal
  Society of London. Series A: Mathematical, Physical and Engineering
  Sciences}, vol.~454, no.~1969, pp.~385--410, 1998.

\bibitem{terhal2005fault}
B.~M. Terhal and G.~Burkard, ``Fault-tolerant quantum computation for local
  non-markovian noise,'' {\em Physical Review A}, vol.~71, no.~1, p.~012336,
  2005.

\bibitem{aharonov2006fault}
D.~Aharonov, A.~Kitaev, and J.~Preskill, ``Fault-tolerant quantum computation
  with long-range correlated noise,'' {\em Physical review letters}, vol.~96,
  no.~5, p.~050504, 2006.

\bibitem{jayashankar2021achieving}
A.~Jayashankar, M.~D.~H. Long, H.~K. Ng, and P.~Mandayam, ``Achieving fault
  tolerance against amplitude-damping noise,'' {\em arXiv preprint
  arXiv:2107.05485}, 2021.

\bibitem{beale2018quantum}
S.~J. Beale, J.~J. Wallman, M.~Guti{\'e}rrez, K.~R. Brown, and R.~Laflamme,
  ``Quantum error correction decoheres noise,'' {\em Physical review letters},
  vol.~121, no.~19, p.~190501, 2018.

\bibitem{eastin2009restrictions}
B.~Eastin and E.~Knill, ``Restrictions on transversal encoded quantum gate
  sets,'' {\em Physical review letters}, vol.~102, no.~11, p.~110502, 2009.

\bibitem{knill2005quantum}
E.~Knill, ``Quantum computing with realistically noisy devices,'' {\em Nature},
  vol.~434, no.~7029, pp.~39--44, 2005.

\bibitem{riesebos2017pauli}
L.~Riesebos, X.~Fu, S.~Varsamopoulos, C.~G. Almudever, and K.~Bertels, ``Pauli
  frames for quantum computer architectures,'' in {\em Proceedings of the 54th
  Annual Design Automation Conference 2017}, pp.~1--6, 2017.

\bibitem{yang2020covariant}
Y.~Yang, Y.~Mo, J.~M. Renes, G.~Chiribella, and M.~P. Woods, ``Covariant
  quantum error correcting codes via reference frames,'' {\em arXiv preprint
  arXiv:2007.09154}, 2020.

\bibitem{paetznick2013universal}
A.~Paetznick and B.~W. Reichardt, ``Universal fault-tolerant quantum
  computation with only transversal gates and error correction,'' {\em Physical
  review letters}, vol.~111, no.~9, p.~090505, 2013.

\bibitem{chamberland2020very}
C.~Chamberland and K.~Noh, ``Very low overhead fault-tolerant magic state
  preparation using redundant ancilla encoding and flag qubits,'' {\em npj
  Quantum Information}, vol.~6, no.~1, pp.~1--12, 2020.

\bibitem{litinski2019magic}
D.~Litinski, ``Magic state distillation: Not as costly as you think,'' {\em
  Quantum}, vol.~3, p.~205, 2019.

\bibitem{bravyi2005universal}
S.~Bravyi and A.~Kitaev, ``Universal quantum computation with ideal clifford
  gates and noisy ancillas,'' {\em Physical Review A}, vol.~71, no.~2,
  p.~022316, 2005.

\bibitem{piltz2014trapped}
C.~Piltz, T.~Sriarunothai, A.~Var{\'o}n, and C.~Wunderlich, ``A
  trapped-ion-based quantum byte with 10- 5 next-neighbour cross-talk,'' {\em
  Nature communications}, vol.~5, no.~1, pp.~1--10, 2014.

\bibitem{heinz2021crosstalk}
I.~Heinz and G.~Burkard, ``Crosstalk analysis for single-qubit and two-qubit
  gates in spin qubit arrays,'' {\em arXiv preprint arXiv:2105.10221}, 2021.

\bibitem{mckay2019three}
D.~C. McKay, S.~Sheldon, J.~A. Smolin, J.~M. Chow, and J.~M. Gambetta,
  ``Three-qubit randomized benchmarking,'' {\em Physical review letters},
  vol.~122, no.~20, p.~200502, 2019.

\bibitem{takita2017experimental}
M.~Takita, A.~W. Cross, A.~C{\'o}rcoles, J.~M. Chow, and J.~M. Gambetta,
  ``Experimental demonstration of fault-tolerant state preparation with
  superconducting qubits,'' {\em Physical review letters}, vol.~119, no.~18,
  p.~180501, 2017.

\bibitem{proctor2019direct}
T.~J. Proctor, A.~Carignan-Dugas, K.~Rudinger, E.~Nielsen, R.~Blume-Kohout, and
  K.~Young, ``Direct randomized benchmarking for multiqubit devices,'' {\em
  Physical review letters}, vol.~123, no.~3, p.~030503, 2019.

\bibitem{treil2013linear}
S.~Treil, {\em Linear algebra done wrong}.
\newblock MTM, 2013.

\bibitem{todd2010operator}
R.~Todd, ``Operator norm,'' {\em Wolfram Mathworld}, 2010.

\bibitem{fellous2020limitations}
M.~Fellous-Asiani, J.~H. Chai, R.~S. Whitney, A.~Auff{\`e}ves, and H.~K. Ng,
  ``Limitations in quantum computing from resource constraints,'' {\em arXiv
  preprint arXiv:2007.01966}, 2020.

\bibitem{kim2018efficient}
T.~Kim and B.-S. Choi, ``Efficient decomposition methods for controlled-r n
  using a single ancillary qubit,'' {\em Scientific reports}, vol.~8, no.~1,
  pp.~1--7, 2018.

\bibitem{gokhale2020full}
P.~Gokhale, {\em Full-Stack, Cross-Layer Optimizations for Quantum Computing}.
\newblock PhD thesis, The University of Chicago, 2020.

\bibitem{murali2019full}
P.~Murali, N.~M. Linke, M.~Martonosi, A.~J. Abhari, N.~H. Nguyen, and C.~H.
  Alderete, ``Full-stack, real-system quantum computer studies: Architectural
  comparisons and design insights,'' in {\em 2019 ACM/IEEE 46th Annual
  International Symposium on Computer Architecture (ISCA)}, pp.~527--540, IEEE,
  2019.

\bibitem{rodrigo2020exploring}
S.~Rodrigo, S.~Abadal, E.~Alarc{\'o}n, and C.~G. Almudever, ``Exploring a
  double full-stack communications-enabled architecture for multi-core quantum
  computers,'' {\em arXiv preprint arXiv:2009.08186}, 2020.

\bibitem{amy2020staq}
M.~Amy and V.~Gheorghiu, ``staq—a full-stack quantum processing toolkit,''
  {\em Quantum Science and Technology}, vol.~5, no.~3, p.~034016, 2020.

\bibitem{li2015resource}
Y.~Li, P.~C. Humphreys, G.~J. Mendoza, and S.~C. Benjamin, ``Resource costs for
  fault-tolerant linear optical quantum computing,'' {\em Physical Review X},
  vol.~5, no.~4, p.~041007, 2015.

\bibitem{kim2021fault}
I.~H. Kim, E.~Lee, Y.-H. Liu, S.~Pallister, W.~Pol, and S.~Roberts,
  ``Fault-tolerant resource estimate for quantum chemical simulations: Case
  study on li-ion battery electrolyte molecules,'' {\em arXiv preprint
  arXiv:2104.10653}, 2021.

\bibitem{di2020fault}
O.~Di~Matteo, V.~Gheorghiu, and M.~Mosca, ``Fault-tolerant resource estimation
  of quantum random-access memories,'' {\em IEEE Transactions on Quantum
  Engineering}, vol.~1, pp.~1--13, 2020.

\bibitem{abah2019energetic}
O.~Abah, R.~Puebla, A.~Kiely, G.~De~Chiara, M.~Paternostro, and S.~Campbell,
  ``Energetic cost of quantum control protocols,'' {\em New Journal of
  Physics}, vol.~21, no.~10, p.~103048, 2019.

\bibitem{ikonen2017energy}
J.~Ikonen, J.~Salmilehto, and M.~M{\"o}tt{\"o}nen, ``Energy-efficient quantum
  computing,'' {\em npj Quantum Information}, vol.~3, no.~1, pp.~1--7, 2017.

\bibitem{campbell2017trade}
S.~Campbell and S.~Deffner, ``Trade-off between speed and cost in shortcuts to
  adiabaticity,'' {\em Physical review letters}, vol.~118, no.~10, p.~100601,
  2017.

\bibitem{deffner2021energetic}
S.~Deffner, ``Energetic cost of hamiltonian quantum gates,'' {\em arXiv
  preprint arXiv:2102.05118}, 2021.

\bibitem{robert2021resource}
A.~Robert, P.~K. Barkoutsos, S.~Woerner, and I.~Tavernelli,
  ``Resource-efficient quantum algorithm for protein folding,'' {\em npj
  Quantum Information}, vol.~7, no.~1, pp.~1--5, 2021.

\bibitem{krinner2019engineering}
S.~Krinner, S.~Storz, P.~Kurpiers, P.~Magnard, J.~Heinsoo, R.~Keller,
  J.~Luetolf, C.~Eichler, and A.~Wallraff, ``Engineering cryogenic setups for
  100-qubit scale superconducting circuit systems,'' {\em EPJ Quantum
  Technology}, vol.~6, no.~1, p.~2, 2019.

\bibitem{li2018crossbar}
R.~Li, L.~Petit, D.~P. Franke, J.~P. Dehollain, J.~Helsen, M.~Steudtner, N.~K.
  Thomas, Z.~R. Yoscovits, K.~J. Singh, S.~Wehner, {\em et~al.}, ``A crossbar
  network for silicon quantum dot qubits,'' {\em Science advances}, vol.~4,
  no.~7, p.~eaar3960, 2018.

\bibitem{gilchrist2005distance}
A.~Gilchrist, N.~K. Langford, and M.~A. Nielsen, ``Distance measures to compare
  real and ideal quantum processes,'' {\em Physical Review A}, vol.~71, no.~6,
  p.~062310, 2005.

\bibitem{carignan2019bounding}
A.~Carignan-Dugas, J.~J. Wallman, and J.~Emerson, ``Bounding the average gate
  fidelity of composite channels using the unitarity,'' {\em New Journal of
  Physics}, vol.~21, no.~5, p.~053016, 2019.

\bibitem{bharti2021noisy}
K.~Bharti, A.~Cervera-Lierta, T.~H. Kyaw, T.~Haug, S.~Alperin-Lea, A.~Anand,
  M.~Degroote, H.~Heimonen, J.~S. Kottmann, T.~Menke, {\em et~al.}, ``Noisy
  intermediate-scale quantum (nisq) algorithms,'' {\em arXiv preprint
  arXiv:2101.08448}, 2021.

\bibitem{peruzzo2014variational}
A.~Peruzzo, J.~McClean, P.~Shadbolt, M.-H. Yung, X.-Q. Zhou, P.~J. Love,
  A.~Aspuru-Guzik, and J.~L. O’brien, ``A variational eigenvalue solver on a
  photonic quantum processor,'' {\em Nature communications}, vol.~5, no.~1,
  pp.~1--7, 2014.

\bibitem{gentini2019noise}
L.~Gentini, A.~Cuccoli, S.~Pirandola, P.~Verrucchi, and L.~Banchi,
  ``Noise-assisted variational hybrid quantum-classical optimization,'' {\em
  arXiv preprint arXiv:1912.06744}, 2019.

\bibitem{villalonga2019flexible}
B.~Villalonga, S.~Boixo, B.~Nelson, C.~Henze, E.~Rieffel, R.~Biswas, and
  S.~Mandr{\`a}, ``A flexible high-performance simulator for verifying and
  benchmarking quantum circuits implemented on real hardware,'' {\em npj
  Quantum Information}, vol.~5, no.~1, pp.~1--16, 2019.

\bibitem{aliceandbob}
A.~. Bob. \url{https://alice-bob.com/fr/}.

\bibitem{starmon5}
QuTech, ``{Starmon-5 quantum processor}.''
  \url{https://www.quantum-inspire.com/backends/starmon-5/}, 2020.

\bibitem{spin2}
QuTech, ``{Spin-2 quantum processor}.''
  \url{https://www.quantum-inspire.com/backends/spin-2/}.

\bibitem{malekakhlagh2020first}
M.~Malekakhlagh, E.~Magesan, and D.~C. McKay, ``First-principles analysis of
  cross-resonance gate operation,'' {\em Physical Review A}, vol.~102, no.~4,
  p.~042605, 2020.

\bibitem{kirchhoff2018optimized}
S.~Kirchhoff, T.~Ke{\ss}ler, P.~J. Liebermann, E.~Ass{\'e}mat, S.~Machnes,
  F.~Motzoi, and F.~K. Wilhelm, ``Optimized cross-resonance gate for coupled
  transmon systems,'' {\em Physical Review A}, vol.~97, no.~4, p.~042348, 2018.

\bibitem{johnson2012optimization}
J.~E.~J. Johnson, {\em Optimization of superconducting flux qubit readout using
  near-quantum-limited amplifiers}.
\newblock University of California, Berkeley, 2012.

\bibitem{reed2014entanglement}
M.~D. Reed, {\em Entanglement and quantum error correction with superconducting
  qubits}.
\newblock Yale University, 2014.

\bibitem{weisend2016cryostat}
J.~Weisend~II, {\em Cryostat Design}.
\newblock Springer, 2016.

\bibitem{enss2011tieftemperaturphysik}
C.~Enss and S.~Hunklinger, {\em Tieftemperaturphysik}.
\newblock Springer-Verlag, 2011.

\bibitem{le202019}
L.~Le~Guevel, G.~Billiot, X.~Jehl, S.~De~Franceschi, M.~Zurita, Y.~Thonnart,
  M.~Vinet, M.~Sanquer, R.~Maurand, A.~G. Jansen, {\em et~al.}, ``19.2 a 110mk
  295$\mu$w 28nm fdsoi cmos quantum integrated circuit with a 2.8 ghz
  excitation and na current sensing of an on-chip double quantum dot,'' in {\em
  2020 IEEE International Solid-State Circuits Conference-(ISSCC)},
  pp.~306--308, IEEE, 2020.

\bibitem{aumentado2020superconducting}
J.~Aumentado, ``Superconducting parametric amplifiers: The state of the art in
  josephson parametric amplifiers,'' {\em IEEE Microwave Magazine}, vol.~21,
  no.~8, pp.~45--59, 2020.

\bibitem{korolev2011note}
A.~Korolev, V.~Shnyrkov, and V.~Shulga, ``Note: Ultra-high frequency ultra-low
  dc power consumption hemt amplifier for quantum measurements in millikelvin
  temperature range,'' {\em Review of Scientific Instruments}, vol.~82, no.~1,
  p.~016101, 2011.

\bibitem{parma2015cryostat}
V.~Parma, ``Cryostat design,'' {\em arXiv preprint arXiv:1501.07154}, 2015.

\bibitem{martin2021energy}
M.~J. Martin, C.~Hughes, G.~Moreno, E.~B. Jones, D.~Sickinger, S.~Narumanchi,
  and R.~Grout, ``Energy use in quantum data centers: Scaling the impact of
  computer architecture, qubit performance, size, and thermal parameters,''
  {\em arXiv preprint arXiv:2103.16726}, 2021.

\bibitem{green2019helium}
M.~A. Green, ``Helium refrigeration during the 50 years since the 1968
  brookhaven summer study,'' {\em IEEE Transactions on Applied
  Superconductivity}, vol.~29, no.~5, pp.~1--5, 2019.

\bibitem{irds}
``Cryogenic electronics and quantum information processing,'' {\em The
  international roadmap for devices and systems}, 2020.

\bibitem{wade2018bandwidth}
M.~Wade, M.~Davenport, M.~D.~C. Falco, P.~Bhargava, J.~Fini, D.~Van~Orden,
  R.~Meade, E.~Yeung, R.~Ram, M.~Popovi{\'c}, {\em et~al.}, ``A
  bandwidth-dense, low power electronic-photonic platform and architecture for
  multi-tbps optical i/o,'' in {\em 2018 European Conference on Optical
  Communication (ECOC)}, pp.~1--3, IEEE, 2018.

\bibitem{shackelford2000crc}
J.~F. Shackelford and W.~Alexander, {\em CRC materials science and engineering
  handbook}.
\newblock CRC press, 2000.

\bibitem{marquardt2002cryogenic}
E.~Marquardt, J.~Le, and R.~Radebaugh, ``Cryogenic material properties
  database,'' in {\em Cryocoolers 11}, pp.~681--687, Springer, 2002.

\bibitem{mukhanov2019scalable}
O.~Mukhanov, A.~Kirichenko, C.~Howington, J.~Walter, M.~Hutchings, I.~Vernik,
  D.~Yohannes, K.~Dodge, A.~Ballard, B.~Plourde, {\em et~al.}, ``Scalable
  quantum computing infrastructure based on superconducting electronics,'' in
  {\em 2019 IEEE International Electron Devices Meeting (IEDM)}, pp.~31--2,
  IEEE, 2019.

\bibitem{benjaminPrivate}
B.~Huard, ``Private communication,'' 2020.

\bibitem{dassonneville2020fast}
R.~Dassonneville, T.~Ramos, V.~Milchakov, L.~Planat, {\'E}.~Dumur, F.~Foroughi,
  J.~Puertas, S.~Leger, K.~Bharadwaj, J.~Delaforce, {\em et~al.}, ``Fast
  high-fidelity quantum nondemolition qubit readout via a nonperturbative
  cross-kerr coupling,'' {\em Physical Review X}, vol.~10, no.~1, p.~011045,
  2020.

\bibitem{heinsoo2018rapid}
J.~Heinsoo, C.~K. Andersen, A.~Remm, S.~Krinner, T.~Walter, Y.~Salath{\'e},
  S.~Gasparinetti, J.-C. Besse, A.~Poto{\v{c}}nik, A.~Wallraff, {\em et~al.},
  ``Rapid high-fidelity multiplexed readout of superconducting qubits,'' {\em
  Physical Review Applied}, vol.~10, no.~3, p.~034040, 2018.

\bibitem{esposito2021perspective}
M.~Esposito, A.~Ranadive, L.~Planat, and N.~Roch, ``Perspective on traveling
  wave microwave parametric amplifiers,'' {\em arXiv preprint
  arXiv:2107.13033}, 2021.

\bibitem{ranadive2021reversed}
A.~Ranadive, M.~Esposito, L.~Planat, E.~Bonet, C.~Naud, O.~Buisson,
  W.~Guichard, and N.~Roch, ``A reversed kerr traveling wave parametric
  amplifier,'' {\em arXiv preprint arXiv:2101.05815}, 2021.

\bibitem{LucaPrivate}
L.~Planat, ``Private communication,'' 2021.

\bibitem{rabaey2003digital}
J.~M. Rabaey, A.~P. Chandrakasan, and B.~Nikoli{\'c}, {\em Digital integrated
  circuits: a design perspective}, vol.~7.
\newblock Pearson education Upper Saddle River, NJ, 2003.

\bibitem{bartolini2014unveiling}
A.~Bartolini, M.~Cacciari, C.~Cavazzoni, G.~Tecchiolli, and L.~Benini,
  ``Unveiling eurora—thermal and power characterization of the most
  energy-efficient supercomputer in the world,'' in {\em 2014 Design,
  Automation \& Test in Europe Conference \& Exhibition (DATE)}, pp.~1--6,
  IEEE, 2014.

\bibitem{cernpowerconsumption}
CERN, ``{Energy consumption}.''
  \url{https://www.lhc-closer.es/taking_a_closer_look_at_lhc/0.energy_consumption}.

\bibitem{cryoconventional}
KEYCOM, ``Semi-rigid cryogenic sus304 coaxial cable.''
  \url{https://www.keycom.co.jp/eproducts/upj/upj2/page.htm}.

\bibitem{tuckerman2016flexible}
D.~B. Tuckerman, M.~C. Hamilton, D.~J. Reilly, R.~Bai, G.~A. Hernandez, J.~M.
  Hornibrook, J.~A. Sellers, and C.~D. Ellis, ``Flexible superconducting nb
  transmission lines on thin film polyimide for quantum computing
  applications,'' {\em Superconductor Science and Technology}, vol.~29, no.~8,
  p.~084007, 2016.

\end{thebibliography}
\end{document}